\documentclass[11pt]{article}

\usepackage{wright}
\usepackage{comment}
\usepackage{tablefootnote}
\usepackage{thm-restate}
\usepackage{IEEEtrantools}
\usepackage{mdframed}
\usepackage{xcolor}

\newcommand{\Stab}{\mathbf{Stab}} \newcommand{\stab}{\Stab}
\newcommand{\T}{\mathbf{T}}
\newcommand{\U}{\mathrm{U}}

\newcommand{\maxcsp}{\text{\sc Max-CSP}\xspace}
\newcommand{\maxcut}{\text{\sc Max-Cut}\xspace}
\newcommand{\qmaxcut}{\text{\sc Quantum} \text{\sc Max-Cut}\xspace}

\newcommand{\heis}{\text{\sc QMax-Cut}\xspace}
\newcommand{\prodval}{\text{\sc Prod}\xspace}
\newcommand{\hsdp}{\text{\sc SDP}_{\text{\sc QMC}}\xspace}
\newcommand{\mcsdp}{\text{\sc SDP}_{\text{\sc MC}}\xspace}
\newcommand{\prodsdp}{\text{\sc SDP}_{\text{\sc Prod}}\xspace}

\newcommand{\val}{\text{\sc Val}\xspace}

\newcommand{\Normed}[1]{\frac{#1}{\NORM{#1}}}
\newcommand{\inr}[2]{\langle #1, #2 \rangle}
\newcommand{\spatial}[3]{\mathcal{S}_{#1,#2} #3}
\newcommand{\valued}[3]{\mathcal{V}_{#1,#2} #3}

\newcommand{\ggraph}{\mathcal{G}}
\newcommand{\hgraph}{\mathcal{H}}

\newcommand{\agw}{\alpha_{\mathrm{GW}}}
\newcommand{\agp}{\alpha_{\mathrm{GP}}}
\newcommand{\abov}{\alpha_{\mathrm{BOV}}}
\newcommand{\apk}{\alpha_{\mathrm{GP}}}
\newcommand{\rgw}{\rho_{\mathrm{GW}}}
\newcommand{\rgp}{\rho_{\mathrm{GP}}}
\newcommand{\rbov}{\rho_{\mathrm{BOV}}}
\newcommand{\rpk}{\rho_{\mathrm{GP}}}

\newcommand{\unifS}{\omega}
\newcommand{\ultraS}{\tilde \omega}

\newcommand{\jnote}[1]{}
\newcommand{\ynote}[1]{}
\newcommand{\knote}[1]{}
\newcommand{\onote}[1]{}
\newcommand{\jnnote}[1]{}

\newcommand{\pE}{\widetilde{\E}}

\let\div\relax
\DeclareMathOperator{\div}{div}

\begin{document}

\title{Unique Games hardness of Quantum Max-Cut,\\
and a conjectured vector-valued Borell's inequality}

\date{}

\author{
\and Yeongwoo Hwang\thanks{yeongwoo@cs.utexas.edu}\\
 \small{\sl The University of Texas at Austin}
\and Joe Neeman\thanks{jneeman@math.utexas.edu}\\
 \small{\sl The University of Texas at Austin}
\and Ojas Parekh\thanks{odparek@sandia.gov}\\
 \small{\sl Sandia National Laboratories}
\and Kevin Thompson\thanks{kevthom@sandia.gov}\\
 \small{\sl Sandia National Laboratories}
\and John Wright\thanks{jswright@berkeley.edu}\\
 \small{\sl The University of California, Berkeley}
}

\maketitle

\begin{abstract}
The \emph{Gaussian noise stability} of a function $f:\R^n \rightarrow \{-1, 1\}$
is the expected value of $f(\bx) \cdot f(\by)$ over $\rho$-correlated Gaussian random variables~$\bx$ and~$\by$.
Borell's inequality states that for $-1 \leq \rho \leq 0$, this is minimized by the mean-zero halfspace $f(x) = \mathrm{sign}(x_1)$.
In this work, we conjecture that a natural generalization of this result holds for functions $f:\R^n \rightarrow S^{k-1}$
which output $k$-dimensional unit vectors.
Our main conjecture, which we call the \emph{vector-valued Borell's inequality}, asserts that  the expectation $\E_{\bx \sim_\rho \by}\langle f(\bx), f(\by)\rangle$
is minimized by the function $f(x) = x_{\leq k} / \Vert x_{\leq k} \Vert$, where $x_{\leq k} = (x_1, \ldots, x_k)$.
We give several pieces of evidence in favor of this conjecture,
including a proof that it does indeed hold in the special case of $n = k$.

As an application of this conjecture,
we show that it implies several hardness of approximation results for a special case of the local Hamiltonian problem related to the anti-ferromagnetic Heisenberg model
known as Quantum Max-Cut.
This can be viewed as a natural quantum analogue of the classical Max-Cut problem
and has been proposed as a useful testbed for developing algorithms.
We show the following, assuming the vector-valued Borell's inequality:
\begin{enumerate}
\item There exists an integrality gap of $0.498$ for the basic SDP, matching the rounding algorithm of Gharibian and Parekh~\cite{GP19}.
	Combined with the work of Anshu, Gosset, and Morenz~\cite{AGM20}, this shows that the basic SDP does not achieve the optimal approximation ratio.
\item It is Unique Games-hard (UG-hard) to compute a $(0.956+\eps)$-approximation to the value of the best product state, matching an approximation algorithm due to Briët, Oliveira, and Vallentin~\cite{BOV10}.  
\item It is UG-hard to compute a $(0.956+\eps)$-approximation to the value of the best (possibly entangled) state. 
\end{enumerate}
Our results also apply to the problem of Rank-$k$ \maxcut
considered by Briët, Oliveira, and Vallentin~\cite{BOV10}
and show that it is UG-hard to outperform their approximation algorithm
for any fixed~$k$, again assuming our conjecture.
\end{abstract}

\thispagestyle{empty}

\newpage
\thispagestyle{empty}

\begingroup
\hypersetup{linktocpage}
\tableofcontents
\endgroup

\thispagestyle{empty}

\newpage

\part{Introduction}

\setcounter{page}{1}

\section{Introduction}
Over the last~$30$ years, starting with the proof of the PCP theorem~\cite{AS98,ALM+98},
researchers have gained a nearly complete understanding of the approximability of classical constraint satisfaction problems (CSPs),
modulo the still-unproven Unique Games Conjecture (UGC) of Khot~\cite{Kho02}.
Due to work of Raghavendra~\cite{Rag08,Rag09},
we know that for each CSP,
there is a ``canonical algorithm'' which achieves a certain approximation ratio~$\alpha > 0$,
and that it is~$\NP$-hard for any polynomial-time algorithm to do better than~$\alpha$, assuming the UGC.
This canonical algorithm is based on the ``basic'' \emph{semidefinite programming (SDP)} relaxation of the CSP,
which is itself derived from the second level\ignore{\footnote{A level-$l$ (or degree-$l$) SoS is a sum of polynomials, each of which are squares, of degree at most $l$.}} of the \emph{sum of squares (SoS)} hierarchy. 

However, for the \emph{quantum} analogue of CSPs, known as the \emph{local Hamiltonian problem}, our understanding remains incomplete.
The quantum analogue of the PCP theorem remains a conjecture,
and there is no general theory of optimal algorithms, be they quantum or classical.
In addition,
there is a natural quantum analogue of the SoS hierarchy,
but in spite of recent progress in analyzing it,
many basic questions remain open,
such as how best to round its solutions.

On top of this, approximation algorithms for the local Hamiltonian problem
face an additional challenge not present for classical CSPs, namely
what kind of quantum state should they output?
The local Hamiltonian problem is an optimization problem over $n$-qubit states;
however, actually outputting a general $n$-qubit state is infeasible on a classical computer,
as it requires exponentially many bits to describe.
Instead, algorithms typically output states from a subset of quantum states, called an \emph{ansatz},
which can be  efficiently represented on a classical computer.
By far the most popular is the ansatz of \emph{product states},
i.e.\ states of the form $\ket{\psi_1} \ot \cdots \ot \ket{\psi_n}$,
which possess no entanglement
but often give surprisingly good approximations to the optimal value~\cite{BH16}.
In general, though, we lack a theory of optimal ansatzes for polynomial-time classical algorithms.

To study these questions, we focus on a special case of the local Hamiltonian problem known as \qmaxcut,
which has been suggested as a useful testbed for designing approximation algorithms~\cite{GP19}.
\qmaxcut is a natural maximization variant of the \emph{anti-ferromagnetic Heisenberg XYZ model},
a classic family of $\QMA$-complete~\cite{CM16,PM17} 2-local Hamiltonians first investigated by Heisenberg~\cite{Hei28} which models magnetic systems where interacting particles have opposing spins.
The constraints in \qmaxcut involve pairs of qubits and, loosely speaking, enforce that the two qubits have opposing values in each of the Pauli~$X$, $Y$, and $Z$ bases. 
As a result, it can be viewed as a quantum analogue of the classical \maxcut problem.

For the special case of finding the best product state in the \qmaxcut problem,
an SDP algorithm of Briët, Oliveira, and Vallentin~\cite{BOV10} gives an approximation ratio of~$0.956$.
For the more general case of finding the best (possibly entangled) state,
Gharibian and Parekh~\cite{GP19} gave an algorithm with approximation ratio~$0.498$.
Their algorithm, which uses the product state ansatz, is based on rounding the corresponding ``basic'' SDP,
derived from the second level of the \emph{noncommutative sum of squares (ncSoS)} hierarchy,
a variant of the SoS hierarchy for optimization problems over matrices.
Following this work, Anshu, Gosset, and Morenz~\cite{AGM20} gave an algorithm with an improved approximation ratio of~$0.531$.
To achieve this,
they used an ansatz which is more expressive than product states---tensor products of one- and two-qubit states---
as it is known that product states cannot surpass approximation ratio $0.5$.
Parekh and Thompson~\cite{PT21}
then showed that a similar algorithm could be captured by the level-4 ncSoS hierarchy,
with a slightly improved  approximation ratio of~$0.533$.

The starting point of our work is a well-known statement from Gaussian geometry
called \emph{Borell's inequality}~\cite{Bor:85}, which has played a central role
in the inapproximability of the classical \maxcut problem
ever since it was introduced to theoretical computer science
in the work of Mossel, O'Donnell, and Oleszkiewicz~\cite{MOO10}.
The earlier work of Khot, Kindler, Mossel, and O'Donnell~\cite{KKMO07}
had shown that the Goemans-Williamson SDP algorithm~\cite{GW95}
is the optimal polynomial-time approximation algorithm for \maxcut under the UGC,
assuming an unproven statement they dubbed ``Majority is Stablest'' and left as a conjecture.
One of the main results of~\cite{MOO10} was to supply a proof of this Majority is Stablest conjecture,
which they did by translating it to Gaussian geometry and applying Borell's inequality,
thereby settling the UG-hardness of \maxcut.
In addition, Borell's inequality serves as an important technical ingredient
in the analysis of the standard SDP integrality gap instance for \maxcut~\cite{FS02,OW08}.
The first contribution of our work is to identify and conjecture a natural generalization of Borell's inequality
suitable for applications to \qmaxcut.

\begin{conjecture}[Vector-valued Borell's inequality; negative~$\rho$ case]\label{conj:vector-borell-intro}
Let $k \leq n$ be positive integers and $-1\leq \rho \leq 0$.
Let $f : \R^n \rightarrow S^{k-1}$.
In addition, let $f_{\mathrm{opt}}:\R^n\rightarrow S^{k-1}$ be defined by $f_{\mathrm{opt}}(x) = x_{\leq k}/\Vert x_{\leq k}\Vert$,
where $x_{\leq k} = (x_1, \ldots, x_k)$.
Then
\begin{equation*}
    \E_{\bu \sim_\rho \bv} \langle f(\bu), f(\bv)\rangle \geq \E_{\bu \sim_\rho \bv}\langle f_{\mathrm{opt}}(\bu), f_{\mathrm{opt}}(\bv)\rangle,
\end{equation*}
where $\bu \sim_\rho \bv$ means that $\bu$ and $\bv$ are $\rho$-correlated Gaussian vectors.
\end{conjecture}

We also conjecture that the optimizers $f_{\mathrm{opt}}$ are unique, up to some natural symmetries; see
\Cref{conj:vector-borell} for the full statement.

The original Borell's inequality is when $k = 1$; in this case $f_{\mathrm{opt}}$
is just defined as $f_{\mathrm{opt}}(x) = \sgn(x_1)$, and the quantity $\E_{\bu \sim_\rho \bv}[f(\bu) f(\bv)]$
is known as the \emph{Gaussian noise stability} of $f$. When the sphere $S^{k-1}$ is replaced by
the probablity simplex $\Delta^{k-1}$, the analogue of \cref{conj:vector-borell-intro} is
a well-known open problem---known as the ``Peace Sign Conjecture''~\cite{IM12}---that was
recently solved when $\rho$ is positive but sufficiently close to zero~\cite{HT20}.

Although we are unable to prove \cref{conj:vector-borell-intro},
we are able to give evidence in favor of it.
For our first piece of evidence, we show that \cref{conj:vector-borell-intro} is true when $n = k$,
which we prove via a spectral argument.
For our second piece of evidence, we instead consider the ``positive $\rho$ case'' of $0 \leq \rho \leq 1$,
in which we would like to identify the function
which \emph{maximizes} the expression $ \E_{\bu \sim_\rho \bv} \langle f(\bu), f(\bv)\rangle$
rather than minimizes it.
In this case, it is clear that our conjectured optimizer $f_{\mathrm{opt}}(x) = x_{\leq k}/\Vert x_{\leq k}\Vert$ is \emph{not} the maximizer, as any constant function will have value~1, but it is plausible that $f_{\mathrm{opt}}$ is still the maximizer among all functions which are mean-zero.
For our second piece of evidence, then,
we show a dimensionality reduction statement in the positive $\rho$ case, implying that
any maximizing mean-zero $f$ is ``intrinsically'' $k$-dimensional;
in other words, it suffices to only consider the $n = k$ case. 
To show this, we extend the recent calculus of variations approach for proving Borell's inequality due to Heilman and Tarter~\cite{HT20}
to output dimensions of size~$k$ larger than~$1$.
Unfortunately, since our reduction to the $n = k$ case only holds for positive $\rho$,
and our proof of the $n = k$ case only holds for negative $\rho$,
these two pieces do not combine to give a full proof of the conjecture.
For more details, along with formal statements of these results, see \cref{part:borell}.

Next, we show that the vector-valued Borell's inequality allows us to address several questions related to the optimality of the SDP algorithms for \qmaxcut.

\begin{theorem}[Main results, informal]\mbox{}
Suppose that the vector-valued Borell's inequality is true.
Then the following statements hold.
\begin{enumerate}
\item There exists an integrality gap of $0.498$ for the basic SDP, matching the rounding algorithm of Gharibian and Parekh~\cite{GP19}.
	This shows that the product state ansatz is optimal for the basic SDP.
	Combined with the work of Anshu, Gosset, and Morenz~\cite{AGM20},
	this also shows that the basic SDP does not achieve the optimal approximation ratio.
\item It is Unique-Games-hard (UG-hard) to compute a $(0.956+\eps)$-approximation to the value of the best product state, matching an approximation algorithm due to Briët, Oliveira, and Vallentin~\cite{BOV10}.  More generally, for any fixed~$k$, it is UG-hard to outperform the approximation algorithm of Briët, Oliveira, and Vallentin~\cite{BOV10} for Rank-$k$ \maxcut ($k=3$ corresponding to the aforementioned $0.956$-approximation).
\item It is UG-hard to compute a $(0.956+\eps)$-approximation to the value of the best (possibly entangled) state. 
\end{enumerate}\label{thm:informal-intro}
\end{theorem}

To our knowledge, this is the first constant-factor hardness of approximation result
for a natural family of local Hamiltonians
which does not already contain a hard-to-approximate classical CSP as a special case,
modulo our conjectured vector-valued Borell's inequality.
We also show that our conjecture implies sharp inapproximability results for \qmaxcut with respect to product states.
Finally, in a striking departure from classical CSPs, in which the level-2 SoS relaxation is optimal under the UGC,
we show that our conjecture implies that the level-4 ncSoS relaxation for \qmaxcut strictly improves upon the level-2 relaxation.
Our results highlight the importance of gaining a better understanding of the level-4 ncSoS relaxation,
as well as the importance of settling \cref{conj:vector-borell-intro}.
The relevance of Conjecture~\ref{conj:vector-borell-intro} to our main results is
discussed in more detail in Section~\ref{sec:tech-overview-overview} below.

\ignore{
Indeed, without additional assumptions on the form of the Hamiltonian, these results establish that \emph{we should not expect} to find a better product state than that achievable by~\cite{BOV10} for the kind of approximation studied.  The first question is also interesting from a physical perspective because it points to limitations on ``local'' theories for describing for describing quantum states.  A feasible solution to the SDP studied in~\cite{BGKT19, GP19} can be interpreted as a ``quasi'' quantum state which respects all $1$-local observables but not some $2$-local observables.  Our results imply performance bounds on these ``low-order theories.''

We also highlight the usefulness of semidefinite programming
as a concrete algorithmic framework for studying the local Hamiltonian problem.
For example,
although answering questions such as ``what is the optimal ansatz?''
might currently be out of reach,
we have shown that they can be answered for a fixed level of the ncSoS hierarchy.
In addition, they suggest that understanding the level-4 ncSoS relaxation is a natural future direction.
}

\paragraph{Related work.}
An earlier draft of this work incorrectly claimed a complete proof of the vector-valued Borell's inequality.
The bug was in the proof of the dimensionality reduction step, which we incorrectly claimed held for both positive \emph{and} negative $\rho$.
We thank Steve Heilman for pointing out the error in this proof.
Subsequent to the original posting of this work,
Parekh and Thompson gave an algorithm for \qmaxcut based on the level-4 ncSoS relaxation
which achieves a $0.5$-approximation and uses the product state ansatz~\cite{PT22}.
This is optimal, as there are simple graphs of value 1 in which product states achieve value at most 0.5~\cite{GP19}.
Combined with our \cref{thm:informal-intro} and assuming the vector-valued Borell's inequality,
this shows that level-4 ncSoS outperforms level-2, even when one is restricted to using the product state ansatz.

In the local Hamiltonian problem,
one is given a set of constraints on a system of~$n$ qubits,
and the objective is to find the ``ground state energy,'',
i.e.\ the optimum energy of a state under these constraints.
This is a central problem in both the fields of condensed matter physics and quantum computing,
and the study of this problem has led to a rich exchange of ideas between the two; see~\cite{GHLS15} for an excellent survey on the topic.
Its importance in condensed matter physics stems from the fact that many interesting real-world systems can be modeled as local Hamiltonians, with the ground state energy corresponding to the energy of the system at zero temperature.
In quantum computing, it is the canonical $\QMA$-complete problem~\cite{KSV02}, and is thus intractable to solve exactly with a classical computer unless $\mathsf{BPP} = \QMA$.

Classical CSPs form a special case of the local Hamiltonian problem,
and so the classical PCP theorem applies to local Hamiltonians as well.
This means that given a general instance of the local Hamiltonian problem, it is $\NP$-hard (though not necessarily $\NP$-complete) to estimate its ground-state energy to a certain constant accuracy.
The \emph{quantum PCP conjecture}~\cite{AAV13} asserts that this task is in fact $\QMA$-complete.
One of the key differences between these two possibilities
is that every local Hamiltonian has an efficient quantum witness of its ground state energy,
namely its $n$-qubit ground state,
but it may not have an efficient \emph{classical} (i.e.\ $\NP$) witness,
as the ground state need not be efficiently representable on a classical computer.

Researchers have designed classical approximation algorithms
for various classes of $2$-local Hamiltonian problems,
the majority of which use the product state ansatz.
These include algorithms
for Hamiltonians whose local terms are positive semi-definite~\cite{GK12,HEP20,PT20,PT22},
traceless Hamiltonians~\cite{HM17,BGKT19,NST09,PT20},
and fermionic Hamiltonians~\cite{BGKT19,HOD22}.
In condensed matter physics, the use of product states as an ansatz is widespread and
 known as \emph{mean-field theory}~\cite{GHLS15};
 in quantum chemistry, it is known as the \emph{Hartree-Fock method}~\cite{BH16}.
 The ubiquity of this ansatz
 stems both from the ease with which it can be analyzed
 and from folklore that product states well-approximate ground states
 in some situations.
This folklore was formalized in the work of Brandão and Harrow~\cite{BH16},
which showed that
product states give a good approximation to the ground states of local Hamiltonians
whose interaction graphs are high degree or sufficiently good expanders,
due to monogamy of entanglement.
As a result, these Hamiltonians cannot serve as hard instances
for the quantum PCP conjecture,
as product states have an efficient classical description.
They then used this structural result
to design approximation algorithms for Hamiltonians
with interaction graphs which are either planar, dense, or have low threshold rank.
This is an example of how the study of approximation algorithms
can shed light on the limitations of the quantum PCP conjecture.

To our knowledge, the only classical approximation algorithms
which do not use the product state ansatz are the aforementioned algorithms of~\cite{AGM20,PT21}, which use tensor products of one- and two-qubit states
(and which still rely heavily on the product state algorithms),
and the intriguing recent algorithms of~\cite{AGM20,AGKS21},
which use an ansatz consisting of states of the form $\ket{u} = U\ket{\psi}$,
where $\ket{\psi}$ is a product state and $U$ is a low-depth quantum circuit.
These two works give several settings in which product states can be ``improved''
by post-processing them with a low-depth quantum circuit;
for example, the former work~\cite{AGM20} shows an algorithm using this ansatz
which is guaranteed to outperform any product state algorithms on degree-3 and 4
instances of \qmaxcut.

\ignore{
To see why this task is in~$\QMA$,
note that every local Hamiltonian has an efficient \emph{quantum} witness for its ground state energy,
namely its lowest energy state, which requires only~$n$ qubits.
As describe above, though, such a state need not be efficiently representable on a classical computer,
and so this problem is not known to be in~$\NP$.
We note that the question of whether each local Hamiltonian
has an efficient classical witness seems related to the question of what the optimal ansatz is\knote{Is it worth mentioning QMA vs. QCMA here?},
though we are not aware of any formal connection between the two.
One distinction is that the ansatz should be able to be efficiently optimized over,
whereas this need not be the case with an efficient classical witness.

Researchers have designed classical approximation algorithms for various classes of maximum $2$-local Hamiltonian problems, which we summarize in \Cref{table:approx_algs}.  As $2$-local Hamiltonian generalizes classical $2$-CSP, most of these problem classes have natural classical counterparts (see Table 1 in~\cite{PT20}).  Most of these approximation algorithms use a product state ansatz, and the more recent works generate an approximate solution using an SDP relaxation of either the optimal energy or optimal energy over product states.}

\section{Technical overview}

In this section, we give a technical overview of our results.
We begin with definitions of the problems we consider.
Then, we state SDP relaxations for these problems
and rounding algorithms for these SDPs that have been considered in the literature.
Finally, we state our results formally and give an overview of our proofs.

\subsection{The \maxcut problem}

The classical analogue of the \qmaxcut problem
is the \maxcut problem.
The simplest of all nontrivial CSPs,
this is the problem of partitioning the vertices of a graph into two sets in order to maximize the number of edges crossing the partition.

\begin{notation}
We use \textbf{boldface} to denote random variables.
\end{notation}

\begin{definition}[Weighted graph]
A \emph{weighted graph} $G = (V, E, w)$ is an undirected graph with weights on the edges specified by $w: E \rightarrow \R^{\geq 0}$.
The weights specify a probability distribution on the edges, i.e.\ $\sum_{e \in E} w(e) = 1$.
We write $\boldsymbol{e} \sim E$ or $(\bu, \bv) \sim E$ for a random edge sampled from this distribution.
We generally use \emph{graph} as shorthand for weighted graph.
\end{definition}

\begin{definition}[\maxcut]\label{def:max-cut}
Given a graph $G = (V, E, w)$, a \emph{cut} is a function $f:V\rightarrow\{\pm 1\}$.
The \emph{value} of the cut is
\begin{equation*}
\E_{(\bu, \bv) \sim E}[ \tfrac{1}{2} - \tfrac{1}{2} f(\bu) f(\bv)].
\end{equation*}
\maxcut is the problem of finding the value of the largest cut, i.e.\ the quantity
\begin{equation*}
\maxcut(G) = \max_{f:V \rightarrow \{\pm 1\}} \E_{(\bu, \bv) \sim E}[ \tfrac{1}{2} - \tfrac{1}{2} f(\bu) f(\bv)].
\end{equation*}
\end{definition}

\maxcut appears on Karp's original list of $\NP$-complete problems~\cite{Kar72},
and it is $\NP$-hard to approximate to a factor better than $\tfrac{16}{17}$~\cite{TSSW00}.
Goemans and Williams gave an algorithm using the basic SDP with approximation ratio $0.878$~\cite{GW95},
and this SDP algorithm was later shown to be optimal by~\cite{KKMO07}, at least assuming the UGC.
Even without assuming the UGC, though,
it is known that the SoS hierarchy
is still the optimal algorithm among a large class of algorithms for \maxcut,
namely those given by polynomial-size SDPs~\cite{LRS15}.

\subsection{The \qmaxcut problem}

\ignore{

The particular subset of Heisenberg Hamiltonians we study will be defined using a weighted graph.  An edge will correspond to an interaction term, and the ``weight'' of an edge will correspond to the strength of that particular term.  

\begin{notation}
We use \textbf{boldface} to denote random variables.
\end{notation}

\begin{definition}[Weighted graph]
A \emph{weighted graph} $G = (V, E, w)$ is an undirected graph with weights on the edges specified by $w: E \rightarrow \R^{\geq 0}$.
The weights are assumed to sum to~$1$, i.e.\ $\sum_{e \in E} w(e) = 1$, and so they specify a probability distribution on the edges.
We will write $\boldsymbol{e} \sim E$ or $(\bu, \bv) \sim E$ for a random edge sampled from this distribution.
We will often use the word \emph{graph} as shorthand for weighted graph.
\end{definition}

The classical analogue of the Heisenberg model
is the \maxcut problem.

\begin{definition}[\maxcut]\label{def:max-cut}
Given a graph $G = (V, E, w)$, a \emph{cut} is a function $f:V\rightarrow\{\pm 1\}$.
The \emph{value} of the cut is
\begin{equation*}
\E_{(\bu, \bv) \sim E}[ \tfrac{1}{2} - \tfrac{1}{2} f(\bu) f(\bv)].
\end{equation*}
\maxcut is the problem of finding the value of the largest cut, i.e.\ the quantity
\begin{equation*}
\maxcut(G) = \max_{f:V \rightarrow \{\pm 1\}} \E_{(\bu, \bv) \sim E}[ \tfrac{1}{2} - \tfrac{1}{2} f(\bu) f(\bv)].
\end{equation*}
\end{definition}

}

\ignore{
Next, we define our central problem of interest, the anti-ferrogmagnetic Heisenberg XYZ model.  
Consistent with our interpretation of this model as a ``quantum max cut,'' we will follow the works of \cite{GP19,AGM20} and define it as a maximization problem with an additional identity term,
rather than as a minimization problem.  Although finding a ground state of the latter is equivalent to finding a maximum-energy state of the former, the two problems differ in terms of approximability, where a $O(\log(n))$-approximation~\cite{BGKT19} is the best known for the minimization version that is of more traditional physical interest.
}

The main focus of this work is the \qmaxcut problem,
a special case of the local Hamiltonian problem first introduced by Gharibian and Parekh in~\cite{GP19}.
Although the local Hamiltonian problem is typically stated as a minimization problem,
they instead defined \qmaxcut to be a maximization problem,
as this makes it more convenient to study from an approximation algorithms perspective and resembles \maxcut.
As stated earlier,
\qmaxcut is a natural maximization variant of the anti-ferromagnetic Heisenberg XYZ model;
we discuss this viewpoint, as well as the Heisenberg model in greater detail, in~\Cref{sec:qmaxcut-heisenberg}.

\begin{definition}[The \qmaxcut interaction]
The \emph{\qmaxcut interaction} is the $2$-qubit operator
$
h = \tfrac{1}{4}(I \ot I - X \ot X - Y \ot Y - Z \ot Z).
$
\end{definition}

Here~$X$, $Y$, and $Z$ refer to the standard \emph{Pauli matrices}.
Intuitively, the \qmaxcut interaction, when applied to a pair of qubits,
enforces that they are opposites in the~$X$, $Y$, and~$Z$ bases.

\begin{definition}[\qmaxcut]
Let $G = (V, E, w)$ be a graph known as the \emph{interaction graph}.
The corresponding instance of the \emph{\qmaxcut problem} is the matrix which acts on $(\C^2)^{\ot V}$
given by
\begin{equation*}
H_G
= \sum_{(u, v) \in E} w_{u, v} \cdot h_{u, v}
=\E_{(\bu, \bv) \sim E} h_{\bu, \bv}.
\end{equation*}
Here, $h_{\bu, \bv}\in \mathbb{C}^{2^{|V|} \times 2^{|V|}}$ is shorthand for $h_{\bu, \bv} \ot I_{V \setminus \{\bu, \bv\}}$, where $h_{\bu,\bv}$ is the \qmaxcut interaction applied to the qubits~$\bu$ and~$\bv$. 
\end{definition}


\begin{definition}[Energy]
Let $H_G$ be an instance of \qmaxcut.
Given a state $\ket{\psi} \in (\C^2)^{\otimes V}$,
its \emph{value} or \emph{energy}
is the quantity $\bra{\psi} H_G \ket{\psi}$.
The \emph{maximum energy} of $H_G$, also referred to as its \emph{value}, is
\begin{equation*}
\heis(G) = \lambda_{\mathrm{max}} (H_G) = \max_{\ket{\psi} \in (\C^2)^{\otimes V}} \bra{\psi} H_G \ket{\psi}.
\end{equation*}
\end{definition}

To see one way \qmaxcut and \maxcut are related, if we let $D_G$ be the diagonal matrix consisting of the diagonal entries of $H_G$ in the computational basis, then $\lambda_{\mathrm{max}} (D_G)$ is half the value of the maximum cut in $G$.

The second problem we consider is the special case of \qmaxcut
when the optimization is only over product states, which has been a common approach in approximation algorithms for the local Hamiltonian problem.

\begin{definition}[Product state value]
The \emph{product state value of $H_G$} is
\begin{equation*}
\prodval(G) = \max_{\forall v \in V, \ket{\psi_v} \in \C^2} \bra{\psi_G} H_G \ket{\psi_G}\text{, where } \ket{\psi_G} = \otimes_{v \in V} \ket{\psi_v}.
\end{equation*}
\end{definition}

There is an alternative expression for the product state value which we will find convenient to use.
It is related to the \emph{Bloch sphere} representation of qubits;
see \Cref{sec:product-states} for a proof.

\begin{definition}[Balls and spheres]
Given a dimension $d \geq 1$, the $d$-dimensional unit ball and sphere are given by $B^d = \{x \in \R^d \mid \Vert x \Vert \leq 1\}$,
and $S^{d-1} = \{x \in \R^d \mid \Vert x \Vert = 1\}$, respectively.
\end{definition}

\begin{proposition}[Rewriting the product state value]\label{prop:rewrite-product}
\begin{equation*}
\prodval(G) = \max_{f:V \rightarrow S^2} \E_{(\bu, \bv) \sim E}[ \tfrac{1}{4} - \tfrac{1}{4} \langle f(\bu), f(\bv)\rangle].
\end{equation*}
\end{proposition}

Note the similarity to the definition of \maxcut (\Cref{def:max-cut}):
aside from the extra factor of $\tfrac{1}{2}$,
the key distinction is that the function~$f$
has range $S^2$ rather than $S^0 = \{-1, 1\}$.
As a result, the product state value can be viewed as an additional 
quantum generalization of \maxcut.
We may further generalize \maxcut by allowing $f$ to have range $S^{k-1}$,
yielding the rank-constrained version of \maxcut studied by Briët, Oliveira, and Vallentin~\cite{BOV10}.
\begin{definition}[Rank-$k$ \maxcut]\label{def:rank-k-maxcut}
\begin{equation*}
\maxcut_k(G) = \max_{f:V \rightarrow S^{k-1}} \E_{(\bu, \bv) \sim E}[ \tfrac{1}{2} - \tfrac{1}{2} \langle f(\bu), f(\bv)\rangle].
\end{equation*}
\end{definition}
Note that $\maxcut_3(G) = 2 \cdot \prodval(G)$.
Indeed, what we refer to as the BOV algorithm for the product state value
was actually originally stated in~\cite{BOV10} as an algorithm for $\maxcut_3(G)$,
though it applies equally well to both cases.

\subsection{Semidefinite programming relaxations}
A standard approach for solving \maxcut
is through its SDP relaxation.

\begin{definition}[The \maxcut SDP]\label{def:mcsdp}
Let $G = (V, E, w)$ be an $n$-vertex graph.
The value of the \maxcut SDP is
\begin{equation*}
\mcsdp(G) = \max_{f:V\rightarrow S^{n-1}} \E_{(\bu, \bv) \sim E}[ \tfrac{1}{2} - \tfrac{1}{2} \langle f(\bu), f(\bv)\rangle].
\end{equation*}
\end{definition}

This is a relaxation because the optimal objective of the \maxcut SDP is at least as large as $\maxcut(G)$ for all graphs~$G$.
It can be approximated in polynomial time,
meaning one can compute an $f:V \rightarrow S^{n-1}$ of value $\mcsdp(G) - \eps$
in time $\poly(n) \cdot \log(1/\eps)$\footnote{The usual considerations show arbitrary additive approximation in polynomial time, e.g.~see~\cite{VB96}.}.
In addition, it is equivalent to the level-2 SoS relaxation for \maxcut.
(Note also that $\mcsdp(G) = \maxcut_n(G)$.)

There is a similar SDP relaxation for the product state value of \qmaxcut.

\begin{definition}[The product state SDP]
\begin{equation*}
\prodsdp(G) = \max_{f:V\rightarrow S^{n-1}} \E_{(\bu, \bv) \sim E}[ \tfrac{1}{4} - \tfrac{1}{4} \langle f(\bu), f(\bv)\rangle].
\end{equation*}
\end{definition}
\noindent
This is in fact the SDP relaxation given by level-2 of the SoS hierarchy applied to the product state value. 
Note that $\prodsdp(G) = \tfrac{1}{2} \mcsdp(G)$.

Now we state the SDP relaxation we will use for the maximum energy of \qmaxcut,
which is equivalent to the level-2 ncSoS relaxation of the maximum energy.
\begin{definition}[The \qmaxcut SDP]\label{def:heis-sdp}
\begin{equation*}
\hsdp(G)
= \max_{f:V\rightarrow S^{n-1}} \E_{(\bu, \bv) \sim E}[ \tfrac{1}{4} - \tfrac{3}{4} \langle f(\bu), f(\bv)\rangle].
\end{equation*}
\end{definition}
\noindent
Deriving this requires a bit more care than either of the \maxcut or product state SDPs.
Indeed, it is not even obvious at first glance that this is a legitimate relaxation!
We include a proof of this fact in \Cref{sec:sdp_proofs} and a discussion of the ncSoS hierarchy in \Cref{sec:ncsos}.
To our knowledge, we are the first to observe this particularly simple form of the \qmaxcut SDP.
Note that $\hsdp(G) = \tfrac{3}{2} \cdot \mcsdp(G) - \tfrac{1}{2}$.

It may be surprising that these three very different optimization problems
yield such similar SDPs.
Indeed, the optimizing function $f:V\rightarrow S^{n-1}$ is the same in each of them.
In spite of this, however,
the quality of the SDP relaxation is different in each case
because each case features a different objective function
to compare the SDP against.

\subsection{Rounding algorithms}\label{sec:rounding}

Most SDP algorithms work by computing the optimal SDP solution
and then converting it to a solution to the original optimization problem in a process known as \emph{rounding}.
We will look at the standard Goemans-Williamson algorithm,
used to round the \maxcut SDP,
and two generalizations of this algorithm, used to round the product state and \qmaxcut SDPs.

\paragraph{Halfspace rounding.}
The Goemans-Williamson algorithm~\cite{GW95}
uses a procedure called ``halfspace rounding''
to round
the $\maxcut$ SDP optimum $f_{\mathrm{SDP}}:V \rightarrow S^{n-1}$
into a random cut $\boldf:V \rightarrow \{-1, 1\}$.
Halfspace rounding works as follows:
\begin{enumerate}
\item Sample a random vector $\bz = (\bz_1, \ldots, \bz_n)$ from the $n$-dimensional Gaussian distribution.
\item For each $u \in V$, set $\boldf(u)$ equal to
	\begin{equation*}
	\boldf(u) := \sgn(\langle \bz, f_{\mathrm{SDP}}(u)\rangle) = \frac{\langle \bz, f_{\mathrm{SDP}}(u)\rangle}{|\langle \bz, f_{\mathrm{SDP}}(u)\rangle |}.
	\end{equation*}
\end{enumerate}
Goemans and Williamson showed that for each $u, v \in V$,
\begin{equation}\label{eq:gw-ineq}
\E_{\boldf}[\tfrac{1}{2}-\tfrac{1}{2}\boldf(u) \boldf(v)]
\geq \agw \cdot[\tfrac{1}{2}-\tfrac{1}{2}\langle f_{\mathrm{SDP}}(u), f_{\mathrm{SDP}}(v)\rangle],
\end{equation}
where $\agw = 0.878567$.
Taking an average over the edges in~$G$,
this shows that the expected value of~$\boldf$ is at least $\agw \cdot \mcsdp(G)$,
and hence at least $\agw \cdot \maxcut(G)$.
This shows that the Goemans-Williamson algorithm is an $\agw$-approximation algorithm.

\paragraph{Projection rounding.}
Briët, Oliveira, and Vallentin~\cite{BOV10} suggested a generalization of halfspace rounding,
which we refer to as ``projection rounding'',
in order to round solutions of $\prodsdp(G)$.
The goal is to convert an SDP solution $f_{\mathrm{SDP}}:V \rightarrow S^{n-1}$
into a function $\boldf:V \rightarrow S^{2}$,
which can then be converted into a product state via~\Cref{def:bloch_vectors}.
Projection rounding works as follows:
\begin{enumerate}
\item Sample a random $3 \times n$ matrix $\bZ$ consisting of $3n$ i.i.d. standard Gaussians.
\item For each $u \in V$, set $\boldf(u) = \bZ f_{\mathrm{SDP}}(u) / \Vert \bZ f_{\mathrm{SDP}}(u)\Vert_2$.
\end{enumerate}
Briët, Oliveira, and Vallentin showed that for each $u, v \in V$,
\begin{equation}\label{eq:bov-ineq}
\E_{\boldf}[\tfrac{1}{4}-\tfrac{1}{4}\langle\boldf(u), \boldf(v)\rangle]
\geq \abov \cdot[\tfrac{1}{4}-\tfrac{1}{4}\langle f_{\mathrm{SDP}}(u), f_{\mathrm{SDP}}(v)\rangle],
\end{equation}
where $\abov = 0.956$.
With the same reasoning as above,
they conclude the following.

\begin{theorem}[Performance of the BOV algorithm~\cite{BOV10}]\label{thm:bov-thm}
The Briët-Oliveira-Vallentin algorithm for the product state value achieves approximation ratio $\abov$.
\end{theorem}

Next, Gharibian and Parekh~\cite{GP19}
used projection rounding to round solutions of $\hsdp(G)$ into product states.
Like~\cite{BOV10}, this involves rounding $f_{\mathrm{SDP}}:V \rightarrow S^{n-1}$,
now the solution of $\hsdp(G)$,
into $\boldf:V \rightarrow S^{2}$.
They establish the inequality
\begin{equation}\label{eq:gp-ineq}
\E_{\boldf}[\tfrac{1}{4}-\tfrac{1}{4}\langle \boldf(u), \boldf(v)\rangle]
\geq \agp \cdot[\tfrac{1}{4}-\tfrac{3}{4}\langle f_{\mathrm{SDP}}(u), f_{\mathrm{SDP}}(v)\rangle],
\end{equation}
where $\alpha_{\mathrm{GP}} = 0.498$.
Note that this inequality possesses an asymmetry
not present in \cref{eq:gw-ineq,eq:bov-ineq}:
the coefficient of $\boldf(u) \boldf(v)$ on the left-hand side is $\tfrac{1}{4}$
but the coefficient of $\langle f_{\mathrm{SDP}}(u), f_{\mathrm{SDP}}(v)\rangle$ on the right-hand side is $\tfrac{3}{4}$.
This asymmetry comes about because they are solving the SDP relaxation for the maximum energy,
which is an optimization over all states,
but only rounding into the set of product states.
Nevertheless, this yields the following theorem.

\begin{theorem}[Performance of the GP algorithm~\cite{GP19}]
The Gharibian-Parekh algorithm for \qmaxcut
achieves approximation ratio $\agp$.
\end{theorem}

\subsection{SDP and algorithmic gaps}

The quality of an SDP relaxation is traditionally measured through its \emph{integrality gap}.

\begin{definition}[Integrality gap]\label{def:integrality-gap}
Let $\mathcal{P}$ denote a maximization problem
and let $\mathrm{SDP}(\cdot)$ be a semidefinite programming relaxation for $\calP$.
Given an instance~$\calI$ of $\calP$, its \emph{integrality gap} is the quantity
\begin{equation*}
\mathrm{GapSDP}(\calI) = \frac{\mathrm{OPT}(\calI)}{\mathrm{SDP}(\calI)}.
\end{equation*}
The \emph{integrality gap} of the SDP is defined to be the minimum integrality gap among all instances,
i.e.
\begin{equation*}
\inf_{\text{instances $\calI$}}\{\mathrm{GapSDP}(\calI)\}.
\end{equation*}
\end{definition}

The integrality gap of an SDP serves as a bound on the approximation ratio of any algorithm based on rounding its solutions.
This is because one typically analyzes a rounding algorithm by comparing the value of its solution
to the value of the SDP, as done for the rounding algorithms in \Cref{sec:rounding}.
We refer to this as the \emph{standard analysis} of rounding algorithms.
We therefore typically view a rounding algorithm as optimal if its worst-case performance matches the integrality gap.

Usually, though, one actually cares about how the rounded solution compares to the optimal value,
not the SDP value.
To show that one has given a tight analysis of an algorithm's approximation ratio,
one must actually exhibit a matching \emph{algorithmic gap}.

\begin{definition}[Algorithmic gap]
Let $\mathcal{P}$ denote a maximization problem.
Let $A$ be an approximation algorithm for $\calP$,
and let $A(\calI)$ be the expected value of the solution it outputs on input~$\calI$.
Given an instance~$\calI$, its \emph{algorithmic gap} is the quantity
\begin{equation*}
\mathrm{Gap}_{A}(\calI) = \frac{A(\calI)}{\mathrm{OPT}(\calI)}.
\end{equation*}
The \emph{algorithmic gap} of~$A$ is defined to be the minimum algorithmic gap among all instances,
i.e.
\begin{equation*}
\inf_{\text{instances $\calI$}}\{\mathrm{Gap}_{A}(\calI)\}.
\end{equation*}
\end{definition}

\subsection{Our results}

Our first set of results are a pair of integrality gaps for the \qmaxcut
and product state SDPs.

\begin{restatable}[Integrality gap for the \qmaxcut SDP]{theorem}{sdpintgap}\label{thm:integrality-gap-heisenberg}
Assuming \cref{conj:vector-borell-intro},
the \qmaxcut semidefinite program $\hsdp(G)$ has integrality gap $\alpha_{\mathrm{GP}}$.
\end{restatable}

Assuming the vector-valued Borell's inequality,
this matches the approximation ratio of the GP algorithm
and shows that projection rounding is optimal for $\hsdp(G)$.
In addition, it shows that product states are the optimal ansatz
for this SDP, as the GP algorithm outputs product states.
Finally, it implies that the GP algorithm is strictly worse than the algorithms of~\cite{AGM20,PT21},
and that level-$4$ of the ncSoS hierarchy
strictly improves upon level-$2$ of the ncSoS hierarchy.

\begin{restatable}[Integrality gap for product state SDP]{theorem}{prodintgap}\label{thm:integrality-gap-prod}
Assuming \cref{conj:vector-borell-intro},
the product state semidefinite program $\prodsdp(G)$ has integrality gap $\abov$.
\end{restatable}

This matches the approximation ratio of the BOV algorithm
and shows that projection rounding is optimal for $\prodsdp(G)$, assuming~\cref{conj:vector-borell-intro}.
Next, we show an algorithmic gap for the product state SDP.
This shows that the ``standard analysis'' of the BOV algorithm is sharp,
and so its approximation ratio is $\abov$ exactly.
We note that this result is unconditional, and therefore not reliant on~\cref{conj:vector-borell-intro}.

\begin{restatable}[Algorithmic gap for product state SDP]{theorem}{algogap}\label{thm:algo-gap}
The Briët-Oliveira-Vallentin algorithm has algorithmic gap $\abov$.
\end{restatable}

Finally, we prove a Unique Games-hardness result for the product state value.
This uses the standard framework of~\cite{KKMO07,Rag08} for translating SDP integrality gaps into inapproximability results.
We apply this framework to the integrality gap from \Cref{thm:integrality-gap-prod}.

\begin{theorem}[Inapproximability of the product state value]\label{thm:main-inapprox}
Assuming \cref{conj:vector-borell-intro} and the Unique Games Conjecture,
it is $\NP$-hard to approximate $\prodval(G)$ to within a factor of $\abov+\eps$, for all $\eps > 0$.
\end{theorem}

This shows that the BOV algorithm is optimal, assuming \cref{conj:vector-borell-intro} and the UGC.
Next, we observe that the \qmaxcut instances which occur in this proof
have interaction graphs of high degree.
Hence, by \cite{BH16} their product state value is roughly identical to their maximum energy.
As a consequence, we also derive a Unique-Games hardness result for the maximum energy.

\begin{theorem}[Inapproximability of \qmaxcut]\label{thm:inapprox_general_state}
Assuming \cref{conj:vector-borell-intro} and the Unique Games Conjecture,
it is $\NP$-hard to approximate $\heis(G)$ to within a factor of $\abov+\eps$, for all $\eps > 0$.
\end{theorem}

This is our one result which is \emph{not} tight, to our knowledge,
as the best known approximation for $\qmaxcut$
achieves approximation $0.533$~\cite{PT21}, which is less than $\abov=0.956$.
The difficulty is that $\hsdp(G)$ is not an optimal SDP,
as it is outperformed by the algorithms of~\cite{AGM20,PT21},
and so we cannot convert an integrality gap for it into a UG-hardness result.

We also generalize our results for the product state value
to hold for $\maxcut_k$ for any fixed~$k$. See~\Cref{sec:rank-constrained} for more details.

\ignore{
We also considered the problem of whether the GP algorithm is optimal,
at least among all algorithms using the product state ansatz.
This potentially \emph{could} be done by converting~\Cref{thm:integrality-gap-heisenberg}
into an inapproximability result using the framework of~\cite{KKMO07,Rag08}.
However, we failed to do so after repeated attempts.

\begin{remark}[Results for $\maxcut_k$] We have $\prodval(G) = \frac{1}{2}\maxcut_3(G)$ and $\prodsdp(G) = \frac{1}{2}\mcsdp(G)$, so the above results may be viewed as results for $\maxcut_3(G)$.  These results generalize to $\maxcut_k(G)$ for $k \geq 3$ using $\mcsdp(G)$ as a relaxation for $\maxcut_k(G)$ (i.e., $\maxcut_k(G) \leq \mcsdp(G)$, for $k \geq 1$).  Although most of our key results are proven for the case of general $k$, for some of them we do focus on the $k=3$ case for the sake of exposition.\onote{John and Yeongwoo, please modify this statement to more accurately reflect what we do.}
\end{remark}
}

\subsection{Proof overview}\label{sec:tech-overview-overview}

Our proof of \Cref{thm:integrality-gap-heisenberg,thm:integrality-gap-prod}
is inspired by a well-known integrality gap construction for the \maxcut SDP due to Feige and Schechtman~\cite{FS02}
which achieves an integrality gap of~$\agw$.
We will begin with an overview of this construction
and a related construction called the ``Gaussian graph'',
and then we will discuss how to modify these constructions
to give integrality gaps for the two \qmaxcut SDPs.

\paragraph{Integrality gap for the \maxcut SDP.}
The construction of the Feige-Schechtman \maxcut integrality gap
is motivated by the following two desiderata.
\begin{enumerate}
\item \label{item:desiderata-one}
	Given an optimal solution  $f_{\mathrm{SDP}}:V \rightarrow S^{n-1}$ to the \maxcut SDP,
	halfspace rounding outputs a random solution~$\boldf$ 
	with expected value exactly equal to $\agw \cdot \mcsdp(G)$.
\item \label{item:desiderata-two}
	This random solution~$\boldf$ is always an optimal cut, i.e.\ it has value $\maxcut(G)$.
\end{enumerate}

As we saw in \Cref{sec:rounding},
the random solution~$\boldf$ output by halfspace rounding
 has expected value at least $\agw \cdot \mcsdp(G)$.
 Hence, if \Cref{item:desiderata-one} were not true,
 one of these solutions would have value strictly bigger than $\agw \cdot \mcsdp(G)$,
 contradicting~$G$ being an integrality gap.
 Likewise, if \Cref{item:desiderata-two} were not true,
 there would exist a cut with value at least~$\agw \cdot \mcsdp(G)$.
 
 Now we use \Cref{item:desiderata-one} to derive a constraint on $f_{\mathrm{SDP}}$.
 Recall from \Cref{eq:gw-ineq} that the value of $\boldf$
is at least $\agw$ times the value of $f_{\mathrm{OPT}}$ in the SDP \emph{for each edge $(u,v)$}.
 In other words,
 \begin{equation}\label{eq:gw-ineq-restated}
\E_{\boldf}[\tfrac{1}{2}-\tfrac{1}{2}\boldf(u) \boldf(v)]
\geq \agw \cdot[\tfrac{1}{2}-\tfrac{1}{2}\langle f_{\mathrm{SDP}}(u), f_{\mathrm{SDP}}(v)\rangle],
\end{equation}
This means that if \Cref{item:desiderata-one} is true,
\Cref{eq:gw-ineq-restated} must be satisfied with equality, for each edge $(u, v)$.
To see what this implies, we first recall the proof of \Cref{eq:gw-ineq-restated}.
Letting $\rho_{u,v}$ denote the inner product $\rho_{u,v} = \langle f_{\mathrm{SDP}}(u), f_{\mathrm{SDP}}(v)\rangle$,
there is an exact formula for the left-hand side, namely
\begin{equation*}
\E_{\boldf}[\tfrac{1}{2}-\tfrac{1}{2}\boldf(u) \boldf(v)]
= \frac{\arccos(\rho_{u,v})}{\pi}.
\end{equation*}
See \cite{GW95} for a proof of this fact.
Then \Cref{eq:gw-ineq-restated} follows as a consequence of the statement 
\begin{equation*}
\min_{-1\leq \rho \leq 1} \frac{\arccos(\rho)/\pi}{\tfrac{1}{2}-\tfrac{1}{2}\rho} = \agw.
\end{equation*}
There is in fact a \emph{unique} minimizer of this expression,
which we write as $\rgw \approx -0.69$.
As a result, \Cref{eq:gw-ineq-restated} is satisfied with equality if and only if $\rho_{u,v} = \rgw$.
Thus, \Cref{item:desiderata-one} implies that $\rho_{u,v} = \rgw$ for each edge $(u,v)$.

Motivated by this, Feige and Schectman consider the \emph{$n$-dimensional sphere graph} $\mathcal{S}^{n-1}_{\rgw}$.
This is an infinite graph with vertex set $S^{n-1}$
in which two vertices $u, v \in S^{n-1}$ are connected whenever $\langle u, v\rangle \approx \rgw$.
There is a natural SDP embedding of the sphere graph 
$f_{\mathrm{SDP}}:S^{n-1} \rightarrow S^{n-1}$,
in which $f_{\mathrm{SDP}}(u) = u$,
for each $u \in S^{n-1}$.
It has value
\begin{equation*}
\E_{(\bu, \bv) \sim E}[\tfrac{1}{2}-\tfrac{1}{2}\langle f_{\mathrm{SDP}}(\bu), f_{\mathrm{SDP}}(\bv)\rangle]
= \E_{(\bu, \bv) \sim E}[\tfrac{1}{2}-\tfrac{1}{2}\langle \bu, \bv \rangle]
\approx \tfrac{1}{2} - \tfrac{1}{2} \rgw.
\end{equation*}
Thus, $\mcsdp(\mathcal{S}^{n-1}) \gtrsim \tfrac{1}{2} - \tfrac{1}{2} \rgw$.
In addition, this graph satisfies our two desiderata:

\begin{enumerate}
\item By construction, $\langle f_{\mathrm{SDP}}(u), f_{\mathrm{SDP}}(v)\rangle \approx \rgw$
	for each edge $(u, v)$.
	Thus, hyperplane rounding will produce a random cut $\boldf$ with average value $\approx \agw \cdot (\tfrac{1}{2} 	- \tfrac{1}{2} \rgw)$.
\item Each cut $\boldf$ is of the form $\boldf(u) = \sgn(\langle \bz, f_{\mathrm{SDP}}(u) \rangle) = \sgn(\langle \bz, u \rangle)$, where~$\bz$ is a random Gaussian.
	By rotational symmetry, all of these cuts have the same value,
	and in particular they have the same value as the case when $\bz = e_1$,
	i.e.\ the cut $f_{\mathrm{opt}}(u) = \sgn(u_1)$.
	The main technical argument of~\cite{FS02} is that this is in fact the optimal cut.
	In other words, for every function $f : S^{n-1} \rightarrow \{-1, 1\}$,
	\begin{equation}\label{eq:whatevs}
	\E_{(\bu, \bv) \sim E}[f(\bu) f(\bv)] \geq \E_{(\bu, \bv) \sim E}[f_{\mathrm{opt}}(\bu) f_{\mathrm{opt}}(\bv)].
	\end{equation}
\end{enumerate}
\noindent
Together, these imply that $\mathcal{S}^{n-1}_{\rgw}$
has integrality gap~$\agw$.

\paragraph{Moving to the Gaussian graph.}
We will actually use a second, related construction of the integrality gap called the \emph{Gaussian graph},
which will turn out to be more convenient to analyze in our case.
It is defined as follows.

\begin{definition}[$\rho$-correlated Gaussian graph]
Let $n$ be a positive integer and $-1\leq \rho \leq 1$.
We define the \emph{$\rho$-correlated Gaussian graph}
to be the infinite graph $\ggraph^n_\rho$ with vertex set $\R^n$
in which a random edge $(\bu, \bv)$ is distributed as two $\rho$-correlated Gaussian vectors.
\end{definition}

For large~$n$, the Gaussian graph, when scaled by a factor of $\tfrac{1}{\sqrt{n}}$, behaves like the sphere graph.
For example, if $\bu$ is a random Gaussian vector,
then $\tfrac{1}{\sqrt{n}}\bu$ is close to a unit vector with high probability.
In addition, if $(\bu, \bv)$ is a random edge, then $\langle \tfrac{1}{\sqrt{n}}\bu, \tfrac{1}{\sqrt{n}}\bv \rangle \approx \rho$
with high probability.
As a result, an argument similar to above shows that it has an integrality gap of $\agw$.

The optimal SDP assignment is the function $f_{\mathrm{SDP}}(u) = u/\Vert u \Vert$.
Hyperplane rounding produces a cut of the form $\boldf(u) = \sgn(\langle \bz, f_{\mathrm{SDP}}(u)\rangle)$,
which is equivalent (up to rotation) to the assignment $f_{\mathrm{opt}}(u) = \sgn(u_1)$.
The analogous statement to \Cref{eq:whatevs} above
that $f_{\mathrm{opt}}$ is optimal is a well-known result in Gaussian geometry
known as \emph{Borell's inequality} or \emph{Borell's isoperimetric theorem}.
There is a version of this theorem for positive and negative~$\rho$,
but we will only need to state the negative~$\rho$ case 
as the integrality gap only requires the $\rgw \approx -0.69$ case. 

\begin{theorem}[Borell's isoperimetric theorem, negative $\rho$ case]\label{thm:borell}
Let $n$ be a positive integer and $-1\leq \rho \leq 0$.
Let $f : \R^n \rightarrow \{-1, 1\}$.
In addition, let $f_{\mathrm{opt}}:\R^n\rightarrow \{-1, 1\}$ be defined by $f_{\mathrm{opt}}(x) = \sgn(x_1)$.
Then
\begin{equation*}
\E_{\bu \sim_\rho \bv}[f(\bu) f(\bv)] \geq \E_{\bu \sim_\rho \bv}[f_{\mathrm{opt}}(\bu) f_{\mathrm{opt}}(\bv)].
\end{equation*}
\end{theorem}

\paragraph{Integrality gap for \qmaxcut.}
The integrality gaps we design for the product state and \qmaxcut SDPs
are constructed along very similar lines to the \maxcut integrality gaps.
In both cases, our goal is to show that projection rounding is the optimal rounding algorithm,
which motivates us to study the conditions in which projection rounding performs worst.
As we saw in \Cref{sec:rounding}, projection rounding applied to the product state and \qmaxcut SDPs
yields a random function $\boldf:V \rightarrow S^2$ satisfying the following conditions, for each edge $(u,v)$:
\begin{align}
\E_{\boldf}[\tfrac{1}{4}-\tfrac{1}{4}\langle\boldf(u), \boldf(v)\rangle]
&\geq \abov \cdot[\tfrac{1}{4}-\tfrac{1}{4}\langle f_{\mathrm{SDP}}(u), f_{\mathrm{SDP}}(v)\rangle],\label{eq:bov-ineq-restated}\\
\E_{\boldf}[\tfrac{1}{4}-\tfrac{3}{4}\langle\boldf(u), \boldf(v)\rangle]
&\geq \agp \cdot[\tfrac{1}{4}-\tfrac{3}{4}\langle f_{\mathrm{SDP}}(u), f_{\mathrm{SDP}}(v)\rangle],\label{eq:pk-ineq-restated}
\end{align}
respectively.
As in the case of \maxcut,
the left-hand sides of these equations depend only on the inner product $\rho_{u, v} = \langle f_{\mathrm{SDP}}(u), f_{\mathrm{SDP}}(v)\rangle$.
This was shown by~\cite{BOV10},
who gave the following exact expression for this quantity:
\begin{equation}\label{eq:might-reference-later}
\E_{\boldf} \langle \boldf(u), \boldf(v)\rangle 
= \frac{2}{3}\left(\frac{\Gamma(2)}{\Gamma(3/2)}\right)^2\rho_{u, v}\cdot\,_2F_1\left(1/2,1/2;5/2;\rho_{u, v}^2\right),
\end{equation}
where $_2F_1(\cdot, \cdot;\cdot;\cdot)$ is the Gaussian hypergeometric function.
Thus, one can compute the approximation ratios $\abov$ and $\agp$
by finding the ``worst case'' values of $\rho_{u,v}$.
For \Cref{eq:bov-ineq-restated}, this is $\rbov \approx -0.584$;
for \Cref{eq:pk-ineq-restated}, this is $\rgp \approx -0.97$.

This suggests finding a graph for the product state SDP and the \qmaxcut SDP
in which $\rho_{u, v} = \rbov$ and $\rho_{u, v} = \rgp$ for each edge $(u, v)$, respectively.
Again, we will use the $\rho$-correlated Gaussian graph,
where again the optimum SDP assignment is the function $f_{\mathrm{SDP}}(u) = u/\Vert u \Vert$.
Projection rounding produces a solution of the form $\boldf(u) = \bZ u / \Vert \bZ u \Vert$,
where $\bZ$ is a random $3 \times n$ Gaussian matrix.
When $n$ is large, this is roughly equivalent to projecting $u$ onto a random 3-dimensional subspace,
in which case it is equivalent (up to rotation) to $f_{\mathrm{opt}}(u) = (u_1, u_2, u_3) / \Vert (u_1, u_2, u_3)\Vert$.

We must now show that $f_{\mathrm{opt}}$ is indeed the optimal solution.
In the case of the product state value, we must show that it is the best among all product states;
equivalently, among all functions $f:\R^n \rightarrow S^2$.
On the other hand, for the case of the ground state energy, we must show that it is the best among all quantum states,
which need not be product.
Fortunately, the Gaussian graph is of high (in fact, infinite) degree,
and so by \cite{BH16} the optimal state is a product state.
\onote{We will need to qualify this since the graph is weighted.}
The optimality of $f_{\mathrm{opt}}$ then follows from our conjectured vector-valued
analogue of Borell's isoperimetric theorem (\cref{conj:vector-borell-intro},
in the case $k=3$).

\section{Conclusion and open questions}

In this work, we have made progress on understanding the approximability of \qmaxcut.
However, there are many interesting questions which remain open,
such as finding the optimal approximation ratio.
We list these below.

\begin{enumerate}
\item Most obviously, is the vector-valued Borell's inequality true?
\item Does there exist an algorithmic gap instance for the GP algorithm with algorithmic gap~$\agp$?
	We believe there is but were unable to find one.
	The key difficulty seems to be that an algorithmic gap instance should be a low-degree graph with an entangled maximum energy state.
	Otherwise, the optimizing state would be close to a product state,
	and in this case the GP algorithm matches the $\abov$ approximation ratio of the BOV algorithm.
	But finding explicit examples of families of \qmaxcut instances with entangled maximum energy states
	for which we can even compute their optimum value is a difficult problem,
	and only a few such examples are known (see~\cite{F17, M13}).
\item Can we perform an optimal analysis of the level-4 ncSoS relaxation of \qmaxcut?
	This would involve designing a rounding algorithm and finding an integrality gap which matches its performance,
	as well as identifying the optimal ansatz, which would need to be more powerful than product states.
	Inspired by the work of~\cite{AGM20}, Parekh and Thompson~\cite{PT21} have considered tensor products of one- and two-qubit states,
	but it is unclear whether this is the optimal ansatz.
\item Could the level-4 ncSoS relaxation actually be optimal for \qmaxcut? 
\item The BOV algorithm is essentially the optimal algorithm for \qmaxcut on high-degree graphs.
	What about low-degree graphs? Can we design improved approximation algorithms in this case as well? Recently, for example, Anshu, Gosset, Morenz Korol and Soleimanifar have demonstrated a way to improve the objective value of a given product state assuming the graph is low degree \cite{AGKS21}.
\item Are there \emph{quantum} approximation algorithms for \qmaxcut whose approximation ratios we can analyze?
\item Can we prove a hardness of approximation result for \qmaxcut which improves upon our~\Cref{thm:inapprox_general_state}?
	This would involve showing a reduction from the Unique Games problem
	which outputs low-degree instances of \qmaxcut.
	This is because high-degree instances have maximum energy states
	which can be approximated by product states,
	and so the BOV algorithm produces an $\abov$-approximation in this case.
	However, traditional Unique Games reductions typically produce high-degree instances,
	and so it seems like new techniques might be needed.
	This is related to our difficulty in producing algorithmic gap instances with entangled maximum energy states,
	and so designing algorithmic gap instances might be a good first step.
\item Is it $\QMA$-hard to approximate \qmaxcut?
	Proving this unconditionally would also prove the quantum PCP conjecture, putting it beyond the range of current techniques.
	But it might be possible to show this \emph{assuming} the quantum PCP conjecture is true.
\item Since the Heisenberg model has an analytic solution only for certain graphs (see for example the ``Bethe Ansatz'' ~\cite{Bet31}), physicists use a set of heuristic algorithms for approximating the ground state of Hamiltonians from the Heisenberg model~\cite{B30, A88, F79}.
	Can we find a rigorous theoretical justification to support the success of these heuristics in practice?
\item How well do the techniques used to design approximation algorithms for \qmaxcut
	carry over to other families of Hamiltonians?
	 One natural family is the set of Hamiltonians in which each local term is a projective matrix.  An $\alpha$-approximation for this case gives an
	 $\alpha$-approximation for any Hamiltonian with positive-semidefinite local terms. 
	 This was considered in \cite{PT20} where it was shown that a rounding algorithm akin to that used in \cite{GP19} also applies to this case.  This was followed up by an approach in \cite{PT22} that applies to more general ansatzes.
	 However, the analysis in these works is not tight so proving the actual performance of the algorithm is an open question.
\end{enumerate}

\section*{Acknowledgments}

We would like to thank Anurag Anshu, Srinivasan Arunachalam, and Penghui Yao
for their substantial contributions to this paper.
We would also like to thank Steve Heilman for pointing out a bug in a previous draft of the paper.

\ignore{
\begin{table}[ht]
\caption{Approximation Algorithms for maximum $2$-local Hamiltonian}\label{table:approx_algs}
\centering
\begin{tabular}{c | c}
\hline\hline
 Problem class & Approximation factor \\
\hline
 \, & \textcolor{red}{0.956}~\cite{BOV10}$^a$\\
 Heisenberg XYZ & \textcolor{red}{0.498}~\cite{GP19}\\
 \, & 0.531~\cite{AGM20}\\
\, & 0.533~\cite{PT21}\\
\hline
 \, & \textcolor{red}{$1/2-\epsilon$} (for dense instances)~\cite{GK12}\\
 Positive semi-definite terms & \textcolor{red}{0.328}~\cite{HEP20}\\
 \, & \textcolor{red}{0.387} (\textcolor{red}{0.467} for maximally entangled)~\cite{PT20}\\
 \hline 

 \, & \textcolor{red}{$\Omega(1/\ell)$}~\cite{HM17}\\
 Traceless & \textcolor{red}{$\Omega(1/\log(n))$}~\cite{BGKT19} \\
 \, & PTAS (for planar instances)~\cite{NST09}\\
\, & \textcolor{red}{0.187} (for bipartite instances)~\cite{PT20}\\
\hline 
 Dense or & \,\\
 Planar or & \textcolor{red}{PTAS}~\cite{BH16}\\
 Low threshold rank & \,\\
\hline 
Fermionic Hamiltonians & \textcolor{red}{$\Omega(1/(n\log(n)))$} \cite{BGKT19}\\
\hline
\end{tabular}

{\footnotesize The first approximation factor is relative to best product state and not comparable to others.  \textcolor{red}{Red} approximation factors use a product state ansatz.  The number of qubits or the number of Fermionic modes is $n$, and $\ell$ is the max number of $2$-local terms a qubit participates in.  PTAS refers to polynomial-time approximation schemes.  }
\end{table}
}
\newpage 
\part{A vector-valued Borell's inequality}\label{part:borell}
In this part, we discuss our main conjecture, \Cref{conj:vector-borell-intro}, in detail
and prove various partial results.

We are primarily concerned with vector-valued functions, $f : \R^n \to B^k$,
where $B^k$ is the unit ball in $\R^k$. Writing $(f_1,\ldots,f_k) = f$ for the coordinate
functions of $f$, note that if $f: \R^n \to B^k$ then each $f_i$ takes values in $[-1, 1]$
and is therefore square-integrable under
the standard Gaussian measure $\gamma$ (i.e. $f_i \in L^2(\gamma))$.  Generally if a function $f$ is square integrable on a space $A$ with measure $\eta$ we will denote this as $f\in L^2(A, \eta)$ or as $f\in L^2(\eta)$ if the space is clear from context.
We say that a vector-valued function $f: \R^n \to \R^k$ is square-integrable
(i.e.\ belongs to $L^2(\gamma)$) if each $f_i \in L^2(\gamma)$.

\paragraph{Gaussian variables.} We will make use of standard notations for multivariate Gaussian random variables: If $\mu\in \mathbb{R}^n$ and $\Sigma\in \mathbb{R}^{n\times n} \succeq 0$ then $\bx\sim \normal(\mu, \Sigma)$ if $\bx$ is a multivariate normal satisfying $\E[\bx_i]=\mu_i$ and $\E[(\bx_i-\mu_i)(\bx_j-\mu_j)]=\Sigma_{i, j}$.  

\begin{definition}[$\rho$-correlated Gaussians]
Let $\bx$ and $\by$ be Gaussian random variables taking values in $\mathbb{R}^n$ and $-1 \leq \rho \leq 1$ be a parameter.
The variables $\bx$ and $\by$ are \emph{$\rho$-correlated Gaussians}, denoted $\bx\sim_\rho \by$, if
\begin{equation*}
(\bx, \by) \sim \normal\!\left(\mathbf{0}, \begin{bmatrix} I_n & \rho I_n \\ \rho I_n & I_n\end{bmatrix} \right),
\end{equation*}
where $\mathbf{0}$ is a vector of zeros of the appropriate size and $I_n$ is the $n\times n$ identity matrix.
\end{definition}

\begin{definition}[Noise operator for vector-valued functions]
\label{def:noise-operator}
Let $f:\R^n \rightarrow \R^k$ be a square-integrable vector-valued function
and $-1 \leq \rho \leq 1$ be a parameter.
Then the \emph{Gaussian noise operator}, $\U_\rho$ is defined as
\begin{equation*}
\U_\rho f(x) = \E_{\bx \sim_\rho \by}[f(\by) \mid \bx = x].   
\end{equation*}
We observe that if $f = (f_1, \ldots, f_k)$ then $\U_\rho f = (\U_\rho f_1,\ldots,\U_\rho f_k)$.
\end{definition}

\begin{definition}[Noise stability for vector-valued functions]\label{def:noise-stab}
Let $f:\R^n \rightarrow \R^k$ be a square-integrable vector-valued function
and $-1 \leq \rho \leq 1$ be a parameter.
Then the \emph{noise stability of $f$ at $\rho$} is
\begin{equation*}
\stab_\rho[f] = \E_{\bx \sim_{\rho} \by}\langle f(\bx), f(\by)\rangle = \E_{\bx \sim \normal(0,1)^n} \langle f(\bx), \U_{\rho} f(\bx) \rangle.
\end{equation*}
\end{definition}


The following vector-valued version of Borell's inequality (for negative
$\rho$) is our main conjecture.

\begin{conjecture}\label{conj:vector-borell}
    Let $1 \le k \le n$ be positive integers and let $-1 \le \rho \le 0$. For all measurable
    functions $f: \R^n \to B^k$,
    \[
        \stab_\rho[f] \geq \stab_\rho[f_{\mathrm{opt}}]
    \]
where $f_{\mathrm{opt}}(x) = x_{\le k} / \|x_{\le k}\|$ and $x_{\le k} = (x_1, \dots, x_k)$.

Moreover, this optimizer is unique up to orthogonal transformations: for any $f: \R^n \to B^k$
for which $\stab_\rho[f] = \stab_\rho[f_\mathrm{opt}]$, there is an orthogonal $n \times n$ matrix $M$
such that $f(x) = f_{\mathrm{opt}}(Mx)$ almost surely.
\end{conjecture}

Note that this recovers Theorem~\ref{thm:borell} in the case that $k=1$, and also that the
assumption $f: \R^n \to B^k$ could be replaced by $f: \R^n \to S^{k-1}$ without having much effect:
since the optimal function $f: \R^n \to B^k$ actually takes values in $S^{k-1}$,
the optimal value is the same if we restrict the optimization to functions $f: \R^n \to S^{k-1}$.

Although our main application of \Cref{conj:vector-borell} involves negative $\rho$, we will
also consider a version with positive $\rho$. In this case, one needs to add the assumption
that $\E[f] = 0$; without this constraint, constant functions like $f(x) = (1, \dots, 0)$ are maximizers.

\begin{conjecture}\label{conj:vector-borell-positive-rho}
    Let $1 \le k \le n$ be positive integers and let $0 \le \rho \le 1$. For all measurable
    functions $f: \R^n \to B^k$ satisfying $\E[f(\bx)] = 0$,
    \[
        \stab_\rho[f] \leq \stab_\rho[f_{\mathrm{opt}}]
    \]
where $f_{\mathrm{opt}}(x) = x_{\le k} / \|x_{\le k}\|$ and $x_{\le k} = (x_1, \dots, x_k)$.

Moreover, this optimizer is unique up to orthogonal transformations: for any $f: \R^n \to B^k$
for which $\stab_\rho[f] = \stab_\rho[f_\mathrm{opt}]$, there is an orthogonal $n \times n$ matrix $M$
such that $f(x) = f_{\mathrm{opt}}(Mx)$ almost surely.
\end{conjecture}

\subsection*{Proposed strategy}

There are a few notable difficulties in establishing \Cref{conj:vector-borell}, compared
with the scalar-valued case. Recall in particular that Borell's theorem is also known when
the expectation of $f$ is constrained: among all functions $f: \R^n \to [-1, 1]$ with $\E[f] = v \in [-1, 1]$,
the noise stability is minimized (for negative $\rho$) by a linear threshold function
with the appropriate expectation, i.e.\ a function of the form $f_{\mathrm{opt}}(x) = \sgn(\inr{a}{x} + b)$.
The formulation in Theorem~\ref{thm:borell} is then recovered just by noting that as a function of the
constraint $v \in [-1, 1]$, the optimal noise stability is minimized when $v = 0$.
In the vector-valued situation $k \ge 2$, we do not solve a constrained version
of the problem. In fact, we do not even know of a good guess for the optimal stability among
functions with $\E[f] = v \ne 0$. The guess that comes from naively extrapolating the $k=1$ solution,
$f(x) = (x_{\le k} - a) / \|x_{\le k} - a\|$ appears not to be optimal.
This inconvenient fact rules out certain proof techniques that work for the scalar-valued case---specifically, recent approaches using $\rho$-convexity~\cite{MN:15} and stochastic calculus~\cite{Eld:15}---because if those methods had worked, they would have also shown optimality
of $f(x) = (x_{\le k} - a) / \|x_{\le k} - a\|$ in the constrained version.
On the other hand, we also do not know how to exploit the older symmetrization-based approaches~\cite{Bor:85}
because of the difficulty of symmetrizing vector-valued functions.

We will propose a three-step strategy for Conjectures~\ref{conj:vector-borell}
and~\ref{conj:vector-borell-positive-rho}. The reason that we have been unable to
complete the proof is that one of the three steps is only proven for positive $\rho$,
while another step is only proven for negative $\rho$.

The first step is to
consider noise stability for functions $f: S^{n-1} \to S^{n-1}$ and we prove
that (in both positive-$\rho$ and negative-$\rho$ cases)
the function $f(x) = x$ has optimal noise stability.  The argument here is
purely spectral: having shown that the eigenvectors of our noise operator
(modified appropriately to live on the sphere) are spherical harmonics, we
expand the function $f$ in the basis of spherical harmonics and show that the
optimal thing to do is to put all the ``weight'' on ``level-1'' coefficients.
One interesting feature of this argument is that it doesn't require $f$ to take
values in $S^{n-1}$: we show that $f(x) = x$ has optimal noise stability among
all functions $f: S^{n-1} \to \R^n$ satisfying $\E[\|f(\bx)\|^2] = 1$.
This step of the proof also works for $0 < \rho \le 1$: in this range, the noise
stability is \emph{maximized}, among functions with $\E[f] = 0$, by $f(x) = x$.

The second step of the proof is to consider functions $f: \R^n \to S^{n-1}$.
We do this by decomposing $\R^n$ into radial ``shells'' and applying the spherical
argument on each shell. This step requires $-1 \le \rho \le 0$.

The final step, which applies only to $0 \le \rho \le 1$
is a kind of
dimension reduction. Specifically, we show that if $f: \R^n \to S^{k-1}$ has optimal
stability then $f$ is ``essentially $k$-dimensional'' in the sense that up to a change of coordinates
there is a function $g: \R^k \to S^{k-1}$ such that $f(x) = g(x_1, \dots, x_k)$. This
essentially reduces the problem to the case of functions $f: \R^k \to S^{k-1}$,
which was already handled in the second step. This step of the proof
uses tools from the
calculus of variations. Essentially, we show that if $f: \R^n \to S^{k-1}$ is
not essentially $k$-dimensional then it can be modified in a way that improves the noise stability.
Arguments of this kind go back to McGonagle and Ross~\cite{MR:15} in the setting
of the Gaussian isoperimetric problem. They were developed in the vector-valued (but still
isoperimetric) setting by Milman and Neeman~\cite{MN:18}, and then applied
to noise stability by Heilman and Tarter~\cite{HT20}.

Note that by combining the first two steps (which both apply for $-1 \le \rho \le 0$) we
can show that \Cref{conj:vector-borell} holds in the case $k = n$.

\section{The spherical case}\label{sec:spherical}

Here we consider the case of $f: S^{n-1} \to B^n$.
Since we want to work with Gaussian noise
on $\R^n$, this shell decomposition imposes a specific noise operator on $S^{n-1}$.
In this section, we will work with a more general noise operator that includes
the ones we will need in later sections.

The ``uniform'' measure on the sphere will be denoted $\unifS$: this is the
unique rotationally invariant probability measure on $S^{n-1}$. If $\bu \sim
\unifS$, we will write $\ultraS$ for the distribution of $\bu_1$ (which is
the same as the distribution of $\inr{v}{\bu}$ for any $v \in S^{n-1}$. Note
that $\ultraS$ is a density on $[-1, 1]$, with density
\begin{equation}\label{eq:ultraspheric-density}
    d\ultraS(t) = \frac{1}{Z_n} (1 - t^2)^\frac{n-3}{2} \, dt,
\end{equation}
where $Z_n$ is a normalizing constant. The proof of Lemma 4.17 in~\cite{FE12} depicts how $d\ultraS$ arises from $d\unifS$.

For a function $g: [-1, 1] \to \R$ satisfying
\begin{equation}\label{eq:g-integrability}
    \int_{-1}^1 |g(t)| \,d\ultraS(t) < \infty,
\end{equation}
define the operator $\U_g$, acting on functions $f: S^{n-1} \to \R^k$ by
\begin{equation}\label{eq:Ug}
    \U_g f(u) = \int_{S^{n-1}} g(\inr uv) f(v) \,d\unifS(v).
\end{equation}
To make the definition fully rigorous, note that the integrability
condition~\eqref{eq:g-integrability} implies that if $f$ is bounded
then $\U_g f$ is defined pointwise. Then Jensen's inequality implies
that $\|\U_g f\|_{L^p(\unifS)} \le C \|\U_g f\|_{L^p(\unifS)}$
for every bounded $f$ -- with $C$ being the left hand side of~\eqref{eq:g-integrability} --
and since bounded functions are dense in $L^p(\unifS)$ it follows that $\U_g$
can be uniquely extended to an operator $L^2(\unifS) \to L^2(\unifS)$.

We will be interested in non-negative $g$, and it might also be convenient to imagine $g$ as
integrating to 1; i.e., with $\int_{-1}^1 g(t)\, d\ultraS(t) = 1$.
In this case $\U_g f(u)$ is an average of values of $f$, much like our Gaussian noise operator $\U_\rho$.

The main result of this section is that if $g$ is monotonic then the function $f(x) = x$ is optimally stable.
\begin{theorem}\label{thm:spherical-noise}
    If $g: [-1, 1] \to [0, \infty)$ satisfies~\eqref{eq:g-integrability} and is non-decreasing then
    for every $f: S^{n-1} \to \R^n$ with $\E_{\bu \sim \unifS} [\|f(\bu)\|^2] = 1$,
    \[
        \E_{\bu \sim \unifS} \inr{f(\bu)}{\U_g f(\bu)}
        \ge
        \E_{\bu \sim \unifS} \inr{f_{\mathrm{opt}}(\bu)}{\U_g f_{\mathrm{opt}}(\bu)},
    \]
    where $f_{\mathrm{opt}}(u) = u$.

    On the other hand, if $g$ is non-\emph{increasing} then 
    for every $f: S^{n-1} \to \R^n$ with $\E_{\bu \sim \unifS} [\|f(\bu)\|^2] = 1$
    and $\E_{\bu \sim \unifS}[f] = 0$,
    \[
        \E_{\bu \sim \unifS} \inr{f(\bu)}{\U_g f(\bu)}
        \le
        \E_{\bu \sim \unifS} \inr{f_{\mathrm{opt}}(\bu)}{\U_g f_{\mathrm{opt}}(\bu)}.
    \]
\end{theorem}

\subsection{Spherical harmonics}

We prove Theorem~\ref{thm:spherical-noise} by decomposing each coordinate of the vector-valued function $f$ into spherical harmonics.  We suggest~\cite{FE12,D13} as introductory references for spherical harmonics.

\begin{definition}[Spherical harmonics]
A \emph{homogeneous polynomial of degree $d$} is a function $p:\R^n \rightarrow \R$ expressible as a linear combination of degree-$d$ monomials.
We define the sets $\{\calH_d: d=0, 1, \dots\}$ by first
setting $\calH_0$ to be the set of constant functions $S^{n-1} \to \R$,
and then inductively defining $\calH_d$ to be the functions $S^{n-1} \to \R$
that can be represented as homogeneous polynomials of degree $d$, and
which are orthogonal to $\bigoplus_{k=0}^{d-1} \calH_k$.
The elements of $\calH_d$ are called degree-$d$ spherical harmonics.
\end{definition}

One subtlety (that will not be particularly important for us) is that distinct polynomials may give rise
to the same function $S^{n-1} \to \R$; for example, the constant function $f(u) = 1$
can be written both as the constant polynomial 1 and as the degree-2 polynomial $u_1^2 + \cdots + u_n^2$,
which evaluates to 1 on the sphere.
The name \emph{spherical harmonics} comes from the fact that $\calH_d$ can
be equivalently defined as the set of homogeneous degree-$d$ polynomials $p$
that are \emph{harmonic} in the sense that $\sum_{i=1}^n \frac{\partial^2}{\partial x_i^2} p(x) = 0$.

The first important thing about the spaces $\calH_d$ is that they form an
orthogonal decomposition of $L^2(S^{n-1},\unifS)$: the $\calH_d$ are orthogonal in the
sense that if $f \in \calH_d$ and $g \in \calH_{d'}$ for $d \ne d'$ then
$\E_{\bu \sim \unifS} [f(\bu) g(\bu)] = 0$; they decompose
$L^2(S^{n-1},\unifS)$ in the sense that if $H_d: L^2(S^{n-1},\unifS) \to \calH_d$ is the
orthogonal projection operator then $f = \sum_{d \ge 0} H_d f$ for every $f \in
L^2(S^{n-1},\unifS)$~\cite{D13}.
In this sense, the decomposition into spherical harmonics is analogous to the decomposition
of a Boolean function into Fourier levels~\cite{OD14}, or the decomposition of a function in $L^2(\R^n, \gamma)$
into Hermite levels~\cite{OD14}.

Then second important thing about the spaces $\calH_d$ is the Funk-Hecke formula, which essentially
says that noise operators like our $\U_g$ act diagonally on spherical harmonics
(much like the usual Boolean and Gaussian noise operators act diagonally on the Fourier
and Hermite bases respectively).

\begin{theorem}[Funk-Hecke formula~\cite{H1917}; see~\cite{D13}, Theorem 2.9]\label{thm:funk-hecke}
    For any $g: [-1, 1] \to \R$ satisfying~\eqref{eq:g-integrability}, let $\U_g$
    be defined as in~\eqref{eq:Ug} (with $k=1$). Then for every $d \in \{0, 1, \dots\}$ there
    exists a $\lambda_d$ such that for every $f \in \calH_d$ and every $u \in S^{n-1}$,
    \[
        \U_g f(u) = \lambda_d f(u).
    \]
    In other words, spherical harmonics are the eigenvalues of $\U_g$ and the eigenvalues
    depend only on the degree $d$.
\end{theorem}

One particularly nice feature of the Funk-Hecke formula is that because the eigenvalues
depend only on the degree $d$, we can compute $\lambda_d$ by choosing the most
convenient $f \in \calH_d$ and $u \in S^{n-1}$. This leads us to the Gegenbauer polynomials,
a family of univariate polynomials that capture the \emph{zonal} spherical harmonics (see~\cite{D13}, Theorem 2.6), those depending only on one direction.

\begin{definition}[Gegenbauer polynomials]\label{def:gegenbauer}
Let $\alpha > -\tfrac{1}{2}$ and $d$ be a nonnegative integer.
The \emph{Gegenbauer polynomial} with Gegenbauer index~$\alpha$ and degree~$d$ is a univariate, real polynomial denoted $C^{(\alpha)}_d$.  The Gegenbauer polynomials correspond to the zonal spherical harmonics of interest to us when $\alpha = \frac{n-2}{2}$ and $n \geq 3$, and we will henceforth make these assumptions on $\alpha$ and $n$.
\end{definition}

Gegenbauer polynomials may be defined recursively, using generating functions, in terms of the Gaussian hypergeometric function, or as special cases of other polynomials (see~\cite{AS72}, Chapter 22).  We will only need a few properties of them.

\begin{proposition}[Properties of the Gegenbauer polynomials]\label{prop:gegenbauer}
\mbox{}
\begin{enumerate}
\item \label{item:geg-low-deg} (\cite{AS72}, 22.4.2) 
We have the following explicit formulas for low-degrees:
\begin{equation*}
C^{(\alpha)}_0(t) = 1\text{, and }
C^{(\alpha)}_1(t) = 2 \alpha t.
\end{equation*}
\item \label{item:geq-inequality} (\cite{AS72}, 22.14.2, \cite{S39}, Theorem 7.4.1)
	For $-1 \leq t \leq 1$ and $\alpha > 0$,
	\begin{equation*}
	|C^{(\alpha)}_d(t)| \leq C^{(\alpha)}_d(1) = \frac{(2 \alpha)_d}{d!},
	\end{equation*}
    with a strict inequality if $d \ge 1$ and $-1 < t < 1$.
\item \label{item:geg-integral} (\cite{AS72}, 22.13.2)
        For each integer $d \geq 0$, define the quantity
	\begin{equation*}
	\mathrm{ratio}_d(t) = \frac{C^{(\alpha)}_d(t)}{C^{(\alpha)}_d(1)} \cdot (1-t^2)^{\alpha - \tfrac{1}{2}}.
	\end{equation*}
	Then
\begin{equation*}
\int \mathrm{ratio}_d(t) \, dt = - \frac{2 (1-t^2)^{\alpha + \tfrac{1}{2}} \alpha}{d (d + 2\alpha)} \cdot \frac{C^{(\alpha + 1)}_{d-1}(t)}{C_d^{(\alpha)}(1)}.
\end{equation*}
\item \label{item:geg-harmonic} (\cite{D13}, Theorem 2.6)
    For each integer $d \ge 0$, the function $S^{n-1} \to \R$ defined by
    $u \mapsto C_d^{(\alpha)}(u_1)$ belongs to $\calH_d$.
\end{enumerate}
\end{proposition}

\subsection{Eigenvalues of the noise operator}

The last property of \Cref{prop:gegenbauer} shows the relevance of Gegenbauer polynomials to the computation of $\lambda_d$:
letting $h(u) = C_d^{(\alpha)}(u_1)$, we have
\[
    \lambda_d = \frac{\U_g h(e_1)}{h(e_1)} = \frac{\E_{\bu \sim \unifS}[h(\bu) g(\inr{\bu}{e_1})]}{h(e_1)}
    = \E_{\bt \sim \ultraS} \left[\frac{C_d^{(\alpha)}(\bt)}{C_d^{(\alpha)}(1)} g(\bt)\right].
\]
Recalling the formula~\eqref{eq:ultraspheric-density} for the density of $\ultraS$, we conclude:

\begin{corollary}[Eigenvalues of $\U_g$]\label{cor:eigenvalue-formula}
    \[
        \lambda_d = \frac{1}{Z_n} \int_{-1}^1 \mathrm{ratio}_d(t) g(t) \, dt.
    \]
\end{corollary}

\begin{remark}[The 2-dimensional case] As noted in \Cref{def:gegenbauer}, we have thus far required $n \geq 3$.  When $n=2$, spherical harmonics reduce to Fourier series.  If $t = \cos(\theta)$, we may define $C_d(t) := \cos(d\theta) = T_d(t)$, where the latter is the degree-$d$ Chebyshev polynomial of the first kind.  Analogues of the properties in \Cref{prop:gegenbauer} hold in this case, allowing us to recover our results for $n=2$.
\end{remark}

Since we are interested in comparing $\lambda_d$ as $d$ varies, the factor $\frac{1}{Z_n}$
is unimportant for us.
The key bound that we will need essentially amounts to considering the case of the indicator function $g = 1_{[-1, t]}$, defined to be $1$ on $[-1,t]$ and $0$ elsewhere.

\begin{lemma}[Key Gegenbauer lemma]\label{lem:key-gegenbauer}
For each integer $d \geq 0$, define the quantity
\begin{equation*}
\nu_d(t) = \int_{-1}^t \mathrm{ratio}_d(w)\, dw,
\end{equation*}
where $-1 < t < 1$. Then $|\nu_d(t)| < - \nu_1(t)$ for $d \geq 1$.
In addition, $\nu_1(t)\leq 0$.
\end{lemma}
\begin{proof}
By \Cref{item:geg-integral} of the Gegenbauer properties,
\begin{equation*}
\nu_d(t) = - \frac{2 (1-t^2)^{\alpha + \tfrac{1}{2}} \alpha}{d (d + 2\alpha)} \cdot \frac{C^{(\alpha + 1)}_{d-1}(t)}{C_d^{(\alpha)}(1)}.
\end{equation*}
Using \Cref{item:geg-low-deg}, we can simplify the $d=1$ case as follows:
\begin{equation*}
\nu_1(t) 
= - \frac{2 (1-t^2)^{\alpha + \tfrac{1}{2}} \alpha}{ (1 + 2\alpha)} \cdot \frac{C^{(\alpha + 1)}_{0}(t)}{C_1^{(\alpha)}(1)}
= - \frac{(1-t^2)^{\alpha + \tfrac{1}{2}}}{ (1 + 2\alpha)}.
\end{equation*}
This is clearly $\leq 0$, as $\alpha > -\tfrac{1}{2}$.
Finally, by \Cref{item:geq-inequality}, we have the bound
\begin{align*}
|\nu_d(t)|
&= \frac{2 (1-t^2)^{\alpha + \tfrac{1}{2}} \alpha}{d (d + 2\alpha)} \cdot \frac{|C^{(\alpha + 1)}_{d-1}(t)|}{C_d^{(\alpha)}(1)}\\
&< \frac{2 (1-t^2)^{\alpha + \tfrac{1}{2}} \alpha}{d (d + 2\alpha)} \cdot \frac{C^{(\alpha + 1)}_{d-1}(1)}{C_d^{(\alpha)}(1)}\\
&= \frac{2 (1-t^2)^{\alpha + \tfrac{1}{2}} \alpha}{d (d + 2\alpha)} \cdot \frac{(2\alpha + 2)_{d-1} \cdot d!}{(d-1)! \cdot (2\alpha)_d}\\
&= \frac{(1-t^2)^{\alpha + \tfrac{1}{2}}}{(2\alpha + 1)}= - \nu_1(t).
\end{align*}
This completes the proof.
\end{proof}

Once we have considered the case of $g = 1_{[-1, t]}$, all other monotonic cases follow simply by
expressing monotonic functions as linear combinations of indicator functions:

\begin{corollary}\label{cor:eigenvalue-bounds}
    In the setting of \Cref{thm:spherical-noise}, if $g$ is non-increasing then $\lambda_1 \le 0$
    and $|\lambda_d| < -\lambda_1$ for all $d \ge 2$. On the other hand, if $g$ is non-decreasing
    then $\lambda_1 \ge 0$ and $|\lambda_d| < \lambda_1$ for all $d \ge 2$.
\end{corollary}

\begin{proof}
    If $g: [-1, 1] \to [0, \infty)$ is non-increasing then in can be written as a linear combination
    of non-increasing indicator functions: there is a measure $\mu$ on $[-1, 1]$ such that
    \[
        g(t) = \int_{-1}^1 1_{[-1, s]}(t) \, d\mu(s).
    \]
    By \Cref{cor:eigenvalue-formula} and Fubini's theorem,
    \[
        \lambda_d = \frac{1}{Z_n} \int_{-1}^1 \int_{-1}^1 \mathrm{ratio}_d(t) 1_{[-1, s]}(t) \, d\mu(s) \, dt
        = \frac{1}{Z_n} \int_{-1}^1 \nu_d(s)\, d\mu(s),
    \]
    where $\nu_d$ is defined as in \Cref{lem:key-gegenbauer}. The claim
    then follows from \Cref{lem:key-gegenbauer}.

    For the case of non-decreasing $g$, note that $\nu_d(1) = 0$ for all $d \ge
    1$, for example because of \Cref{item:geg-harmonic} and the fact that
    spherical harmonics of degree $d \ge 1$ are orthogonal to constant
    functions (which are the spherical harmonics of degree $0$).
    Then we represent $g$ a linear combination of non-decreasing indicator functions
    by choosing $\mu$ such that
    \[
        g(t) = \int_{-1}^1 1_{[s, 1]}(t) \, d\mu(s) = \int_{-1}^1 1 - 1_{[-1, s]}(t) \, d\mu(s);
    \]
    and finally, we have
    \[
        \lambda_d = \frac{1}{Z_n} \int_{-1}^1 \int_{-1}^1 \mathrm{ratio}_d(t) (1 - 1_{[-1, s]}(t)) \, d\mu(s) \, dt
        = -\frac{1}{Z_n} \int_{-1}^1 \nu_d(s)\, d\mu(s),
    \]
    and we conclude as before using \Cref{lem:key-gegenbauer}.
\end{proof}

Finally, \Cref{thm:spherical-noise} follows from \Cref{cor:eigenvalue-bounds} simply by decomposing
the function $f$ in spherical harmonics.

\begin{proof}[Proof of \Cref{thm:spherical-noise}]
    Assume first that $g$ is non-increasing, and choose $f: S^{n-1} \to \R^n$ with $\E[\|f\|^2] = 1$.
    Recall that $H_d: L^2(S^{n-1}) \to \calH_d$ is the orthogonal projection onto degree-$d$
    spherical harmonics. We extend $H_d$ to act on vector-valued functions coordinate-wise,
    so that if $f_1, \dots, f_n$ are the coordinate functions of $f$ then $H_d f = (H_d f_1, \dots, H_d f_n)$.  We also have $\U_g f = (\U_g f_1, \ldots, \U_g f_n)$.
    Recall that $f = \sum_{d \ge 0} H_d f$; then \Cref{thm:funk-hecke} implies that
    \[
        \E_{\bu \sim \unifS} \inr{f(\bu)}{\U_g f(\bu)}
        = \E_{\bu \sim \unifS} \left\langle
            \sum_{d \ge 0} (H_d f)(\bu),
            \sum_{d' \ge 0} (\U_g H_{d'} f)(\bu)
        \right\rangle
        = \sum_{d \ge 0} \lambda_d \E_{\bu \sim \unifS}[\|H_d f(\bu)\|^2],
    \]
    where the cross-terms with $d \ne d'$ vanished because of the orthogonality
    of spherical harmonics.

    On the other hand, the orthogonality of the decomposition $f = \sum_{d \ge 0} H_d f$ implies
    that
    \[
        \sum_{d \ge 0} \lambda_d \E_{\bu \sim \unifS} [\|H_d f(\bu)\|^2] = \E_{\bu \sim \unifS} [\|f(\bu)\|^2] = 1.
    \]
    Since (by \Cref{cor:eigenvalue-bounds} and the fact that $\lambda_0 = \E_{\bt \sim \ultraS}[g(\bt)] \ge 0$)
    $\lambda_1$ is the most-negative of all eigenvalues,
    \[
        \E_{\bu \sim \unifS} \inr{f(\bu)}{\U_g f(\bu)} \ge \lambda_1.
    \]
    Since $f(u) = u$ is a degree-1 spherical harmonic, we get equality in this case. This
    completes the proof for non-increasing $g$.

    When $g$ is non-decreasing, the argument is the same except that the assumption $\E[f] = 0$ implies
    that $H_0 f = 0$, and then \Cref{cor:eigenvalue-bounds} implies that $\lambda_1$ is the most positive
    among all remaining eigenvalues. Therefore,
    \[
        \E_{\bu \sim \unifS} \inr{f(\bu)}{\U_g f(\bu)} \le \lambda_1,
    \]
    and as before we have equality for $f(u) = u$.
\end{proof}

\section{The full-dimensional case}

Here we consider the case when $f$ is an assignment from $\R^n$ to the ball $B^k$ when $k=n$, and we consider a correlation parameter
$-1 \le \rho \le 0$.

\begin{theorem}[Vector-valued Borell's inequality; $n$-dimensional outputs]\label{thm:n-dim-borell}
Let $f:\R^n \rightarrow B^n$.
In addition, let $f_{\mathrm{opt}} :\R^n \rightarrow B^n$
be defined by $f_{\mathrm{opt}}(x) = x / \Vert x\Vert$. Let $-1 \leq \rho \leq 0$. Then
\begin{equation*}
		\stab_\rho[f] \geq \stab_\rho[f_{\mathrm{opt}}].
\end{equation*}
Moreover, if $\stab_\rho[f] = \stab_\rho[f_{\mathrm{opt}}]$ then there is an orthogonal
matrix $M$ such that $f(x) = f_{\mathrm{opt}}(Mx)$ almost surely.
\end{theorem}

Our goal is to lower bound the expression
\begin{equation*}
\E_{\bx \sim_{\rho} \by} \langle f(\bx), f(\by)\rangle.
\end{equation*}
Fourier analysis is a natural tool to bring to bear on this problem.  Although it is natural to consider the Hermite polynomials since the expectation is ``diagonal'' in this set of polynomials (e.g.~\cite{OD14}, Proposition 11.33), the optimal~$f_{\mathrm{opt}}$'s expansion in the Hermite basis of polynomials is complicated (see~\cite{PT20} for a Hermite expansion yielding that of $f_{\mathrm{opt}}$).
As a result, it is difficult to compare the value of $f$ to the value of $f_{\mathrm{opt}}$ using the Hermite basis.

Instead, we reparameterize $\bx$ as $\bx = \br \cdot \bu$,
where $\br$ is the length (or radius) of~$\bx$ and $\bu$ is the unit vector in the direction of~$\bx$.
Similarly, we will reparameterize $\by$ as $\by = \bs \cdot \bv$.
For each value~$r$ that the random variable~$\br$ may take, we will think of $f$ as specifying a separate function on the unit sphere $S^{n-1}$.
We denote this function as $f_{r}:S^{n-1} \rightarrow B^n$ and define it by
\begin{equation*}
f_{r}(\bu) := f(r \cdot \bu) = f(\bx).
\end{equation*}
Using this, we can rewrite our original expectation as
\begin{equation}\label{eq:reparameterized-expectation}
\E_{\br, \bs} \E_{\bu, \bv} \langle f_{\br}(\bu), f_{\bs}(\bv)\rangle.
\end{equation}
What is nice about this reparameterization is that it simplifies our optimizer $f_{\mathrm{opt}}$.
In particular, for each fixed $r \geq 0$, $(f_{\mathrm{opt}})_r(u)$ is simply equal to $u$.

To analyze \Cref{eq:reparameterized-expectation}, we first condition on fixed values of $r, s \geq 0$.
This gives the expression
\begin{equation*}
\E_{\bu, \bv} \langle f_{r}(\bu), f_{s}(\bv)\rangle.
\end{equation*}
This is just an expectation involving two functions on the sphere
(under a distribution on~$\bu$ and~$\bv$ described in the next section). If we take $\U_{\rho}^{r,s}$ as the standard Gaussian noise operator (\Cref{def:noise-operator}) conditioned on $r$,$s$, then we can then further rewrite the above expectation as
\begin{equation}\label{eq:introduce-noise-operator}
\E_{\bu} \langle f_{r}(\bu), \U_\rho^{r, s} f_{s}(\bu)\rangle,
\end{equation}
where we may think of $\U_{\rho}^{r, s} f_s(u)$ as the average of $f_s(\bv)$ over a random~$\bv$, conditioned on~$r$, $s$, and $\bu=u$. 
This noise operator turns out to fall into the setting that we considered in the
previous section, and so applying \Cref{thm:spherical-noise} for each fixed $r,s$ will allow us to
prove \Cref{thm:n-dim-borell}.

\subsection{The induced noise operator}

If $(\bx, \by)$ are $\rho$-correlated random variables then the probability density function (PDF) can be written as
\begin{equation*}
    G_\rho(x, y)= \frac 1A_\rho e^{-\frac{\Vert x \Vert^2+\Vert y \Vert^2- 2 \rho\langle x, y \rangle }{2(1-\rho^2)}}
    = \frac 1A_\rho e^{-\frac{\Vert x \Vert^2+\Vert y \Vert^2}{2(1-\rho^2)}} e^{\frac{\rho r s\langle u, v \rangle}{(1-\rho^2)}},
\end{equation*}
where $A_\rho$ is a normalizing constant.
When we reparameterize according to $(r, s, u, v)$ we obtain:
\begin{equation*}
    G_\rho(r, s, u, v)
    = \frac 1A_\rho (rs)^{n-1} e^{-\frac{r^2+s^2}{2(1-\rho^2)}}e^{\frac{\rho r s \langle u, v \rangle}{(1-\rho^2)}},
\end{equation*}
where the $(rs)^{n-1}$ factor arises from the change of variables.

We will be more interested, however, in the conditional distributions:
\begin{definition}[Conditioned correlated Gaussians] 
We denote by $G^{r, s}_\rho(u,v)$ the PDF of $(\bu,\bv)$, with respect to the measure $\omega$,
conditioned on the values~$r$,~$s$.  We write $(\bu,\bv) \sim \normal^{r, s}_\rho$ for correlated random variables drawn from this distribution.

We denote by $G^{r, s}_\rho(v \mid u)$ the PDF of $\bv$, with respect to the measure $\omega$,
conditioned on the values~$r$,~$s$, and~$u$. This can be written as
\begin{align}\label{eq:conditioned_pdf}
G^{r, s}_\rho(v \mid u)
&= \frac{1}{A^{r, s}_\rho } e^{\frac{\rho r s \langle u, v\rangle}{ (1-\rho^2)}},
\end{align}
where $A^{r, s}_\rho$ is a normalizing constant that depends on $r$, $s$, and $\rho$.  
\end{definition}

We note that \eqref{eq:conditioned_pdf} depends only on the quantity $\langle u, v\rangle$,
and it is monotonically decreasing in this quantity because $r,s \geq 0$ and $\rho \leq 0$.

\begin{definition}[Conditioned Gaussian noise operator]\label{def:noise-operator-on-shell}
The \emph{conditioned Gaussian noise operator} is an operator on $L_2(S^{n-1},\unifS)$ which acts on a function $f:S^{n-1}\rightarrow \mathbb{R}$ as:
\begin{equation*}
    \U^{r, s}_\rho f(u) = \E_{(\bu,\bv) \sim \normal^{r,s}_\rho}[f(\bv) \mid \bu = u] = \int_{S^{n-1}} G^{r,s}_\rho(v \mid u) f(v)\, d\unifS(v).
\end{equation*}
\end{definition}

Note in particular that for every $r,s \ge 0$ and for every $-1 \le \rho \le 0$, $\U^{r,s}_\rho$ is
a noise operator of the form~\eqref{eq:Ug}, for the non-increasing function $g(t) = \frac{1}{A^{r,s}_\rho} e^{\frac{\rho r s t}{1-\rho^2}}$.  We have that for any $u$,
\begin{equation*}
1 = \int_{S^{n-1}} G^{r,s}_\rho(v \mid u)\, d\unifS(v) = \int^1_{-1} g(t)\,d\ultraS(t),
\end{equation*}
by the definitions of $\unifS$ and $\ultraS$, demonstrating that $g \geq 0$ satisfies \eqref{eq:g-integrability}.  Recalling from \Cref{thm:funk-hecke} that the eigenfunctions of $\U^{r,s}_\rho$ are spherical harmonics,
let $\lambda_d^{r,s}$ be the eigenvalue corresponding to $\calH_d$; i.e., $\U^{r,s}_\rho h = \lambda_d^{r,s} h$
for all $h \in \calH_d$.
Then \Cref{cor:eigenvalue-bounds} implies that
\begin{equation}\label{eq:lambda1-bound}
    \lambda_1^{r,s} \le -|\lambda_d|^{r,s} \text{ for every $d \ge 0$},
\end{equation}
with equality only if $d = 1$.
Moreover, the fact that $G^{r,s}_\rho$ is a probability density implies that $\lambda_0^{r,s} = 1$.

The following lemma gives our main lower-bound.
It shows that if $f_r$ and $f_s$ have mean zero,
then the average inner product is lower-bounded by $\lambda_1^{r,s}$,
exactly the value that the optimizer $f_{\mathrm{opt}}$ would achieve.
However, when they are not mean-zero,
they can outperform the optimizer;
consider $f_r = (1, 0, \ldots, 0)$ and $f_s = (-1, 0, \ldots, 0)$,
which have average inner-product $-1$.
To compensate for this, the lemma includes a correction factor depending on the means of~$f_r$ and~$f_s$.

\begin{lemma}[Main lower bound]\label{lem:degree-0-and-1-lower-bound}
\begin{equation*}
\E_{(\bu, \bv) \sim \normal_\rho^{r, s}} \langle f_r(\bu), f_s(\bv)\rangle
\geq \langle \E[f_r], \E[f_s]\rangle + \lambda_1^{r,s},
\end{equation*}
with equality if and only if there is an orthogonal matrix $M$ so that $f_r(u) = f_s(u) = M u$.
\end{lemma}
\begin{proof}
    Recalling that $H_d$ is the orthogonal projection onto $\calH_d$ and that $U^{r,s}_\rho h = \lambda^{r,s}_d h$
    for all $h \in \calH_d$, we have
\begin{align*}
\E_{(\bu, \bv) \sim \normal_\rho^{r, s}} \langle f_r(\bu), f_s(\bv)\rangle
&= \E_{\bu \sim \unifS} \inr{f_r(\bu)}{\U_\rho^{r,s} f_s(\bu)} \\
&= \sum_{d \ge 0} \E_{\bu \sim \unifS} \inr{f_r(\bu)}{\U_\rho^{r,s} H_d f_s(\bu)} \\
&= \sum_{d \ge 0} \lambda_d^{r,s} \E_{\bu \sim \unifS} \inr{f_r(\bu)}{H_d f_s(\bu)} \\
&= \sum_{d \ge 0} \lambda_d^{r,s} \E_{\bu \sim \unifS} \inr{H_d f_r(\bu)}{H_d f_s(\bu)}
\end{align*}
Recall that $\lambda^{r, s}_0 = 1$ and that $H_0 f_r$ is the constant
function $\E_{\bu \sim \unifS}[f_r(\bu)]$; then
\[
\E_{(\bu, \bv) \sim \normal_\rho^{r, s}} \langle f_r(\bu), f_s(\bv)\rangle
= \langle \E[f_r], \E[f_s]\rangle + \sum_{d \ge 1} \lambda_d^{r,s} \E_{\bu \sim \unifS} \inr{H_d f_r(\bu)}{H_d f_s(\bu)}.
\]
For the second term, Cauchy-Schwarz (twice) and~\eqref{eq:lambda1-bound} imply that
\begin{align}
\label{eq:same-dim-first-cs}
\sum_{d \ge 1} \lambda_d^{r,s} \E_{\bu \sim \unifS} \inr{H_d f_r(\bu)}{H_d f_s(\bu)}
& \ge -\sum_{d \ge 1} |\lambda_d^{r,s}| \sqrt{\E [\|H_d f_r\|^2] \E[\|H_d f_s\|^2]} \\
\label{eq:same-dim-lambda}
& \ge \lambda_1^{r,s} \sum_{d \ge 1} \sqrt{\E [\|H_d f_r\|^2] \E[\|H_d f_s\|^2]} \\
\notag
& \ge \lambda_1^{r,s} \sqrt{\sum_{d \ge 1} \E [\|H_d f_r\|^2] \cdot \sum_{d \ge 1} \E[\|H_d f_s\|^2]}
\end{align}
Finally, recalling that $\sum_{d \ge 0} \E [\|H_d f_r\|^2] = \E[\|f_r\|^2] \le 1$, we have
\[
    \sum_{d \ge 1} \lambda_d^{r,s} \E_{\bu \sim \unifS} \inr{H_d f_r(\bu)}{H_d f_s(\bu)} \ge \lambda_1^{r,s}.
\]
This completes the proof of the inequality.

The equality simply follows because if $f_r(u) = f_s(u) = u$
then all inequalities in this proof are equalities.
By the equality cases in~\eqref{eq:lambda1-bound},
we have equality in~\eqref{eq:same-dim-lambda} if and only if $f_r$ and $f_s$ are both affine functions:
$f_r(u) = \E[f_r] + M_r u$ for some $n \times n$
matrix $M_r$ and $f_s(u) = \E[f_s] + M_s u$ for some $n \times n$ matrix $M_s$.
Then we have equality in~\eqref{eq:same-dim-first-cs} if and only if $M_s$ is a
non-negative scalar multiple of $M_r$.
Because $f_r(u)$ takes values in $B^n$, we must have $\|M_r\|_{\textrm{op}} \le 1$ and $\|M_s\|_{\textrm{op}} \le 1$.
But in order to have equality in $\E[\|H_1 f_r\|^2] \le 1$, we must have $\|M_r\|_2^2 = 1$, and so $M_r$
is an orthogonal matrix. Similarly $M_s$ must be an orthogonal matrix, and since it is a non-negative multiple
of $M_r$ they must be equal. Finally, $\E[\|f_r\|^2] = \|\E[f_r]\|^2 + \E[\|H_1 f_r\|^2]  \le 1$,
and so if $\E[\|H_1 f_1\|^2]= 1$ then we must have $\E[f_r] = 0$; similarly for $\E[f_s]$.
\end{proof}

Now we prove \Cref{thm:n-dim-borell}.
At a high-level, the correction term in \Cref{lem:degree-0-and-1-lower-bound} showed that one can improve on the optimizer
for fixed~$r$ and~$s$ by using nonzero means $\E[f_r]$ and $\E[f_s]$.
What we will now show is that although this is true for fixed~$r$ and~$s$,
when averaged over random~$\br$ and~$\bs$ this correction term no longer helps.
In other words, we will show that $\E_{\br, \bs} \langle \E_{\bu}[f_{\br}(\bu)], \E_{\bu}[f_{\bs}(\bu)]\rangle$ is nonnegative,
and so it can only increase the average inner product.

\begin{proof}[Proof of \Cref{thm:n-dim-borell}]
Our goal is to lower-bound
\begin{equation}\label{eq:first-equation-in-big-proof}
\E_{\bx \sim_\rho \by} \langle f(\bx), f(\by)\rangle
 = \E_{\br, \bs} \E_{\bu, \bv} \langle f_{\br}(\bu), f_{\bs}(\bv)\rangle.
\end{equation}
Setting $g(r) = \E_{\bu} f_r(\bu)$,
\Cref{lem:degree-0-and-1-lower-bound} implies that this is at least
\begin{equation*}
\eqref{eq:first-equation-in-big-proof}
\geq \E_{\br, \bs}[\langle g(\br), g(\bs)\rangle + \lambda_1^{\br,\bs}]
= \E_{\br, \bs}\langle g(\br), g(\bs) \rangle + \E_{\br, \bs} \lambda_1^{\br,\bs}.
\end{equation*}
The second term is exactly the value of our conjectured optimizer via~\Cref{lem:degree-0-and-1-lower-bound}.
As a result, it suffices to show that the first term is nonnegative.
We will begin by rewriting it as
\begin{equation}\label{eq:first-rewrite}
    \E_{\bx \sim_\rho \by} \langle g(\Vert \bx \Vert), g(\Vert \by \Vert)\rangle.
\end{equation}
Consider the following method of drawing two $\rho$-correlated strings $\bx$ and $\by$:
first, sample $\bz, \bz', \bz'' \sim \normal(0,1)^n$.
Next, set
\begin{equation*}
\bx = \rho' \cdot\bz + \sqrt{1 - (\rho')^2} \cdot\bz',
\qquad
\by = - (\rho'\cdot \bz + \sqrt{1 - (\rho')^2} \cdot\bz''),
\end{equation*}
where $\rho' = \sqrt{-\rho}$.
Then conditioned on $\bz$,
$\bx$ and $-\by$ are independent and identically distributed random variables.
Hence, we can write
\begin{equation*}
\eqref{eq:first-rewrite}
= \E_{\bz} \E_{\bx, \by} \langle g(\Vert \bx \Vert),  g(\Vert \by \Vert)\rangle
= \E_{\bz} \langle \E_{\bx} g(\Vert \bx \Vert), \E_{\by}g(\Vert \by \Vert)\rangle
=\E_{\bz} \langle \E_{\bx} g(\Vert \bx \Vert), \E_{\by}g(\Vert -\by \Vert)\rangle.
\end{equation*}
Note that the last equality holds because $\Vert -\by \Vert=\Vert \by \Vert$.  For each~$\bz$, the two terms in the inner product are equal, and so this is nonnegative.
This completes the proof of the inequality.

To see the equality cases, recall that we applied the bound of \Cref{lem:degree-0-and-1-lower-bound}
for every $r$ and $s$. If equality is attained in the inequality, we must have equality
in \Cref{lem:degree-0-and-1-lower-bound} for almost every $r$ and $s$. It follows that the matrix $M$
of \Cref{lem:degree-0-and-1-lower-bound} must be independent of $r$ and $s$, and the claimed characterization
of equality cases follows.
\end{proof}

\subsection{The positive-$\rho$ case}
We assumed in this section that $\rho \le 0$. In the case $\rho > 0$, the Gaussian noise model induces
a spherical noise model of the form~\eqref{eq:Ug} with an \emph{increasing} function $g$.
By the results of Section~\ref{sec:spherical},
\begin{equation}\label{eq:positive-rho-eigen-inequality}
\lambda_1^{r,s} \ge |\lambda_d^{r,s}|
\end{equation}
for all $d \ge 2$, and so \Cref{lem:degree-0-and-1-lower-bound} may be extended
to the $\rho > 0$ case, with the opposite inequality. The problem comes from
the first term on the right hand side of \Cref{lem:degree-0-and-1-lower-bound};
this term has a non-negative sign, which is is our favor when $\rho < 0$
but against us when $\rho > 0$. It is possible that this
non-negative term is cancelled out by the difference between the two sides
of~\eqref{positive-rho-eigen-inequality}, but we were not able to show this.

\section{Dimension reduction}
\label{sec:dim-reduction}

We will eventually be concerned with $3$-dimensional assignments to points
which lie in a $n$-dimensional sphere, $S^{n-1}$. \Cref{thm:n-dim-borell} shows
if we are allowed $n$-dimensional assignments, $f_\mathrm{opt}$ minimizes the
noise stability for negative $\rho$. In this section, we will show that the optimization
over $k$-dimensional assignments (for $k \le n$) reduces to the optimization
over $n$-dimensional assignments, \emph{but only for non-negative $\rho$}.
We do this by showing that optimally stable functions are ``at most $k$-dimensional,'' in the
sense that they can be defined on $\R^k$ and not on $\R^n$.

We say that $f: \R^n \to B^k$ is optimally stable with parameter $\rho \in [0, 1]$ if
    $\rho > 0$ and $\E_{\bx \sim_\rho \by}[\inr{f(\bx)}{f(\by)}]$ is maximal among all
        functions $f: \R^n \to B^k$ with $\E_{\bx}[f(\bx)] = 0$.

\begin{theorem}\label{thm:dimension-reduction}
    For every $n, k \ge 1$ and every $\rho \in [0, 1]$, there is an optimally stable function $f$.
    Moreover, if $k \le n$ and $\rho \in (0, 1)$ then for every optimally stable function $f$,
    after a change of coordinates on $\R^n$,
    $f(x)$ depends only on $x_1, \dots, x_k$.
\end{theorem}

Let's address the existence part first, because it's easier.
\begin{proof}[Proof of existence in Theorem~\ref{thm:dimension-reduction}]
    When $\rho \in \{0, 1\}$, existence is trivial because every function is optimally stable;
    from now on, assume $\rho \in (0, 1)$.

    Choose an optimizing sequence $f_n$, i.e.\ a sequence of functions $f_n: \R^n \to B^k$
    such that $\E[f_n] = 0$ and $\E[\inr{f(\bx)}{f(\by)}]$ converges to the optimal value.
    Since $f_n$ are uniformly bounded in $L^2(\gamma)$,
    after passing to a subsequence we may assume that $f_n$ converges weakly to, say, $f$.
    By testing weak convergence against a constant function, it follows that $\E[f] = 0$.

    Recall that the noise operator $\U_\rho: L^2(\gamma) \to L^2(\gamma)$ is compact -- for example, because
    it acts diagonally on the Hermite basis, with eigenvalues that converge to zero.
    It follows that $\U_\rho f_n$ converges strongly in $L^2(\gamma)$ to $\U_\rho f$ and hence
    \begin{align*}
        \E_{\bx \sim_\rho \by} [\inr{f_n(\bx)}{f_n(\by)}]
        &= \E_{\bx} [\inr{f_n(\bx)}{\U_\rho f_n(\bx)}] \\
        &= \E[\inr{f_n}{\U_\rho f_n - \U_\rho f}] + \E[\inr{f_n}{\U_\rho f}] \\
        &\to \E[\inr{f}{\U_\rho f}],
    \end{align*}
    where the first term converged to zero because $\|f_n\|$ is bounded and $\|\U_\rho f_n - \U_\rho f\| \to 0$,
    and the second term converged to $\E[\inr{f}{\U_\rho f}]$ by the weak convergence of $f_n$.
    Since $\E[\inr{f}{\U_\rho f}] = \E_{\bx \sim_\rho \by} [\inr{f(\bx)}{f(\by)]}$, the limit function $f$
    is optimally stable.
\end{proof}

For a differentiable function $f: \R^n \to \R^k$, we write $Df(x)$ for the $k \times n$ matrix
of partial derivatives at the point $x \in \R^n$. For $v \in \R^n$, we will write $D_v f(x) \in \R^k$
for the directional derivative of $f$ in the direction $v$. Of course, $D_v f(x)$ is just an abbreviation
for $(D f(x)) \cdot v$.

\subsection{Outline of the dimension reduction}

The main idea behind the proof of Theorem~\ref{thm:dimension-reduction} is perturbative: we
show that if the function depends on more than $k$ coordinates, there is a perturbation
$\tilde f$ of $f$ that satisfies $\E[\tilde f] = 0$ but has a better noise stability.
We will consider two families of perturbations: ``value'' perturbations of the
form $\tilde f(x) = f(x) + \epsilon \psi(x) + o(\epsilon)$, and ``spatial''
perturbations of the form $\tilde f(x) = f(x + \epsilon \Psi(x) + o(\epsilon))$; our final perturbation will
be a combination of these.

The perturbation $\tilde f$ will never be written down very explicitly. In most
of our analysis, we will rather consider a one-parameter family $f_\epsilon$ of
perturbations, and we will establish the \emph{existence} of a good perturbation
by studying the derivatives of $f_\epsilon$ at $\epsilon = 0$.

There are many technical details, partly because we are considering an infinite-dimensional
optimization problem (over all $f: \R^n \to B^k$) and partly because the a priori the optimal
functions could be almost arbitrarily nasty. However, most of our arguments have simple analogues
for finite-dimensional constrained optimization. In particular, suppose that we are trying
to maximize a differentiable function $\psi: \R^m \to \R$ while obeying the constraint
$g(x) = 0$, for a differentiable $g: \R^m \to \R^k$. Classical Lagrangian
theory for this problem implies that if $x_0 \in \R^m$ is a maximizer and $D g(x_0)$ has rank $k$ then
there is some $\lambda \in \R^k$ such that $D \psi(x_0) = \lambda^T Dg(x_0)$: if this were not the case,
there would be a curve $c: [-\delta, \delta] \to \R^n$ with $c(0) = x_0$,
$g(c(t)) \equiv 0$,
and $\left.\frac{d}{dt}\right|_{t=0} \psi(c(t)) \ne 0$, contradicting the maximality of $x_0$.

The classical theory extends to second-order (at least, if $\psi$ and $g$ are twice-differentiable): if
$x_0$ is a maximizer and $D g(x_0)$ has rank $k$
then the matrix $D^2 \psi - \sum_i \lambda_i D^2 g_i$ acts negatively
on the kernel of $D g(x_0)$ (where $\lambda = (\lambda_1, \dots, \lambda_k)$ is the one whose existence
was guaranteed by the first-order theory).
This is essentially the constrained-optimization analogue of the statement that a function has
a negative-semidefinite Hessian at a maximizer, and it can be proven by showing that if
it fails to hold then there is a curve $c: [-\delta, \delta] \to \R^n$ with $c(0) = c_0$,
$g(c(t)) \equiv 0$,
and $\left.\frac{d^2}{dt^2}\right|_{t=0} \psi(c(t)) > 0$, contradicting the maximality of $x_0$.

To prove Theorem~\ref{thm:dimension-reduction}, we first find analogues of the first- and second-order
variational principles above. For the first-order conditions, we show (Lemma~\ref{lem:first-variation-lagrangian})
that there exists $\lambda \in \R^k$ such that
\begin{equation}\label{eq:first-order-condition}
    |\U_\rho f - \lambda/2| f = \U_\rho f - \lambda/2.
\end{equation}
For the second-order conditions, we show that for the same $\lambda$ and for any nice
enough vector field $\Psi: \R^n \to \R^n$ satisfying $\E [D_{\Psi(\bx)} f(\bx)] = 0$,
\begin{equation}\label{eq:second-variation-for-smooth}
    \E_{\bx \sim_\rho \by} [\inr{D_{\Psi(\bx)} f(\bx)}{D_{\Psi(\by)} f(\by)}]
    - \E_\bx[|\U_\rho f - \lambda/2| \cdot |D_{\Psi(\bx)} f(\bx)|^2] \le 0.
\end{equation}
Note that the expression above is a quadratic function of the vector field $\Psi$, which
can be though of as a ``direction'' along which we perturb $f$. In particular, our second-order
condition really says -- as in the finite-dimensional case -- that a certain quadratic form
acts non-positively on a certain subspace.

Finally, we test~\eqref{eq:second-variation-for-smooth} by substituting constant vector fields
$\Psi(x) \equiv v \in \R^n$, and show that either
\[
    \E_{\bx \sim_\rho \by} [\inr{D_{v} f(\bx)}{D_{v} f(\by)}]
    - \E_\bx[|\U_\rho f - \lambda/2| \cdot |D_{v} f(\bx)|^2] > 0
\]
or $D_v f \equiv 0$. Hence, for every $v \in \R^n$, $\E[D_v f(\bx)] = 0$ implies $D_v f \equiv 0$.
The function $v \mapsto \E[D_v f(\bx)]$ is linear, so if $W \subset \R^n$ is its kernel then $W$ has
codimension at least $k$. After applying a change of variables so that $\mathrm{span} \{e_1, \dots, e_k\} \subseteq W^\perp$ 
the fact that $D_v f \equiv 0$ for $v \in W$ implies that $f$ is a function only of $x_1, \dots, x_k$.

\subsection{Technicalities}

One problem with the outline above is that we wrote ``$D_v f$'' several times, but no one
told us that the optimal function $f$ was differentiable.

We get around this difficulty by exploiting the ``smoothness'' of our objectives and constraints.
For example, we don't care so much about the derivatives of $f$ as we do about how $\E[f_\epsilon]$
changes as we vary $\epsilon$. But $\E[f_\epsilon]$ has as many derivatives (in $\epsilon$) as we wish,
because we may write $\E[f_\epsilon]= \int f(x + \epsilon \Psi(x) + o(\epsilon)) \frac{d\gamma}{dx}\, dx$
and then use a change of variables to pass the spatial perturbation onto the (very smooth)
Gaussian density. Organizing the computations with this explicit change of variables is tedious,
so what we actually do is to first derive our perturbative formulas for smooth
functions $f$, then integrate by parts to push the derivatives onto $\frac{d\gamma}{dx}$. We then
get formulas that make sense for non-smooth $f$; we show that they actually hold for non-smooth $f$
by taking smooth approximations.

The rest of this section is about the integration-by-parts formulas and uniform approximations
that make everything go through rigorously. In particular, we prove several non-smooth analogues of
statements that are trivial for differentiable functions.

For a $\calC^1$ vector field $W$, define
\[
    \div_\gamma W(x) = \div W(x) - \inr{W(x)}{x}.
\]
Note that this satisfies the product rule $\div_\gamma (fW) = f \div_\gamma W + \nabla_W f$ for $\calC^1$
functions $f: \R^n \to \R$.
The point of this definition is the formula
\[
    \int \div_\gamma W \, d\gamma = 0
\]
for compactly supported $W$. Using the product rule, this is equivalent to
\begin{equation}\label{eq:derivative}
    \int f \div_\gamma W\, d\gamma = - \int \nabla_W f \, d\gamma
\end{equation}
for compactly supported $W$ and/or $f$. Now think of the left hand side
as \emph{defining} the derivative of $f$ in a weak sense, noting that the left
hand side makes sense for non-smooth $f$.

Because of the way~\eqref{eq:derivative} expresses derivatives of $f$ in
terms of derivatives of $W$, we will need to impose regularity on the vector
fields $W$ that we consider.
\begin{definition}
    A vector field $W$ is \emph{tame} if it's bounded, $\calC^\infty$-smooth, and if
    its derivatives of all orders are bounded.
\end{definition}

Next, we define our spatial perturbations and our main tool for approximating it by smooth functions:
let $W: \R^n \to \R^n$ be a tame vector field
and let $\{F_t: t \in \R\}$ be the flow
along $W$, defined as the unique function satisfying $F_0(x) = x$ and
\[
    \frac{dF_t(x)}{dt} = W(F_t(x))
\]
for all $t \in \R$ and $x \in \R^n$. Then $F_t$ is a $\calC^\infty$ diffeomorphism for all $t$.
Given an optimal function $f: \R^n \to \R^k$, we may consider the competitor function
$\spatial tWf$ given by
\[
    (\spatial tWf)(x) = f(F_t^{-1}(x)).
\]

It is well-known that
functions in $L^2(\gamma)$ can be approximated (for example, by truncating and mollifying)
using smooth functions. The point here is that we can do this approximation in such
a way that it also applies \emph{uniformly in $t$} to the spatial perturbations $\spatial tWf$.

\begin{lemma}\label{lem:uniform-approximation}
    If $f: \R^n \to \R^k$ is bounded and $W$ is tame then there is a sequence uniformly bounded functions $f_n \in \calC^\infty_c$
    such that
    \[
        \sup_{t \in [-1, 1]} \|\spatial tWf - \spatial tW{f_n}\|_{L_2(\gamma)} \to 0.
    \]
\end{lemma}

\begin{proof}
    Using a change of variables, we can write
    \[
        \|\spatial tWg\|_{L_2(\gamma)}^2
        = \int \|g(x)\|^2 |D F_t(x)| \phi(F_t(x))\, dx,
    \]
    where $|D F_t|$ denotes the Jacobian determinant of $F_t$.
    Now define $\tilde \phi(x) = \sup_{t \in [-1, 1]} |D F_t(x)| \phi(F_t(x))$. Then
    $\tilde \phi$ is integrable: $|D F_t(x)|$ is uniformly bounded for $t \in [-1, 1]$;
    also $|F_t(x) - x|$ is uniformly bounded and so
    \[
        \phi(F_t(x)) \le C \exp(-((|x| - C)_+)^2/2)
    \]
    for some $C$, which is integrable.
    Define the finite measure $d \tilde \gamma = \tilde \phi \, dx$; note that our definition of
    $\tilde \phi$ ensures that
    \[
        \|\spatial tWg\|_{L_2(\gamma)} \le \|g\|_{L_2(\tilde \gamma)}
    \]
    for every $t \in [-1, 1]$.

    Finally, take a uniformly bounded sequence of functions $f_n \in \calC^\infty_c$ such that
    $f_n \to f$ in $L^2(\tilde \gamma)$. Then the claim follows, because
    \[
        \sup_{t \in [-1, 1]} \|\spatial tWf - \spatial tW{f_n}\|_{L_2(\gamma)}
        =
        \sup_{t \in [-1, 1]} \|\spatial tW{(f - f_n)}\|_{L_2(\gamma)}
        \le
        \|f - f_n\|_{L_2(\tilde \gamma)}.
    \]
\end{proof}

If $f: \R^n \to S^{k-1}$ is differentiable
then for any $x, v \in \R^n$, $D_v f(x)$ is tangent to $S^{k-1}$ at $f(x)$; or in other words,
$\inr{D_v f}{f} \equiv 0$. Here is an analogue for certain non-smooth $f$.

\begin{lemma}\label{lem:tangential-derivative}
    Suppose $f: \R^n \to B^k$ is measurable and $\psi: \R^n \to [0, \infty)$ is a bounded, Lipschitz function
    that is differentiable on $\{\psi > 0\}$.
    Assume that $f(x) \in S^{k-1}$ whenever $\psi(x) > 0$, and that $\psi f$ has a uniformly bounded derivative. Then, for every tame vector field $W$,
    \[
        \sum_i \int f_i \div_\gamma (\psi f_i W) \, d\gamma = 0.
    \]
\end{lemma}

(To see why this is an analogue of the easy fact above, note that integration by parts shows that
the left hand side is $-\int \psi \inr{D_W f}{f}\, d\gamma$ in the case of smooth $f$.)

\begin{proof}
    Since $\div_\gamma(\psi f_i W)$ is integrable, by the dominated convergence theorem it suffices to find
    uniformly bounded functions $f^\epsilon$ converging pointwise to $f$ such that
    \begin{equation}\label{eq:tangential-derivative-goal}
        \lim_{\epsilon \to 0}
        \sum_i \int f_i^\epsilon \div_\gamma (\psi f_i W) \, d\gamma = 0.
    \end{equation}
    Fix a constant $C \ge 1$ large enough to be larger than the the uniform bound on $\psi$ and the Lipschitz constants of both $\psi$ and $\psi f$.

    For $\epsilon > 0$, let $\eta^\epsilon: [0, \infty) \to [0, \infty)$ be a $\calC^1$ function satisfying
    \begin{itemize}
        \item $\eta^\epsilon(s) = s$ for $s \ge \epsilon$,
        \item $\eta^\epsilon(s) \ge \frac{\epsilon}{2}$ for all $s$,
        \item $(\eta^\epsilon)'(s) \le 1$ for all $s$,
        \item $(\eta^\epsilon)'(0) = 0$.
    \end{itemize}
    Now define $\psi^\epsilon = \eta^\epsilon \circ \psi$ and $f^\epsilon = \frac{\psi}{\psi^\epsilon} f$. Note
    that $f^\epsilon = f$ whenever $\psi \ge \epsilon$. Moreover, $\psi^\epsilon$ is $\calC^1$
    with
    \[
        |D_v \psi^\epsilon| \le |(\eta^\epsilon)' \circ \psi| |D_v \psi| \le C
    \]
    for any unit vector $v$.
    Since $\eta^\epsilon \ge \frac{\epsilon}{2}$, it follows that $f^\epsilon$ is $\calC^1$ with
    \begin{equation}\label{eq:tangential-derivative-bound}
        |D_v f^\epsilon| \le \frac{2}{\epsilon} |D_v (\psi f)| + |\psi f| \frac{|D_v \psi^\epsilon|}{\epsilon/2} \le \frac{4C^2}{\epsilon}
    \end{equation}
    for every unit vector $v$.

    Since $f^\epsilon = f$ whenever $\psi \ge \epsilon$, we have $|f^\epsilon(x)| = 1$ on $\{\psi \ge \epsilon\}$.
    It follows then that $\inr{D_v f^\epsilon(x)}{f(x)} \equiv 0$ on $\{\psi \ge \epsilon\}$.
    Therefore,
    \begin{align*}
        \sum_i \int f_i^\epsilon \div_\gamma (\psi f_i W) \, d\gamma
        &= -\int \inr{D_W f^\epsilon}{\psi f} \, d\gamma \\
        &= -\int_{\{\psi < \epsilon\}} \inr{D_W f^\epsilon}{\psi f} \, d\gamma.
    \end{align*}
    By~\eqref{eq:tangential-derivative-bound},
    \[
        \left|\int_{\{\psi < \epsilon\}} \inr{D_W f^\epsilon}{\psi f} \, d\gamma\right|
        \le 4 C^2 \int_{\{\psi < \epsilon\}} \frac{\psi |W|}{\epsilon} \, d\gamma
        \le 4 C^2 \int_{\{0 < \psi < \epsilon\}} |W| \, d\gamma.
    \]
    Since $|W|$ is uniformly bounded, the final bound converges to zero as $\epsilon \to 0$.
    This establishes~\eqref{eq:tangential-derivative-goal} and thus completes the proof.
\end{proof}

Here's a simple bound on the derivatives of $\U_\rho g$ for any $L^2$ function $g$.
\begin{lemma}\label{lem:gradient-bound}
    For any $g \in L^2(\gamma)$ and any $-1 < \rho < 1$, $\U_\rho g$ is $\calC^\infty$ smooth and satisfies
    \[
        \E[\|\nabla^k \U_\rho g\|^2] \le C(\rho, k) \E [g^2]
    \]
    for some constant $C(\rho, k) < \infty$, where $\|\nabla^k g\|_2^2$
    denotes the sum of squares of all $k$th order partial derivatives of $g$.
\end{lemma}
\begin{proof}
    With the change of variables $z = \rho x + \sqrt{1-\rho^2} y$, we can write
    \begin{align*}
        \U_\rho g(x)
        &= (2\pi)^{-n/2} \int g(\rho x + \sqrt{1-\rho^2} y) e^{-|y|^2/2}\, dy \\
        &= (2\pi(1-\rho^2))^{-n/2} \int g(z) e^{-\frac{|z - \rho x|^2}{2(1-\rho^2)}} \, dz.
    \end{align*}
    This last formula is clearly differentiable in $x$.

    To show the claimed bound, recall that if $H_\alpha$ are the orthonormal
    Hermite functions (where $\alpha$ is a multi-index) then
    $\frac{\partial}{\partial x_i} H_\alpha = \sqrt{\alpha_i} H_{\alpha - e_i}$.
    Hence, if $g = \sum_\alpha H_\alpha \hat g_\alpha$ is the Hermite expansion
    of $g$ then
    \[
        \E\left[\Big(\frac{\partial}{\partial x_i} \U_\rho g\Big)^2\right]
    = \sum_\alpha \alpha_i e^{-2 \rho |\alpha|} \hat g_\alpha^2
    \]
    and so
    \[
        \E\left[|\nabla \U_\rho g|^2\right]
    = \sum_\alpha |\alpha| e^{-2 \rho |\alpha|} \hat g_\alpha^2
    \]
    The claimed inequality for $k=1$
    follows because, for $x \ge 0$, $x e^{-\rho x}$ is bounded by a constant depending on $\rho$;
    for larger $k$ it follows by induction on $k$.
\end{proof}

Here is an integrated-by-parts version of the obvious fact that if $\E[|D_w f|^2] = 0$ then $f(x)$
is ``independent of $w$'' in the sense that $f(x) = f(y)$ whenever $x$ and $y$ differ by a multiple of $w$.
\begin{lemma}\label{lem:direction-independent}
    For $f \in L^2(\gamma)$, $w \in \R^n$, and $0 < \rho < 1$,
    \[
        \sum_i \E [f_i \div_\gamma((D_w \U_\rho f_i) w)] \le 0,
    \]
    with equality if and only if there is a function $g: w^\perp \to \R$ with $f(x) = g(\Pi_{w^\perp} x)$
    almost surely.
\end{lemma}

\begin{proof}
    Fix $s < 1$; since $\U_s f$ is sufficiently smooth (e.g.\ by Lemma~\ref{lem:gradient-bound}),
    \begin{align*}
        \sum_i \E [\U_s f_i \div_\gamma((D_w \U_\rho f_i) w)] 
        &= -\sum_i \E[\inr{D_w \U_s f_i}{D_w \U_\rho f_i}] \\
        &= -\sum_i \E[\inr{D_w \U_s f_i}{D_w \U_{\sqrt{\rho/s}} \U_{\sqrt{\rho s}} f_i}] \\
        &= -\sqrt{\frac s\rho} \sum_i \E[\inr{D_w \U_s f_i}{\U_{\sqrt{\rho/s}} D_w \U_{\sqrt{\rho s}} f_i}] \\
        &= -\sqrt{\frac s\rho} \sum_i \E[\inr{\U_{\sqrt{\rho/s}} D_w \U_s f_i}{ D_w \U_{\sqrt{\rho s}} f_i}] \\
        &= - \sum_i \E[\inr{D_w \U_{\sqrt{\rho s}} f_i}{ D_w \U_{\sqrt{\rho s}} f_i}] \\
        &= -\E \|D_w \U_{\sqrt{\rho s}} f\|_2^2.
    \end{align*}
    Taking the limit as $s \to 1$, we obtain the identity
    \[
        \sum_i \E [f_i \div_\gamma((D_w \U_\rho f_i) w)] 
        = -\sum_i \E[\inr{D_w f_i}{D_w \U_\rho f_i}] = -\E \|D_w \U_{\sqrt{\rho}} f\|_2^2
    \]
    for any $f \in L^2(\gamma)$.
    The non-positivity claim
    follows easily, and it is also clear that zero is attained if and only if
    $\U_{\sqrt \rho} f$ is independent of $w$. To see that the same is true for $f$,
    suppose without loss of generality that $w = e_i$. Then $f$ is independent of $w$ if and only if
    $f$'s Hermite coefficients $\hat f_\alpha$ are zero whenever $\alpha_i > 0$. Since
    $\U_{\sqrt \rho}$ acts diagonally and non-degenerately on the Hermite basis, $f$ is independent
    of $w$ if and only if $\U_{\sqrt \rho} f$ is.
\end{proof}

\subsection{The first-order conditions}

To derive the first-order optimality condition~\eqref{eq:first-order-condition}, we introduce
the ``value'' perturbations described in the outline.
For $f: \R^n \to \R^k$ and a vector field $W: \R^n \to \R^k$,
define (for $t \in \R$)
\[
    (\valued tWf)(x) = \tilde N(f(x) + t W(x)),
\]
where $\tilde N(x) = x/\max\{1, \|x\|\}$.

\begin{lemma}
    For any measurable $f: \R^n \to B^k$ and any bounded, measurable vector field $W: \R^n \to \R^k$,
    \[
        \left.\frac{d}{dt}\right|_{t=0} \E[\valued tWf] = \E[W - \inr{f}{W}_+ f 1_{\{\|f\| = 1\}}],
    \]
    where $a_+ = \max\{a, 0\}$.
\end{lemma}

\begin{proof}
    If $\|f(x)\| < 1$, $\tilde N(f(x) + t W(x)) = f(x) + t W(x)$ for sufficiently small $t$.
    On the other hand, if $\|f(x)\| = 1$ then
    Taylor expansion gives
    \begin{equation}\label{eq:taylor-on-values}
        \tilde N(f(x) + t W(x)) = f(x) + tW(x) -
        t\inr{W(x)}{f(x)}_+ f(x) + O(t^2)
    \end{equation}
    for any $x \in \R^n$.
    The $O(t^2)$ term
    is uniform in $x$ because we assume $W$ to be uniformly bounded, and hence
    \[
        \frac{\E[\tilde N(f + t W) - f]}{t} = \E[W - \inr{f}{W}_+ f 1_{\{\|f\| = 1\}}] + O(t),
    \]
    and the claim follows by taking the limit as $t \to 0$.
\end{proof}

It is important to note that we can perturb $\E[f]$ in all possible directions; this is our analogue
of the fact that for the finite-dimensional constrained-optimization theory to hold, the constraint
function should have a full-rank derivative.

\begin{lemma}\label{lem:spanning-perturbations}
    For any measurable $f: \R^n \to B^k$,
    there exists a set $W_1, \dots W_k$ of vector fields such that
    \[
        \left\{\left.\frac{d}{dt}\right|_{t=0} \E[\valued t{W_i}f]: i = 1, \dots, k\right\}
    \]
    spans $\R^k$.
\end{lemma}

\begin{proof}
    If $\{x: \|f(x)\| < 1\}$ has positive measure, the claim is clear because
    for vector fields $W$ supported on $\{x: \|f(x\| < 1\}$, $\left.\frac{d}{dt}\right|_{t=0} \E[\valued tWf] = \E[W]$. From now on, we will assume that $f: \R^n \to S^{k-1}$.

    First, choose $v_0 \in S^{k-1}$ belonging to the support of $f$. For some
    sufficiently small $\epsilon > 0$, let $A := \{x: |f(x) - v_0| < \epsilon\}$, and note
    that $A$ has positive measure.
    Let $w_1, \dots, w_{k-1}$ be a basis for $v_0^\perp$ and let $w_k = -v_0$; then $\{w_1, \dots, w_k\}$ spans
    $\R^k$. Define (for $i=1,\dots,k$) $W_i = w_i 1_A / \gamma(A)$.

    First, consider any $i = 1, \dots, k-1$.
    Since $\inr{W_i}{v_0} \equiv 0$, and since $W_i = 0$ whenever $|f(x) - v_0| \ge \epsilon$,
    \[
        |\inr{f(x)}{W_i(x)}| \le \epsilon |W_i(x)|
    \]
    for every $x$. Therefore,
    \[
        \tilde w_i := \left.\frac{d}{dt}\right|_{t=0} \E[\valued t{W_i}f]
            = \E[W_i] - \E[\inr{f}{W_i}_+ f] = w_i + O(\epsilon).
    \]

    On the other hand, for $i=k$ we have
    \[
        |\inr{f(x)}{W_k(x)} + 1| \le \epsilon |W_k(x)|,
    \]
    meaning in particular that $\inr{f}{W_k} \le 0$ pointwise as soon as $\epsilon < 1$.
    Therefore,
    \[
        \tilde w_k := \left.\frac{d}{dt}\right|_{t=0} \E[\valued t{W_k}f] = \E[W_k] = w_k.
    \]
    Since $\{\tilde w_1,\dots,\tilde w_k\}$ is an arbitrarily small perturbation of $\{w_1, \dots, w_k\}$,
    if $\epsilon > 0$ is sufficiently small then $\{\tilde w_1, \dots, \tilde w_k\}$ spans $\R^k$.
\end{proof}

The next step in establishing the first-order conditions is to show that it's enough
to consider derivatives: if it's possible to improve the objective to first-order while preserving
the constraints to first-order then it's also possible to improve the objective to first-order
while preserving the constraints exactly.

\begin{lemma}\label{lem:first-variation}
    If $\rho \in (0, 1)$ and $f: \R^n \to B^k$ is optimally stable then for every bounded, measurable vector field $W$,
    $\left.\frac{d}{dt}\right|_{t=0} \E[\valued tWf] = 0$ implies
    \[
        \left.\frac{d}{dt}\right|_{t=0} \E_{\bx \sim_\rho \by}[\inr{(\valued tWf)(\bx)}{(\valued tWf)(\by)}] = 0.
    \]
\end{lemma}

\begin{proof}
    Let $W$ be any vector field with
    $\left.\frac{d}{dt}\right|_{t=0} \E[\valued tWf] = 0$,
    and choose vector fields
    $W_1, \dots, W_k$ as in Lemma~\ref{lem:spanning-perturbations}.
    Consider the competitor function $f_{\alpha,\beta}(x) = \tilde N(f(x) + \sum_i \alpha_i W_i + \beta W)$.
    We define
    $L: \R^{k+1} \to \R^k$ by
    \[
        L(\alpha, \beta) = \E[f_{\alpha,\beta}].
    \]
    Then
    \begin{equation}\label{eq:diff-beta}
        \frac{\partial L}{\partial \beta} (0, 0) = \left.\frac{d}{dt}\right|_{t=0} \E[\valued tWf] = 0
    \end{equation}
    and
    \begin{equation}\label{eq:diff-alpha}
        \frac{\partial L}{\partial \alpha_i} (0, 0)
        = \left.\frac{d}{dt}\right|_{t=0} \E[\valued t{W_i}f],
    \end{equation}
    which by our choice of $W_i$ implies that $D L$ (as a $(k+1) \times k$ matrix) has rank $k$.
    By the implicit function theorem, there is some interval $(-\epsilon, \epsilon)$ and a differentiable
    curve $\eta: (-\epsilon, \epsilon) \to \R^{k+1}$ such that $\eta(0) = 0$, $\eta'(0) \ne 0$ and
    $L(\eta(t)) = 0$ for all $t \in (-\epsilon, \epsilon)$.
    We'll write $\alpha(t)$ for the first $k$ coordinates of $\eta$, and $\beta(t)$ for the last coordinate.

    Now, the fact that $\eta'(0) \ne 0$ implies that at least one of $\alpha'(0)$ or $\beta'(0)$ is non-zero.
    But the chain rule and the fact that $L(\eta(t))$ is constant gives
    \[
        0 = \left.\frac{d}{dt}\right|_{t=0} L(\eta(t)) = \sum_i \frac{\partial L}{\partial \alpha_i}(0,0) \alpha'_i(0)
            + \frac{\partial L}{\partial \beta} (0,0) \beta'(0);
    \]
    the term involving $\beta$ vanishes because of~\eqref{eq:diff-beta}, while~\eqref{eq:diff-alpha} implies
    that the vectors $\frac{\partial}{\partial \alpha_i} L(0,0)$ are linearly independent. It follows that
    $\alpha'(0) = 0$, and so then we must have $\beta'(0) \ne 0$.

    Finally, we consider the objective value
    \[
        J(\alpha, \beta) = \E_{\bx \sim_\rho \by} \left[\left\langle
            f_{\alpha,\beta}(\bx)
            ,
            f_{\alpha,\beta}(\by)
    \right\rangle\right].
    \]
    Since $f$ is optimally stable and $f_{\alpha(t),\beta(t)}$ satisfies the constraints,
    $\left.\frac{d}{dt}\right|_{t=0} J(\alpha(t), \beta(t)) = 0$.
    On the other hand, the chain rule gives
    \begin{multline*}
        0 = \left.\frac{d}{dt}\right|_{t=0} J(\alpha(t), \beta(t))
            = \sum_i \alpha'_i(0)
            \left.\frac{d}{dt}\right|_{t=0} \E_{\bx \sim_\rho \by}[\inr{(\valued t{W_i}f)(\bx)}{(\valued t{W_i}f)(\by)}] \\
            + \beta'(0)
            \left.\frac{d}{dt}\right|_{t=0} \E_{\bx \sim_\rho \by}[\inr{(\valued t{W}f)(\bx)}{(\valued t{W}f)(\by)}].
    \end{multline*}
    We showed that $\alpha'_i(0) = 0$ for all $i$, and $\beta'(0) \ne 0$, we conclude that
    \[
        \left.\frac{d}{dt}\right|_{t=0} \E_{\bx \sim_\rho \by}[\inr{(\valued t{W}f)(\bx)}{(\valued t{W}f)(\by)}] = 0.
            \qedhere
    \]
\end{proof}

There is also a Lagrangian interpretation of Lemma~\ref{lem:first-variation}:
\begin{lemma}\label{lem:first-variation-lagrangian}
    If $\rho \in (0, 1)$ and $f: \R^n \to B^k$ is optimally stable then there exists some $\lambda \in \R^k$ such that for every
    bounded, measurable vector field $W$,
    \[
        \left.\frac{d}{dt}\right|_{t=0} \E_{\bx \sim_\rho \by}[\inr{(\valued t{W}f)(\bx)}{(\valued t{W}f)(\by)}] =
            \left\langle
                \lambda,
                \left.\frac{d}{dt}\right|_{t=0} \E[\valued tWf]
            \right\rangle.
    \]
\end{lemma}

\begin{proof}
    For a bounded, measurable vector field $W$, let
    \[
        \phi(W) = \left.\frac{d}{dt}\right|_{t=0} \E[\valued tWf] \in \R^k
    \]
    and
    \[
        \psi(W) = 
        \left.\frac{d}{dt}\right|_{t=0} \E_{\bx \sim_\rho \by}[\inr{(\valued t{W}f)(\bx)}{(\valued t{W}f)(\by)}]
            \in \R
    \]
    noting that both $\phi(W)$ and $\psi(W)$ are linear functions of $W$.
    Let $\calX$ be any finite-dimensional subspace of bounded measurable vector fields for which
    $\{\phi(W): W \in \calX\}$ spans $\R^k$,
    and consider the linear map $L: \calX \to \R^{k+1}$ given by $L(W) = (\phi(W), \psi(W))$.
    By Lemma~\ref{lem:first-variation}, $(0, \dots, 0, 1)$ does not belong to the range of $L$;
    it follows that there exists $\lambda = \lambda(\calX)$ such that $(-\lambda, 1)$ is orthogonal
    to the range of $L$ (for example, $\lambda$ can be found by rescaling the residual
    of the orthogonal projection of $(0, \dots, 0, 1)$ onto the range of $L$). For this $\lambda$, we have
    $\inr{\lambda}{\phi(W)} = \psi(W)$ for all $W \in \calX$. Note that $\lambda(\calX)$ is unique,
    because the range of $L$ has dimension at least $k$.

    Now, if $\calX \subset \calX'$ are two vector spaces satisfying the spanning property above
    then $\lambda(\calX') = \lambda(\calX)$ (because $(-\lambda(\calX'), 1)$ is orthogonal to $L(\calX')$
    and hence also $L(\calX)$,
    and $\lambda(\calX)$ is the unique vector with that property).
    It follows then that there is a $\lambda$ satisfying $\inr{\lambda}{\phi(W)} = \psi(W)$
    for all bounded, measurable $W$: take any $\calX$ for which $\{\phi(W): W \in \calX\}$ spans $\R^k$,
    and take
    $\lambda = \lambda(\calX)$. Then for any bounded, measurable $W$, consider $\calX' = \mathrm{span}(W \cup \calX)$; since $\lambda(\calX') = \lambda$, it follows that $\inr{\lambda}{\phi(W)} = \psi(W)$.
\end{proof}

To make Lemma~\ref{lem:first-variation-lagrangian} more useful, we test it on
``local'' vector fields $W$ to extract valuable \emph{pointwise} information
about optimally stable functions.
First, note that the Taylor expansion~\eqref{eq:taylor-on-values} implies that
\begin{align*}
    \left.\frac{d}{dt}\right|_{t=0} \E_{\bx \sim_\rho \by}[\inr{(\valued t{W}f)(\bx)}{(\valued t{W}f)(\by)}]
        &= 2 \E_{\bx \sim_\rho \by} [\inr{f(\bx)}{W(\by) - \inr{W(\by)}{f(\by)}_+ f(\by)}] \\
        &= 2 \E[\inr{W - \inr{W}{f}_+ f}{\U_\rho f}].
\end{align*}
Therefore, Lemma~\ref{lem:first-variation-lagrangian} implies that there exists $\lambda \in \R^k$ such that
\begin{equation}\label{eq:first-variation-in-space}
    2 \E[\inr{W - \inr{W}{f}_+ f 1_{\{\|f\| = 1\}}}{\U_\rho f - \lambda/2}] = 0
\end{equation}
for every bounded, measurable $W$.

\begin{lemma}\label{lem:first-order-conditions}
    If $\rho \in (0, 1)$ and $f: \R^k \to B^k$ is optimally stable then for the $\lambda$ of
    Lemma~\ref{lem:first-variation-lagrangian}, we have
    \[
        |\U_\rho f - \lambda/2| f = \U_\rho f - \lambda/2 \text{ a.e.}
    \]
\end{lemma}

\begin{proof}
    Suppose not, and choose some small $\epsilon > 0$ such that the set
    \[
        A := \{x: |\U_\rho f - \lambda/2| > \epsilon \text{ and } |f - N(\U_\rho f - \lambda/2)| > \epsilon\}
    \]
    has positive measure. Then find some $v \in B^k$ such that
    \[
        C := \{x: x \in A \text{ and } |f(x) - v| < \epsilon^3\}
    \]
    has positive measure. Since
    $|f - N(\U_\rho f - \lambda/2)| > \epsilon$ on $C$, it follows that $|v -
    N(\U_\rho f - \lambda/2)| > \epsilon - \epsilon^3$ on $C$, and so (for small
    enough $\epsilon > 0$) we can find some unit vector $w \in v^\perp$ 
    such that $\inr{N(\U_\rho f - \lambda/2)}{w} \ge \epsilon/2$ on $C$. Now set $W = w 1_C$.

    On the set $C$, $w \in v^\perp$ and $|f(x) - v| \le \epsilon^3$ imply that $|\inr{W}{f}_+ f| \le \epsilon^3$.
    On the other hand, we also have (still on the set $C$)
    \[
        \inr{W}{\U_\rho f - \lambda/2} = \inr{W}{N(\U_\rho f - \lambda/2)} |\U_\rho f - \lambda/2| \ge \frac{\epsilon}{2} |\U_\rho f - \lambda/2| \ge \frac{\epsilon^2}{2},
    \]
    and if $\epsilon > 0$ is small enough then this contradicts~\eqref{eq:first-variation-in-space}.
\end{proof}

\subsection{Spatial perturbations}

Next, we compute the first and second derivatives of the objectives and constraints
for the spatial perturbations $\spatial tWf$ which, recall, is defined by letting
$\{F_t: t \in \R\}$ be the flow along $W$,
and setting
\[
    (\spatial tWf)(x) = f(F_t^{-1}(x)).
\]

\begin{lemma}
    For any bounded function $f: \R^n \to \R^k$ and any tame vector field $W$,
    $\E[\spatial tWf]$ is differentiable in $t$ and satisfies
    \[
        \left.\frac{d}{dt}\right|_{t=0} \E[\spatial tWf] = \E [f \div_\gamma W].
    \]
\end{lemma}

\begin{proof}
    For $f \in \calC_c^1$, this follows by writing out the definition of $\spatial tWf$, differentiating
    inside the integral, and integrating by parts using~\eqref{eq:derivative}:
    \[
        \left.\frac{d}{dt}\right|_{t=0} \E [\spatial tWf] = \left.\frac{d}{dt}\right|_{t=0} \E[f\circ F_{-t}]= -\E [D_W f] = \E[f \div_\gamma W].
    \]
    Because $\spatial sW{\spatial tWf} = \spatial{s+t}Wf$, this also implies that
    \begin{equation}\label{eq:smooth-derivative}
        \frac{d}{dt} \E[\spatial tWf] = \E[\spatial tWf \div_\gamma W]
    \end{equation}
    for $f \in \calC_c^1$.

    Next we handle the case of general bounded $f$.
    Take an approximating sequence $f_n$ as in Lemma~\ref{lem:uniform-approximation}.
    Defining $\phi(t) = \E[\spatial tWf]$ and $\phi_n(t) = \E[\spatial tWf_n]$,
    we see from~\eqref{eq:smooth-derivative} and the uniform boundedness
     of $f_n$ that 
    $\phi_n'(t)$ is continuous in $t$, and bounded uniformly in $n$ and $t$.
    Moreover, Lemma~\ref{lem:uniform-approximation} ensures that $\phi_n(t) \to \phi(t)$ uniformly for $t \in [-1, -1]$, and that $\phi'_n(t)$ converges uniformly.
    It follows that $\phi(t)$ is differentiable in $t$ and satisfies
    \[
        \phi'(0) = \lim_{n\to\infty} \phi_n'(0) = \E [f \div_\gamma W].
    \]
\end{proof}

Next, we do a similar computation for the objective function:
\begin{lemma}
    For any bounded, measurable function $f: \R^n \to \R^k$ and any tame vector field $W$,
    $\E_{\bx \sim_\rho \by}[\inr{\spatial tWf (\bx)}{\spatial tWf(\by)}]$ is differentiable in $t$ and satisfies
    \[
        \left.\frac{d}{dt}\right|_{t=0} \E_{\bx \sim_\rho \by}[\inr{\spatial tWf(\bx)}{\spatial tWf(\by)}]
            = 2\E [\inr{f}{\U_\rho f} \div_\gamma W + \inr{f}{D_W \U_\rho f}].
    \]
\end{lemma}

\begin{proof}
    If $f$ is $\calC_c^\infty$, we compute by calculus that
    \begin{align*}
        \left.\frac{d}{dt}\right|_{t=0} \E_{\bx \sim_\rho \by}[\inr{\spatial tWf(\bx)}{\spatial tWf(\by)}]
            &= -2 \E_{\bx \sim_\rho \by}[\inr{D_{W(\bx)} f(\bx)}{f(\by)}] \\
            &= -2 \E[\inr{D_{W} f}{\U_\rho f}]
    \end{align*}
    For each coordinate, we integrate by parts using~\eqref{eq:derivative}:
    \[
        -2 \E  [D_{W} f_i \cdot \U_\rho f_i]
        = 2\E [f_i \div_\gamma( (\U_\rho f_i) W)]
        = 2\E [f_i \U_\rho f_i \div_\gamma(W) + f_i D_W \U_\rho f_i].
    \]
    Summing over $i$ completes the proof for $f \in \calC_c^\infty$; note
    that because $\spatial sW{\spatial tWf} = \spatial{s+t}Wf$, we also have
    the derivative at $t \ne 0$:
    \begin{equation*}
        \frac{d}{dt} \E_{\bx \sim_\rho \by}[\inr{\spatial tWf(\bx)}{\spatial tWf(\by)}]
        = 2\E [\inr{\spatial tWf}{\U_\rho \spatial tWf} \div_\gamma W + \inr{\spatial tWf}{D_W \U_\rho \spatial tWf}].
    \end{equation*}

    Now consider a bounded function $f$ and choose an approximating sequence $f_n$ as in
    Lemma~\ref{lem:uniform-approximation}. Letting
    $\phi(t) = \E_{\bx \sim_\rho \by} [\inr{\spatial tWf(\bx)}{\spatial tWf(\by)}]$
    and
    $\phi_n(t) = \E_{\bx \sim_\rho \by} [\inr{\spatial tW{f_n}(\bx)}{\spatial tW{f_n}(\by)}]$,
    we note from the formula above that $\phi_n'(t)$ is uniformly bounded and converging uniformly (for $t \in [-1, 1]$,
    where the boundedness of the term involving $D_W \U_\rho \spatial tW{f_n}$
    follows from Lemma~\ref{lem:gradient-bound}).
    Since $\phi_n(t) \to \phi(t)$ uniformly for $t \in [-1, 1]$,
    it follows that $\phi(t)$ is differentiable and $\phi'(0) = \lim_{n\to\infty} \phi_n'(0)$.
\end{proof}

\subsection{Second variation}

Here, we establish the second-order optimality condition for spatial perturbations.
First, let us observe that everything is twice differentiable in $t$:
\begin{lemma}
    For any measurable $f: \R^n \to B^k$ and any tame vector field $W$,
    both $\E[\spatial tWf]$ and $\E_{\bx \sim_\rho \by} [\inr{\spatial tWf(\bx)}{\spatial tWf(\by)}]$
    are twice differentiable in $t$.
\end{lemma}

\begin{proof}
    This can be seen simply by writing out the definitions and changing
    variables so that the derivatives fall only on the Gaussian kernel, much as in the proof
    of Lemma~\ref{lem:gradient-bound}.
\end{proof}

\begin{lemma}\label{lem:stability}
    Suppose $f$ is optimally stable. If $\rho > 0$ and $W$ is a tame vector field such that
    \[
        \left.\frac{d}{dt}\right|_{t=0} \E [\spatial tWf] = 0,
    \]
    then
    \[
        \left.\frac{d^2}{dt^2}\right|_{t=0}
        \E_{\bx \sim_\rho \by} [\inr{\spatial tWf(\bx)}{\spatial tWf(\by)}]
        - \left\langle
            \lambda
            ,
            \left.\frac{d^2}{dt^2}\right|_{t=0} \E[\spatial tWf]
        \right \rangle
        \le 0.
    \]
    If $\rho < 0$ then under the same assumptions, the left hand side above is \emph{at least} zero.
\end{lemma}

This motivates the definition:
\begin{definition}
    For a tame vector field $W$, define
    the \emph{index form}
    \[
        Q(W) =
        \left.\frac{d^2}{dt^2}\right|_{t=0}
        \E_{\bx \sim_\rho \by} [\inr{\spatial tWf(\bx)}{\spatial tWf(\by)}]
        - \left\langle
            \lambda
            ,
            \left.\frac{d^2}{dt^2}\right|_{t=0} \E[\spatial tWf]
        \right \rangle
     \]
\end{definition}

\begin{proof}
    Take vector fields $W_1, \dots, W_m$ as in the proof of Lemma~\ref{lem:first-variation},
    and define
    \[
        f_{\alpha,\beta}(x) = N(f(F_\beta(x)) + \sum_i \alpha_i W_i(x)),
    \]
    where $F_\beta$ is the flow along the vector field $W$.
    As in the proof of Lemma~\ref{lem:first-variation}, defining $L(\alpha, \beta) = \E [f_{\alpha,\beta}]$
    implies that $\frac{\partial L}{\partial \beta} (0, 0) = 0$, while $D L(0, 0)$ has rank $k$.
    Therefore (as in the proof of Lemma~\ref{lem:first-variation}) we can find smooth curves
    $\alpha(t) \in \R^k$ and $\beta(t) \in \R$ such that $\alpha'(0) = 0$, $\beta'(0) \ne 0$,
    and $L(\alpha(t), \beta(t)) \equiv 0$ for $t$ in some interval $(-\epsilon, \epsilon)$.
    Taking second derivatives with respect to $t$, we have
    \begin{equation}\label{eq:stability-second-derivative-of-constraints}
        0 = \left.\frac{d^2}{dt^2}\right|_{t=0} L(\alpha(t), \beta(t)) = \sum_i \alpha_i''(0) \frac{\partial L}{\partial \alpha_i}(0,0) + (\beta'(0))^2 \frac{\partial^2 L}{\partial \beta^2}(0,0)
    \end{equation}

    Define $J(\alpha, \beta) = \E_{\bx \sim_\rho \by} \inr{f_{\alpha,\beta}(\bx)}{f_{\alpha,\beta}(\by)}$ and
    let $K(t) = J(\alpha(t), \beta(t))$.
    First, note that (because $\alpha'(0) = 0$) $K'(0) = \beta'(0)
    \frac{\partial J}{\partial \beta}(0,0)$.  Since $\beta'(0) \ne 0$, we must
    have $\frac{\partial J}{\partial \beta}(0, 0) = 0$ -- otherwise, there
    would be some small $t$ (either positive or negative) giving a contradiction to the optimality of $f$.
    Taking another derivative, the optimality of $f$ implies that
    \[
        0 \ge K''(0)
          = \sum_i \alpha_i''(0) \frac{\partial J}{\partial \alpha_i}(0, 0) + (\beta'(0))^2 \frac{\partial^2 J}{\partial \beta^2}(0,0).
    \]
    Finally, recall from Lemma~\ref{lem:first-variation-lagrangian} that $\frac{\partial J}{\partial \alpha_i} = \inr{\lambda}{\frac{\partial L}{\partial \alpha_i}}$; going back to~\eqref{eq:stability-second-derivative-of-constraints}, we obtain
    \[
        0 \ge K''(0) = -(\beta'(0))^2 \left\langle \lambda,
        \frac{\partial^2 L}{\partial \beta^2}(0,0)
        \right \rangle
        + (\beta'(0))^2 \frac{\partial^2 J}{\partial \beta^2}(0,0);
    \]
    dropping the (positive) $(\beta'(0))^2$ terms and untangling the notation, this is equivalent to the claim.
\end{proof}

\subsection{The index form for translations}

Here, we compute the index form $Q$ for constant vector fields $W \equiv w$,
giving the rigorous, integrated-by-parts analogue of~\eqref{eq:second-variation-for-smooth}
(at least, for $W \equiv w$, which is all that we will need).

\begin{lemma}\label{lem:Q}
    For any measurable $f: \R^n \to B^k$ and any $w \in \R^n$,
    \begin{equation*}
        Q(w)
        = 2 \sum_i \left(
            \E [f_i \div_\gamma(\div_\gamma((\U_\rho f_i - \lambda/2) w) w)]
            - \frac{1}{\rho} \E [f_i \div_\gamma((D_w \U_\rho f_i) w)]
        \right).
    \end{equation*}
\end{lemma}

\begin{proof}
    Note that $\spatial twf(x) = f(x - tw)$.
    If $f \in \calC^\infty$, we simply compute
    \[
        \left.\frac{d^2}{dt^2}\right|_{t=0} \E[\spatial twf]
            = \E [D_w (D_w f)]
    \]
    and
    \begin{align*}
        \left.\frac{d^2}{dt^2}\right|_{t=0}
        & \E_{\bx \sim_\rho \by} [\inr{\spatial twf(\bx)}{\spatial twf(\by)}] \\
        &= 2\E[ \inr{D^2_{w,w} f(\bx)}{f(\by)} + \inr{D_w f(\bx)}{D_w f(\by)}] \\
        &= 2\E[\inr{D^2_{w,w} f}{\U_\rho f} + \inr{D_w f}{\U_\rho D_w f}] \\
        &= 2\E[\inr{D^2_{w,w} f}{\U_\rho f}] + \frac{2}{\rho} \E[\inr{D_w f}{D_w \U_\rho f}],
    \end{align*}
    and hence
    \begin{equation}\label{eq:Q-for-smooth-functions}
        Q(w) = 
        2\E [\inr{D^2_{w,w} f}{\U_\rho f - \lambda/2}]+ \frac{2}{\rho} \E[\inr{D_w f}{D_w \U_\rho f}].
     \end{equation}
    Integrating the first term by parts twice gives
    \begin{align*}
        \E [D^2_{w,w} f_i \cdot (\U_\rho f_i - \lambda_i/2)]
        &= -\E [D_w f_i \cdot \div_\gamma((\U_\rho f_i - \lambda_i/2) w)] \\
        &= \E [f_i \div_\gamma(\div_\gamma((\U_\rho f_i - \lambda_i/2) w) w)];
    \end{align*}
    integrating the second term of~\eqref{eq:Q-for-smooth-functions} by parts gives
    \[
        \E [D_w f_i \cdot D_w \U_\rho f_i] = -\E [f_i \div_\gamma((D_w \U_\rho f_i) w)].
    \]
    Overall, we obtain
    \begin{equation*}
        Q(w)
        = 2 \sum_i \E\left[
            f_i \div_\gamma(\div_\gamma((\U_\rho f_i - \lambda_i/2) w) w)
            - \frac{1}{\rho} f_i \div_\gamma((D_w \U_\rho f_i) w)
        \right].
    \end{equation*}

    By the familiar approximation argument (noting that the terms involving
    second derivatives of $\U_\rho f_i$ are controlled by Lemma~\ref{lem:gradient-bound}),
    the same formula applies for all bounded, measurable functions $f$.
\end{proof}

Our formula for $Q(w)$ in Lemma~\ref{lem:Q} doesn't require $f$ to be optimally stable.
However, if $f$ is optimally stable, the first-order conditions
allow us to find a simpler formula.

\begin{lemma}\label{lem:Q-simpler}
    If $f: \R^n \to B^k$ is optimally stable then for any $w \in \R^n$,
    \[
        Q(w) = 2\frac{\rho - 1}{\rho} \sum_i \E[f_i \div_\gamma ((D_w \U_\rho f_i) w)].
    \]
\end{lemma}

\begin{proof}
    The point is to show that
    \[
            \sum_i \E [f_i \div_\gamma(\div_\gamma((\U_\rho f_i - \lambda_i/2) w) w)]
            = \sum_i \E [f_i \div_\gamma((D_w \U_\rho f_i) w)];
    \]
    then the claim follows immediately.
    To show the identity above, note that (by the product rule)
    \[
        \div_\gamma ((\U_\rho f_i - \lambda_i/2) w) = (\U_\rho f_i - \lambda_i/2) \div_\gamma w
        + D_w \U_\rho f_i;
    \]
    plugging this in above, it suffices to show that
    \begin{equation}\label{eq:Q-simpler-subgoal}
        \sum_i \E [f_i \div_\gamma((\U_\rho f_i - \lambda_i/2) \div_\gamma w \cdot w)] = 0
    \end{equation}
    But Lemma~\ref{lem:first-variation-lagrangian} implies that $\U_\rho f -
    \lambda/2 = |\U_\rho f - \lambda/2| f$;~\eqref{eq:Q-simpler-subgoal} then
    follows from
    Lemma~\ref{lem:tangential-derivative} with $\psi = |\U_\rho f - \lambda/2|$ and $W = (\div_\gamma w) w$.
\end{proof}

\begin{proof}[Proof of Theorem~\ref{thm:dimension-reduction}]
    Suppose $f$ is optimally stable and suppose that $\rho \in (0, 1)$. By Lemma~\ref{lem:stability},
    for any $w$ with
    \[
        \left.\frac{d}{dt}\right|_{t=0} \E[\spatial twf] = 0,
    \]
    we have $Q(w) \le 0$. Using the formula for $Q$ in Lemma~\ref{lem:Q-simpler}, for such $w$ we have
    \[
        \sum_i \E[f_i \div_\gamma ((D_w \U_\rho f_i) w)] \ge 0.
    \]
    But then Lemma~\ref{lem:direction-independent} implies that $f(x)$ can be written as a function
    of $\Pi_{w^\perp} x$.

    Note that the map
    \[
        L(w) = \left.\frac{d}{dt}\right|_{t=0} \E[\spatial twf]
    \]
    is a linear map $\R^n \to \R^k$. Then $\ker L$ has dimension at least $n-k$
    After applying a change of coordinates
    in $\R^n$, we may assume that $\ker L$ contains the span of $e_{k+1}, \dots, e_n$,;
    then the previous paragraph
    implies that $f(x)$ depends only on $x_1, \dots, x_m$.
\end{proof}

\subsection{The case of negative $\rho$}

Many of the technical results we developed above apply to the case of negative $\rho$, with a few sign changes.
Notably, the sign in Lemma~\ref{lem:first-order-conditions} changes to
\[
|U_\rho f - \lambda/2| f = \lambda/2 - U_\rho f;
\]
and because the negative-$\rho$ case is a minimization problem instead of a maximization problem,
the sign of the second-order conditions flips also: if $\rho < 0$ then the final inequality of
Lemma~\ref{lem:stability} is reversed.

Because of these sign changes, the constant vector fields $w$ turn out not to contradict any stability.
In fact, with $\rho < 0$ then $Q(w) \le 0$ for every $w$, whether or not $\spatial twf$ preserves expectations
to first-order. It remains plausible that there are some other vector fields that will imply
Theorem~\ref{thm:dimension-reduction} in the case of negative $\rho$, but we were unable to find them.

\newpage
\part{Hardness of \qmaxcut}
The goal of this part is to use the results from \Cref{part:borell}
to derive hardness of approximation results for \qmaxcut
assuming the vector-valued Borell's inequality.
This part is organized as follows:
we begin with preliminaries in \Cref{sec:hardness-prelims}.
Then we show our integrality gaps in \Cref{sec:integrality-gaps} and our algorithmic gap in \Cref{sec:algorithmic-gap}.
Finally, in \Cref{sec:dictator-test} we develop a \emph{dictator test} for product states,
which is a crucial ingredient in our hardness proof, contained in \Cref{sec:ug-hardness}.

\section{Preliminaries}\label{sec:hardness-prelims}
\ignore{
\subsection{The \maxcut problem}

\begin{definition}[Weighted graph]
A \emph{weighted graph} $G = (V, E, w)$ is an undirected graph with weights on the edges specified by $w: E \rightarrow \R^{\geq 0}$.
The weights are assumed to sum to~$1$, i.e.\ $\sum_{e \in E} w(e) = 1$, and so they specify a probability distribution on the edges.
We will write $\boldsymbol{e} \sim E$ or $(\bu, \bv) \sim E$ for a random edge sampled from this distribution.
We will often use the word \emph{graph} as shorthand for weighted graph.
\end{definition}

\begin{definition}[\maxcut]\label{def:max-cut}
Given a graph $G$, \maxcut is the problem of partitioning the vertices of~$G$ into two sets
to maximize the probability that a random edge is cut.
The optimum value is given by
\begin{equation*}
\maxcut(G) = \max_{f:V \rightarrow \{\pm 1\}} \E_{(\bu, \bv) \sim E}[ \tfrac{1}{2} - \tfrac{1}{2} f(\bu) f(\bv)].
\end{equation*}
\end{definition}
}

\subsection{Pauli matrices}

\begin{definition}[Pauli matrices]\label{def:pauli}
The \emph{Pauli matrices} are the Hermitian matrices
\begin{equation*}
X =
\begin{pmatrix}
0 & 1\\
1 & 0
\end{pmatrix},
\quad
Y =
\begin{pmatrix}
0 & -i\\
i & 0
\end{pmatrix},
\quad
Z =
\begin{pmatrix}
1 & 0\\
0 & -1
\end{pmatrix}.
\end{equation*}
\end{definition}

\begin{notation}
We will generally use $P$ or $Q$ for a variable in $\{I, X, Y, Z\}$.
\end{notation}

\begin{proposition}[Properties of the Pauli matrices]\label{prop:pauli-props}
The Pauli matrices have the following properties.
\begin{enumerate}
\item $X^2 = Y^2 = Z^2 = I$.
\item $XY = i Z$, $YZ = iX$, and $ZX = i Y$.
\item $XY = -YX$, $YZ = - ZY$, and $ZX = - XZ$.
\item $\tr[X] = \tr[Y] = \tr[Z] = 0$.
\end{enumerate}
\end{proposition}

\subsection{On \qmaxcut and the Heisenberg model}\label{sec:qmaxcut-heisenberg}

The \emph{quantum Heisenberg model} is a family of Hamiltonians first studied by Heisenberg in~\cite{Hei28}.
Given an unweighted graph $G = (V, E)$,
a Hamiltonian from this model is written as
\begin{equation*}
H = 
- \E_{(u, v) \in E} (J_X \cdot X_u X_v + J_Y \cdot Y_u Y_v + J_Z \cdot Z_u Z_v) - m \sum_{u \in V} Z_u,
\end{equation*}
where $P_u$ for $P \in \{X, Y, Z\}$ refers to the Pauli matrix~$P$ applied to the $u$-th qubit,
$J_X, J_Y, J_Z$ are real-valued coefficients known as \emph{coupling constants},
and $m$ is a real-valued coefficient known as the \emph{external magnetic field}.
As is typical in Hamiltonian complexity,
and unlike in \qmaxcut,
the ground state energy of this Hamiltonian is defined to be its \emph{minimum} (rather than maximum) eigenvalue,
and the ground state is defined to be the corresponding eigenvector.
The \emph{ferromagnetic case} refers to the case when $J_X, J_Y, J_Z \geq 0$,
in which case neighboring qubits tend to have the same values in the $X,Y,Z$ bases,
and the \emph{anti-ferromagnetic case} is when $J_X, J_Y, J_Z \leq 0$,
in which case they have opposing values.

The \emph{anti-ferromagnetic Heisenberg XYZ model} (which we will henceforth simply refer to as the ``Heisenberg model''),
is the case when $J_X = J_Y = J_Z = -1$ and $m = 0$.
It is natural to allow for the graph $G = (V, E, w)$ to be weighted,
in which case we can write a Hamiltonian from this model as
\begin{equation*}
H_G^{\text{\sc Heis}} = \E_{(\bu, \bv) \sim E}[X_{\bu} X_{\bv} + Y_{\bu} Y_{\bv} + Z_{\bu} Z_{\bv}].
\end{equation*}
As we have mentioned before, \qmaxcut was defined to be a natural maximization version of the Heisenberg model.
Indeed,
if $H_G$ is the \qmaxcut instance corresponding to~$G$,
then $H_G = (I - H_G^{\text{\sc Heis}})/4$.
This means that if $\ket{\psi}$ is the minimum energy state of $H_G^{\text{\sc Heis}}$ and has energy $\nu$,
then $\ket{\psi}$ is also the \emph{maximum} energy state of $H_G$ and has energy $(1-\nu)/4$.

Where the two variants differ is in their approximability.
As we have seen throughout this work,
one can achieve a constant-factor approximation to the \qmaxcut objective in polynomial time.
On the other hand, the best known approximation algorithm for the Heisenberg model objective
is due to~\cite{BGKT19} and achieves a $1/O(\log(n))$ approximation.
In particular, if the minimum energy of $H_G^{\text{\sc Heis}}$ is $\nu$,
this algorithm finds a product state with energy no greater than $\nu/O(\log(n))$.
The source of this difference is the identity term $I \ot I$ in the \qmaxcut objective,
which ``inflates'' the energy of a state relative to its energy in the Heisenberg model.
For example, a tensor product of maximally mixed qubits always has objective value~$1/4$ in \qmaxcut due to these identity terms
(which, in turn, implies one can always trivially achieve an approximation ratio of $1/4$).
In the Heisenberg model, however, its objective value is~$0$, and so it gives no approximation to the optimum value.
An analogous situation occurs in the classical world,
where the \maxcut objective
\begin{equation*}
\max_{f:V\rightarrow \{-1, 1\}}
\E_{(\bu, \bv) \sim E}[\tfrac{1}{2} - \tfrac{1}{2} f(\bu) f(\bv)]
\end{equation*}
has a constant-factor $0.878567$-approximation~\cite{GW95},
but the shifted and rescaled objective
\begin{equation*}
\min_{f:V\rightarrow \{-1, 1\}}
\E_{(\bu, \bv) \sim E}[f(\bu) f(\bv)]
\end{equation*}
gives us the (anti-ferromagnetic) \emph{Ising model} problem, for which the best-known algorithm is due to Charikar and Wirth~\cite{CW04} and achieves an approximation ratio of $1/O(\log(n))$.

The Heisenberg model
is ``notoriously difficult to solve even on
bipartite graphs, in contrast to \maxcut''~\cite{GP19}.
Only a few explicit solutions have been found,
several of which are well-known results in the physics literature.
These include the Heisenberg model on the cycle graph, whose solution due to Bethe is known as the ``Bethe ansatz''~\cite{Bet31},
and on the complete bipartite graph, known as the ``Lieb-Mattis model''~\cite{LM62}.
To our knowledge, \cite[Section 5.2]{CM16} contains a complete list of known explicit solutions.
This difficulty of the finding solutions for the Heisenberg model
was explained by the works of~\cite{CM16,PM17},
who showed that it is a $\QMA$-complete problem.
This implies that \qmaxcut is also $\QMA$-complete.

\subsection{Alternative expressions for the \qmaxcut interaction}

There are several alternative ways of writing the \qmaxcut interaction
\begin{equation*}
h = \tfrac{1}{4}(I \ot I - X \ot X - Y \ot Y - Z \ot Z)
\end{equation*}
which are common in the literature.
The first involves the singlet state.

\begin{definition}[Singlet state]
The \emph{two-qubit singlet state} is
\begin{equation*}
\ket{s} = \tfrac{1}{\sqrt{2}} \ket{01} - \tfrac{1}{\sqrt{2}} \ket{10}.
\end{equation*}
It is also known as the \emph{two-qubit anti-symmetric state},
and as the element $\ket{\Psi^-}$ of the \emph{Bell basis}
of two-qubit states
\begin{equation*}
\ket{\Phi^{\pm}} = \tfrac{1}{\sqrt{2}} \ket{00} \pm \tfrac{1}{\sqrt{2}} \ket{11},
\quad
\ket{\Psi^{\pm}} = \tfrac{1}{\sqrt{2}} \ket{01} \pm \tfrac{1}{\sqrt{2}} \ket{10}.
\end{equation*}
\end{definition}

\ignore{
\begin{definition}[The anti-ferromagnetic Heisenberg interaction]
The \emph{anti-ferromagnetic Heisenberg interaction} is the $2$-qubit operator
$
h = \ket{s}\bra{s}.
$
\end{definition}
}

The following proposition gives a convenient expression for the \qmaxcut interaction
as the projector on the singlet state.

\begin{proposition}[Rewriting the \qmaxcut interaction]\label{prop:rewrite-interaction}
\begin{equation*}
h
 = \tfrac{1}{4}\cdot (I \ot I - X \ot X) \cdot (I \ot I - Z \ot Z)
= \ket{s}\bra{s}.
\end{equation*}
\end{proposition}
\begin{proof}
The first equality follows from $(X \ot X) (Z \ot Z) = - Y \ot Y$.
We verify the second equality by checking that both sides have the same eigendecomposition.
Let
\begin{equation}\label{eq:bell-basis}
\tfrac{1}{\sqrt{2}}\ket{0,a} + \tfrac{1}{\sqrt{2}}(-1)^b \ket{1, 1+a}
\end{equation}
be a member of the Bell basis, for $a, b \in \{0, 1\}$.
This is a $1$-eigenvector of~$\ket{s}\bra{s}$ if $a, b = 1$ and a $0$-eigenenvector otherwise.
Now we verify this holds for the LHS:
\begin{align*}
& \tfrac{1}{4} \cdot (I \ot I - X \ot X) \cdot (I \ot I - Z \ot Z) \cdot \big(\tfrac{1}{\sqrt{2}}\ket{0,a} + \tfrac{1}{\sqrt{2}}(-1)^b \ket{1, 1+a}\big)\\
={}&\tfrac{1}{4\sqrt{2}} \cdot (I \ot I - X \ot X) \cdot \big(\ket{0,a} - (-1)^a \ket{0, a} + (-1)^b \ket{1, 1+a}- (-1)^{b} \cdot (-1)^{2+a} \ket{1, 1+a}\big)\\
={}&\tfrac{1}{4\sqrt{2}} \cdot (1 - (-1)^a) \cdot (I \ot I - X \ot X) \cdot \big(\ket{0,a}  + (-1)^b \ket{1, 1+a}\big)\\
={}&\tfrac{1}{4\sqrt{2}} \cdot (1 - (-1)^a) \cdot \big(\ket{0,a} - \ket{1, 1+a} + (-1)^b \ket{1, 1+a} - (-1)^b \ket{0, a}\big)\\
={}& \tfrac{1}{4\sqrt{2}}\cdot (1 - (-1)^a) \cdot (1- (-1)^b) \cdot (\ket{0, a} - \ket{1, 1+a}).
\end{align*}
If $a, b = 1$, then this is equal to~\eqref{eq:bell-basis}, showing that it is a $1$-eigenvector.
Otherwise, this is zero, which completes the proof.
\end{proof}

\ignore{

\begin{definition}[The anti-ferromagnetic Heisenberg model]
Let $G = (V, E, w)$ be a graph known as the \emph{interaction graph}.
Consider a quantum state in $(\C^2)^{V}$
containing a qubit for each vertex $v \in V$.
For any $u, v \in V$, we write $h_{u, v}$ as shorthand for the operator
\begin{equation*}
h_{u, v} \otimes I_{V \setminus \{u, v\}},
\end{equation*}
where~$h$ acts on the~$u$ and~$v$ qubits
and~$I$ acts on the rest.
The corresponding instance of the \emph{anti-ferromagnetic Heisenberg model} is given by
\begin{equation*}
H_G = \E_{(\bu, \bv) \sim E} h_{\bu, \bv}.
\end{equation*}
\end{definition}

\begin{definition}[Energy]
Let $H_G$ be an instance of the \emph{anti-ferromagnetic Heisenberg model}.
Its \emph{maximum energy} is 
\begin{equation*}
\heis(G) = \lambda_{\mathrm{max}} (H_G) = \max_{\ket{\psi} \in (\C^2)^{V}} \bra{\psi} H_G \ket{\psi}.
\end{equation*}
We may also refer to this as the \emph{value} of $H_G$.
\end{definition}
}

Though we will not need it,
we note the following additional way of rewriting the \qmaxcut interaction
for didactic purposes.
The proof is left to the reader.

\begin{proposition}
\begin{equation*}
h
 = \tfrac{1}{2}\cdot (I \ot I - \mathsf{SWAP}),
\end{equation*}
where $\mathsf{SWAP}$ is the two-qubit swap gate.
\end{proposition}

\subsection{Product states}\label{sec:product-states}

\ignore{
\begin{definition}[Product state value]
The \emph{product state value of $H_G$} is
\begin{equation*}
\prodval(G) = \max_{\forall v \in V, \ket{\psi_v} \in \C^2} \bra{\psi_G} H_G \ket{\psi_G},
\end{equation*}
where $\ket{\psi_G} = \otimes_{v \in V} \ket{\psi_v}$.
\end{definition}

\begin{definition}[Balls and spheres]
Given a dimension $d \geq 1$, the $d$-dimensional unit ball and sphere are given by
\begin{equation*}
B^d = \{x \in \R^d \mid \Vert x \Vert \leq 1\},
\qquad
S^{d-1} = \{x \in \R^d \mid \Vert x \Vert = 1\}.
\end{equation*} 
\end{definition}
}

The following definition gives a convenient decomposition for single-qubit quantum states.
It can be derived using the properties in \Cref{prop:pauli-props}.

\begin{definition}[Bloch spheres and Bloch vectors]\label{def:bloch_vectors}
Let $\rho$ be a one qubit density matrix.
Then there exists a coefficient vector $c = (c_X, c_Y, c_Z) \in B^3$ such that
\begin{equation*}
\rho = \frac{1}{2}\cdot (I + c_X  X + c_Y  Y + c_Z  Z).
\end{equation*}
In addition, $\rho$ is a pure state if and only if $\Vert c\Vert = 1$; equivalently, if $c \in S^2$. We'll refer to the vector $c$ as the \textit{Bloch vector} for $\rho$.
\end{definition}

Using this, we can now prove the alternate form for the product state value.

\begin{proposition}[Rewriting the product state value; \Cref{prop:rewrite-product} restated]\label{prop:rewrite-product-restated}
\begin{equation*}
\prodval(G) = \max_{f:V \rightarrow S^2} \E_{(\bu, \bv) \sim E}[ \tfrac{1}{4} - \tfrac{1}{4} \langle f(\bu), f(\bv)\rangle].
\end{equation*}
\end{proposition}
\begin{proof}
Let $\ket{\psi_G} = \otimes_{v \in V} \ket{\psi_v}$ be a product state.
For each $v \in V$, let $f(v) = (v_X, v_Y, v_Z) \in S^2$
be its Bloch sphere coefficient vector.
In other words,
\begin{equation*}
\ket{\psi_v}\bra{\psi_v} = \frac{1}{2}\cdot (I + v_X  X + v_Y  Y + v_Z  Z).
\end{equation*}
Then the energy of $\ket{\psi_G}$ is given by
\begin{equation}\label{eq:first-step-of-energy-calculation}
\tr[H_G \cdot \ket{\psi_G}\bra{\psi_G}]
= \E_{(\bu, \bv) \sim E}\tr[h_{\bu, \bv} \cdot \ket{\psi_G}\bra{\psi_G}]
= \E_{(\bu, \bv) \sim E}\tr[h_{\bu, \bv} \cdot \ket{\psi_{\bu}}\bra{\psi_{\bu}} \ot \ket{\psi_{\bv}} \bra{\psi_{\bv}}].
\end{equation}
For any edge $(u, v) \in E$, the energy of the $(u,v)$-interaction is
\begin{align*}
&\tr[h_{u, v} \cdot \ket{\psi_u}\bra{\psi_u} \ot \ket{\psi_v} \bra{\psi_v}]\\
={}& \tr\Big[\frac{1}{4} \cdot ( I \ot I - X \ot X - Y \ot Y - Z \ot Z)\\
		&\qquad \qquad \qquad \qquad\cdot \frac{1}{2}\cdot (I + u_X  X + u_Y  Y + u_Z  Z) 
												\ot \frac{1}{2}\cdot (I + v_X  X + v_Y  Y + v_Z  Z)\Big]\\
={}& \frac{1}{4} \cdot (1 - u_X v_X - u_Y v_Y - u_Z v_Z) \tag{by \Cref{prop:pauli-props}}\\
={}& \frac{1}{4} \cdot (1 - \langle f(u), f(v)\rangle).
\end{align*}
Substituting into \Cref{eq:first-step-of-energy-calculation}, the energy of $\ket{\psi_G}$ is
\begin{equation*}
\E_{(\bu, \bv) \sim E}[\tfrac{1}{4}  - \tfrac{1}{4}\langle f(\bu), f(\bv)\rangle].
\end{equation*}
This concludes the proof.
\end{proof}

\begin{remark}\label{rem:prod-state-bad}
One consequence of \Cref{prop:rewrite-product-restated} is that $\prodval(G) \leq 1/2$ always,
even though $\heis(G)$ can be as large as~$1$.
Indeed, $\bra{\psi_1, \psi_2} h \ket{\psi_1,\psi_2} \leq 1/2$ 
for any qubit states $\ket{\psi_1}, \ket{\psi_2} \in \C^2$,
even though $\bra{s} h \ket{s} = 1$ by \Cref{prop:rewrite-interaction}.
\end{remark}

Although \Cref{rem:prod-state-bad} shows that the product states can in general give a poor approximation to the energy,
there are interesting special cases in which they still give good approximations.
The following result shows that this holds
provided that the degree of~$G$ is large.

\begin{theorem}[Corollary 4 of~\cite{BH16}]\label{thm:BH}
Let $G  = (V, E, w)$ be a $D$-regular graph with uniform edge weights. Then
\begin{equation*}
\prodval(G) \geq \heis(G) - O\left(\frac{1}{D^{1/3}}\right).
\end{equation*}
\end{theorem}

Unfortunately, we will not be able to apply this theorem directly
because our graphs will not be precisely unweighted, $D$-regular graphs,
but instead high-degree graphs with different weights on different edges.
In this case, Brandão and Harrow provide the following bound.

\begin{theorem}[Theorem 8 of~\cite{BH16}]\label{thm:BH-nonuniform}
Let $G  = (V, E, w)$ be a weighted graph.
Define
\begin{enumerate}
\item the probability distribution $(p_u)_{u \in V}$ such that $p_u = \tfrac{1}{2} \Pr_{\boldsymbol{e} \sim E}[\text{$\boldsymbol{e}$ contains $u$}]$,
\item the $|V| \times |V|$ matrix $A$ such that $A_{u, v} = \Pr_{(\bu', \bv') \sim E}[\bu' = u \mid \bv' = v].$
\end{enumerate}
Then the following inequality holds.
\begin{equation*}
\heis(G) \leq \prodval(G) + 20\cdot (\tr[A^2] \Vert p \Vert_2^2)^{1/8} + \Vert p \Vert_2^2.
\end{equation*}
\end{theorem}

We note that in the case of a $D$-regular graph,
$\tr[A^2] = n/D$ and $\Vert p \Vert_2^2 = 1/n$.
We will use the following corollary, which we will find easier to apply.

\begin{corollary}\label{cor:BH-nonuniform-easy-to-use}
In the setting of \Cref{thm:BH-nonuniform}, the following inequality holds.
\begin{equation*}
\heis(G) \leq \prodval(G) + 20\cdot (n \cdot \max_{u, v} \{A_{u, v}\} \cdot \max_u \{p_u\})^{1/8} + \max_u \{p_u\},
\end{equation*}
where $n$ is the number of vertices in~$G$.
\end{corollary}
\begin{proof}
First, we bound the $\tr[A^2]$ term:
\begin{equation*}
\tr[A^2] = \sum_{u, v} A_{u, v} \cdot A_{v, u}
\leq \max_{v, u} \{A_{v, u}\} \cdot \sum_{u, v} A_{u, v}
=  \max_{v, u} \{A_{v, u}\} \cdot \sum_{v} 1
= \max_{v, u} \{A_{v, u}\} \cdot n.
\end{equation*}
Next, we bound the $\Vert p \Vert_2^2$ term:
\begin{equation*}
\Vert p \Vert_2^2
= \sum_{u} p_u^2 \leq \max_u \{p_u\} \cdot \sum_u p_u = \max_u \{p_u\}.
\end{equation*}
Substituting these bounds into \Cref{thm:BH-nonuniform} completes the proof.
\end{proof}

\subsection{Deriving the basic SDPs}\label{sec:sdp_proofs}

Now we will show to derive the basic SDP for \qmaxcut.
As a warm-up, we will recall a standard method for deriving the \maxcut SDP
\begin{equation}\label{eq:mcsdp-restated}
\mcsdp(G) = \max_{f:V\rightarrow S^{n-1}} \E_{(\bu, \bv) \sim E}[ \tfrac{1}{2} - \tfrac{1}{2} \langle f(\bu), f(\bv)\rangle].
\end{equation}
One way of deriving this SDP is as follows:
let $f:V\rightarrow \{-1, 1\}$ be an assignment to the vertices.
Consider the $n \times n$ matrix $M$ defined as $M(u, v) = f(u) \cdot f(v)$.
Then~$M$ is a real, PSD matrix such that $M(v, v) = 1$ for all $v \in V$. 
Furthermore, we can write the value of~$f$ in terms of~$M$ as
\begin{equation}\label{eq:objective-in-terms-of-M}
\E_{(\bu, \bv) \sim E}[ \tfrac{1}{2} - \tfrac{1}{2} M(\bu, \bv)].
\end{equation}
Now, we relax our problem and consider optimizing \Cref{eq:objective-in-terms-of-M}
over all real, PSD matrices~$M$ such that $M(v, v) = 1$.
Such an~$M$ can be written as the Gram matrix of a set of real vectors of dimension~$n$;
i.e. there is a function $f:V \rightarrow \R^{n}$ such that $M(u, v) = \langle f(u), f(v)\rangle$.
This yields the SDP in \Cref{eq:mcsdp-restated}.

\paragraph{The basic SDP for \qmaxcut.}
Now we show how to derive the basic SDP for \qmaxcut.
To begin, let $G = (V, E, w)$ be an $n$-vertex graph.
Let $\ket{\psi} \in (\C^2)^{V}$ be a quantum state.
Consider the set of $3n$ vectors
$X_u \ket{\psi}$, $Y_u \ket{\psi}$, $Z_u\ket{\psi}$,
where $P_u$ denotes the Pauli matrix~$P$ acting on qubit~$u$.
The Gram matrix of these vectors, denoted $M(\cdot,\cdot)$,
is the $3n \times 3n$ matrix whose rows and columns are indexed by
Pauli matrices $P_u$ such that
\begin{equation*}
M(P_u, Q_v) = \bra{\psi} P_u Q_v\ket{\psi}.
\end{equation*}
Using $M$,
we can express the energy of $\ket{\psi}$ as follows:
\begin{align}
\bra{\psi}H_G\ket{\psi}
 = \E_{(\bu, \bv) \sim E} \bra{\psi} h_{\bu, \bv}\ket{\psi}
& = \E_{(\bu, \bv) \sim E} \tfrac{1}{4} \cdot \bra{\psi}( I_{\bu} \ot I_{\bv} - X_{\bu} \ot X_{\bv} - Y_{\bu} \ot Y_{\bv} - Z_{\bu} \ot Z_{\bv})\ket{\psi}\nonumber\\
&= \tfrac{1}{4}\cdot\E_{(\bu, \bv) \sim E}[1 - M(X_{\bu}, X_{\bv}) - M(Y_{\bu}, Y_{\bv}) - M(Z_{\bu}, Z_{\bv})].\label{eq:new-objective}
\end{align}
Let us derive some constraints on this matrix:
\begin{enumerate}
\item \textbf{PSD:} $M$ is Hermitian and PSD.
\item \textbf{Unit length:} For each $P_u$, $M(P_u, P_u) = 1$.
\item \textbf{Commuting Paulis:} For each $P_u, Q_v$ such that $u \neq v$, $P_u$ commutes with~$Q_v$.
	This implies that 
	$
	M(P_u, Q_v) = M(Q_v, P_u)
	$
	and is therefore real because $M$ is Hermitian.
\item \textbf{Anti-commuting Paulis:} For each $P_u, Q_u$ such that $P \neq Q$, $P_u$ anti-commutes with $Q_u$.\label{item:anticommute}
	This implies that $
	M(P_u, Q_u) = - M(Q_u, P_u)
	$
	and therefore has no real part because $M$ is Hermitian.
\end{enumerate}
Now we relax our problem and consider optimizing \Cref{eq:new-objective}
over all matrices~$M$ that satisfy these four conditions.
This is a relaxation because not all matrices~$M$ correspond to Gram matrices of vectors of the form $P_u \ket{\psi}$.

Prior to stating the SDP, we perform one final simplification.
Given such an~$M$,
consider the matrix $M' = \tfrac{1}{2}(M + M^T)$.
This satisfies all four conditions, has the same energy as~$M$,
and moreover satisfies $M'(P_u, Q_u) = 0$ for $P \neq Q$.
We can therefore replace~\Cref{item:anticommute}
with this stronger condition, which implies that~$M'$ is real.
Thus, $M'$ is a real, symmetric $3n \times 3n$ PSD matrix,
so we can write it as the Gram matrix of a set of real vectors of dimension $3n$.
In other words, there are functions
\begin{equation*}
f_X, f_Y, f_Z: V \rightarrow \R^{3n}
\end{equation*}
such that $M'(P_u, Q_v) = \langle f_P(u) , f_Q(v)\rangle$.
Putting everything together, we have the following SDP.

\begin{proposition}[\qmaxcut SDP]\label{prop:sum-of-squares-equals-sdp}
Let $G = (V, E, w)$ be an $n$-vertex graph.
The value of the SDP for \qmaxcut can be written as
\begin{align}
\hsdp(G)
= \max~& \frac{1}{4} \cdot \E_{(\bu, \bv) \sim E}[1 - \langle f_X(\bu), f_X(\bv)\rangle- \langle f_Y(\bu), f_Y(\bv)\rangle - \langle f_Z(\bu), f_Z(\bv)\rangle], \label{eq:heis-first-sdp}\\
\mathrm{s.t.}~& \langle f_P(v), f_Q(v) \rangle = 0, \quad \forall v \in V,~P \neq Q \in \{X, Y, Z\},\nonumber\\
	&f_X, f_Y, f_Z:V \rightarrow S^{3n-1}.\nonumber
\end{align}
\end{proposition}

This SDP can also be viewed as the degree-2 relaxation for \qmaxcut
in the \emph{non-commutative Sum of Squares (ncSoS)} hierarchy.
We give a didactic treatment of this perspective in \Cref{sec:ncsos}
(which is not necessary to understand the rest of this paper).
We now further simplify the SDP relaxation for \qmaxcut
and derive the expression from \Cref{def:heis-sdp}.

\begin{proposition}[\qmaxcut SDP, simplified version; \Cref{def:heis-sdp}]\label{prop:equivalence}
\begin{equation}\label{eq:sdp-mc-heis}
\hsdp(G)
= \max_{f:V\rightarrow S^{n-1}} \E_{(\bu, \bv) \sim E}[ \tfrac{1}{4} - \tfrac{3}{4} \langle f(\bu), f(\bv)\rangle].
\end{equation}
\end{proposition}
\begin{proof}
We first show how to convert a solution for~\eqref{eq:heis-first-sdp} into a solution for~\eqref{eq:sdp-mc-heis}
without decreasing the value.
To begin, we can rewrite~\eqref{eq:heis-first-sdp} as
\begin{equation}\label{eq:split-up}
\tfrac{1}{3} \cdot  \E_{(\bu, \bv) \sim E}[ \tfrac{1}{4} - \tfrac{3}{4} \langle f_X(\bu), f_X(\bv)\rangle]
+\tfrac{1}{3} \cdot  \E_{(\bu, \bv) \sim E}[ \tfrac{1}{4} - \tfrac{3}{4} \langle f_Y(\bu), f_Y(\bv)\rangle]
+\tfrac{1}{3} \cdot  \E_{(\bu, \bv) \sim E}[ \tfrac{1}{4} - \tfrac{3}{4} \langle f_Z(\bu), f_Z(\bv)\rangle].
\end{equation}
Pick the term $P \in \{X, Y, Z\}$ with the largest value, and set $f = f_P$.
Then $f$ has value in \eqref{eq:sdp-mc-heis} at least the value of $f_X, f_Y, f_Z$ in \eqref{eq:split-up}.
The only caveat is that $f$ maps into $S^{3n-1}$ rather than $S^{n-1}$.
However, $f$ only outputs $n$ different vectors, so these can be represented in~$n$-dimensional space while preserving inner products.

Next, we reverse. Let $f:V\rightarrow S^{n-1}$ be a solution to \eqref{eq:sdp-mc-heis}.
We define
\begin{equation*}
f_X(v) = e_1 \otimes f(v),
\quad
f_Y(v) = e_2 \otimes f(v),
\quad
f_Z(v) = e_3 \otimes f(v),
\end{equation*}
where $e_1, e_2, e_3$ are standard basis vectors in $\R^3$.
Then $\langle f_P(v), f_Q(v)\rangle = 0$ for $P \neq Q$ because $\langle e_i, e_j\rangle = 0$ for $i \neq j$.
In addition, $\langle f_P(u), f_P(v) \rangle = \langle f(u), f(v)\rangle$.
Thus, the value of this assignment
\begin{align*}
\eqref{eq:heis-first-sdp}
& = \frac{1}{4} \cdot \E_{(\bu, \bv) \sim E}[1 - \langle f_X(\bu), f_X(\bv)\rangle- \langle f_Y(\bu), f_Y(\bv)\rangle - \langle f_Z(\bu), f_Z(\bv)\rangle]\\
& = \frac{1}{4} \cdot \E_{(\bu, \bv) \sim E}[1 - \langle f(\bu), f(\bv)\rangle- \langle f(\bu), f(\bv)\rangle - \langle f(\bu), f(\bv)\rangle]\\
& = \E_{(\bu, \bv) \sim E}[ \tfrac{1}{4} - \tfrac{3}{4} \langle f(\bu), f(\bv)\rangle].
\end{align*}
As a result, the value remains unchanged. This completes the proof.
\end{proof}

\subsection{Projection rounding}

We recall the performance of projection rounding, first stated in \Cref{eq:might-reference-later}.

\begin{theorem}[{\cite[Lemma 2.1]{BOV10}}]\label{thm:exact-formula-for-average-inner-product}
Let $-1 \leq \rho \leq 1$,
and let $u$ and $v$ be two $n$-dimensional unit vectors such that $\langle u, v\rangle = \rho$.
Let $\bZ$ be a random $k \times n$ matrix consisting of $kn$ i.i.d.\ standard Gaussians. Then
\begin{equation*}
F^*(k, \rho)
:= \E_{\bZ}\left\langle \frac{\bZ u}{\Vert \bZ u\Vert}, \frac{\bZ v}{\Vert \bZ v\Vert}\right\rangle
= \frac{2}{k}\left(\frac{\Gamma((k+1)/2)}{\Gamma(k/2)}\right)^2 \langle u, v\rangle \,_2F_1\left(1/2,1/2;k/2 + 1;\langle u, v\rangle^2\right),
\end{equation*}
where $_2F_1(\cdot, \cdot;\cdot;\cdot)$ is the Gaussian hypergeometric function.
\ignore{
\begin{equation*}
\E_{\boldf} \langle \boldf(u), \boldf(v)\rangle 
= \frac{2}{3}\left(\frac{\Gamma(2)}{\Gamma(3/2)}\right)^2\rho_{u, v}\cdot\,_2F_1\left(1/2,1/2;5/2;\rho_{u, v}^2\right),
\end{equation*}}
\end{theorem}

In the \qmaxcut SDP,
if an edge $(u, v)$ has value $\tfrac{1}{4} - \tfrac{3}{4} \rho$,
then projection rounding will produce a solution whose value on this edge
is $\tfrac{1}{4} - \tfrac{1}{4} F^*(3, \rho)$ in expectation.
Similarly, in the product state SDP,
an edge with value $\tfrac{1}{4} - \tfrac{1}{4}\rho$
will be rounded into a solution with value $\tfrac{1}{4} - \tfrac{1}{4} F^*(3, \rho)$ on this edge.
We can then define our approximation ratios as the worst case rounding over all values of~$\rho$.

\begin{definition}[Approximation ratios]\label{def:ratios}
The constant $\apk$ is defined as the solution to the minimization problem
\begin{equation*}
\apk = 
\min_{-1 \leq \rho < 1/3} \frac{\tfrac{1}{4} - \tfrac{1}{4} F^*(3, \rho)}{\tfrac{1}{4} - \tfrac{3}{4} \rho},
\end{equation*}
and the constant $\rpk$ is defined as the minimizing value of~$\rho$.
In addition, the constant $\abov$ is defined as the solution to the minimization problem
\begin{equation*}
\abov = 
\min_{-1 \leq \rho \leq 1} \frac{\tfrac{1}{4} - \tfrac{1}{4} F^*(3, \rho)}{\tfrac{1}{4} - \tfrac{1}{4} \rho},
\end{equation*}
and the constant $\rbov$ is defined as the minimizing value of~$\rho$.
\end{definition}

We note that the minimization for $\apk$ is only over $\rho \leq 1/3$, because when $\rho \geq 1/3$ the denominator is $\leq 0$.

\begin{proposition}[Formula for the optimum value]\label{prop:opt-formula}
Let $f_{\mathrm{opt}}:\R^n \rightarrow S^{k-1}$ be defined by $f_{\mathrm{opt}}(x) = x_{\leq k}/\Vert x_{\leq k} \Vert$,
where $x_{\leq k} = (x_1, \ldots, x_k)$. Then
\begin{equation*}
\E_{\bx \sim_\rho \by}\langle f_{\mathrm{opt}}(\bx), f_{\mathrm{opt}}(\by)\rangle = F^*(k, \rho).
\end{equation*}
\end{proposition}

\begin{proof}
Let $u, v \in \R^n$ be any unit vectors with $\langle u, v \rangle = \rho$.
Let $\bZ$ be an $n \times n$ matrix consisting of~$n^2$ i.i.d.\ standard Gaussians.
Then $\tfrac{1}{\sqrt{n}} \bZ u $ and $\tfrac{1}{\sqrt{n}}\bZ v$ are distributed as $\rho$-correlated Gaussians.
If $\Pi_{\leq k}$ is the $k \times n$ matrix which projects a vector down to its first $k$ coordinates, we have
\begin{align*}
\E_{\bx \sim_\rho \by}\langle f_{\mathrm{opt}}(\bx), f_{\mathrm{opt}}(\by)\rangle
&= \E_{\bZ} \langle f_{\mathrm{opt}}(\tfrac{1}{\sqrt{n}} \bZ u), f_{\mathrm{opt}}(\tfrac{1}{\sqrt{n}} \bZ v)\rangle\\
&= \E_{\bZ} \langle f_{\mathrm{opt}}(\bZ u), f_{\mathrm{opt}}(\bZ v)\rangle
= \E_{\bZ} \left\langle \frac{\Pi_{\leq k} \bZ u}{\Vert \Pi_{\leq k} \bZ u\Vert},
	\frac{\Pi_{\leq k} \bZ v}{\Vert \Pi_{\leq k}  \bZ v\Vert}\right\rangle.
\end{align*}
Now note that $\Pi_{\leq k} \bZ$ is distributed as a random $k \times n$ matrix consisting of $kn$ i.i.d.\ standard Gaussians.
Applying \Cref{thm:exact-formula-for-average-inner-product}, this is exactly equal to $F^*(k, \rho)$.
\end{proof}

\subsection{Fourier analysis on the hypercube}\label{sec:fourier-analysis}

We will review basic concepts in the Fourier analysis of Boolean functions.
See~\cite{OD14} for further details.

\begin{definition}[Fourier transform]
Let $f:\{-1, 1\}^n \rightarrow \R$ be a function.
Then it has a unique representation as a multilinear polynomial known as the \emph{Fourier transform}, given by
\begin{equation*}
f(x) = \sum_{S \subseteq [n]} \widehat{f}(S) \chi_S(x),
\end{equation*}
where $\widehat{f}(S)$ is a real coefficient called the \emph{Fourier coefficient},
and $\chi_S(x)$ is the monomial $\prod_{i \in S} x_i$.

We extend this definition to functions $f:\{-1, 1\}^n \rightarrow \R^k$ as follows:
let $f = (f_1, \ldots, f_k)$. Then
\begin{equation*}
f(x)
= (f_1(x), \ldots, f_k(x))
= \sum_{S \subseteq [n]} (\widehat{f}_1(S), \ldots, \widehat{f}_k(S)) \chi_S(x)
= \sum_{S \subseteq [n]} \widehat{f}(S) \chi_S(x),
\end{equation*}
where $\widehat{f}(S) = (\widehat{f}_1(S), \ldots, \widehat{f}_k(S))$ is a vector-valued Fourier coefficient.
\end{definition}

\begin{definition}[Variance and influences]
Let $f:\{-1, 1\}^n \rightarrow \R^k$. Then
\begin{equation*}
\E_{\bx\sim \{-1, 1\}^n} \Vert f(\bx)\Vert_2^2 = \sum_{S \subseteq [n]} \Vert \widehat{f}(S) \Vert_2^2.
\end{equation*}
Its \emph{variance} is the quantity
\begin{equation*}
\Var[f]
= \E_{\bx \sim \{-1, 1\}^n} \Vert f(\bx) - \E[f] \Vert_2^2
= \sum_{S \subseteq [n], S \neq \emptyset} \Vert \widehat{f}(S)\Vert_2^2.
\end{equation*}
Given a coordinate~$i$,
we define its \emph{influence} as
\begin{equation*}
\Inf_i[f]
= \sum_{S \subseteq [n], S \ni i} \Vert \widehat{f}(S)\Vert_2^2
= \sum_{S \subseteq [n], S \ni i} \sum_j  \widehat{f_j}(S)^2
= \sum_j \Inf_i[f_j].
\end{equation*}
We will also need truncated versions of these two measures:
\begin{equation*}
\Inf^{\leq m}_i[f] = \sum_{|S| \leq m, S \ni i} \Vert \widehat{f}(S)\Vert_2^2,
\qquad
\Var[f^{>m}] = \sum_{|S| > m} \Vert \widehat{f}(S) \Vert_2^2.
\end{equation*}
\end{definition}

\begin{proposition}[Only few noticeable coordinates]\label{prop:influence-bound}
Let $f:\{-1, 1\}^n \rightarrow B^k$.
Then there are at most $m/\delta$ coordinates~$i$ such that $\Inf^{\leq m}_i[f] \geq \delta$.
\end{proposition}
\begin{proof}
Let $N$  be the set of all such coordinates. Then
\begin{multline*}
|N| \cdot \delta
\leq \sum_{i \in N} \Inf^{\leq m}_i[f]
= \sum_{i \in N}\sum_{|S| \leq m, S \ni i} \Vert \widehat{f}(S)\Vert_2^2
= \sum_{|S| \leq m} |S \cap N| \cdot \Vert \widehat{f}(S)\Vert_2^2\\
\leq \sum_{|S| \leq m} m\cdot \Vert \widehat{f}(S)\Vert_2^2
\leq m \cdot \sum_{S \subseteq [n]} \Vert \widehat{f}(S)\Vert_2^2
= m \cdot \E_{\bx} \Vert f(\bx) \Vert_2^2
\leq m.
\end{multline*}
Rearranging this gives $|N| \leq m /\delta$.
\end{proof}

\begin{definition}[Correlated Boolean variables~{\cite[Definition 2.40]{OD14}}]
Given a fixed $x \in \{-1, 1\}^n$,
we say that $\by \in \{-1, 1\}^n$ is \emph{$\rho$-correlated to~$x$}
if each coordinate $\by_i$ is sampled independently according to the following distribution:
\begin{equation*}
\by_i =
\left\{\begin{array}{rl}
x_i & \text{with probability $\tfrac{1}{2} + \tfrac{1}{2}\rho$}\\
-x_i & \text{with probability $\tfrac{1}{2} - \tfrac{1}{2}\rho$}.
\end{array}\right.
\end{equation*}
In addition, we say that $\bx$ and $\by$ are \emph{$\rho$-correlated $n$-dimensional Boolean strings}
if $\bx$ is sampled from $\{-1, 1\}^n$ uniformly at random
and $\by$ is $\rho$-correlated to~$\bx$.
Note that for each $i$, $\E[\bx_i \by_i] = \rho$.
\end{definition}

\begin{definition}[Noise stability]
Let $f: \{-1, 1\}^n \rightarrow \R^k$, and let $-1 \leq \rho \leq 1$.
Given an input $x \in \{-1, 1\}^n$, we write
\begin{equation*}
\T_\rho f(x) = \E_{\substack{\text{$\by$ which is}\\\text{$\rho$-correlated to~$x$}}}[f(\by)] = \sum_{S \subseteq [n]} \rho^{|S|} \widehat{f}(S) \chi_S(x).
\end{equation*}
Then the \emph{Boolean noise stability} of $f$ at~$\rho$ is
\begin{equation*}
\Stab_\rho[f]
= \E_{\substack{\text{$(\bx, \by)$ $\rho$-correlated}\\\text{$n$-dim Boolean strings}}}\langle f(\bx), f(\by)\rangle
= \E_{\bx \sim \{-1, 1\}^n} \langle f(\bx), \T_\rho f(\bx)\rangle = \sum_{S \subseteq [n]} \rho^{|S|} \widehat{f}(S)^2.
\end{equation*}
This coincides with the Gaussian noise sensitivity of~$f$.
To see this, note that the Fourier expansion allows us to extend~$f$'s domain to all of $\R^n$.
Then
\begin{equation*}
\E_{\bx \sim_\rho \by}\langle f(\bx), f(\by) \rangle = \sum_{S \subseteq [n]} \rho^{|S|} \widehat{f}(S)^2.
\end{equation*}
Hence, we use $\Stab_\rho[f]$ as for both notions.
\end{definition}

\section{Integrality gaps}\label{sec:integrality-gaps}

In this section, we prove \Cref{thm:integrality-gap-heisenberg,thm:integrality-gap-prod}, which we restate here.

\begin{theorem}[Integrality gap for the \qmaxcut SDP; \Cref{thm:integrality-gap-heisenberg} restated]
\label{thm:integrality-gap-heisenberg-restated}
Assuming \cref{conj:vector-borell-intro},
the \qmaxcut semidefinite program $\hsdp(G)$ has integrality gap $\alpha_{\mathrm{GP}}$.
\end{theorem}

\begin{theorem}[Integrality gap for product state SDP; \Cref{thm:integrality-gap-prod} restated]
\label{thm:integrality-gap-prod-restated}
Assuming \cref{conj:vector-borell-intro},
the product state semidefinite program $\prodsdp(G)$ has integrality gap $\abov$.
\end{theorem}

Recalling \Cref{def:integrality-gap}, our goal is to compute
$$\inf_{\text{instances }\mathcal I\text{ of }\mathcal P}\left\{\frac{\mathrm{OPT}(\mathcal I)}{\mathrm{SDP}(\mathcal I)}\right\},$$
where $\mathcal P$ is the problem of either computing the value or the product state value of a \qmaxcut instance.
To upper bound this quantity, we construct a specific instance $\mathcal I$ and give an upper-bound for $\mathrm{OPT}(\mathcal I)$ and a lower-bound for $\mathrm{SDP}(\mathcal I)$. Note that for the product state case, we only optimize over product states, whereas for the general \qmaxcut we consider all quantum states. However, the specific instance $\mathcal I$ we consider will correspond to a graph of high degree, and for such graphs Brandao and Harrow \cite{BH16} show it suffices to consider product states.

The instance we use for both integrality gaps is the
$\rho$-correlated Gaussian graph.
This was also used as an integrality gap for the \maxcut problem in the work of~\cite{OW08},
and we will follow their proof closely.
As in their proof, we will have to deal with the technicality that the Gaussian graph is actually an infinite graph,
and so our final integrality gap instance will
involve discretizing the Gaussian graph to produce a finite graph.

\subsection{The Gaussian graph as an integrality gap}
To begin, we provide a lower-bound for the \qmaxcut SDP.

\begin{lemma}[\qmaxcut SDP Lower Bound]\label{lem:heis-sdp-lower}
$$\hsdp(\mathcal G_\rho^n) \geq \tfrac{1}{4} - \tfrac{3}{4} \rho - O(\sqrt{\log n/n}).$$
\end{lemma}
\begin{proof}
Consider the feasible solution $f_\mathrm{ident}(x) = x/\Vert x \Vert$ to the \qmaxcut SDP.
It has value
\begin{equation}\label{eq:value-of-ident}
\E_{\bx \sim_{\rho} \by}\left[ \frac{1}{4} - \frac{3}{4}\left\langle \frac{\bx}{\Vert \bx \Vert}, \frac{\by}{\Vert \by \Vert}\right\rangle\right]
= \frac{1}{4} - \frac{3}{4}\E_{\bx \sim_{\rho} \by}\left\langle \frac{\bx}{\Vert \bx \Vert}, \frac{\by}{\Vert \by \Vert}\right\rangle.
\end{equation}
Intuitively, we expect $\langle \bx/ \Vert \bx \Vert, \by/ \Vert \by \Vert \rangle$ to roughly be equal to~$\rho$,
at least when~$n$ is large.
Formally, we will use the inequality
\begin{equation*}
\E_{\bx \sim_{\rho} \by}\left\langle \frac{\bx}{\Vert \bx \Vert}, \frac{\by}{\Vert \by \Vert}\right\rangle \leq \rho + O(\sqrt{\log n/n}),
\end{equation*}
which was shown in the proof of~\cite[Theorem~$4.3$]{OW08}.
This implies that
\begin{equation*}
\eqref{eq:value-of-ident} \geq  \tfrac{1}{4} - \tfrac{3}{4} \rho - O(\sqrt{\log n/n}).
\end{equation*}
As this lower-bounds the value of $f_{\mathrm{ident}}$, it also lower-bounds the value of the SDP.
\end{proof}
An essentially identical proof also yields the following lemma, which gives a lower-bound for the product state SDP.
\begin{lemma}[Product State SDP Lower Bound]
$$\prodsdp(\ggraph_\rho^n) \geq \tfrac{1}{4} - \tfrac{1}{4} \rho - O(\sqrt{\log n/n}).$$
\end{lemma}

On the other hand, assuming \Cref{conj:vector-borell-intro},
the optimal product state assignment is given by $f_{\mathrm{opt}}(x) = x_{\leq 3} / \Vert x_{\leq 3} \Vert$,
and so by~\Cref{prop:opt-formula} the product state value can be computed exactly as
\begin{equation*}
\prodval(\ggraph^n_\rho)  = \tfrac{1}{4} - \tfrac{1}{4} F^*(3, \rho).
\end{equation*}
Not only that, $\ggraph_\rho^n$ is a weighted, regular graph of infinite degree.
Thus, we should now be able to apply \Cref{thm:BH} (or one of its nonuniform analogues)
to show that its maximum energy is exactly equal to its product value.
Strictly speaking, the maximum energy $\heis(\ggraph^n_\rho)$ is not well-defined because $\ggraph^n_\rho$ is an infinite graph. However, we will define $\heis(\ggraph^n_\rho)$ to be $\prodval(\ggraph^n_\rho)$ and show in the following section that these quantities are indeed approximately equal in the discretized graph.

\subsection{Discretizing the Gaussian graph}

The following lemma shows $\ggraph^n_\rho$ can be discretized  with a negligible loss in value. The proof is given in \Cref{sec:card-reduc}. 

\begin{lemma}[Graph Discretization]\label{lem:discretization}
Let $G = \mathcal G^n_\rho$ be the $\rho$-correlated Gaussian graph.
Then for every $\eps > 0$, there exists a finite, weighted graph~$G'$ such that
\begin{IEEEeqnarray*}{rClrCl}
\hsdp(G') &\geq& \hsdp(G) - \eps, \quad&\quad\heis(G') &\leq& \heis(G) + \eps, \\
\prodsdp(G') &\geq& \prodsdp(G) - \eps, \quad&\quad \prodval(G') &\leq& \prodval(G) + \eps.
\end{IEEEeqnarray*}
\end{lemma}

With the instance in hand, we now prove that it yields our desired integrality gaps. 

\begin{proof}[Proof of \Cref{thm:integrality-gap-heisenberg-restated,thm:integrality-gap-prod-restated}]
We start with the \qmaxcut SDP.
Assume \cref{conj:vector-borell-intro}.
Then combining \Cref{lem:heis-sdp-lower,lem:discretization} and taking the dimension~$n$ suitably large,
there exists a graph~$G$ such that
\begin{equation*}
\frac{\heis(G)}{\hsdp(G)} \leq \frac{\tfrac{1}{4} - \tfrac{1}{4}F^*(3,\rho)}{\tfrac{1}{4} -\tfrac{3}{4}\rho} + \eps,
\end{equation*}
for each $\eps > 0$. Taking the infimum over~$\eps$, the integrality gap of $\hsdp$ is at most
\begin{equation*}
\frac{\tfrac{1}{4} - \tfrac{1}{4}F^*(3,\rho)}{\tfrac{1}{4} - \tfrac{3}{4}\rho}.
\end{equation*}
This is minimized by $\rho = \rgp$, in which case it is equal to $\agp$.
Hence, the integrality gap of the \qmaxcut SDP is at most~$\agp$.
On the other hand, the integrality gap is at least~$\agp$
because the GP algorithm shows there always exists a solution of value at least $\agp \cdot \hsdp(G)$.
As a result, the integrality gap is exactly~$\agp$, concluding the proof.

The case of the product state SDP follows by a similar argument.
\end{proof}

\section{Algorithmic gap for the BOV algorithm}\label{sec:algorithmic-gap}

The main goal of this section is to prove \Cref{thm:algo-gap}, which we restate here.
We note that it does not require assuming \cref{conj:vector-borell-intro}.

\begin{theorem}[Algorithmic gap for product state SDP; \Cref{thm:algo-gap} restated]\label{thm:algo-gap-restated}
The Briët-Oliveira-Vallentin algorithm has algorithmic gap $\abov$.
\end{theorem}

Recall that the algorithmic gap is the quantity
$$\inf_{\text{graphs }G}\left\{\frac{A_{\mathrm{BOV}}(G)}{\prodval(G)}\right\},$$
where $A_{\mathrm{BOV}}(G)$ is the average value of the product state output by the BOV algorithm on graph~$G$.
This is at least~$\abov$ by~\cref{thm:bov-thm},
and so we need to show that it is also at most~$\abov$,
which entails finding a graph~$G$ in which the BOV algorithm outputs a solution of value $\abov \cdot \prodval(G)$.
Our construction is based on a classic algorithmic gap instance for the Goemans-Williamson SDP
called the \emph{noisy hypercube graph}, essentially due to Karloff~\cite{Kar99} (cf.\ the exposition in \cite{OD08}). 

The BOV algorithm solves the product state SDP and rounds its solution using projection rounding.
The SDP is only guaranteed to return \emph{some} optimal (or near-optimal) solution,
but we are free to choose which of these optimal solutions to provide to the BOV algorithm (see \cite{OD08} for more details).
As in the construction of the integrality gaps in \Cref{sec:integrality-gaps},
the algorithmic gap instance and the optimal SDP solution are motivated by the fact that the BOV algorithm performs worst
on edges $(u, v)$ where the SDP vectors have inner product $\langle f_{\mathrm{SDP}}(u), f_{\mathrm{SDP}}(v)\rangle = \rbov$.
The graph is an analogue of the $\rho$-correlated sphere graph on the Boolean hypercube,
known as the noisy hypercube.

\begin{definition}[Noisy hypercube graph]
Let $n$ be a positive integer and $-1 \leq \rho \leq 1$.
We define the \emph{$\rho$-noisy hypercube} to be the graph $\hgraph^n_\rho$ with vertex set $\{-1, 1\}^n$
in which a random edge $(\bx, \by)$ is distributed as two $\rho$-correlated Boolean strings.
\end{definition}

As defined, the noisy hypercube does not correspond to a legitimate \qmaxcut instance,
as it contains self-loops.
For now, we will analyze the noisy hypercube as if this is not an issue,
and we will remove the self-loops at the end of the section.

The SDP solution we will consider is the ``identity solution'', i.e. the function $f_{\mathrm{ident}}:\{-1, 1\}^n \rightarrow S^{n-1}$
defined by $f_{\mathrm{ident}}(x) = \tfrac{1}{\sqrt{n}} x$ for each $x \in \{-1, 1\}^n$. The following three lemmas (\ref{prod-sdp-value}, \ref{lem:opt-sdp-value}, \ref{lem:proj-rounding-value}) establish a upper bound on $A_{\mathrm{BOV}}(\hgraph^n_\rho)$.

\begin{lemma}[Value of the SDP solution]\label{prod-sdp-value}
The feasible SDP solution $f_{\mathrm{ident}}(x)$ has value $\tfrac{1}{4} - \tfrac{1}{4} \rho$.
\end{lemma}
\begin{proof}
We compute
\begin{equation*}
\E_{\substack{\text{$(\bx, \by)$ $\rho$-correlated}\\\text{$n$-dim Boolean strings}}}[\langle f_{\mathrm{ident}}(\bx), f_{\mathrm{ident}}(\by)\rangle]
= \tfrac{1}{n}\E_{\bx, \by} \langle \bx, \by\rangle
= \tfrac{1}{n} \sum_{i=1}^n \E_{\bx, \by}[\bx_i \by_i]
= \rho.
\end{equation*}
As a result, the value of $f_{\mathrm{ident}}$ is
\begin{equation*}
\E_{\bx, \by}[\tfrac{1}{4} - \tfrac{1}{4}\langle f_{\mathrm{ident}}(\bx), f_{\mathrm{ident}}(\by)\rangle] = \tfrac{1}{4} - \tfrac{1}{4} \rho.\qedhere
\end{equation*}
\end{proof}

Our next lemma shows that $f_{\mathrm{ident}}$ achieves the optimal SDP value,
and so it is fair for the BOV algorithm to receive it as a solution to the product state SDP.

\begin{lemma}[Value of the SDP]\label{lem:opt-sdp-value}
The value of the SDP is $\prodsdp(\hgraph^n_\rho) = \tfrac{1}{4} - \tfrac{1}{4}\rho$.
As a result, $f_{\mathrm{ident}}$ is an optimal SDP solution.
\end{lemma}
\begin{proof}
The function $f_{\mathrm{ident}}$ is a feasible solution with value $\tfrac{1}{4}-\tfrac{1}{4}\rho$,
and so it suffices to show that $\prodsdp(\hgraph^n_\rho) \leq \tfrac{1}{4} - \tfrac{1}{4}\rho$.
This entails showing that
\begin{equation*}
\E_{\substack{\text{$(\bx, \by)$ $\rho$-correlated}\\\text{$n$-dim Boolean strings}}}[\tfrac{1}{4} - \tfrac{1}{4}\langle f(\bx), f(\by)\rangle]
\leq \tfrac{1}{4} - \tfrac{1}{4}\rho,
\end{equation*}
for all functions $f:\{-1, 1\}^n \rightarrow S^{N-1}$, where $N = 2^n$ is the number of vertices in $\hgraph^n_\rho$.
Equivalently, we will show $\E_{\bx, \by}\langle f(\bx), f(\by)\rangle \geq \rho$ for all such~$f$.

Write $f = (f_1,\dots,f_N)$ where $f_i : \{-1,1\}^n \rightarrow \mathbb R$ for each~$i$.
Then, 
\begin{equation*}
\E_{\substack{\text{$(\bx, \by)$ $\rho$-correlated}\\\text{$n$-dim Boolean strings}}}\langle f(\bx), f(\by)\rangle
= \sum_{i=1}^N \E_{\bx, \by} [f_i(\bx) f_i(\by)]
\geq \sum_{i=1}^N \rho \cdot \E_{\bx} [f_i(\bx)^2]
= \rho \cdot \E_{\bx}\left[\sum_{i=1}^N f_i(\bx)^2\right]
= \rho.
\end{equation*}
The inequality here is due to \Cref{prop:stab-bound}
and we defer it to the appendix. This completes the proof.
\end{proof}

Given $f_{\mathrm{opt}} = f_{\mathrm{ident}}$, the BOV algorithm performs projection rounding and outputs the solution.
The following lemma shows the value of the output.

\begin{lemma}[Value of the projection rounding]\label{lem:proj-rounding-value}
Given the SDP solution $f_{\mathrm{ident}}$,
projection rounding will produce a random product state whose average value is at most
\begin{equation*}
\tfrac{1}{4} - \tfrac{1}{4} F^*(3, \rho) + O(\sqrt{\log(n)/n}).
\end{equation*}
\end{lemma}
\begin{proof}
By the Chernoff bound, 
\begin{equation}\label{eq:totally-sick-chernoff-bound}
\Pr_{\substack{\text{$(\bx, \by)$ $\rho$-correlated}\\\text{$n$-dim Boolean strings}}}[\langle \bx, \by\rangle = \rho \cdot n \pm O(\sqrt{n \log n})]
\leq O(1/\sqrt{n}).
\end{equation}
Let $(x, y)$ be an edge in $\hgraph^n_\rho$ such that $\langle x, y\rangle = \rho \cdot n \pm O(\sqrt{n \log n})$.
Then
\begin{equation*}
\rho_{x, y} := \langle f_{\mathrm{ident}}(x), f_{\mathrm{ident}}(y)\rangle = \rho \pm O(\sqrt{\log n/n}).
\end{equation*}
On this edge, the random product state produced by projection rounding has average value $\tfrac{1}{4} - \tfrac{1}{4} F^*(3, \rho_{x, y})$
due to \Cref{thm:exact-formula-for-average-inner-product}.
Since $\rho_{x, y}$ is within $O(\sqrt{\log n / n})$ of~$\rho$, and $F^*(3, \cdot)$ is Lipschitz by \Cref{lem:inner_prod_lipschitz},
we can bound the average value of this edge by
\begin{equation*}
\tfrac{1}{4} - \tfrac{1}{4} F^*(3, \rho) + O(\sqrt{\log n/n}).
\end{equation*}
As for the remaining edges, we can trivially bound the average value of each by~$1$,
and this contributes an extra $O(1/\sqrt{n})$
to the total value of the product state due to \Cref{eq:totally-sick-chernoff-bound}.
\end{proof}

Now, we compute the optimal product state value for $\hgraph^n_\rho$. To do so, we consider a family of product state solutions $f:\{-1, 1\}^n \rightarrow S^2$ we call \emph{embedded dictators},
in which there exists an $i \in [n]$ such that $f(x) = (x_i, 0, 0)$ for all $x$.

\begin{lemma}[Value of embedded dictators]\label{lem:embedded-dictators}
Embedded dictators $f(x) = (x_i, 0, 0)$ achieve value $\tfrac{1}{4} - \tfrac{1}{4}\rho$.
Hence, the product state value of $\hgraph_\rho^n$ is 
$\prodval(\hgraph_\rho^n) = \tfrac{1}{4}-\tfrac{1}{4}\rho$.
\end{lemma}
\begin{proof}
Consider the product state corresponding to the function $f_{\mathrm{dict}}:\{-1,1\}^n \rightarrow S^2$
given by $f_{\mathrm{dict}}(x) = (x_i, 0, 0)$.
It has value
\begin{equation*}
\tfrac{1}{4} - \tfrac{1}{4}\E_{\substack{\text{$(\bx, \by)$ $\rho$-correlated}\\\text{$n$-dim Boolean strings}}}[\langle f_{\mathrm{dict}}(\bx), f_{\mathrm{dict}}(\by)\rangle]
= \tfrac{1}{4} - \tfrac{1}{4}\E_{\bx \sim_\rho \by}[\bx_i \by_i]
= \tfrac{1}{4} - \tfrac{1}{4} \rho.
\end{equation*}
This shows the product state value is at least $\tfrac{1}{4} - \tfrac{1}{4} \rho$.
It is also at most $\tfrac{1}{4} - \tfrac{1}{4} \rho$ because this is the value of the product state SDP,
which is an upper bound on the product state value.
\end{proof}

\begin{proof}[Proof of \Cref{thm:algo-gap-restated}]
To begin, we will remove the self-loops from $\hgraph^n_\rho$,
which have total weight
\begin{equation*}
w_{\mathrm{loops}} = \Pr_{\substack{\text{$(\bx, \by)$ $\rho$-correlated}\\\text{$n$-dim Boolean strings}}}[\bx = \by] = \left(\tfrac{1}{2} + \tfrac{1}{2}\rho\right)^n.
\end{equation*}
Let $\calH'$ be the graph with vertex set $\{-1, 1\}^n$ in which a random edge $(\bx, \by)$ is distributed as two $\rho$-correlated Boolean strings,
conditioned on $\bx \neq \by$.
Then for each edge $(x, y)$, if $w(x, y)$ is its weight in $\hgraph^n_\rho$ and $w'(x, y)$ is its weight in $\calH'$,
we have that $w'(x, x) = 0$, and
\begin{equation*}
w'(x, y) = \frac{1}{1 - w_{\mathrm{loops}}} \cdot w(x, y)
\end{equation*}
for $x \neq y$.
Consider an SDP solution $f:\{-1, 1\}^n \rightarrow S^{N-1}$.
It has value~$0$ on each self-loop in $\hgraph^n_\rho$.
As a result, if it has value~$\nu$ in $\hgraph^n_\rho$,
then it has value $\nu / (1 - w_{\mathrm{loops}})$ in $\calH'$.
This argument applies to the value of product states as well.

In summary, all values are ``scaled up'' by a factor of $1/(1-w_{\mathrm{loops}})$ in $\calH'$.
This implies that $f_{\mathrm{ident}}$ is still an optimal SDP solution, and by \Cref{lem:proj-rounding-value} and \Cref{lem:embedded-dictators}, the ratio of the average value of the resulting solution to $\prodval(\calH')$ is
\begin{equation*}
\frac{\tfrac{1}{4} - \tfrac{1}{4} F^*(3, \rho) + O(\sqrt{\log(n)/n})}{\tfrac{1}{4} - \tfrac{1}{4}\rho}.
\end{equation*}
Taking an infimum over $n$, we can upper bound the algorithmic gap by
\begin{equation*}
\frac{\tfrac{1}{4} - \tfrac{1}{4} F^*(3, \rho)}{\tfrac{1}{4} - \tfrac{1}{4}\rho}.
\end{equation*}
This is minimized at $\rho = \rbov$, in which case it is $\abov$,
matching the approximation ratio of the BOV algorithm.
\end{proof}
\ignore{
\section{A dictator test for the product state value}

Now we show that the noisy hypercube serves as a dictatorship test for functions of the form $f:\{-1, 1\}^n \rightarrow B^k$.
Informally, this means that if $f$ is an embedded dictator, it should have high value,
and if it is ``far'' 'from a dictator (in the sense that if it has no ``notable'' input coordinates),
then it should have low value.
This will be an important ingredient in our Unique-Games hardness proof in~XXX.

We have already shown that embedded dictators achieve value $\tfrac{1}{4} - \tfrac{1}{4} \rho$ on the noisy hypercube in~XXX.
Now we will upper-bound the value that functions ``far'' from dictators achieve.
We will show that their value, up to small error, is at most the optimum product state value on the Gaussian graph
$\ggraph^n_\rho$, which we have shown to be $\tfrac{1}{4} - \tfrac{1}{4} F^*(k, \rho)$.

\begin{theorem}[Dictatorship test soundness]\label{def:embedded-dic}
Let $-1 < \rho \leq 0$.
Then for any $\epsilon > 0$,
there exists a small enough $\delta = \delta(\epsilon, \rho) > 0$
and large enough $m = m(\epsilon, \rho) \geq 0$ such that the following is true.
Let $f: \{-1, 1\}^n \rightarrow B^k$ satisfy $\Inf_i^{\leq k}[f] \leq \delta$ for all $i \in [n]$.
Then  $\Stab_\rho[f] \geq F^*(k, \rho) - \eps$.
In other words, the value of~$f$ on the noisy hypercube $\hgraph^n_\rho$ is at most
$
\tfrac{1}{4} - \tfrac{1}{4} F^*(k, \rho) + \eps.
$
\end{theorem}

The $k = 1$ case is the Majority is Stablest theorem of~\cite{MOO10},
which serves as the soundness case for the $\maxcut$ dictatorship test;
our theorem generalizes Majority is Stablest to larger values of~$k$.
The proof follows the same outline as the proof of Majority is Stablest:
we apply an ``invariance principle'' to exchange $f$'s Boolean inputs
with Gaussians of the same mean and variance.
Then we use our results on the noise stability of functions in Gaussian space
to upper-bound the value of~$f$. 
The invariance principle we will use is the following one due to Isaksson and Mossel~\cite{IM12},
which applies to vector-valued functions.

\begin{theorem}[Vector-valued invariance principle \cite{IM12}]\label{thm:vector-invariance}
Fix $\tau, \gamma \in (0,1)$ and set $d = \tfrac{1}{18}\log\tfrac{1}{\tau} / \ln(2)$. Let $f = (f_1,\dots,f_k)$ be a $k$-dimensional multilinear polynomial such that $\Var[f_j] \leq 1$, $\Var[f_j^{> d}] < (1-\gamma)^{2d}$, and $\Inf^{\leq d}_i[f_j] \leq \tau$ for each $j \in [k]$ and $i \in [n]$. Let $\bx$ be a uniformly random string over $\{-1,1\}^n$
and $\by$ be an $n$-dimensional standard Gaussian random variable.
Furthermore, let $\Psi : \mathbb R^k \rightarrow \mathbb R$ be Lipschitz continuous with Lipschitz constant $A$. Then,
$$|\E[\Psi(f(\bx))] - \E[\Psi(f(\by))]| \leq C_kA \tau^{\gamma/(18 \ln 2)}$$
where $C_k$ is a parameter depending only on $k$.
\end{theorem}

\begin{proposition}\label{prop:zeta-lipschitz}
Let $\mathrm{dist}_{B^k}:\R^k\rightarrow \R$ be the function defined by
\begin{equation*}
\mathrm{dist}_{B^k}(y)
 = \left\{\begin{array}{cl}
	0 & \text{if $\Vert y \Vert_2 \leq 1$,}\\
	\Vert y - \tfrac{y}{\Vert y \Vert} \Vert_2 & \text{otherwise,}
\end{array}\right.
\end{equation*}
which computes the distance of a point $y \in \mathbb R^k$ from $B^{k}$.
Then $\zeta$ is $1$-Lipschitz.
\end{proposition}
\begin{proof}
It suffices to show that $\zeta(y + z) \leq \zeta(y) + \Vert z \Vert_2$ for all $y, z \in \R^n$.
This is clearly true if $\Vert y + z \Vert_2 \leq 1$, as $\zeta(y+1) = 0$ in that case.
On the other hand, when $\Vert y + z \Vert_2 \geq 1$,
\begin{equation*}
\zeta(y+z)
= \Vert y + z- \tfrac{y+z}{\Vert y+z \Vert} \Vert_2
= \Vert y 
\end{equation*}
\end{proof}

\begin{theorem}[Vector Valued Majority is Stablest]\label{thm:vec-mis}
Fix $\rho \in (-1,0]$. Then for any $\epsilon > 0$, there exists small enough $\delta = \delta(\epsilon, \rho)$ and large enough $k = k(\epsilon, \rho) \geq 0$ such that if $f : \{-1,1\}^n \rightarrow B^{m-1}$ is any function satisfying
$$\Inf^{\leq k}_i[f] = \sum_{j=1}^m \Inf^{\leq k}_i[f_j] \leq \delta \text{ for all }i = 1, \dots, n$$
then
$$\E_{x \sim_\rho y} [f(x)f(y)] = \Stab_\rho[f] \geq F^*(m,\rho) - \epsilon$$
\end{theorem}

\begin{proof}
Throughout, we'll use $\bx$ to denote string in $\{-1,1\}^n$ and $\by$ to denote a vector in $\mathbb R^n$.
Since the statement of  the vector-valued invariance principle (\Cref{thm:vector-invariance}) requires a function with low high-degree variance, we consider $g = \T_{1-\gamma} f$. Then for each $j \in [m]$,
$$\Var[g_j^{\geq d}] = \sum_{|S| \geq d} (1-\gamma)^{2|S|}\widehat {f}_j(S)^2 \leq (1-\gamma)^{2d} \Var[f_j^{\geq d}] \leq (1-\gamma)^{2d}.$$
Also, $\Inf_i[g_j] \leq \Inf_i[f_j]$. Next, we bound the error in the quantity $\Stab_\rho[f]$ when we consider $g$ in place $f$.
\begin{align}
\ABS{\Stab_\rho[f] - \Stab_\rho[g]} &= \ABS{\sum_S \rho^{|S|} \norm{\widehat{f}(S)}_2^2 - \sum_S (\rho(1-\gamma)^2)^{|S|} \norm{\widehat{f}(S)}_2^2}\nonumber\\
&=\sum_S \ABS{\rho^{|S|}\big(1 - (1-\gamma)^{2|S|}\big)} \norm{\widehat{f}(S)}_2^2\nonumber\\
&\leq \rho(1-(1-\gamma)^2) \sum_S\norm{\hat f(S)}_2^2\nonumber\\
&\leq \rho(1-(1-\gamma)^2) \E[\norm{f}_2^2]\label{eq:bound1}
\end{align}
$f$ has range $B^{m-1}$ and thus $\E[\norm{f}_2^2] \leq 1$. By choosing $\gamma(\epsilon, \rho)$ to be sufficiently small we can bound this quantity by $\epsilon/2$.
Next, since $\Stab_{\rho}$ is a function of the coefficients of the polynomial corresponding to $f$, we have that $\Stab_\rho[g(\bx)] = \Stab_\rho[g(\by)]$.
However, although $g(x) \in B^{k}$ for all $x \in \{-1, 1\}^n$ (since $f(x) \in B^{k}$ by assumption),
$g(y)$ may take values outside the unit ball for $y \in \R^n$.
This prevents us from directly applying \Cref{thm:k-dim-borell} to~$g$.
We'll instead apply the theorem to the function $g':\R^n \rightarrow B^k$ defined as
$$g'(y) = 
\left\{\begin{array}{cl}
g(y) & \text{if $g(y) \in B^{k}$,}\\
\frac{g(y)}{\NORM{g(y)}} & \text{otherwise.}
\end{array}\right.$$
Applying \Cref{thm:k-dim-borell} yields that $\Stab_{\rho}[g'] \geq F^\star(k,\rho)$. It remains to bound $\ABS{\Stab_{\rho}[g] - \Stab_{\rho}[g']}$.
\begin{align*}
\ABS{\Stab_\rho[g] - \Stab_\rho[g']} &= \ABS{\langle g, \T_\rho g\rangle - \langle g', \T_\rho g'\rangle}\\
&= \ABS{\langle g, \T_\rho g\rangle - \langle g', \T_\rho g\rangle + \langle g', \T_\rho g\rangle - \langle g', \T_\rho g'\rangle}\\
&= \ABS{\langle g - g', \T_\rho g\rangle +\langle g', \T_\rho g - \T_\rho g'\rangle}\\
&= \ABS{\langle g - g', \T_\rho g\rangle} +\ABS{\langle g', \T_\rho g - \T_\rho g'\rangle}\\
&\leq \norm{g - g'}_2\norm{\T_\rho g}_2 + \norm{g'}_2\norm{\T_\rho g - \T_\rho g'}_2 \tag{by Cauchy-Schwartz}\\
&\leq \norm{g - g'}_2\norm{g}_2 + \norm{g'}_2\norm{g - g'}_2 \tag{because $\T_\rho$ is a contraction}\\
&= \norm{g - g'}_2(\norm{g}_2 + \norm{g'}_2) \\
&\leq 2\norm{g - g'}_2. \tag{because $\Vert g' \Vert_2 \leq \Vert g \Vert_2 \leq 1$}
\end{align*}
Now we just need to bound $\norm{g - g'}_2$ by $\epsilon/4$.
Note that $\norm{g - g'}_2^2 = \E[\zeta(g(y))]$, where $\zeta$ is the 1-Lipschitz function from \Cref{prop:zeta-lipschitz}.
Let $\delta = (\epsilon^2/(36\cdot C_m))^{18/\gamma}$. We can now apply \Cref{thm:vector-invariance} with $\Psi = \zeta$ and $\tau = \delta$, which yields
$$\E_{\by}[\zeta(g)] = |\E_{\bx}[\zeta(g(\bx))] - \E_{\by}[\zeta(g(\by))]| \leq \epsilon^2/36$$
Then, $\norm{g-g'}_2 \leq \sqrt{\epsilon^2/36} = \epsilon/6$ and $\ABS{\Stab_\rho[g] - \Stab_\rho[g']} \leq \epsilon/2$. Combining this with our bound of $\epsilon/2$ for \Cref{eq:bound1} concludes the proof.
\end{proof}

\section{Unique Games hardness of the quantum Heisenberg model}

\begin{definition}[Unique Label Cover]\label{def:unique-games}
The Unique Label Cover problem, $\mathcal L(V, W, E, [M], \{\sigma_{v,w}\}_{(v,w) \in E})$ is defined as follows: Given is a bipartite graph with left side vertices $V$, right side vertices $W$, and a set of edges $E$. The goal is to assign one `label' to every vertex of the graph, where $[M]$ is the set of allowed labels. The labeling is supposed to satisfy certain constraints given by bijective maps $\sigma_{v,w} : [M] \rightarrow [M]$. There is one such map for every edge $(v,w) \in E$. A labeling $L : V \cup W \rightarrow [M]$ `satisfies' an edge $(v,w)$ if
$$\sigma_{v,w}(L(w)) = L(v)$$
\end{definition}

\begin{conjecture}[Unique Games Conjecture \cite{Kho02}]
For any $\eta, \gamma > 0$, there exists a constant $M = M(\eta, \gamma)$ such that it is $\NP$-hard to distinguish whether a Unique Label Cover problem with label set size $M$ has optimum at least $1 - \eta$ or at most $\gamma$.
\end{conjecture}

Our hardness result is stated as follows.

\begin{theorem}[UG-Hardness of Approximating Quantum \maxcut]\label{thm:ug-hardness}
For any $\rho \in [-1,0)$ and $\epsilon > 0$, there exists an instance of the Heisenberg Hamiltonian such that deciding whether the maximum energy is greater than $\frac{1- \rho}{4} - \epsilon$ or less than $\frac{1 - F^\star(3, \rho)}{4} + \epsilon$ is Unique Games-Hard. In more standard notation, we say that it is UG-Hard to $(\frac{1 - F^\star(3,\rho)}{4} + \epsilon, \frac{1- \rho}{4} - \epsilon)$-approximate Quantum \maxcut.
\end{theorem}

\begin{remark}
$\frac{1 - F^\star(3,\rho)}{4}$ is exactly the product state value obtained by the rounding algorithm of \cite{GP19} on the integrality gap instance in \Cref{sec:integrality-gaps}.
\end{remark}

\begin{remark}
Minimizing the ratio $\frac{1-\rho}{4}/\frac{1-F^\star(3,\rho)}{4}$ w.r.t $\rho$ yields $\alpha_\mathrm{BOV}$, which shows \Cref{thm:main-inapprox}.
\end{remark}

\begin{theorem}\label{thm:quant-ug-hardness}
Fix a CSP over domain $S^{m-1}$ with predicate set $\Psi$. Suppose there exists an $(\alpha,\beta)$-Embedded-Dictator-vs.-No-Notables test using predicate set $\Psi$. Then for all $\delta > 0$, it is ``UG-hard'' to $(\alpha+\delta,\beta- \delta)$-approximate $\maxcsp(\Psi)$.\footnote{For simplicity, we assume each constraint in $T$ has width 2, but we could easily extend this to width $c$ constraints.}
\end{theorem}

\begin{proof}
The proof is by reduction from the Unique Games problem.
Pick constants $\eta, \gamma$ and $M = M(\eta, \gamma)$ satisfying the Unique Games Conjecture.
Let $\mathcal L(U, V, E, [M], \{\pi_{u,v}\}_{(u,v) \in E})$ be a Unique Games instance.
The reduction produces a Heisenberg Hamiltonian instance with graph~$G$
whose vertex set is $U \times \{-1, 1\}^M$.
A random edge in $G$ is sampled as follows: pick~$\bu \in U$ uniformly at random,
and sample two uniformly random neighbors $\bv, \bv' \sim N(\bu)$ independently.
Let $\bx$ and $\by$ be $\rho$-correlated $M$-dimensional Boolean strings.
Output the edge between $(\bv, \bx \circ \pi_{\bu, \bv'})$ and $(\bv', \by \circ \pi_{\bu, \bv'})$,
where $w \circ \sigma$ is the string in which $(w \circ \sigma)_i = w_{\sigma(i)}$.

A product state assignment to~$G$ corresponds to a function $f_v:\{-1, 1\}^M \rightarrow S^2$ for each $v \in V$.
It has value
\begin{equation*}
\E_{\bu \sim U}\E_{\bv, \bv' \sim N(\bu)} \E_{\substack{\text{$(\bx, \by)$ $\rho$-correlated}\\\text{$n$-dim Boolean strings}}}
	\left[\tfrac{1}{4} - \tfrac{1}{4}\langle f_{\bv}(\bx \circ \pi_{\bu, \bv}), f_{\bv}(\by \circ \pi_{\bu, \bv'})\rangle\right].
\end{equation*}

\textit{Completeness.} Assume $\mathcal L$ has a labeling $L:U \cup V \rightarrow [m]$ satisfying more than $(1- \eta)$-fraction of the edges.
For each $v \in V$, let $f_v(x) = (x_{L(v)}, 0, \dots, 0)$.
To analyze the performance of~$f$,
let us first fix a vertex $u \in U$ and two neighbors $v, w \in N(u)$,
and condition on the case that $L$ satisfies both edges $(u, v)$ and $(u, w)$.
This means that $\pi_{v\rightarrow u}(L(v)) = L(u) = \pi_{w\rightarrow u}(L(w))$.
Thus, for each $x, y \in \{-1, 1\}^M$,
\begin{align*}
f_v(x \circ \pi_{v\rightarrow u}) &= ((x \circ \pi_{v \rightarrow u})_{L(v)}, 0, 0) = (x_{\pi_{v \rightarrow u}(L(v))},0,0) = (x_{L(u)},0,0),
\end{align*}
and similarly $f_w(y \circ \pi_{w \rightarrow u}) = (y_{L(u)}, 0, 0)$.
As a result, the value of $f$ conditioned on $u$, $v$, and $w$ is
\begin{align*}
\E_{\bx, \by}
	\left[\tfrac{1}{4} - \tfrac{1}{4}\langle f_{v}(\bx \circ \pi_{v \rightarrow u}), f_{w}(\by \circ \pi_{w \rightarrow u})\rangle\right]
=\E_{\bx, \by}
	\left[\tfrac{1}{4} - \tfrac{1}{4}\langle (\bx_{L(u)}, 0, 0), (\by_{L(u)}, 0, 0)\rangle\right],
\end{align*}
which is just the value of the $L(u)$-th embedded dictator on the noisy hypercube, i.e.\ $1/4 - 1/4 \rho$.

Now we average over $\bu, \bv, \bw$.
Because $\mathcal L$ is a biregular Unique Games instance,
picking a random vertex $\bu \in U$ and neighbor $\bv \in N(\bu)$ is equivalent to picking a uniformly random edge from~$E$.
Therefore, by the union bound, the probability that the assignment~$L$
satisfies both edges $(\bu, \bv)$ and $(\bu, \bw)$ is at least $1-2\eta$.
As we have seen, conditioned on this event, the assignment~$f$ has value at least $1/4 - 1/4\rho$.
As a result, we can lower-bound the value of~$f$ by
\begin{equation*}
(1-2\eta) \cdot (\tfrac{1}{4} - \tfrac{1}{4}\rho) \geq \tfrac{1}{4} - \tfrac{1}{4}\rho - \eta.
\end{equation*}
This completes the completeness case.

\textit{Soundness.}
We will show the contrapositive.
Suppose there is a product state assignment $\{f_v\}_{v \in V}$ to $G$ with value at least $\tfrac{1}{4} - \tfrac{1}{4} F^*(3, \rho) + \eps$.
We will use this to construct a randomized assignment $\bL:U \cup V \rightarrow [M]$ whose average value is at least~$\gamma$,
which implies that the Unique Games instance has value at least~$\gamma$.

For each $u \in U$, we define the function $g_u: \{-1, 1\}^M \rightarrow B^2$ as
\begin{equation*}
g_u(x) = \E_{\bv \sim N(u)}[ f_{\bv}(x \circ \pi_{\bv \rightarrow u}) ].
\end{equation*}
Then we can rewrite the value of the assignment $\{f_v\}$ as
\begin{align*}
\E_{\bu} \E_{\bv, \bw \sim N(\bu)} \E_{\bx, \by}\left[\tfrac{1}{4} - \tfrac{1}{4}\langle f_{\bv}(\bx \circ \pi_{\bv \rightarrow \bu}), f_{\bw}(\by \circ \pi_{\bw \rightarrow \bu})\rangle\right]
&= \E_{\bu} \E_{\bx, \by}[\tfrac{1}{4} - \tfrac{1}{4}\langle g_{\bu}(\bx), g_{\bu}(\by)\rangle]\\
&= \E_{\bu}[\tfrac{1}{4} - \tfrac{1}{4}  \Stab_\rho[g_{\bu}]].
\end{align*}
Since~$f$ has value at least $\tfrac{1}{4} - \tfrac{1}{4} F^*(3, \rho) + \eps$,
an averaging argument implies that at least an $\eps/2$ fraction of $u \in U$ satisfy 
$\tfrac{1}{4} - \tfrac{1}{4}  \Stab_\rho[g_{u}] \geq \tfrac{1}{4} - \tfrac{1}{4} F^*(3, \rho) + \eps/2$.
Rearranging, these $u$'s satisfy
\begin{equation*}
\Stab_\rho[g_u] \leq F^*(3, \rho) + 2 \eps.
\end{equation*}
We call any such~$u$ ``good''.
We now apply the soundness of our dictatorship test to the good $u$'s:
by XXX, any such $u$ has a ``notable'' coordinate, i.e. an $i$ such that $\Inf^{\leq m}_i[g_u] > \delta$.
Our random assignment will then use this~$i$ as its label for the vertex~$u$:
$\bL(u) = i$. (If~$u$ has multiple notable coordinates, then we pick one of these arbitrarily as the label for~$u$.)

Next we'll need to obtain labels for the neighbors of~$u$.
We will use the condition that~$u$ is good to derive that many of~$u$'s neighbors~$v$ have notable coordinates.
This requires relating the Fourier spectrum of $g_u$ to the Fourier spectra of the neighboring $f_v$'s.
To begin,
\begin{equation*}
\chi_S(x \circ \pi_{v \rightarrow u})
 = \prod_{i \in S}(x \circ \pi_{v \rightarrow u})_i
 = \prod_{i \in S} x_{\pi_{v\rightarrow u}(i)}
 = \prod_{j \in T} x_j
 = \chi_T(x),
\end{equation*}
where $T = \pi_{v \rightarrow u}(S) = \{\pi_{v\rightarrow u}(i) : i \in S\}$.
As a result,
\begin{equation*}
f_v(x \circ \pi_{v \rightarrow u})
= \sum_{S \subseteq [n]} \widehat{f}(S) \chi_{S}(x \circ \pi_{v \rightarrow u})
= \sum_{T \subseteq [n]} \widehat{f}(\pi_{u \rightarrow v}(T)) \chi_T(x).
\end{equation*}
Averaging over all $\bv \in N(u)$,
\begin{equation*}
g_u(x)
= \E_{\bv \sim N(u)}[ f_{\bv}(x \circ \pi_{\bv \rightarrow u}) ]
= \sum_{T \subseteq [n]} \E_{\bv \sim N(u)}[\widehat{f}(\pi_{u \rightarrow v}(T))] \chi_T(x)
= \sum_{T \subseteq [n]} \widehat{g}(T) \chi_T(x).
\end{equation*}
Hence,
\begin{align*}
\delta
&< \Inf^{\leq m}_i[g_u]\\
&= \sum_{|T| \leq m:  T \ni i} \Vert \widehat{g}(T) \Vert^2_2\\
&= \sum_{|T| \leq m:  T \ni i} \left\Vert \E_{\bv \sim N(u)}[\widehat{f}(\pi_{\bu \rightarrow v}(T))] \right\Vert^2_2\\
& \leq \sum_{|T| \leq m:  T \ni i} \E_{\bv \sim N(u)} \left\Vert \widehat{f}(\pi_{\bu \rightarrow v}(T)) \right\Vert^2_2\tag{by XXX}\\
& = \E_{\bv \sim N(u)}\left[\sum_{|T| \leq m:  T \ni i} \left\Vert \widehat{f}(\pi_{\bu \rightarrow v}(T)) \right\Vert^2_2\right]\\
& =  \E_{\bv \sim N(u)}[\Inf^{\leq m}_{\pi_{\bv \rightarrow u}(i)}[f_{\bv}]]
\end{align*}
By another averaging argument,
a $\delta/2$-fraction of $u$'s neighbors~$v$ satisfy $\Inf^{\leq m}_{\pi_{v \rightarrow u}(i)}[f_{\bv}] \geq \delta/2$.
We call these the ``good neighbors''.
For each good neighbor~$v$, the set of possible labels
$$S_v = \{j : \Inf^{\leq m}_j[f_v] \geq \delta/2 \}$$
is non-empty.
In addition, one of these labels~$j$ satisfies $j = \pi_{v \rightarrow u}(i)$.
On the other hand, by \Cref{lem:influence-bound},  $|S_v| \leq XXX$ and so this set is not too large either.
For each good neighbor,
we assign the label of $\bL(v)$ by picking a uniformly random $j \in S_v$.
For all other vertices (i.e.\ those which are not good or good neighbors),
we assign $\bL$ a random label.

Now we consider the expected number of edges in $\mathcal{L}$ satisfied by $\bL$.
Given a random edge $(\bu, \bv)$,
the probability that $\bu$ is good is at least $\eps/2$;
conditioned on this, the probability that $\bv$ is a good neighbor is at least $\delta/2$.
Assuming both hold, since $S_{\bv}$ is of size at most XXX and contains one label equal to $\pi_{v \rightarrow u}(L(u))$,
then $\bL$ satisfies the edge $(\bu, \bv)$ with probability at least XXX.
In total, $\bL$ satisfies at least an
$$XXX$$
fraction of the edges. Setting $\gamma =XXX$ concludes the proof.
\end{proof}
}

\ignore{

\section{Gap Instances}

In this section, we prove theorems \ref{thm:integrality-gap-heisenberg}, \ref{thm:integrality-gap-prod}, and \ref{thm:algo-gap}.

\subsection{Integrality Gaps}\label{sec:integrality-gaps}

\sdpintgap*

\prodintgap*

Recalling \Cref{def:integrality-gap}, our goal is to bound
$$\inf_{\text{instances }\mathcal I\text{ of }\mathcal P}\left\{\frac{\mathrm{OPT}(\mathcal I)}{\mathrm{SDP}(\mathcal I)}\right\}$$
where $\mathcal P$ is either the Heisenberg Hamiltonian SDP, or the product state SDP. To upper bound this quantity, we construct a specific instance $\mathcal I$, and give an upperbound for $\mathrm{OPT}(\mathcal I)$ and a lowerbound for $\mathrm{SDP}(\mathcal I)$. Note that for product state approximation, we only optimize over product state, whereas for the general Heisenberg Hamiltonian we consider all quantum states. However, the specific instance $\mathcal I$ we consider will correspond to a high-degree graph, and for such graphs Brandow and Harrow \cite{BH16} show it suffices to consider product states.

\begin{theorem}[Product-State Approximations \cite{BH16}]
Let $G = (V,E)$ be a $D$-regular graph, and let $H$ be a Hamiltonian on $n = |V|$ qubits given by
$$H = \E_{(i,j) \in E} H_{i,j} = \frac{2}{nD}\sum_{(i,j) \in E} H_{i,j}$$
where $H_{i,j}$ acts only on qubits $i,j$ and $\norm{H_{i,j}} \leq 1$. Then there exists a state $\ket \psi = \ket{\psi_1} \otimes \cdots \otimes \ket{\psi_n}$ such that
$$\tr[H \psi] \geq \lambda_{\textrm{max}}(H) - O(D^{-1/3})$$
\end{theorem}

Therefore, by considering Hamiltonians on graphs with sufficiently large degree, we have that $\prodval(G) \geq \mathrm{OPT}(G) - \epsilon$ and for both theorems \ref{thm:integrality-gap-heisenberg} and \ref{thm:integrality-gap-prod}, it suffices to upper bound $\prodval(G)$, for the same instance $G$. Next, recall from equations \ref{eq:bov-ineq} and \ref{eq:gp-ineq} that the SDP values are quite similar, with just a multiplicative difference. As such, the analysis for the upper bound in both cases will be highly similar.

The instance we construct bears resemblance to the integrality instance for the \maxcut SDP, given by \cite{FS02}. In particular, we take our graph $G$ to correspond to all points in $S^{n-1}$, the $n$-dimensional sphere, with edges drawn from a $\rho$-correlated multivariate Gaussian normal distribution.

\begin{definition}[$\rho$-correlated Sphere Graph]\label{def:graph}
The $\rho$-correlated sphere graph, denoted $\mathcal G^n_\rho$, has vertices on $S^{n-1}$ and edges $(\bx/||\bx||,\by/||\by||)$, where, recalling the notation from \Cref{sec:notation}, $\bx \sim_\rho \by$ are $\rho$-correlated $n$-dimensional Gaussians.
\end{definition}

Although one might object that $\mathcal G^n_\rho$ is an infinite graph, the following lemma shows $G_{\rho,n}$ can be discretized  with a negligible loss in value. The proof is given in \Cref{sec:card-reduc}. 

\begin{lemma}[Graph Discretization]\label{lem:discretization}
Let $G = \mathcal G^n_\rho$ be a $\rho$-correlated sphere graph such that $\mathrm{SDP}(G) \geq c$ and $\prodval(G) \leq s$. Then, there exists a finite, weight graph $G'$ (with $n = 1/\epsilon^{O(d)}$ vertices) containing no self-loops and satisfying $\mathrm{SDP}(G') \geq c - \epsilon$ and $\prodval(G') \leq s + \epsilon$.\ynote{Do we need loopless?}
\end{lemma}

With the instance $G_{\rho,n}$ in hand, we state our upperbound for $\prodval(G_{\rho,n})$.

\begin{lemma}[Product State Value Upper Bound]\label{lem:prod-upper}
Recall that a product state assignment to $G = \mathcal G^n_\rho$ is given by a function $f : S^{n-1} \rightarrow S^2$ and obtains value,
\begin{equation}\label{eq:prod-val}
\val_{G}(f) = \frac{1}{4}\PAREN{1-\E_{\bx \sim_\rho \by}\Inner{f\PAREN{\Normed{\bx}}, f\PAREN{\Normed{\by}}}} 
\end{equation}
Then,
$$\prodval(G) = \max_{f : S^{n-1} \rightarrow S^2} \val_{G}(f) \leq \val_{G}(f_\mathrm{opt})$$
\end{lemma}
Setting $k = 3$ in the following theorem gives a characterization of $\val_{\mathcal G^n_\rho}(f_\mathrm{opt})$ solely as a function of $\rho$.
\begin{theorem}[\cite{BOV10}]\label{thm:hypergeometric}
Let $u,v$ be unit vectors in $\mathbb R^n$ with $\langle u,v\rangle = \rho$ and let $\bZ \in \mathbb R^{k\times n}$ be a random matrix whose entries are distributed independently according to the standard normal distribution with mean $0$ and variance $1$. Then,
$$F^\star(k,\rho) := \E\left[\frac{\bZ u}{||\bZ u||}\cdot\frac{\bZ v}{||\bZ v||}\right] = \frac{2}{k}\left(\frac{\Gamma((k+1)/2)}{\Gamma(k/2)}\right)^2\rho  \cdot {}_2F_1(1/2,1/2; k/2+1;\rho^2)$$
where ${}_2F_1$ is the hypergeometric function and $\Gamma$ is the standard gamma function.
\end{theorem}

Next, the following theorem gives an upperbound of for the Heisenberg Hamiltonian SDP.

\begin{lemma}[Heisenberg Hamiltonian SDP Lower Bound]\label{lem:heis-sdp-lower}
The assignment $f_\mathrm{ident}(x) = x$ for $x \in S^{n-1}$ is a feasible solution to the Heisenberg Hamiltonian SDP on $G^n_\rho$ and obtains value $\frac{1}{4}(1 -3\rho)$. Then,
$$\frac{1}{4}(1-3\rho) \leq \hsdp(\mathcal G_\rho^n)$$
\end{lemma}
An essentially identical proof also yields the following corollary, which gives a upper bound for the product state SDP.
\begin{corollary}[Product State SDP Lower Bound]\label{lem:prod-sdp-lower}
The assignment $f_\mathrm{ident}(x) = x$ for $x \in S^{n-1}$ is a feasible solution to the product state SDP on $G_{\rho,n-1}$ and obtains value $\frac{1}{4}(1 -\rho)$. Then,
$$\frac{1}{4}(1-\rho) \leq \prodsdp(\mathcal G^n_\rho)$$
\end{corollary}
Then, the proof of \Cref{thm:integrality-gap-heisenberg} follows simply from \Cref{lem:prod-upper} and \Cref{lem:heis-sdp-lower}.
\begin{proof}[Proof of \Cref{thm:integrality-gap-heisenberg}]
Combining \Cref{lem:prod-upper} and \Cref{lem:heis-sdp-lower} yields
\begin{equation}
\label{eq:prod-ratio}
\frac{\prodval(\mathcal G_\rho^n)}{\hsdp(\mathcal G_\rho^n)} \leq \frac{\val_{\mathcal G_\rho^n}(f_\mathrm{opt})}{\frac{1}{4}(1-3\rho)} = \frac{ \frac{1}{4}\PAREN{1 - F^\star(3,\rho)}}{\frac{1}{4}(1-3\rho)}\\
\end{equation}
This quantity is minimzied at $\rho = \rgp$ and is equal to $\agp$.
\end{proof}
Finally, using $\frac{1}{4}(1-\rho)$ in place of $\frac{1}{4}(1-3\rho)$ yields a value of $\abov$, obtained at $\rho = \rbov$. To conclude, we give the proofs from \Cref{lem:prod-upper} and \Cref{lem:heis-sdp-lower}.
\begin{proof}[Proof of \Cref{lem:prod-upper}]
This case is a fairly straightforward consequence of \Cref{thm:k-dim-borell}, with $k = 3$. However, note that the theorem considers functions $f : \R^n \rightarrow B^3$, whereas we have $f : \mathbb S^{n-1} \rightarrow S^2$. Observe that since $S^{n-1} \subseteq \R^n$ and $S^2 \subseteq B^3$, functions of the latter class are always a member of the former. Therefore,
\begin{align}
\label{eq:stab-upper}
\prodval(G_{\rho,n}) = \max_{f : S^{n-1} \rightarrow S^2} \val_{G_{\rho,n}}(f) &\leq \max_{f : \R^n \rightarrow B^3} \frac{1}{4}\PAREN{1 - \E_{\bx\sim_\rho \by} \Inner{f(\bx), f(\by)}}\nonumber\\
&= \max_{f : \R^n \rightarrow B^3} \frac{1}{4}\PAREN{1 - \Stab_\rho[f]}\nonumber\\
&= \frac{1}{4}\PAREN{1 - \Stab_\rho[f_\mathrm{opt}]}
\end{align}
Where (\ref{eq:stab-upper}) follows from the $\rho < 0$ case of \Cref{thm:k-dim-borell}. In order to show that this is in fact equal to $\val_{G_{\rho,n}}(f_\mathrm{opt})$, we need to restrict the domain of $f_\mathrm{opt}$ to $S^{n-1}$. Note that
$$f_\mathrm{opt}\PAREN{\Normed{x}} = \Normed{\Pi^{(k)}\Normed{x}} = \frac{\frac{1}{\norm x}\Pi^{(k)}x}{\frac{1}{\norm x} \NORM{\Pi^{(k)}x}} = \Normed{\Pi^{(k)}x} = f_\mathrm{opt}(x)$$
Therefore, $\frac{1}{4}(1-\Stab_{\rho}[f_\mathrm{opt}]) = \val_{G_{\rho,n}}(f_\mathrm{opt})$ and we get the desired upperbound.
\end{proof}

\begin{proof}[Proof of \Cref{lem:heis-sdp-lower}]
Given the assignment $f_\mathrm{ident}$, the SDP obtains value
$$\frac{1}{4}\E_{\bx \sim_\rho \by}[1 - 3\inner{\bx,\by}] = \frac{1}{4}\PAREN{1-\E_{\bx \sim_\rho \by} \inner{\bx,\by}} = \frac{1}{4}(1 - 3\rho)$$
Since $f_\mathrm{ident}$ is a feasible solution, this is a lower bound on $\hsdp(G_{\rho,n})$.
\end{proof}

\subsection{Algorithmic Gaps}

Next we prove \Cref{thm:algo-gap},

\algogap*

This proof closely resembles that of \cite{kar99}, who gave an algorithmic gap instance for the \maxcut SDP. Recall, we are trying to bound the quantity,
$$\inf_{\text{instances }\mathcal I\text{ of }\mathcal P}\left\{\frac{\mathrm{A}(\mathcal I)}{\mathrm{OPT}(\mathcal I)}\right\}$$
Furthermore, the algorithm of \cite{BOV10} shows that for any instance $\mathcal I$,
$$\abov \leq \frac{\mathrm{A}(\mathcal I)}{\hsdp(\mathcal I)} \leq \frac{\mathrm{A}(\mathcal I)}{\mathrm{OPT}(\mathcal I)}$$
Thus, to show an algorithmic gap of $\abov$, it suffices to demonstrate an instance $\mathcal I$ such that $\frac{\mathrm{A}(\mathcal I)}{\mathrm{OPT}(\mathcal I)} \leq \abov$. As before our instance will be a graph, with a Hamiltonian term on each edge. Much like the rounding algorithm of \cite{GW95}, the analysis of \cite{BOV10} is tight for edges $(u,v)$ where $\langle u,v \rangle = \rbov$. This informs our graph construction.
\begin{definition}[Discrete Embedded Graph]\label{def:discrete-graph}
The discrete embedded graph, denoted $\mathcal B^n_\rho$ has vertex set $\{-\tfrac{1}{\sqrt n},\tfrac{1}{\sqrt n}\}^n$. The edge distribution corresponds to the following random experiment,
\begin{enumerate}
\item Pick $v \in \{-\tfrac{1}{\sqrt n},\tfrac{1}{\sqrt n}\}^n$
\item Set $u = v$
\item For each $i \in [n]$, set $u_i = -u_i$ with probability $\tfrac{1}{2}-\tfrac{1}{2}\rho$
\end{enumerate}
\end{definition}
As desired, $\E_{(\bu,\bv) \sim \mathcal B^n_\rho} [\inner{\bu, \bv}] = \rho$. On the graph $\mathcal B^n_\rho$, value of the SDP relaxation is equal to the identity assignment $f_\mathrm{ident}(v) = v$.
\begin{lemma}[Value of the SDP Relaxation]
Let $G(V,E) = \mathcal B^n_\rho$. Then,
$$\hsdp(G) = \max_{f:V\rightarrow S^{n-1}} \E_{(\bu, \bv) \sim E}[ \tfrac{1}{4} - \tfrac{3}{4} \inner{f(\bu), f(\bv)}] \leq  \frac{1}{4}(1 - 3\rho) = \E_{(\bu, \bv) \sim E}[ \tfrac{1}{4} - \tfrac{3}{4} \langle f_\mathrm{ident}(\bu), f_\mathrm{ident}(\bv)\rangle]$$
\end{lemma}
\begin{proof}
The last equality follows by construction of $G$, so we just need to show the inequality. To do so, we will show,
$$\min_{f\colon V \rightarrow S^{n-1}} \E_{(u,v) \sim E}[\langle f(\bu),f(\bv)\rangle] \geq \rho$$
Fix any assignment $f$. Although the domain for $f$ is $\{-\tfrac{1}{\sqrt n}, \tfrac{1}{\sqrt n}\}^n$, by scaling, it suffices to consider functions $f : \{-1,1\}^n \rightarrow \mathbb S^{n-1}$. Write $f = (f_1,\dots,f_n)$ where $f_1 : \{-1,1\}^n \rightarrow \mathbb R$ with the added constraint that for all $x$, $\sum_i f_i(x)^2 = 1$. For any edge $(\bu,\bv)$, $\bu$ and $\bv$ are $\rho$ correlated and by definition, $\E_{(\bu,\bv) \sim E)}[\langle f(\bu),f(\bv)\rangle] = \Stab_\rho[f]$. Then, 
\begin{align*}
\E_{(\bu,\bv) \sim E)}[\langle f(\bu),f(\bv)\rangle] &= \sum_{i=1}^n \E_{(\bu,\bv) \sim E)}[ f_i(\bu),f_i(\bv)]\\
&= \sum_{i=1}^n \Stab_\rho[f_i]\\
&\geq \sum_{i=1}^n \rho \cdot \E[f_j(x)^2] && \text{By \Cref{prop:stab-bound}}\\
&= \rho \E[\sum_{i=1}^n f_j(x)^2] = \rho
\end{align*}
\end{proof}
But we also have that $f_\mathrm{ident}$ is a feasible SDP relaxation so $\val_{\mathcal B^n_\rho}(f_\mathrm{ident}) = \hsdp(\mathcal B^n_\rho)$ and $f_\mathrm{ident}$ is a possible output for the SDP assignment. Letting $\rho = \rbov$ be the worst angle for the rounding scheme of \cite{BOV10}, we get that
\begin{equation}\label{eq:alg-upper-bound}
\mathrm A(\mathcal B^n_\rho) = \abov \cdot \val_{\mathcal B^n_\rho}(f_\mathrm{ident}) = \abov \cdot \hsdp(\mathcal B^n_\rho)
\end{equation}
Next, consider the assignment $f^{(1)}(x) = (x_1/|x_1|,0,0)$. This assignment achieves value,
\begin{align*}
\val_{\mathcal B^n_\rho}(f^{(1)}) = \E_{(\bu,\bv) \sim E}[\tfrac{1}{4}-\tfrac{3}{4}\langle f^{(1)}(\bu),f^{(1)}(\bv)\rangle] = \E_{(\bu,\bv) \sim E}[\tfrac{1}{4}-\tfrac{3}{4}\tfrac{\bu_1}{|\bu_1|}\tfrac{\bv_1}{|\bv_1|}] = \tfrac{1}{4}(1-3\rho) = \hsdp(\mathcal B^n_\rho)
\end{align*}
Since $\val_{\mathcal B^n_\rho}(f^{(1)})$ is an lower bound on $\mathrm{OPT}(\mathcal B^n_\rho)$, we get that 
\begin{equation}\label{eq:opt-lower-bound}
\hsdp(\mathcal B^n_\rho) \leq \mathrm{OPT}(\mathcal B^n_\rho).
\end{equation}
Combining equations \ref{eq:alg-upper-bound} and \ref{eq:opt-lower-bound} yields
$$\mathrm A(\mathcal B^n_\rho) \leq \abov \mathrm{OPT}(\mathcal B^n_\rho) \implies \frac{\mathrm A(\mathcal B^n_\rho)}{\mathrm{OPT}(\mathcal B^n_\rho)} \leq \abov$$
}
\section{A dictator test for the product state value}\label{sec:dictator-test}

Now we show that the noisy hypercube serves as a dictatorship test for functions of the form $f:\{-1, 1\}^n \rightarrow B^k$, assuming \cref{conj:vector-borell-intro}.
Informally, this means that if $f$ is an embedded dictator, it should have high value,
and if it is ``far'' from a dictator (in the sense that it has no ``notable'' input coordinates)
then it should have low value.
This will be an important ingredient in our Unique-Games hardness proof in~\Cref{sec:ug-hardness} below.

We have already shown that embedded dictators achieve value $\tfrac{1}{4} - \tfrac{1}{4} \rho$ on the noisy hypercube in~\Cref{lem:embedded-dictators}.
Now we will upper-bound the value that functions ``far'' from dictators achieve.
We will show that their value, up to small error, is at most the optimum product state value on the Gaussian graph
$\ggraph^n_\rho$, which we have shown to be $\tfrac{1}{4} - \tfrac{1}{4} F^*(k, \rho)$. Throughout this section, we will make heavy use of the various Fourier analytic quantites defined in \Cref{sec:fourier-analysis}.

\begin{theorem}[Dictatorship test soundness]\label{thm:dictator-test}
Assume \cref{conj:vector-borell-intro}.
Let $-1 < \rho \leq 0$.
Then for any $\epsilon > 0$,
there exists a small enough $\delta = \delta(\epsilon, \rho) > 0$
and large enough $m = m(\epsilon, \rho) \geq 0$ such that the following is true.
Let $f: \{-1, 1\}^n \rightarrow B^k$ be any function satisfying 
$$\Inf^{\leq m}_i[f] = \sum_{j=1}^k \Inf^{\leq m}_i[f_j] \leq \delta,\quad \text{for all $i = 1, \dots, n$}.$$
Then
$$\E_{\substack{\text{$(\bx, \by)$ $\rho$-correlated}\\\text{$n$-dim Boolean strings}}} [f(\bx)f(\by)] = \Stab_\rho[f] \geq F^*(k,\rho) - \epsilon.$$
In other words, the value of~$f$ on the noisy hypercube $\hgraph^n_\rho$ is at most
$
\tfrac{1}{4} - \tfrac{1}{4} F^*(k, \rho) + \eps.
$
\end{theorem}

The $k = 1$ case is the negative~$\rho$ case of the Majority is Stablest theorem of~\cite{MOO10},
which serves as the soundness case for the $\maxcut$ dictatorship test;
our theorem generalizes the negative~$\rho$ case of Majority is Stablest to larger values of~$k$.
The proof follows the same outline as the proof of Majority is Stablest appearing in~\cite[Chapter 11.7]{OD14}:
we apply an ``invariance principle'' to exchange $f$'s Boolean inputs
with Gaussians of the same mean and variance.
Then we use our \cref{conj:vector-borell-intro} on the noise stability of functions in Gaussian space
to upper-bound the value of~$f$. 
The invariance principle we will use is the following one due to Isaksson and Mossel~\cite{IM12},
which applies to vector-valued functions.

\begin{theorem}[Vector-valued invariance principle \cite{IM12}]\label{thm:vector-invariance}
Fix $\delta, \gamma \in (0,1)$ and set $m = \tfrac{1}{18}\log\tfrac{1}{\delta}$. Let $f = (f_1,\dots,f_k)$ be a $k$-dimensional multilinear polynomial such that $\Var[f_j] \leq 1$, $\Var[f_j^{> m}] \leq (1-\gamma)^{2m}$, and $\Inf^{\leq m}_i[f_j] \leq \delta$ for each $j \in [k]$ and $i \in [n]$. Let $\bx$ be a uniformly random string over $\{-1,1\}^n$
and $\by$ be an $n$-dimensional standard Gaussian random variable.
Furthermore, let $\Psi : \mathbb R^k \rightarrow \mathbb R$ be Lipschitz continuous with Lipschitz constant $A$. Then,
$$|\E[\Psi(f(\bx))] - \E[\Psi(f(\by))]| \leq C_kA \delta^{\gamma/(18 \ln 2)}$$
where $C_k$ is a parameter depending only on $k$.
\end{theorem}

As a first step in proving \Cref{thm:dictator-test}, we claim that we can assume $f$ is odd (i.e.\ $f(x) = -f(-x)$) without loss of generality.

\begin{lemma}\label{lem:decreasing-stab}
Fix $\rho \in (-1,0]$ and $f : \{-1,1\}^n \rightarrow B^k$.
Define the function $g(x) = \tfrac{1}{2}(f(x) - f(-x))$. Then, $g$ is odd, has range~$B^k$, and satisfies
$$\Stab_\rho[f] \geq \Stab_\rho[g]$$
and
$$\forall i \in [n],\ \Inf_i^{\leq m}[f] \geq \Inf_i^{\leq m}[g].$$
\end{lemma}
\begin{proof}
Oddness follows from $g(-x) = \tfrac{1}{2}(f(-x) - f(x)) = -g(x)$.
In addition, $g(x)$ maps into $B^k$ because it is the average of two points in $B^k$.
The Fourier transform of $f$ is $\sum_{S \subseteq [n]} \hat f(S) \chi_S(x)$.
Noting that $\chi_S(x) = \chi_S(-x)$ for $|S|$ even and $-\chi_S(-x)$ for $|S|$ odd,
the Fourier transform of $g$ is
$$g(x) = \sum_{\text{$|S|$ odd}} \hat f(S) \chi_S(x).$$
Thus, for any $i \in [n]$,
$$\Inf_i^{\leq m}[g] = \sum_{\substack{|S| \leq m, S \ni i, \\ \text{$|S|$ odd}}} \norm{\hat f(S)}_2^2
\leq  \sum_{|S| \leq m, S \ni i} \norm{\hat f(S)}_2^2 = \Inf_i^{\leq m}[f].$$
Furthermore, using that $\rho^{|S|} \geq 0$ for $|S|$ even,
\begin{equation*}\Stab_\rho[g] = \sum_{\text{$|S|$ odd}} \rho^{|S|}\norm{\hat f(S)}_2^2
\leq  \sum_{S \subseteq [n]} \rho^{|S|} \norm{\hat f(S)}_2^2 = \Stab_\rho[f].\qedhere
\end{equation*}
\end{proof}

Therefore, to lower bound $\Stab_\rho[f]$, it suffices to lower bound the noise stability of the odd function $g$.
Then, the proof of \Cref{thm:dictator-test} is essentially by reduction to case of \emph{positive} $\rho$, i.e.\ $\rho \in [0,1)$, but only for odd functions.
This requires an analogue of \cref{conj:vector-borell-intro}
for the case of positive $\rho$ and odd~$f$, which we derive as follows.

\begin{corollary}[Vector-valued Borell's inequality; positive $\rho$ and odd $f$ case]\label{cor:vecor-borell-positive-rho}
Assume \cref{conj:vector-borell-intro}.
Then it holds in the reverse direction when $\rho \in [0,1]$, with the additional assumption that $f$ is odd. In particular,
$$\Stab_\rho[f] \leq \Stab_\rho[f_\mathrm{opt}],$$
where $f_{\mathrm{opt}}(x) = x_{\leq k}/\Vert x_{\leq k}\Vert$ and $x_{\leq k} = (x_1, \ldots, x_k)$.
\end{corollary}
\begin{proof}
Because $f$ is odd,
\begin{equation}\label{eq:odd-stab}
\Stab_\rho[f] = \E_{\bx \sim_\rho \by}[\langle f(\bx),f(\by)\rangle] =  -\E_{\bx \sim_\rho \by}[\langle f(\bx),f(-\by)\rangle] =  -\E_{\bx \sim_{-\rho} \by}[\langle f(\bx),f(\by)\rangle] = -\Stab_{-\rho}[f],
\end{equation}
where we have used the fact that if $\bx$ and $\by$ are $\rho$-correlated, then $\bx$ and $-\by$ are $(-\rho)$-correlated.
Recalling that $\rho \geq 0$, we apply \cref{conj:vector-borell-intro} and obtain $\Stab_{-\rho}[f] \geq \Stab_{-\rho}[f_\mathrm{opt}]$.
Observing that $f_{\mathrm{opt}}$ is an odd function, we further have
$$\Stab_{-\rho}[f] \geq \Stab_{-\rho}[f_\mathrm{opt}] = -\Stab_\rho[f_\mathrm{opt}].$$
Thus, negating the above inequality and combining with \Cref{eq:odd-stab} yields
$$\Stab_{\rho}[f] = -\Stab_{-\rho}[f] \leq \Stab_{\rho}[f_\mathrm{opt}].$$
This concludes the proof.
\end{proof}

Now we prove the main theorem of this section.

\begin{proof}[Proof of \Cref{thm:dictator-test}]
To begin, we choose parameters
\begin{align*}
\gamma &= \tfrac{1}{6}(1+\rho)\epsilon, \tag{dictated by \Cref{eq:apply-gamma}}\\
\delta &= \PAREN{\tfrac{\epsilon}{12k C_k}}^{(18 \ln 2) / \gamma}, \tag{so that the error in \Cref{thm:vector-invariance} is at most $\epsilon/(6k)$}\\
m &= \tfrac{1}{18}\log \tfrac{1}{\delta}. \tag{dictated by \Cref{thm:vector-invariance}}
\end{align*}
Using \Cref{lem:decreasing-stab}, assume without loss of generality that $f$ is odd.
Throughout, we'll use $x$ to denote a string in $\{-1,1\}^n$ and $y$ to denote a vector in $\mathbb R^n$. To prove the claim, we'll perform a series of modifications to the initial function $f$ so that we can apply \Cref{thm:vector-invariance}. At each step, we'll show that each modified function has noise stability close to $\Stab_\rho [f]$. In particular, we'll consider the following functions:
\begin{enumerate}
\item $g(x) = \T_{1-\gamma} f(x)$,
\item $g(y)$, where $y \in \mathbb R^n$,
\item $\mathcal R \circ g(y)$, which is $g(y)$ rounded to $B^{k}$.
\end{enumerate}

\textbf{Step 1.} Since the statement of the vector-valued invariance principle (\Cref{thm:vector-invariance}) requires a function with low high-degree variance, we consider $g = \T_{1-\gamma} f$. Then for each $j \in [k]$,
$$\Var[g_j^{> m}] = \sum_{|S| > m} (1-\gamma)^{2|S|}\widehat {f}_j(S)^2 \leq (1-\gamma)^{2m} \sum_{S \subseteq [n]} \widehat{f}_j(S)^2 \leq (1-\gamma)^{2m}.$$
Also, $\Inf_i^{\leq m}[g_j] \leq \Inf_i^{\leq m}[f_j]$ for all $i \in [n]$.
Furthermore, note that since $f$ is odd, $g$ is also odd. Next, we bound the error in the quantity $\Stab_\rho[f]$ when we consider $g$ in place $f$. Since $g = \T_{1-\gamma} f$, we see that $\Stab_\rho[g] = \Stab_{\rho(1-\gamma)^2}[f]$, and so it suffices to bound $|\Stab_{\rho}[f] - \Stab_{\rho(1-\gamma)^2}[f]|$. To do so, we use \Cref{eq:rho-stab-diff}. However, note this proposition only applies for $\rho > 0$. But since $f$ is odd, we have,
$$|\Stab_{\rho}[f] - \Stab_{\rho(1-\gamma)^2}[f]|=|-\Stab_{-\rho}[f] + \Stab_{-\rho(1-\gamma)^2}[f]|=|\Stab_{-\rho}[f] - \Stab_{-\rho(1-\gamma)^2}[f]|,$$
which allows us to apply \Cref{eq:rho-stab-diff} with $\rho' = -\rho$. This yields
\begin{equation}\label{eq:apply-gamma}
|\Stab_{\rho}[f] - \Stab_{\rho(1-\gamma)^2}[f]| \leq \tfrac{2\gamma}{1+\rho} \Var[f] \leq \tfrac{2\gamma}{1+\rho}\E[\norm{f}_2^2]
\leq\tfrac{2\gamma}{1+\rho},
\end{equation}
where we used $\Var[f] \leq \E[\norm{f}_2^2]$ for all functions
and the fact that $f$'s range is $B^k$.
For our choice of $\gamma$, this is equal to $\epsilon/3$.\\

\textbf{Step 2.} Next, we bound the error accrued when we apply $g$ on Gaussian inputs. Consider the function
$$\Psi(v) = 
\left\{\begin{array}{cl}
\norm{v}_2^2 & \text{if $\norm{v}_2 \leq 1$,}\\
1 & \text{otherwise.}
\end{array}\right.$$
By \Cref{lem:psi-lipschitz}, $\Psi$ is $2$-Lipschitz. We'll rewrite $\Stab_\rho[g]$ using $\Psi$ in order to apply \Cref{thm:vector-invariance}.
\begin{align*}
\Stab_\rho[g(\bx)] &= \E_\bx [\langle g(\bx), \T_\rho g(\bx)\rangle]\\
&= \E_\bx [\langle \T_{\sqrt{-\rho}} g(\bx), \T_{-\sqrt{-\rho}} g(\bx) \rangle]\\
&= -\E_\bx [\langle \T_{\sqrt{-\rho}} g(\bx), \T_{\sqrt{-\rho}} g(\bx) \rangle] \tag{Since $g$ is odd}\\
&= -\E_\bx [\langle \norm{\T_{\sqrt{-\rho}} g(\bx)}_2^2 \rangle]\\
&= -\E_{\bx}\BRAC{\Psi(\T_{\sqrt{-\rho}}g(\bx))},
\end{align*}
where the last step used that $\T_{\sqrt{-\rho}} g(\bx) \in B^k$ and hence is unchanged by $\Psi(\cdot)$.
Applying \Cref{thm:vector-invariance}, we get
$$\ABS{\E_{\bx}\BRAC{\Psi(\T_{\sqrt{-\rho}}[g(\bx)])} - \E_{\by}\BRAC{\Psi(\T_{\sqrt{-\rho}}[g(\by)])}} \leq \epsilon/(6k) \leq \epsilon/3.$$
Here, we chose the ``$f$'' function in \Cref{thm:vector-invariance} to be $T_{\sqrt{-\rho}} g$,
which satisfies the low variance and influence properties because~$g$ does.

\textbf{Step 3.} The term $-\E_{\by}[\Psi(\T_{\sqrt{-\rho}}[g(\by)])]$ is \textit{almost} ready for application of \cref{conj:vector-borell-intro} through \Cref{cor:vecor-borell-positive-rho}. However, although $g(x)$ is bounded in $B^k$, the same might not hold for $g(y)$, where $y \in \mathbb R^n$. To fix this, we consider $\mathcal R \circ g$, where
$$\mathcal R(v) = 
\left\{\begin{array}{cl}
v & \text{if $\norm{v}_2 \leq 1$,}\\
\tfrac{v}{\norm{v}} & \text{otherwise.}
\end{array}\right.$$
In other words, $\calR(v)$
rounds a vector $v$ to the unit ball $B^m$. Then
\begin{align}
&\ABS{\E_{\by}\BRAC{\Psi(\T_{\sqrt{-\rho}}[g(\by)])} - \E_{\by}\BRAC{\Psi(\T_{\sqrt{-\rho}}[\mathcal R \circ g(\by)])}}\nonumber\\
&\leq \E_{\by} \ABS{\Psi(\T_{\sqrt{-\rho}}[g(\by)]) -\Psi(\T_{\sqrt{-\rho}}[\mathcal R \circ g(\by)])}\nonumber\\
&\leq 2\E_{\by} \NORM{\T_{\sqrt{-\rho}}[g(\by)] -\T_{\sqrt{-\rho}}[\mathcal R \circ g(\by)]}_2\tag{$\Psi$ is $2$-Lipschitz}\nonumber\\
&\leq 2\sum_{i=1}^k \E_{\by} \ABS{\T_{\sqrt{-\rho}}[g_i(\by)] -\T_{\sqrt{-\rho}}[(\mathcal R \circ g)_i(\by)]}\tag{$\norm{\cdot}_2 \leq \norm{\cdot}_1$}\nonumber\\
&\leq 2\sum_{i=1}^k \E_{\by} \ABS{g_i(\by) -(\mathcal R \circ g)_i(\by)}\tag{$\T_{\sqrt{-\rho}}$ is a contraction}\nonumber\\
&= 2\sum_{i=1}^k \E_{\by} |\Phi_i(g(\by))|,\label{eq:exp-sum}
\end{align}
where we define the map $\Phi_i : \mathbb R^k \rightarrow \mathbb R$ to be $\Phi_i(v) = v_i - \mathcal R(v)_i$, which is $2$-Lipschitz by \Cref{cor:phi-lipschitz}. To evaluate this, we first recall that on \emph{Boolean} inputs~$x$,
$g(x) \in B^k$, and so we have that $\Phi_i(g(x)) = 0$. Therefore
$$\ABS{\E_{\bx}[\Phi_i(g(\bx))] - \E_{\by}[\Phi_i(g(\by))]} = \ABS{\E_{\by}[\Phi_i(g(\by))]},$$
and applying \Cref{thm:vector-invariance} one last time, we can upper bound this by $\epsilon/{6k}$. The sum in \Cref{eq:exp-sum} is in turn upper bounded by $\epsilon/3$.
Finally, applying \Cref{cor:vecor-borell-positive-rho} to
\begin{equation*}
-\E_{\by}\BRAC{\Psi(\T_{\sqrt{-\rho}}[\mathcal R \circ g(\by)])} = -\Stab_{-\rho}[\mathcal R \circ g]
\end{equation*}
yields a lower bound of $-\Stab_{-\rho}[f_\mathrm{opt}] = \Stab_\rho[f_\mathrm{opt}]$, for which \Cref{prop:opt-formula} and \Cref{thm:exact-formula-for-average-inner-product} give an explicit formula. Through the three transformations, we accrue an error of at most $\epsilon$, which proves the claim.
\end{proof}

\section{Unique Games hardness of \qmaxcut}\label{sec:ug-hardness}

Now we prove hardness of \qmaxcut.
Our starting point is the Unique Games problem.

\begin{definition}[Unique Games]\label{def:unique-games}
The Unique Games problem is defined as follows.
An instance is a tuple $\calI(U, V, E, [M], \{\pi_{u\rightarrow v}\}_{(u, v) \in E})$, corresponding to a bipartite graph with left side vertices~$U$,
right side vertices~$V$,
and edges~$E$,
in addition to a bijection $\pi_{u \rightarrow v}:[M] \rightarrow [M]$ for each $(u, v) \in E$. We will also write $\pi_{v \rightarrow u}$ for $\pi_{u \rightarrow v}^{-1}$.
A labeling of the vertices is a function $L : U \cup V \rightarrow [M]$,
which satisfies the edge $(u, v) \in E$ if $\pi_{u \rightarrow v}(L(u)) = L(v)$.
The value of~$L$ is the fraction of edges it satisfies,
and the value of the instance~$\calI$ is the maximum value of any labeling.
\end{definition}

\begin{conjecture}[Unique Games Conjecture \cite{Kho02}]\label{con:ugc}
For any $\gamma > 0$, there exists a constant $M = M(\gamma)$ such that it is $\NP$-hard to distinguish whether an instance of the Unique Games problem with label set size $M$ has value at least $1 - \gamma$ or at most $\gamma$.
Furthermore, we may assume the constraint graph $\mathcal{C} = (U \cup V, E)$ is biregular.
\end{conjecture}

The fact that the constraint graph may be taken to be biregular is a consequence of the result of Khot and Regev~\cite{KR08},
as pointed out by Bansal and Khot~\cite{BK10}.
Our hardness result is stated as follows.

\begin{theorem}[UG-Hardness of Approximating \qmaxcut]\label{thm:ug-hardness}
Assume \cref{conj:vector-borell-intro}.
For any $\rho \in (-1,0)$ and $\epsilon > 0$,
given an instance of \qmaxcut,
the Unique Games Conjecture implies that the following two tasks are $\NP$-hard:
\begin{enumerate}
\item distinguishing if the product state value is greater than $\tfrac{1}{4} - \tfrac{1}{4}\rho - \eps$
or less than $\tfrac{1}{4} - \tfrac{1}{4}F^*(3, \rho) + \eps$,
\item distinguishing if the maximum energy is greater than $\tfrac{1}{4} - \tfrac{1}{4}\rho - \eps$
or less than $\tfrac{1}{4} - \tfrac{1}{4}F^*(3, \rho) + \eps$.
\end{enumerate}
Choosing $\rho = \rbov$, these imply that approximating the product state value and maximum energy
to a factor $\abov+\eps$ is $\NP$-hard, assuming the Unique Games Conjecture.
\end{theorem}

\ignore{
\begin{remark}
We note that our proof also extends to the $\maxcut_k$ problem.
In particular, it shows that for $\rho \in (-1, 0)$, it is $\NP$-hard (assuming the UGC)
to distinguish if the $\maxcut_k$ value is greater than $\tfrac{1}{2} - \tfrac{1}{2} \rho - \eps$
or smaller than $\tfrac{1}{2} - \tfrac{1}{2} F^*(k, \rho) + \eps$.
This shows that the rank-$k$ projection rounding algorithm of~\cite{BOV10}
is optimal for each fixed~$k$.  We note that the worst~$\rho$ for each fixed~$k$ is in the range $(-1,0)$,
which follows because the sign of $\rho$ and $F^*(k, \rho)$ are the same and $|F^*(k, \rho)| \leq |\rho|$; so it suffices to just prove the inapproximability result for this range.
\end{remark}
}

The proof mostly follows the standard outline for UG-hardness proofs introduced in~\cite{KKMO07}.
In particular, the graph produced by the reduction is exactly the same as the one produced in their \maxcut reduction,
with the one exception that we will eventually eliminate all self-loops so that it is a well-defined \qmaxcut instance.
The chief new difficulty is that in order to estimate the maximum energy of the graph produced by the reduction,
we use \Cref{cor:BH-nonuniform-easy-to-use} to relate it to the product state value;
however, the error term this theorem produces is a somewhat odd analogue of degree for weighted graphs,
and bounding it is slightly tedious.

\begin{proof}[Proof of \Cref{thm:ug-hardness}]
The proof is by reduction from the Unique Games problem.
To begin, we choose parameters
\begin{align*}
\delta &= \delta(\eps/2, \rho), \tag{$\delta(\cdot, \cdot)$ from \Cref{thm:dictator-test}}\\
m &= m(\eps/2, \rho), \tag{$m(\cdot, \cdot)$ from \Cref{thm:dictator-test}}\\
M &= \max\{M(\gamma), \tfrac{8 \log(\eps/200)}{\log(1/2 - \rho/2)}\}. \tag{$M(\cdot)$ from \Cref{con:ugc}; dictated by \Cref{eq:m-fact,eq:m-fact-2}}\\
\gamma & =  \frac{\eps \delta^2}{16M}, \tag{dictated by \Cref{eq:used-gamma,eq:used-gamma-again}}
\end{align*}
Let $\calI(U, V, E, [M], \{\pi_{u\rightarrow v}\}_{(u,v) \in E})$ be a biregular instance of the Unique Games problem.
The reduction produces a \qmaxcut instance with graph~$G$
whose vertex set is $V \times \{-1, 1\}^M$.
A random edge in $G$ is sampled as follows: pick~$\bu \in U$ uniformly at random,
and sample two uniformly random neighbors $\bv, \bw \sim N(\bu)$ independently,
where $N(\bu)$ is the set of $\bu$'s neighbors.
Let $\bx$ and $\by$ be $\rho$-correlated $M$-dimensional Boolean strings.
Output the edge between $(\bv, \bx \circ \pi_{\bv\rightarrow \bu})$ and $(\bw, \by \circ \pi_{\bw \rightarrow \bu})$.
Given $w \in \{-1, 1\}^M$ and $\sigma:[M]\rightarrow [M]$,
we write $w \circ \sigma \in \{-1, 1\}^M$ for the string in which $(w \circ \sigma)_i = w_{\sigma(i)}$.

A product state assignment to~$G$ corresponds to a function $f_v:\{-1, 1\}^M \rightarrow S^2$ for each $v \in V$.
It has value
\begin{equation*}
\E_{\bu \sim U}\E_{\bv, \bw \sim N(\bu)} \E_{\substack{\text{$(\bx, \by)$ $\rho$-correlated}\\\text{$n$-dim Boolean strings}}}
	\left[\tfrac{1}{4} - \tfrac{1}{4}\langle f_{\bv}(\bx \circ \pi_{\bv \rightarrow \bu}), f_{\bw}(\by \circ \pi_{\bw \rightarrow \bu})\rangle\right].
\end{equation*}

\textit{Completeness.} Assume $\mathcal I$ has a labeling $L:U \cup V \rightarrow [M]$ satisfying more than $(1- \gamma)$-fraction of the edges.
For each $v \in V$, let $f_v(x) = (x_{L(v)}, 0, \dots, 0)$.
To analyze the performance of~$f$,
let us first fix a vertex $u \in U$ and two neighbors $v, w \in N(u)$,
and condition on the case that $L$ satisfies both edges $(u, v)$ and $(u, w)$.
This means that $\pi_{v\rightarrow u}(L(v)) = L(u) = \pi_{w\rightarrow u}(L(w))$.
Thus, for each $x \in \{-1, 1\}^M$,
\begin{align*}
f_v(x \circ \pi_{v\rightarrow u}) &= ((x \circ \pi_{v \rightarrow u})_{L(v)}, 0, 0) = (x_{\pi_{v \rightarrow u}(L(v))},0,0) = (x_{L(u)},0,0),
\end{align*}
and similarly $f_w(y \circ \pi_{w \rightarrow u}) = (y_{L(u)}, 0, 0)$ for each $y \in \{-1, 1\}^M$.
As a result, the value of $f$ conditioned on $u$, $v$, and $w$ is
\begin{align*}
\E_{\bx, \by}
	\left[\tfrac{1}{4} - \tfrac{1}{4}\langle f_{v}(\bx \circ \pi_{v \rightarrow u}), f_{w}(\by \circ \pi_{w \rightarrow u})\rangle\right]
=\E_{\bx, \by}
	\left[\tfrac{1}{4} - \tfrac{1}{4}\langle (\bx_{L(u)}, 0, 0), (\by_{L(u)}, 0, 0)\rangle\right],
\end{align*}
which is just the value of the $L(u)$-th embedded dictator on the noisy hypercube, i.e.\ $1/4 - 1/4 \rho$.

Now we average over $\bu, \bv, \bw$.
Because $\mathcal L$ is a biregular Unique Games instance,
it is in particular left-regular,
and so
picking a random vertex $\bu \in U$ and neighbor $\bv \in N(\bu)$ is equivalent to picking a uniformly random edge from~$E$.
Therefore, by the union bound, the probability that the assignment~$L$
satisfies both edges $(\bu, \bv)$ and $(\bu, \bw)$ is at least $1-2\gamma$.
As we have seen, conditioned on this event, the assignment~$f$ has value at least $1/4 - 1/4\rho$.
Due to our choice of $\gamma$, we can lower-bound the value of~$f$ by
\begin{align}
(1-2\gamma) \cdot (\tfrac{1}{4} - \tfrac{1}{4}\rho)
&\geq \tfrac{1}{4} - \tfrac{1}{4}\rho - \gamma \nonumber\\
&\geq \tfrac{1}{4} - \tfrac{1}{4}\rho - \tfrac{1}{2} \eps.\label{eq:used-gamma}
\end{align}
This completes the completeness case.

\textit{Soundness.}
We will show the contrapositive.
Suppose there is a product state assignment $\{f_v\}_{v \in V}$ to $G$ with value at least $\tfrac{1}{4} - \tfrac{1}{4} F^*(3, \rho) + \tfrac{1}{2}\eps$.
We will use this to construct a randomized assignment $\bL:U \cup V \rightarrow [M]$ whose average value is at least~$\gamma$,
which implies that the Unique Games instance has value at least~$\gamma$.

For each $u \in U$, we define the function $g_u: \{-1, 1\}^M \rightarrow B^3$ as
\begin{equation*}
g_u(x) = \E_{\bv \sim N(u)}[ f_{\bv}(x \circ \pi_{\bv \rightarrow u}) ].
\end{equation*}
Then we can rewrite the value of the assignment $\{f_v\}$ as
\begin{align*}
\E_{\bu} \E_{\bv, \bw \sim N(\bu)} \E_{\bx, \by}\left[\tfrac{1}{4} - \tfrac{1}{4}\langle f_{\bv}(\bx \circ \pi_{\bv \rightarrow \bu}), f_{\bw}(\by \circ \pi_{\bw \rightarrow \bu})\rangle\right]
&= \E_{\bu} \E_{\bx, \by}[\tfrac{1}{4} - \tfrac{1}{4}\langle g_{\bu}(\bx), g_{\bu}(\by)\rangle]\\
&= \E_{\bu}[\tfrac{1}{4} - \tfrac{1}{4}  \Stab_\rho[g_{\bu}]].
\end{align*}
Since~$f$ has value at least $\tfrac{1}{4} - \tfrac{1}{4} F^*(3, \rho) + \tfrac{1}{2}\eps$,
an averaging argument implies that at least an $\eps/4$ fraction of $u \in U$ satisfy 
$\tfrac{1}{4} - \tfrac{1}{4}  \Stab_\rho[g_{u}] \geq \tfrac{1}{4} - \tfrac{1}{4} F^*(3, \rho) + \eps/4$.
Rearranging, these $u$'s satisfy
\begin{equation*}
\Stab_\rho[g_u] \leq F^*(3, \rho) - \eps.
\end{equation*}
We call any such~$u$ ``good''.
We now apply the soundness of our dictatorship test to the good $u$'s:
by \Cref{thm:dictator-test}, any such $u$ has a ``notable'' coordinate, i.e.\ an $i$ such that $\Inf^{\leq m}_i[g_u] > \delta$.
Our random assignment will then use this~$i$ as its label for the vertex~$u$:
$\bL(u) = i$. (If~$u$ has multiple notable coordinates, then we pick one of these arbitrarily as the label for~$u$.)

Next we'll need to obtain labels for the neighbors of~$u$.
We will use the condition that~$u$ is good to derive that many of~$u$'s neighbors~$v$ have notable coordinates.
This requires relating the Fourier spectrum of $g_u$ to the Fourier spectra of the neighboring $f_v$'s.
To begin, for any subset $S \subseteq [M]$,
\begin{equation*}
\chi_S(x \circ \pi_{v \rightarrow u})
 = \prod_{i \in S}(x \circ \pi_{v \rightarrow u})_i
 = \prod_{i \in S} x_{\pi_{v\rightarrow u}(i)}
 = \prod_{j \in T} x_j
 = \chi_T(x),
\end{equation*}
where $T = \pi_{v \rightarrow u}(S) = \{\pi_{v\rightarrow u}(i) : i \in S\}$.
As a result,
\begin{equation*}
f_v(x \circ \pi_{v \rightarrow u})
= \sum_{S \subseteq [n]} \widehat{f}_v(S) \chi_{S}(x \circ \pi_{v \rightarrow u})
= \sum_{T \subseteq [n]} \widehat{f}_v(\pi_{u \rightarrow v}(T)) \chi_T(x).
\end{equation*}
Averaging over all $\bv \in N(u)$,
\begin{equation*}
g_u(x)
= \E_{\bv \sim N(u)}[ f_{\bv}(x \circ \pi_{\bv \rightarrow u}) ]
= \sum_{T \subseteq [n]} \E_{\bv \sim N(u)}[\widehat{f}_{\bv}(\pi_{u \rightarrow \bv}(T))] \chi_T(x)
= \sum_{T \subseteq [n]} \widehat{g}(T) \chi_T(x).
\end{equation*}
Hence,
\begin{align*}
\delta
&< \Inf^{\leq m}_i[g_u]\\
&= \sum_{|T| \leq m:  T \ni i} \Vert \widehat{g}(T) \Vert^2_2\\
&= \sum_{|T| \leq m:  T \ni i} \left\Vert \E_{\bv \sim N(u)}[\widehat{f}_{\bv}(\pi_{u \rightarrow \bv}(T))] \right\Vert^2_2\\
& \leq \sum_{|T| \leq m:  T \ni i} \E_{\bv \sim N(u)} \left\Vert \widehat{f}_{\bv}(\pi_{u \rightarrow \bv}(T)) \right\Vert^2_2\tag{because $\Vert \cdot \Vert^2_2$ is convex}\\
& = \E_{\bv \sim N(u)}\left[\sum_{|T| \leq m:  T \ni i} \left\Vert \widehat{f}_{\bv}(\pi_{u \rightarrow \bv}(T)) \right\Vert^2_2\right]\\
& =  \E_{\bv \sim N(u)}[\Inf^{\leq m}_{\pi_{u \rightarrow \bv}(i)}[f_{\bv}]].
\end{align*}
By another averaging argument,
a $\delta/2$-fraction of $u$'s neighbors~$v$ satisfy $\Inf^{\leq m}_{\pi_{u \rightarrow v}(i)}[f_{v}] \geq \delta/2$.
We call these the ``good neighbors''.
For each good neighbor~$v$, the set of possible labels
$$S_v = \{j : \Inf^{\leq m}_j[f_v] \geq \delta/2 \}$$
is non-empty.
In addition, one of these labels~$j$ satisfies $j = \pi_{u \rightarrow v}(i)$.
On the other hand, by \Cref{prop:influence-bound},  $|S_v| \leq 2m/\delta$ and so this set is not too large either.
For each good neighbor,
we assign the label of $\bL(v)$ by picking a uniformly random $j \in S_v$.
For all other vertices (i.e.\ those which are not good or good neighbors),
we assign $\bL$ a random label.

Now we consider the expected number of edges in $\mathcal{I}$ satisfied by $\bL$.
Given a random edge $(\bu, \bv)$,
the probability that $\bu$ is good is at least $\eps/4$;
conditioned on this, the probability that $\bv$ is a good neighbor is at least $\delta/2$.
Assuming both hold, since $S_{\bv}$ is of size at most $2M/\delta$ and contains one label equal to $\pi_{\bu \rightarrow \bv}(L(\bu))$,
then $\bL$ satisfies the edge $(\bu, \bv)$ with probability at least $\delta/2M$.
In total, $\bL$ satisfies at least an
\begin{equation}\label{eq:used-gamma-again}
\frac{\eps \delta^2}{16 M} = \gamma
\end{equation}
fraction of the edges. This concludes the proof.

\textit{Moving from the product state value to the maximum energy.}
First, we modify the graph~$G$ to remove any self-loops.
To do this, we modify the distribution on edges $(\bv, \bx \circ \pi_{\bv \rightarrow \bu})$ and $(\bw, \by \circ \pi_{\bw \rightarrow \bu})$
so that $\bx$ and $\by$ are distributed as $\rho$-correlated Boolean strings \emph{conditioned on them not being equal}.
This removes all self-loops, as any self-loop in the graph must have $\bv = \bw$ and $\bx \circ \pi_{\bv \rightarrow \bu} = \by \circ \pi_{\bw \rightarrow \bu}$,
which implies that $\bx = \by$.
(Note that this also removes some edges which are \emph{not} self-loops,
namely those for which $\bv \neq \bw$.)
Given that $\rho$-correlated $\bx$ and $\by$ are equal with probability
\begin{equation}\label{eq:m-fact}
(\tfrac{1}{2} + \tfrac{1}{2} \rho)^M \leq \eps/4,
\end{equation}
removing this event can only change the product state value of the graph by at most~$\eps/4$.

Next, we apply~\Cref{cor:BH-nonuniform-easy-to-use} to bound the value of~$H_G$ over general states in the soundness case.
To do so, let us write $E$ for the edges of~$G$,
and define $p_{v, x} = \tfrac{1}{2} \Pr_{\boldsymbol{e} \sim E}[\text{$\boldsymbol{e}$ contains $(v, x)$}]$.
Note that the distribution of a random edge $(\bv, \bx \circ \pi_{\bv\rightarrow \bu})$ and $(\bw, \by \circ \pi_{\bw \rightarrow \bu})$
is symmetric and never contains self-loops, and so
$p_{v, x} = \Pr[(\bw, \by \circ \pi_{\bw\rightarrow \bu}) = (v, x)]$.
But the UG instance $\calI$ is biregular,
and so $\bw$ is just a uniformly random element of~$V$,
and $\by$ is just a uniformly random string in $\{-1, 1\}^M$.
Hence, $p_{v, x} =|V|^{-1}2^{-M}$ for each $v, x$,
and so $\max_{v, x}\{p_{v, x}\} = |V|^{-1}2^{-M}$.

The next thing we have to bound to apply~\Cref{cor:BH-nonuniform-easy-to-use}
is the maximum of
\begin{equation*}
\Pr[(\bv, \bx \circ \pi_{\bv\rightarrow \bu}) = (v', x') \mid (\bw, \by \circ \pi_{\bw \rightarrow \bu}) = (w', y')]
\end{equation*}
over all $v', w' \in V$ and $x', y' \in \{-1, 1\}^M$.
Note that if we condition on a fixed value for $\bu$
and on the event that $\bv = v'$,
then this is just the maximum probability that $\bx \circ \pi_{\bv \rightarrow \bu}$
equals a fixed string, given that $\bx$ is $\rho$-correlated but not equal to~$y'$.
Given that $\bx$ is most likely to be $-y'$ since $\rho$ is negative, this probability is
\begin{equation*}
\tfrac{1}{1 - (\tfrac{1}{2} + \tfrac{1}{2} \rho)^M} \cdot (\tfrac{1}{2} - \tfrac{1}{2} \rho)^M \leq 2 \cdot  (\tfrac{1}{2} - \tfrac{1}{2} \rho)^M,
\end{equation*}
where we normalized by $1 - (\tfrac{1}{2} + \tfrac{1}{2} \rho)^M$ due to the condition that $\bx \neq \by$.
Then averaging over $\bu$ and $\bv$ can only decrease this bound.

\ignore{
Next, let $A$ be the matrix such that
\begin{equation*}
A_{(v, x), (w, y)} = \Pr_{((\bv', \bx'), (\bw', \by')) \sim E}[(\bv', \bx') = (v, x) \mid (\bw', \by') = (w, y)].
\end{equation*}
$A$ can be viewed as the transition matrix of a random walk on the vertices of~$G$,
in which the probability of transitioning from vertex $(w, y)$ to $(v, x)$ is given by $A_{(v, x), (w, y)}$.
Let $(\bv_0, \bx_0)$ be a uniformly random starting vertex,
and let $(\bv_1, \bx_1)$ and $(\bv_2, \bx_2)$ be the next two steps of the walk.
Then
\begin{equation*}
\tr[A^2] = \sum_{(w, x)} \Pr[(\bv_2, \bx_2) = (w, x) \mid (\bv_0, \bx_0) = (w, x)].
\end{equation*}
As a result, because $(\bv_0, \bx_0)$ is uniformly random,
\begin{equation*}
\tr[A^2] \Vert p \Vert_2^2
= \sum_{(w, x)} \Pr[(\bv_0, \bx_0) = (\bv_2, \bx_2) = (w, x)]
= \Pr[(\bv_0, \bx_0) = (\bv_2, \bx_2)].
\end{equation*}
Let us condition this probability on fixed values of $(v_0, x_0)$ and $(v_1, x_1)$.
Then we can bound the probability that $(\bv_2, \bx_2) = (v_0, x_0)$
by the probability that $\bx_2 = x_0$.
But given $x_1$, $\bx_2$ will be generated by taking a $\rho$-correlated sample $\by \sim_{\rho} x_1$
and applying some permutation to $\by$.
So we can bound this by the maximum probability that a~$\rho$-correlated sample equals a fixed string, which is
\begin{equation*}
\tfrac{1}{1 - (\tfrac{1}{2} + \tfrac{1}{2} \rho)^M} \cdot (\tfrac{1}{2} - \tfrac{1}{2} \rho)^M \leq 2 \cdot  (\tfrac{1}{2} - \tfrac{1}{2} \rho)^M,
\end{equation*}
where we recall that all edge weights are normalized by $1 - (\tfrac{1}{2} + \tfrac{1}{2} \rho)^M$
because we removed self-loops.
}

Now we can apply \Cref{cor:BH-nonuniform-easy-to-use},
which states that
\begin{equation*}
\heis(G) \leq \prodval(G) + 20\cdot (2 \cdot  (\tfrac{1}{2} - \tfrac{1}{2} \rho)^M)^{1/8} + \tfrac{1}{|V| 2^M},
\end{equation*}
which is at most $\eps/4$ by our choice of~$M$.
Hence,
\begin{equation}\label{eq:m-fact-2}
\heis(G) \leq \tfrac{1}{4} - \tfrac{1}{4} F^*(3, \rho) + \eps,
\end{equation}
which completes the proof.
\end{proof}

\ignore{
\section{Unique Games Hardness of the Quantum Heisenberg Model}

\subsection{Definitions and Preliminaries}

In this section, we apply \Cref{thm:vector-borell} to show that solving ``difficult'' instances of the Unique Label Cover reduces to $\alpha_\mathrm{BOV}$-approximating the product state value of the Heisenberg model. Assuming the Unique Games conjecture of Khot \cite{Kho02}, this yields a $\NP$-hardness result for product state approximation of the Heisenberg model.

\begin{definition}[Unique Label Cover]\label{def:unique-games}
The Unique Label Cover problem, $\mathcal L(V, W, E, [M], \{\sigma_{v,w}\}_{(v,w) \in E})$ is defined as follows: Given is a biregular, bipartite graph with left side vertices $V$, right side vertices $W$, and a set of edges $E$. The goal is to assign one `label' to every vertex of the graph, where $[M]$ is the set of allowed labels. The labeling is supposed to satisfy certain constraints given by bijective maps $\sigma_{v,w} : [M] \rightarrow [M]$. There is one such map for every edge $(v,w) \in E$. A labeling $L : V \cup W \rightarrow [M]$ `satisfies' an edge $(v,w)$ if
$$\sigma_{v,w}(L(w)) = L(v)$$
\end{definition}
	
\begin{conjecture}[Unique Games Conjecture \cite{Kho02}]
For any $\eta, \gamma > 0$, there exists a constant $M = M(\eta, \gamma)$ such that it is $\NP$-hard to distinguish whether a Unique Label Cover problem with label set size $M$ has optimum at least $1 - \eta$ or at most $\gamma$.
\end{conjecture}

Our hardness result is stated as follows,

\begin{theorem}[UG-Hardness of Approximating Quantum \maxcut]\label{thm:ug-hardness}
For any $\rho \in [-1,0)$ and $\epsilon > 0$, there exists an instance of the Heisenberg Hamiltonian such that deciding whether the maximum energy is greater than $\frac{1- \rho}{4} - \epsilon$ or less than $\frac{1 - F^\star(3, \rho)}{4} + \epsilon$ is Unique Games-Hard. In more standard notation, we say that it is UG-Hard to $(\frac{1 - F^\star(3,\rho)}{4} + \epsilon, \frac{1- \rho}{4} - \epsilon)$-approximate Quantum \maxcut.
\end{theorem}

\begin{remark}
$\frac{1 - F^\star(3,\rho)}{4}$ is exactly the product state value obtained by the rounding algorithm of \cite{GP19} on the integrality gap instance in \Cref{sec:integrality-gaps}.
\end{remark}

\begin{remark}
Minimizing the ratio $\frac{1-\rho}{4}/\frac{1-F^\star(3,\rho)}{4}$ w.r.t $\rho$ yields $\alpha_\mathrm{BOV}$, which shows \Cref{thm:main-inapprox}.
\end{remark}

To show this result, we slightly generalize the framework of showing Unique Games hardness results for CSPs, developed in \cite{KKMO07} for the case of classical $\maxcut$. In that work, the authors show that in order to show hardness of approximation results for a general CSP over domain $\{-1,1\}$, it suffices to develop a ``Dictator-vs.-No-Notables'' test (also known as Dicator-vs.-Quasirandomness). Such a test distinguishes between ``dictators'', which are functions of the form $f(x_1,\dots,x_n) = x_i$, and functions which have no influential coordinates.fi

\begin{definition}[$(\alpha,\beta)$-Dictator-vs.-No-Notables \cite{KKMO07,OD14}]\label{def:dictator-no-notables}
Let $\Psi$ be a finite set of predicates over the domain $\Omega = \{-1,1\}$. Let $0 < \alpha < \beta \leq 1$ and let $\lambda : [0,1] \rightarrow [0,1]$ satisfy $\lambda(\epsilon) \rightarrow 0$ as $\epsilon \rightarrow 0$. Suppose that for each $n \in \mathbb N^+$ there is a local tester for functions $f : \{-1,1\}^n \rightarrow \{-1,1\}$ with the following properties:
\begin{itemize}
	\item \textbf{Completeness}. If $f$ is a dictator then the test accepts with probability at least $\beta$.
	\item \textbf{Soundness}. If $f$ has no $(\epsilon,\epsilon)$-notable coordinates -- i.e., $\Inf_i^{\leq 1-\epsilon}[f] \leq \epsilon$ for all $i \in [n]$ then the test accepts with probability at most $\alpha + \lambda(\epsilon)$.
	\item The tester’s accept/reject decision uses predicates from $\Psi$; i.e., the tester can be viewed as an instance of Max-CSP($\Psi$).
\end{itemize}
\end{definition}

In the case of \maxcut, the appropriate test is the Long Code test, which queries $f$ at two $\rho$-correlated points $x,x'$ then checks $f(x) \not= f(x')$. Completeness is straightforward to verify. However, the showing soundness requires the Majority is Stablest theorem, which bounds the maximum stability of a function with no influential coordinates. The proof of Majority is Stablest uses an ``invariance principle'' to relate functions over discrete domains (e.g. $\{-1,1\}^n$) and no highly influential coordinates to functions over continuous domains ($\mathbb R^n$).

\begin{theorem}[Invariance Principle \cite{MOO10, OD14}]\label{thm:invariance}
Fix $\epsilon, \gamma \in [0,1]$ and set $d = \log\tfrac{1}{\epsilon}$. Let $f$ be a $n$-variate multilinear polynomial 
$$f(x) = \sum_{S \subseteq N} \hat f(S) \prod_{i \in S} x_i$$
such that $\Var[f] \leq 1$, $\Var[f^{> d}] < (1-\gamma)^{2d}$, and $\Inf^{\leq d}_i[f] \leq \epsilon$ for all $i \in [n]$. Let $\bx$ be a uniformly random string over $\{-1,1\}^n$ and $\by$ a $n$-dimensional standard Gaussian random variable. Assume $\psi : \mathbb R \rightarrow \mathbb R$ is $c$-Lipschitz continuous, then
$$|\E[\psi(f(\bx))] - \E[\psi(f(\by))] \leq O(c) \cdot 2^k \epsilon^{1/4}$$
\end{theorem}
The Majority is Stablest theorem is the following,
\begin{theorem}[Majority is Stablest \cite{MOO10, KKMO07}\footnote{This theorem was proved in \cite{MOO10} for positive $\rho$ and the condition $\Inf_i(f) \leq \delta$. The version we state here is due to \cite{KKMO07} and was originally given as a conjecture.}]\label{thm:mis}
Fix $\rho \in (-1,0]$. For any $\epsilon > 0$, there is a small enough $\delta = \delta(\epsilon, \rho)$ and large enough $k = k(\epsilon, \rho)$ such that if $f : \{-1,1\}^n \rightarrow [-1,1]$ is any function satisfying $\Inf^{\leq k}_i[f] \leq \delta$ for all $i \in [n]$, then
$$\mathbb S_\rho(f) \geq \Lambda_\rho(\mu) + \epsilon$$
\end{theorem}
Here, $\Lambda_\rho(\mu)$ is the \textit{Gaussian quadrant probability function} \cite{OD14}. Now, given a $(\alpha,\beta)$-Dictator-vs.-No-Notables test one can show this yields a UG-hardness result for a related constraint satisfaction problem.

\begin{theorem}[\cite{KKMO07}]\label{thm:classical-ug-hardness}
Fix a CSP over domain $\Omega = \{-1,1\}$ with predicate set $\Psi$. Suppose there exists an $(\alpha,\beta)$-Dictator-vs.-No-Notables test using predicate set $\Psi$. Then for all $\delta > 0$, it is ``UG-hard'' to $(\alpha+\delta,\beta- \delta)$-approximate $\maxcsp(\Psi)$.
\end{theorem}

\subsubsection{Example: Max Cut}

Assuming the existence of a test satisfying \Cref{def:dictator-no-notables}, let's consider how \Cref{thm:classical-ug-hardness} works for classical $\maxcut$. Given a Unique Label Cover instance on graph $G = (V,E)$ with labels $[M]$, we observe that a labeling to vertex, i.e. $L(v)$ for $v \in V$, can be represented by the dicator $f_{L(v)}(x) = x_{L(v)}$, where $f : \{-1,1\}^M \rightarrow \{-1,1\}$. When $G$ has a satisfying assignment, given two neighbors $w,w'$ of $v$, we have that $L(v)) = \sigma_{(v,w)}(L(w)) = \sigma_{(v,w')}(L(w'))$. Therefore, we have that,
$$f_{L(v)}(x) = f_{L(w)}(x \circ \sigma_{(v,w)}) = f_{L(w')}(x \circ \sigma_{(v,w')})$$
where $x \circ \sigma$ is used to represent the string $(x_{\sigma(1)},x_{\sigma(2)},\dots,x_{\sigma(n)})$. In particular, $g_w = f_{L(w)}(x \circ \sigma_{(v,w)})$ yields the same dicator as $g_w'$ and applying the $(\alpha,\beta)$-Dictator-vs.-No-Notables test yields success with probability $\geq \beta$. Note that since $\maxcut$ is $\maxcsp(\not=)$, we use the $\not=$ predicate.

On the other hand, if the test succeeds with probability $\alpha + \lambda(\epsilon)$, then for some constant fraction of $v \in V$, the test passes with high probability on neighbors $w,w'$ of $v$. By the second property in \Cref{def:dictator-no-notables}, $g_w, g_{w'}$ have some influential coordinates. Decoding these coordinates as a labeling for $v$ yields a labeling for $G$ which satisfies a large fraction of the edges.

\subsection{Vector-Valued Majority is Stablest}

\subsubsection{Definitions}

Throughout this section, we'll work extensively with functions expressed as multilinear polynomials. Given the function $f : \Omega_1 \times \cdots \times \Omega_n \rightarrow \mathbb R$, the corresponding polynomial is of the form
$$Q(x) = \sum_S c_S \chi_S$$
where $\chi_S = \Pi_{i \in S} x_i$, and $x_i$ is a variable over $\Omega_i$. The \textit{degree} of $Q(x)$ is $\max \{|S| : c_S \not= 0\}$. We also use the notation
$$Q^{\leq d}(x) = \sum_{S : |S| \leq d} c_S \chi_S$$
These polynomials will be evaluated on independent, identically distributed collections of random variables $(x_1,\dots,x_n)$. Furthermore, these variables are restricted to be \textit{orthonormal}. Then, we are able to obtain the following formulas, familiar from the analysis of boolean functions,

\begin{proposition}\label{prop:fourier-properties}
Let $\bx = (x_1,\dots,x_n)$ be a collection of orthonormal, i.i.d. random variables. Then,
$$\E[Q(\bx)] = c_\emptyset \qquad \E[Q(\bx)^2] = \sum_S c^2_S \qquad \Var[Q(\bx)] = \sum_{S : |S| > 0} c^2_S$$
$$\Inf_i[Q(\bx)] = \sum_{S : i \in S} c^2_S \qquad \T_\rho[Q(\bx)] = \sum_S \rho^{|S|} c_S \chi_S \qquad \Stab_{\rho}[Q(\bx)] = \sum_S \rho^{|S|} c^2_S$$
\end{proposition}
Note that the above formulas, with the exception of $\T_\rho$, do not depend on the random variables $\bx$, and instead on the polynomial $Q$. $\T_\rho$ should be interpretted as an operator on multilinear polynomials. $Q$ will often correspond to a function $f$, and as a slight abuse of notation, we will use $f$ in place of $Q$ as an argument to the above definitions (e.g. $\Stab_\rho[f]$). In this case, $\hat f(S)$ refers to the coefficient $c_S$. The standard definition for the norm of a function defined over random variables $\bx$ is,
$$\norm{f}_2^2 = \langle f, f \rangle = \E_{x\sim \bx}[f(x)^2]$$
Observe that this is exactly $\sum_S c_S^2 = \sum_S \hat f(S)^2$. The inner product can be extended to two different functions $f,g$ and yields $\langle f, g \rangle = \sum_S \hat f(S) \hat g(S)$. Next, we consider the case of vector-valued functions. We will work most often with functions mapping boolean strings $\{-1,1\}^n$ to $\mathbb R^m$. Such a function can be expressed as $m$ different functions for each output coordinate, written as $f = (f_1,\dots,f_m)$. In general, the definitions in \Cref{prop:fourier-properties} can be applied to each $f_i$, and extended to $f$ by taking the sum. For instance,
\begin{equation}\label{eq:stab}
\Stab_\rho[f] = \sum_{i\in[n]} \Stab_\rho[f_i] = \sum_{i\in[n]} \sum_S \rho^{|S|} {\hat f_i(S)}^2
\end{equation}
We might want to consider vector-valued polynomial, rather than $m$ different polynomials for each $f_i$. We define,
$$f = \sum_S \hat{f}(S) \chi_S$$
where $\hat{f} = (\hat f_1,\dots,\hat f_m)$. Noting that $\norm{(x_1,\dots,x_m)}_2 = \sqrt{x_1^2 + \cdots + x_m^2}$, we can simplify \Cref{eq:stab} as
$$\Stab_\rho[f] = \sum_{i\in[n]} \sum_S \rho^{|S|} {\hat f_i(S)}^2 = \sum_S \rho^{|S|}\norm{\hat{f}(S)}_2^2$$
A special case of one of the above definitions will be ``low-degree'' influences.
$$\Inf^{\leq d}_i[Q(\bx)] = \sum_{S : |S| \leq d, i \in S} c^2_S$$

\subsubsection{Majority is Stablest}

Now, we will generalize the Majority is Stablest theorem to vector-valued functions. Fortunately, a suitable generalization of the invariance principle already exists.

\begin{theorem}[Vector Invariance Principle \cite{IM12}]\label{thm:vector-invariance}
Fix $\tau, \gamma \in [0,1]$ and set $d = \log\tfrac{1}{\tau}$. Let $f = (f_1,\dots,f_m)$ be a $m$-dimensional multilinear polynomial such that $\Var[f_j] \leq 1$, $\Var[f_j^{> d}] < (1-\gamma)^{2d}$, and $\Inf^{\leq d}_i[f_j] \leq \tau$ for each $j \in [m]$ and $i \in [n]$. Let $\bx$ be a uniformly random string over $\{-1,1\}^n$ and $\by$ a $n$-dimensional standard Gaussian random variable. Furthermore, let $\Psi : \mathbb R^m \rightarrow \mathbb R$ be Lipschitz continuous with Lipschitz constant $A$. Then,
$$|\E[\Psi(f(\bx))] - \E[\Psi(f(\by))]| \leq C_mA \tau^{\gamma/18}$$
where $C_m$ is a parameter depending only on $m$.
\end{theorem}

Carrying out a proof similar to that of \Cref{thm:mis} yields the following, 

\begin{theorem}[Vector Valued Majority is Stablest]\label{thm:dictator-test}
Fix $\rho \in (-1,0]$. Then for any $\epsilon > 0$, there exists small enough $\delta = \delta(\epsilon, \rho)$ and large enough $k = k(\epsilon, \rho) \geq 0$ such that if $f : \{-1,1\}^n \rightarrow B^{m-1}$ is any function satisfying
$$\Inf^{\leq k}_i[f] = \sum_{j=1}^m \Inf^{\leq k}_i[f_j] \leq \delta \text{ for all }i = 1, \dots, n$$
then
$$\E_{x \sim_\rho y} [f(x)f(y)] = \Stab_\rho[f] \geq F^*(m,\rho) - \epsilon$$
\end{theorem}

\begin{proof}
Throughout, we'll use $\bx$ to denote string in $\{-1,1\}^n$ and $\by$ to denote a vector in $\mathbb R^n$. Since the statement of \Cref{thm:vector-invariance} requires a function with low high-degree variance, we consider $g = \T_{1-\gamma} f$. Then for each $j \in [m]$,
$$\Var[g_j^{\geq d}] = \sum_{|S| \geq d} (1-\gamma)^2\hat {g}_j(S)^2 \leq (1-\gamma)^{2d} \Var[g_j^{\geq d}] \leq (1-\gamma)^{2d}$$
Also, $\Inf_i[g_j] \leq \Inf_i[f_j]$. Next, we bound the error in the quantity $\Stab_\rho[f]$ when we consider $g$ in place $f$.
\begin{align}
\ABS{\Stab_\rho[f] - \Stab_\rho[g]} &= \ABS{\sum_S \rho^{|S|} \norm{\widehat{f}(S)}_2^2 - \sum_S (\rho(1-\gamma)^2)^{|S|} \norm{\widehat{f}(S)}_2^2}\nonumber\\
&=\sum_S \ABS{(\rho(1 - (1-\gamma)^2))^{|S|}} \norm{\widehat{f}(S)}_2^2\nonumber\\
&\leq \rho(1-(1-\gamma)^2) \sum_S\norm{\hat f(S)}_2^2\nonumber\\
&\leq \rho(1-(1-\gamma)^2) \E[\norm{f}_2^2]\label{eq:bound1}
\end{align}
$f$ has range $B^{m-1}$ and thus $\E[\norm{f}_2^2] \leq 1$. By choosing $\gamma(\epsilon, \rho)$ to be sufficiently small we can bound this quantity by $\epsilon/2$.
Next, since $\Stab_{\rho}$ is a function of the coefficients of the polynomial corresponding to $f$, we have that $\Stab_\rho[g(\bx)] = \Stab_\rho[g(\by)]$. However, although $g(\bx) \in B^{m-1}$ (since $f(\bx) \in B^{m-1}$ by definition), $g(y)$ may take values outside the unit ball. This prevents us from directly applying \Cref{thm:vector-borell}. We'll instead apply the theorem to the function
$$g'(\by) = \begin{cases} 
g(\by) & \text{if }g(\by) \in B^{m-1}\\
\frac{g(\by)}{\NORM{g(\by)}} & \text{otherwise}
\end{cases}$$
Applying \Cref{thm:vector-borell} yields that $\Stab_{\rho}[g'] \geq \Stab_{\rho}[f_\mathrm{opt}] = F^\star(m,\rho)$. It remains to bound $\ABS{\Stab_{\rho}[g] - \Stab_{\rho}[g']}$.
\begin{align*}
\ABS{\Stab_\rho[g] - \Stab_\rho[g']} &= \ABS{\langle g, \T_\rho g\rangle - \langle g', \T_\rho g'\rangle}\\
&= \ABS{\langle g, \T_\rho g\rangle - \langle g', \T_\rho g\rangle - \PAREN{\langle g', \T_\rho g'\rangle - \langle g', \T_\rho g\rangle}}\\
&= \ABS{\langle g - g', \T_\rho g\rangle +\langle g', \T_\rho g - \T_\rho g'\rangle}\\
&\leq \ABS{\norm{g - g'}_2\norm{\T_\rho g}_2 + \norm{g'}_2\norm{\T_\rho g - \T_\rho g'}_2} && \text{Cauchy-Schwartz}\\
&\leq \ABS{\norm{g - g'}_2\norm{g}_2 + \norm{g'}_2\norm{g - g'}_2} && \text{$\T_\rho$ is a contraction}\\
&\leq \norm{g - g'}_2(\norm{g}_2 + \norm{g'}_2)\\
&\leq \norm{g - g'}_2(\norm{g - g' + g'}_2 + \norm{g'}_2)\\
&\leq \norm{g - g'}_2(\norm{g - g'}_2 + 2\norm{g'}_2)\\
&\leq 3\norm{g - g'}_2
\end{align*}
where we have used that for $\norm{g - g'}_2 \leq 1$, $\norm{g - g'}_2^2 \leq \norm{g - g'}_2$ and that $\norm{g'}_2 \leq 1$. Now we just need to bound $\norm{g - g'}_2$ by $\epsilon/6$. Note that $\norm{g - g'}_2^2 = \E[\zeta(g(y))]$, where $\zeta$ is the 1-Lipschitz \ynote{Should we include a proof of this?} function yielding the squared distance of a point $y \in \mathbb R^n$ from $B^{m-1}$
\begin{equation*}
\zeta(y) = \begin{cases}
	0 & \text{if }\norm{y}_2 \leq 1\\
	\norm{y - \tfrac{y}{\norm{y}}}_2^2 & \text{otherwise}
\end{cases}
\end{equation*}
Let $\delta = (\epsilon^2/(36\cdot C_m))^{18/\gamma}$. We can now apply \Cref{thm:vector-invariance} with $\Psi = \zeta$ and $\tau = \delta$, which yields
$$\E_{\by}[\zeta(g)] = |\E_{\bx}[\zeta(g(\bx))] - \E_{\by}[\zeta(g(\by))]| \leq \epsilon^2/36$$
Then, $\norm{g-g'}_2 \leq \sqrt{\epsilon^2/36} = \epsilon/6$ and $\ABS{\Stab_\rho[g] - \Stab_\rho[g']} \leq \epsilon/2$. Combining this with our bound of $\epsilon/2$ for \Cref{eq:bound1} concludes the proof.
\end{proof}

\subsection{Embedded Dictator-vs.-No-Notables}

Recall that in the case of quantum \maxcut, a product state assignment corresponds to a Bloch vector $v \in S^2$. With this in mind, we extend the domain of predicates in \Cref{def:dictator-no-notables} to $S^2$. Additionally, rather than working with a tester $T$ which output either \texttt{YES}/\texttt{NO} (equivalently $1$/$0$), our tester will output a value $v \in [-1,1]$. As such, rather than characterizing $T$ using a success probability, we work with the expected value of $T$.\footnote{Note for $1$/$0$-valued testers, the expected value and success probability are exactly the same.}.

\begin{definition}[Real-Valued Tester]
A real-valued tester $T$ for functions $f : \{-1,1\}^n \rightarrow S^{m-1}$ of type $\Psi : (S^{m-1})^r \rightarrow [-1,1]$ is a randomized algorithm which selects $r$ strings $x_1,\dots,x_r \in \{-1,1\}^n$ from a distribution $\mathcal D$, queries $f(x_1),\dots,f(x_r)$, then outputs the result of applying predicate $\Psi$. The value of such a tester is
$$\val_T(f) = \E_{\bx_1,\dots,\bx_r \sim \mathcal D}[\Psi(f(\bx_1),\dots,f(\bx_r))]$$
\end{definition}
\begin{notation} Given a tester $T$ with a length $l$ predicate and a collection of $l$ functions $F = \{f^{(1)},\dots,f^{(l)}\}$ we use the notation $T(F)$ to denote evaluating $T$ on the collection $F$ as follows,
$$\E_{\bx_1,\dots,\bx_r \sim \mathcal D}[\Psi(f^{(1)}(\bx_1),\dots,f^{(l)}(\bx_r))]$$
\end{notation}
Next, to generalize the notion of a dictator to vector-valued functions, we introduce the \textit{embedded dictators}.

\begin{definition}[Embedded Dictator]
Let $f : \{-1,1\}^n \rightarrow S^{m-1}$ be a vector-valued function, with $f(x) = (f_1(x),\dots,f_{n+1}(x))$. Then, we say $f$ is a embedded dictator if $f_i$ is a dictator, for some $i$ and $f_j = 0$ for all other $j \not= i$.
\end{definition}

Then, we have the following test.

\begin{definition}[$(\alpha,\beta)$-Embedded-Dictator-vs.-No-Notables]\label{def:embedded-dic}
Let $\Psi$ be a finite set of predicates over the domain $S^2$. Let $-1 \leq \beta < 0 < \alpha \leq 1$ and let $\lambda : [0,1] \rightarrow [0,1]$ satisfy $\lambda(\epsilon) \rightarrow 0$ as $\epsilon \rightarrow 0$. Suppose that for each $n \in \mathbb N^+$ there is a real-valued tester for functions $f : \{-1,1\}^n \rightarrow S^2$ with the following properties:
\begin{itemize}
	\item \textbf{Completeness}. If $f$ is an embedded dictator then the test has expected value at least $\beta$.
	\item \textbf{Soundness}. If $f$ has no $(\epsilon,\epsilon)$-notable coordinates -- i.e., $\Inf_j^{\leq 1-\epsilon}[f]  = \sum_{i \in [3]} \Inf_j^{\leq 1-\epsilon}[f_i] \leq \epsilon$ for all $j \in [n]$ then the test has expected value at most $\alpha + \lambda(\epsilon)$.
\end{itemize}
Then this yields a family of testers called a $(\alpha,\beta)$-Embedded-Dicator-vs.-No-Notables test using predicate $\Psi$.
\end{definition}

We'll also need a simple lemma, which bounds the maximum number of coordinates with large low degree influence. For a proof see \Cref{sec:fourier-proofs}.

\begin{lemma}[Bounds on Influential Coordinates]\label{lem:influence-bound}
Fix a function $f : \{-1,1\}^n \rightarrow S^{m-1}$. Then the set,
$$S = \{i \in [n] : \Inf^{\leq 1-\epsilon}_i[f] > \tau\}$$
has size at most $m(1-\epsilon)/\tau$.
\end{lemma}

Before giving a tester satisfying these properties, we state our analogue of \Cref{thm:classical-ug-hardness}, which shows that such a tester is useful for showing hardness-of-approximation results.

\begin{theorem}\label{thm:quant-ug-hardness}
Fix a CSP over domain $S^{m-1}$ with predicate set $\Psi$. Suppose there exists an $(\alpha,\beta)$-Embedded-Dictator-vs.-No-Notables test using predicate set $\Psi$. Then for all $\delta > 0$, it is ``UG-hard'' to $(\alpha+\delta,\beta- \delta)$-approximate $\maxcsp(\Psi)$.\footnote{For simplicity, we assume each constraint in $T$ has width 2, but we could easily extend this to width $c$ constraints.}
\end{theorem}

\begin{proof}
As in \cite{KKMO07}, the proof will be a reduction from a hard instance of Unique Label Cover $\mathcal L$ to a Max-CSP($\Psi$) instance $\mathcal I$. Pick constants $\eta, \gamma$ and $M = M(\eta, \gamma)$ satisfying the Unique Games Conjecture. Then, take an instance of Unique Label Cover $\mathcal L(V, W, E, [M], \{\sigma_{v,w}\}_{(v,w) \in E})$ and for each vertex $v \in V \cup W$, associate the function $f_v : \{-1,1\}^M \rightarrow S^{m-1}$. Furthermore, let $T = \Psi(f(x_1),f(x_2))$ be a tester satisfying the properties given in \Cref{def:embedded-dic}. To form our instance $\mathcal I$, consider all triples $(v,w, w')$ where $v \in V$ and $w, w' \in N(v) \subseteq W$. Then, evaluate $T$ on the collection $\{g_w,g_{w'}\}$, which yields an instance of \maxcsp($\Psi$) of the form
\begin{equation}\label{eq:constraint}
\mathcal I = \E_{(\bv,\bw,\bw')}[T(\{g_{\bw},g_{\bw'}\})]
\end{equation}
As before, $g_{w}(x) = f_w(x \circ \sigma_{v,w})$. Observe when $(v,w)$ is satisfied under $\sigma_{v,w}$, $g_w(x) = f_v(x)$.

In the completeness case, we show that when $\mathcal L$ has value at least $1- \eta$, we can construct $f_v$ such that evaluating each constraint of \Cref{eq:constraint} is analogous to applying $T$ to an embedded dictator $g$. For soundness, we use the contrapositive and show that if $\mathcal I$ has value more than $\alpha + 2\lambda(\epsilon)$, then there is a labeling satisfying at least $\gamma$ portion of the edges.\\

\textit{Completeness.} Assume $\mathcal L$ has a labeling $L$ satisfying more than $(1- \eta)$-fraction of the edges. For each $v \in W \cup V$, let $f_v(x) = (x_{L(v)}, 0, \dots, 0)$. Select triples $(v,w,w')$ by first picking a vertex $v \in V$ and neighbors $w,w' \in N(v) \subseteq W$. By a result in \cite{KR08}, we can assume graph for $\mathcal L$ is uniform on $V$'s side; then, picking a random vertex $v \in V$ and neighbor $w \in N(v)$ is equivalent to picking a uniformly random edge from $E$. Therefore, we by applying the union bound, we see that the probability edges $(v,w)$ and $(v,w')$ are both satisfied in the labeling is at least $1 - 2\eta$. Denote this event $\mathrm{Sat}_{v,w,w'}$. Then we can rewrite $\mathcal I$ as
\begin{align*}
\E_{(\bv,\bw,\bw')}[T(\{g_{\bw}, g_{\bw'}\})] &= \E_{(\bv,\bw,\bw')}[T(\{g_{\bw}, g_{\bw'}\}) | \mathrm{Sat}_{\bv,\bw,\bw'}]\Pr[\mathrm{Sat}_{\bv,\bw,\bw'}]\\
&\qquad + \E_{(\bv,\bw,\bw')}[T(\{g_{\bw}, g_{\bw'}\}) | \neg\mathrm{Sat}_{\bv,\bw,\bw'}]\Pr[\neg\mathrm{Sat}_{\bv,\bw,\bw'}]\\
&\geq \E_{(\bv,\bw,\bw')}[T(\{g_{\bw}, g_{\bw'}\}) | \mathrm{Sat}_{\bv,\bw,\bw'}](1-2\eta)
\end{align*}
When $(v,w)$ and $(v,w')$ are both satisfied, $g_w(x) = g_{w'}(x) = f_v(x)$ and this is equivalent to applying $T$ to the embedded dictator function $g_v$. Thus, $\mathcal I$ has value at least $(1-2\eta)\beta = \beta - O(\eta)$.

\textit{Soundness.} Assume instance $\mathcal I$ has value more than $\alpha + 2\lambda(\epsilon)$. Let $F = \{g_v\}_{v \in V \cup W}$ be an assignment which obtains this value. We observe that,
\begin{align*}
\val_\mathcal{I}(F) &= \E_{(\bv,\bw,\bw')}[T(\{g_{\bw},g_{\bw'}\})]\\
&= \E_{(\bv,\bw,\bw')}\BRAC{\E_{\bx_1,\bx_2 \sim \mathcal D}[\Psi(g_{\bw}(\bx_1),g_{\bw'}(\bx_2))]}\\
&= \E_{\bv}\BRAC{\E_{\bx_1,\bx_2 \sim \mathcal D}[\Psi(\E_{\bw} g_{\bw}(\bx_1),\E_{\bw} g_{\bw}(\bx_2))]}\\
&= \E_{\bv} \BRAC{T\BRAC{\E_{\bw \in N({\bv})} g_{\bw}}}
\end{align*}
Where the second to last equality is because $(w,w') \in N(v)$ are selected uniformly and independently at random. Thus, we can write $\val_\mathcal{I}(F)$ as $\E_{\bv \in V}[T(h_{\bv})]$, where $h_v = \E_{\bw \in N(v)}g_{\bw}$. Let $V'$ be a subset of $V$ such that for each $v \in V'$, $T(h_v) \geq \alpha + \lambda(\epsilon)$. We call these \textit{good} vertices. Furthermore, let $c = |V'|/|V|$. By writing $\val_\mathcal{I}(F)$ as $c\E_{\bv \in V'}[T(h_{\bv})] + (1-c)\E_{\bv \in V \setminus V'}[T(h_{\bv})]$, one can see that $c \geq \lambda(\epsilon)$ and we can select $\lambda(\epsilon)$ good vertices. Thus, since $T$ a $(\alpha,\beta)$-Embedded-Dicatator-vs.-No-Notables test, for each $v$ the set
$$\mathrm{Notables}_v = \{j \in [M] : \Inf_j^{\leq 1-\epsilon}[h_v] > \epsilon\}$$
is non-empty. Selecting just one $j$ from each $\mathrm{Notables}_v$ yields an assignment to the $\lambda(\epsilon)$ good vertices $v \in V'$.

Next we'll need to obtain labels for neighbors of the good vertices. Fix a good vertex $v$. Note that the function $\Inf_j$ is convex and therefore, $\epsilon < \Inf_j^{\leq 1-\epsilon}[h_v] \leq \mathbb E_{\bw \in N(v)}[\Inf_j^{\leq 1-\epsilon}[g_{\bw}]]$. Applying the same averaging argument yields that for $\epsilon/2$ many neighbors $w$, $\Inf_j^{\leq 1-\epsilon}[g_w] \geq \epsilon/2$. Recalling that $g_w(x) = f_w(x \circ \sigma_{(v,w)})$, we equivalently obtain $\Inf_{\sigma^{-1}_{(v,w)}(j)}^{\leq 1-\epsilon}[f_w] \geq \epsilon/2$. This shows that the set of good labels
$$S_w = \{j : \Inf^{\leq 1-\epsilon}_j[f_w] \geq \epsilon/2 \}$$
is non-empty. Next, by \Cref{lem:influence-bound}, we note $|S_w| \leq 6(1-\epsilon)/\epsilon$. Assign a label to $w$ by picking a uniformly random $j \in S_w$.

Finally, we consider the expected number of satisfied edges. The probability of selecting a good vertex $v \in V$ is at least $\lambda(\epsilon)$. For each such vertex, $\epsilon/2$ of its neighbors are assigned a label $j \in S_w$. Since this set has size at most $6(1-\epsilon)/\epsilon$, this yields a probability of $\epsilon/(6(1-\epsilon))$ of selecting the appropriate label. Thus, in expectation we satisfy
$$\frac{\lambda(\epsilon)\epsilon^2}{6(1-\epsilon)}$$
of the edges. Setting $\gamma =\frac{\lambda(\epsilon)\epsilon^2}{6(1-\epsilon)}$ concludes the proof.
\end{proof}

\subsection{Test for Quantum Heisenberg Model}

With the reassurance that constructing a test satisfying \Cref{def:dictator-no-notables} yields a $(\alpha+\delta, \beta-\delta)$-hardness result, we give our test $T$ \ynote{name this test?} for the predicate $\Psi(\bx,\by) = 1 - \langle \bx, \by \rangle$ on functions $f \colon \{-1,1\}^n \rightarrow S^{m-1}$. Let $\mathcal D_T$ be a distribution corresponding to the following sampling procedure
\begin{enumerate}
	\item Pick $x \in \{-1,1\}^n$ randomly.
	\item Construct $y \in \{-1,1\}^n$ by copying $x$, then flipping each coordinate independently with probability $\rho$.
\end{enumerate}
Given $\bx,\by \sim \mathcal D_T$, return $1 - \langle f(\bx), f(\by)\rangle$. Then, $\val_T(f) = \E_{\bx,\by \sim \mathcal D_T}[1 - \langle f(\bx), f(\by)\rangle]$.
\begin{lemma}
Test $T$ is a $(1-F^\star(m, \rho),1- \rho)$-Embedded-Dictators-vs.-No-Notables test for the predicate $\Psi(\bx,\by) = 1 - \langle \bx, \by \rangle$.
\end{lemma}
\begin{proof}
\textit{Completeness}. Let $f$ be an embedded dictator. WLOG, let $f(x) = (x_i, 0, \dots, 0)$. Then,
$$\val_T(f) = \E_{\bx,\by \sim \mathcal D_T}[1 - \langle f(\bx), f(\by)\rangle] = \E_{\bx,\by \sim \mathcal D_T}[1 - \bx_i \by_i] = 1-\rho$$
\textit{Soundness}. Assume $f$ has no $(\epsilon, \epsilon)$-notable coordinates. Then, by the Vector Majority is Stablest Theorem in the form of \Cref{thm:dictator-test} yields,
\begin{align*}
\E_{x,y \sim \mathcal D_T}[1 - \langle f(x), f(y)\rangle] \leq 1 - F^\star(m, \rho) + \lambda(\epsilon)
\end{align*}
where $\lambda$ is a function depending only on $\epsilon$ and approaches $0$ as $\epsilon \rightarrow 0$.
\end{proof}
}
\newpage
\part{Appendix}

\appendix 

\section{The non-commutative Sum of Squares hierarchy}\label{sec:ncsos}

The Sum of Squares (SoS) hierarchy gives a canonical method for strengthening the basic SDP to achieve better approximation ratios.
It features a tunable parameter $d$;
as~$d$ is increased, the quality of the approximation improves,
but the runtime needed to compute the optimum increases as well.
We will give a didactic overview of the SoS hierarchy
in order to explain how our basic SDP for \qmaxcut
arises naturally as the level-$2$ SoS relaxation.
For a more extensive treatment of Sum of Squares,
consult the excellent notes at~\cite{BS16}.

\subsection{Sum-of-squares relaxations for \maxcut}

We begin with the sum-of-squares relaxation for the \maxcut problem,
which generalizes the basic SDP from \Cref{def:mcsdp}.
In fact, we will state the SoS hierarchy in terms of
a general polynomial optimization problem over Boolean (i.e.\ $\pm 1$) variables.
Let $\calI$ be a finite set and $x = \{x_i\}_{i \in \calI}$ be a set of indeterminates indexed by~$\calI$.
Let $p(x)$ be a polynomial in the $x_i$'s,
and consider the optimization problem
\begin{align*}
\max &~~p(x) \\
\text{s.t.} &~~x_i^2 = 1,~\forall i \in \calI.
\end{align*}
For example, \maxcut is the case when $\calI = V$ and $p(x) = \E_{(\bu, \bv) \sim E}[\frac{1}{2}- \frac{1}{2} x_{\bu} x_{\bv}]$.
An alternative way to write this maximization is over probability distributions $\mu$ on $\pm 1$ assignments,
i.e. functions $\mu: \{-1, 1\}^{\calI} \rightarrow \R^{\geq 0}$ such that $\sum_x \mu(x) = 1$.
Then the optimum value is equal to
\begin{align*}
\max &~~\E_{\mu}[p(\bx)] = \sum_{x \in \{-1, 1\}^{\calI}} \mu(x) \cdot p(x) \\
\text{s.t.} &~~\text{$\mu$ is a probability distribution,}
\end{align*}
because we can take $\mu$ to have support only on the optimizing $x$'s.
Note that because~$\mu$ is a probability distribution,
it satisfies two properties:
(i) $\E_{\mu}[1] = 1$,
(ii) for each polynomial~$q(x)$, $\E_{\mu}[q(\bx)^2] \geq 0$.
Indeed, both of these properties hold pointwise, for all~$x$.
The SoS hierarchy replaces this optimization over probability distributions
with an optimization over ``pseudo-distributions'' while partially maintaining these two properties.

\begin{definition}[The SoS hierarchy for Boolean optimization problems]
Let $\mu: \{-1, 1\}^\calI \rightarrow \R$ be a function. Given a polynomial $q(x)$, we write
\begin{equation*}
\pE_\mu[q(x)] = \sum_{x \in \{-1, 1\}^{\calI}} \mu(x) \cdot q(x).
\end{equation*}
We say that $\mu$ is a \emph{degree-$d$ pseudo-distribution}
if $\pE_\mu[1] = 1$ and $\pE_\mu[q(x)^2] \geq 0$ for all polynomials~$q$ of degree at most~$d/2$.
In this case, we say that $\pE_\mu[\cdot]$ is a \emph{degree-$d$ pseudo-expectation}.
The value of the degree-$d$ SoS relaxation is simply the maximum of $\pE_\mu[p(x)]$
over all pseudo-distributions~$\mu$.
\end{definition}

It can be shown that the value of the degree-$2$ SoS is equal to the basic SDP.

\subsection{Sum-of-squares relaxations for \qmaxcut}\label{sec:heis-sos}

Now, we extend the Sum of Squares hierarchy to optimization problems over quantum states.
We will consider states consisting of qubits indexed by a finite set $\calI$, i.e.\ unit vectors in $(\C^2)^{\otimes \calI}$.
Let $H$ be a square matrix acting on $(\C^2)^{\otimes \calI}$ and consider the optimization problem
\begin{align*}
\max &~~\tr[\rho \cdot H] \\
\text{s.t.} &~~\text{$\rho$ is a density matrix.}
\end{align*}
For example, \qmaxcut is the case when $\calI = V$ and $H = \E_{(\bu, \bv) \sim E}[h_{\bu, \bv}]$.
Because~$\rho$ is a density matrix, it satisfies three properties:
(1) $\rho$ is Hermitian,
(2) $\tr[\rho \cdot I] = 1$,
and (3) for any matrix $M$, $\tr[\rho \cdot M^\dagger M] \geq 0$.
The last of these is because $A = M^\dagger M$ is PSD,
and a matrix $\rho$ is PSD if and only if $\tr[\rho \cdot A] \geq 0$ for all PSD matrices~$A$.
The SoS hierarchy will instead optimize over ``pseudo-density matrices'' while partially maintaining these three properties.
To begin, we first define a matrix analogue of degree.

\begin{definition}[The basis of Pauli matrices]
The Pauli matrices $\{I, X, Y, Z\}^{\otimes n}$ form an orthogonal basis for the set of $2^n \times 2^n$ matrices.
Given $P, Q \in \{I, X, Y, Z\}^{\otimes n}$, they satisfy
\begin{equation*}
\tr[P Q]
= \left\{\begin{array}{cl}
2^n & \text{if $P = Q$,}\\
0 & \text{otherwise}.
\end{array}\right.
\end{equation*}
Given a $2^n \times 2^n$ matrix $M$, we write $\widehat{M}(P)$ for the coefficient of~$M$ on~$P$ in this basis.
In other words,
\begin{equation*}
M = \sum_{P \in \{I, X, Y, Z\}^{\otimes n}} \widehat{M}(P) \cdot P.
\end{equation*}
In addition, $M$ is Hermitian if and only if $\widehat{M}(P)$ is real, for all~$P$.
\end{definition}

\begin{definition}[Degree of a matrix]\label{def:matrix-degree}
Given $P \in \{I, X, Y, Z\}^{\otimes n}$,
the \emph{degree} of $P$, denoted $|P|$, is the number of qubits on which~$P$ is not the $2 \times 2$ identity matrix.
More generally, we say that a $2^n \times 2^n$ matrix~$M$ has degree-$d$
if $\widehat{M}(P) = 0$ for all $|P| > d$.
\end{definition}

Now we describe the analogue of the sum-of squares hierarchy for quantum states,
which is known as the \emph{NPA} or \emph{non-commutative Sum of Squares (ncSoS) hierarchy}. 

\begin{definition}[The ncSoS hierarchy for quantum optimization problems]
Let $\rho$ and $M$ be square matrices acting on $(\C^2)^{\calI}$.
We write
\begin{equation*}
\pE_{\rho}[M]= \tr[\rho \cdot M].
\end{equation*}
We say that $\rho$ is a \emph{degree-$d$ pseudo-density matrix}
if $\rho$ is Hermitian, $\pE_\rho[I] = 1$,
and $\pE_\rho[M^\dagger M] \geq 0$ for all matrices~$M$ of degree at most~$d/2$ (cf.\ \Cref{def:matrix-degree}).
In this case, we say that $\pE_\rho[\cdot]$ is a \emph{degree-$d$ pseudo-expectation}.
The value of the degree-$d$ ncSoS relaxation is simply the maximum of $\pE_\rho[H]$
over all pseudo-distributions~$\rho$.
\end{definition}

\begin{remark}[Convergence of the SoS relaxation]
When $d = 2 n$, where $n = |\calI|$, the SoS relaxation solves the optimization problem exactly.
This is because every square matrix $M$ acting on $(\C^2)^\calI$ is degree-$n$;
thus $ \tr[\rho \cdot M^\dagger M] = \pE_\rho[M^\dagger  M] \geq 0$, and so~$\rho$ must be positive semidefinite.
\end{remark}

\subsection{Degree-two non-commutative Sum of Squares}

Now we analyze the degree-2 ncSoS relaxation for \qmaxcut
and show that it coincides with the basic SDP we considered in \Cref{sec:sdp_proofs}.
We begin with a definition.

\begin{definition}[Degree-$d$ slice of a matrix]\label{def:slice}
Given a $2^n \times 2^n$ matrix~$M$, we write $M^{=d}$ for its degree-$d$ component, i.e.
\begin{equation*}
M^{=d}
= \sum_{P : |P| = d} \widehat{M}(P) \cdot P.
\end{equation*}
\end{definition}

Let~$\rho$ be a feasible solution to the degree-2 ncSoS relaxation.
We will begin by showing that we may assume without loss of generality
that $\rho$ only has degree~$0$ and~$2$ components,
i.e.\ that $\rho = \rho^{=0} + \rho^{=2}$.
Prior to showing this, we will need a technical lemma.

\begin{lemma}\label{prop:pauli-commute-negation}
Let $P, Q, R \in \{I, X, Y, Z\}^{\otimes n}$.
Then $\tr(P Q R) = (-1)^{|P| + |Q| + |R|} \cdot \tr(P R Q)$.
\end{lemma}
\begin{proof}
First, we prove this for~$n = 1$.
When $n = 1$, both sides are zero unless $QR$ is the same Pauli matrix as~$P$, up to a multiplicative constant.
Suppose this is so.
If one of~$P$, $Q$, or $R$ is the identity matrix,
then the other two are equal to each other, and so $\tr(P Q R) = \tr(P R Q)$.
This satisfies the equality because $|P| + |Q| + |R|$ is either $0$ or $2$ in this case.
Otherwise, none of $P$, $Q$, or $R$ is the identity matrix,
and so they are distinct Pauli matrices,
which means~$Q$ and~$R$ anticommute.
So $\tr(P Q R) = - \tr(P R Q)$, satisfying the equality because $|P| + |Q| + |R| =3$  in this case.

Now, the general~$n$ case follows from the~$n= 1$ case because
\begin{equation*}
\tr(P Q R)
= \prod_{i=1}^n \tr(P_1 Q_1 R_1)
= \prod_{i=1}^n \Big((-1)^{|P_1| + |Q_1| + |R_1|} \cdot\tr(P_1 R_1 Q_1)\Big)
= (-1)^{|P| + |Q| + |R|} \cdot \tr(P R Q).
\end{equation*}
This completes the proof.
\end{proof}

\begin{proposition}[Restricting to the degree-0 and degree-2 slices]
Let $\rho$ be a feasible solution. Then $\rho^{=0} + \rho^{=2}$ is a feasible solution with the same value as~$\rho$.
\end{proposition}
\begin{proof}
To begin, we claim that $\rho' = \rho^{=0} + \rho^{=1} + \rho^{=2}$ is a feasible solution with the same value as~$\rho$.
This is because the constraints $\pE_\rho[I] = 1$ and $\pE_\rho[M^\dagger M] \geq 0$
and the objective $\pE_\rho[H_G]$ feature matrices of degree at most~$2$,
and so these values are unchanged if we replace~$\rho$ with $\rho^{=0} + \rho^{=1} + \rho^{=2}$.

Next, we claim that $\rho'' = \rho^{=0} - \rho^{=1} + \rho^{=2}$ is a feasible solution with the same value as $\rho'$.
The constraint $\pE_{\rho''}[I] = 1$ and the value $\pE_{\rho''}[H_G]$ feature matrices which have no degree-1 terms,
so negating $\rho^{=1}$ doesn't affect these expressions.
As for the remaining constraint, for each degree-1 matrix~$M$,
\begin{align*}
\pE_{\rho''}[M^\dagger M]
& = \tr((\rho^{=0} - \rho^{=1} + \rho^{=2}) \cdot M^\dagger M)\\
& = \sum_{|P| \leq 2} \sum_{|Q|, |R| \leq 1} (-1)^{|P|} \cdot \widehat{\rho}(P) \widehat{M}^\dagger(Q)\widehat{M}(R) \cdot \tr(P Q R) \\
& = \sum_{|P| \leq 2} \sum_{|Q|, |R| \leq 1} (-1)^{|R| + |Q|} \cdot \widehat{\rho}(P) \widehat{M}^\dagger(Q)\widehat{M}(R) \cdot \tr(P R Q) \tag{by \Cref{prop:pauli-commute-negation}}\\
& = \tr((\rho^{=0} + \rho^{=1} + \rho^{=2}) \cdot (M^{=0} - M^{=1}) (M^{=0} - M^{=1})^\dagger)\\
& = \pE_{\rho'}[(M^{=0} - M^{=1}) (M^{=0} - M^{=1})^\dagger] \geq 0.
\end{align*}
Hence, $\rho''$ satisfies this constraint as well.

We conclude by noting that $\tfrac{1}{2}(\rho' + \rho'') = \rho^{=0} + \rho^{=2}$
is a feasible solution with the same value as~$\rho$,
because the constraints and the objective are linear functions of~$\rho$.
\end{proof}

Henceforth, we assume $\rho = \rho^{=0} + \rho ^{=2}$.
Using this, we note that the constraint $\pE_\rho[M^\dagger M] \geq 0$
holding for all $M$ which are degree-$1$ is equivalent to it holding only for~$M$ which are homogeneous degree-$1$
(i.e.\ with no degree-$0$ term).
This is because if we write $M = M^{=0} + M^{=1} = \widehat{M}(I) \cdot I + M^{=1}$, then
\begin{align*}
\tr(\rho \cdot M^\dagger M)
&= |\widehat{M}(I)|^2 \cdot \tr(\rho) + \widehat{M}(I)^\dagger \cdot \tr(\rho \cdot M^{=1})
	+ \widehat{M}(I) \cdot \tr(\rho \cdot (M^{=1})^\dagger)
	+ \tr(\rho \cdot (M^{=1})^\dagger M^{=1})\\
&= |\widehat{M}(I)|^2 + \tr(\rho \cdot (M^{=1})^\dagger M^{=1}) \tag{because $\rho$ has no degree-$1$ component}\\
&\geq \tr(\rho \cdot (M^{=1})^\dagger M^{=1}),
\end{align*}
which is $\geq 0$ because $\pE_\rho[(M^{=1})^\dagger M^{=1}] \geq 0$.

Now, we let $R(\cdot, \cdot)$ be the $3n \times 3n$ matrix whose rows and columns are indexed by degree-$1$ Pauli matrices such that
\begin{equation*}
R(P_i, Q_j) = \tr(\rho \cdot P_i Q_j).
\end{equation*}
for all $P, Q \in \{X, Y, Z\}$ and $i, j \in \{1, \ldots, n\}$.
When $i \neq j$ or $i = j$ and $P = Q$, then $P_i Q_j$ is a Pauli matrix in $\{I, X, Y, Z\}^{\otimes n}$,
and so
\begin{equation*}
R(P_i, Q_j) = \tr(\rho \cdot P_i Q_j)
= 2^n \cdot \widehat{\rho}(P_i Q_j),
\end{equation*}
which is a real number. On the other hand, when $i = j$ but $P \neq Q$,
then $P_i Q_j = P_i Q_i$ is a degree-$1$ Pauli matrix times a phase of $i$ or $-i$.
In this case,
\begin{equation*}
R(P_i, Q_i) = \tr(\rho \cdot P_i Q_i) = 0,
\end{equation*}
because~$\rho$ has no degree-$1$ component.
Put together, these imply that $R$ is a real-valued matrix.

We can now rewrite our constraints and objective function in terms of this matrix.
First, the constraint $\tr(\rho) = 1$ corresponds to
\begin{equation*}
R(P_i, P_i) = \tr(\rho \cdot P_i P_i) = \tr(\rho) = 1
\end{equation*}
for any $P_i$.
Next, the objective function is
\begin{align*}
\tr(\rho \cdot H_G)
& = \E_{(\bi, \bj) \sim E} \tr(\rho \cdot h_{\bi, \bj})\\
& = \E_{(\bi, \bj) \sim E} \tr(\rho \cdot \tfrac{1}{4} \cdot( I_{\bi} \ot I_{\bj} - X_{\bi} \ot X_{\bj} - Y_{\bi} \ot Y_{\bj} - Z_{\bi} \ot Z_{\bj}))\\
&= \tfrac{1}{4} - \tfrac{1}{4} \E_{(\bi, \bj) \sim E} \sum_{P \in \{X, Y, Z\}} \tr(\rho \cdot P_{\bi} \ot P_{\bj})\\
&= \tfrac{1}{4} - \tfrac{1}{4} \E_{(\bi, \bj) \sim E} \sum_{P \in \{X, Y, Z\}} R(P_{\bi}, P_{\bj}).
\end{align*}
For the last constraint, let $M = \sum_{P_i} \widehat{M}(P_i) \cdot P_i$ be any homogeneous degree-1 matrix.
Then
\begin{align*}
0 \leq \tr(\rho \cdot M^\dagger M)
&= \sum_{P_i, Q_j} \widehat{M}(P_i)^\dagger \widehat{M}(Q_j) \cdot \tr(\rho \cdot P_i Q_j)\\
&= \sum_{P_i, Q_j} \widehat{M}(P_i)^\dagger \widehat{M}(Q_j) \cdot R(P_i,Q_j)
= \mathrm{vec}(M)^\dagger\cdot R \cdot\mathrm{vec}(M),
\end{align*}
where $\mathrm{vec}(M)$ is the height-$3n$ vector with $\mathrm{vec}(M)(P_i) = \widehat{M}(P_i)$.
As the $\widehat{M}(P_i)$'s are allowed to be arbitrary complex numbers,
this condition is equivalent to $R$ being positive semidefinite.
 As a result, this matrix has the exact same form and objective as the matrix $M'(\cdot, \cdot)$ from \Cref{sec:sdp_proofs};
 following the steps in that proof, one can then convert $R$ into a solution to the basic SDP.
 This completes the proof.

\ignore{
We claim that $\rho' = \rho^{=0} + \rho^{=1} + \rho^{=2}$ is a feasible solution with the same value as~$\rho$.
This is because the constraints $\pE_\rho[I] = 1$ and $\pE_\rho[M^\dagger M] \geq 0$
and the objective $\pE_\rho[H_G]$ feature matrices of degree at most~$2$,
and so these values are unchanged if we replace~$\rho$ with $\rho^{=0} + \rho^{=1} + \rho^{=2}$.

}
\ignore{
Next, consider the $(3n + 1) \times (3n+1)$ matrix $R$ whose rows and columns are indexed by the 
degree-0 and degree-1 Pauli matrices $P, Q$ in which $R(P, Q) = \tr[\rho' \cdot P Q]$.
Using~$R$, we may express the SoS objective as
\begin{align}
\tr[\rho \cdot H_G]
& = \E_{(\bi, \bj) \sim E} \tr[\rho \cdot h_{\bi, \bj}]\nonumber\\
& = \E_{(\bi, \bj) \sim E} \tr[\rho \cdot \tfrac{1}{4} \cdot( I_{\bi} \ot I_{\bj} - X_{\bi} \ot X_{\bj} - Y_{\bi} \ot Y_{\bj} - Z_{\bi} \ot Z_{\bj})]\nonumber\\
&= \tfrac{1}{4} - \tfrac{1}{4} \E_{(\bi, \bj) \sim E} \sum_{P \in \{X, Y, Z\}} \tr[\rho \cdot P_{\bi} \ot P_{\bj}]\nonumber\\
&= \tfrac{1}{4} - \tfrac{1}{4} \E_{(\bi, \bj) \sim E} \sum_{P \in \{X, Y, Z\}} R(P_{\bi}, P_{\bj}).\label{eq:new-objective}
\end{align}
Let us derive some constraints on~$R$:
\begin{enumerate}
\item \textbf{PSD:} $R$ is Hermitian and PSD.
			It is Hermitian because $\rho'$ is Hermitian,
			and so $$R(P, Q)^\dagger = \tr[\rho' \cdot PQ]^\dagger = \tr[\rho' \cdot QP] = R(Q, P).$$
			It is PSD because for any vector $v = (v_P)_P$ with a component for each degree-0 and degree-1 Pauli matrix~$P$,
			\begin{equation*}
			v^\dagger R v
			= \sum_{P, Q} v_P^\dagger v_Q \tr[\rho' \cdot P Q]
			= \tr\Big[\rho' \cdot \Big(\sum_P v_P P\Big)^\dagger \Big(\sum_Q v_Q Q\Big)\Big] \geq 0,
			\end{equation*}
			because $M = \sum_P v_P P$ is a degree-$1$ matrix. 
\item \textbf{Unit length:} For each $P$, $R(P, P) = 1$.
\item \textbf{Commuting Paulis:} If $P$ and $Q$ are commuting then $R(P, Q) = R(Q, P)$.
\item \textbf{Anti-commuting Paulis:} If $P$ and $Q$ are anti-commuting then $R(P, Q) = R(Q, P)$.
\end{enumerate}
Indeed, $R$ satisfies these constraints if and only if $\rho'$ is a feasible solution to the level-2 ncSoS relaxation.
Thus, the problem of maximizing~\eqref{eq:new-objective} subject to~$R$ satisfying these four constraints
is equivalent to the level-2 ncSoS relaxation.

We note that the objective~\eqref{eq:new-objective} only contains degree-$2$ terms,
and so it remains unchanged if we replace it with~$R'$,
in which we zero out all entries in~$R$ of the form $R(I, P)$ and $R(P, I)$,
where $P$ is a degree-$1$ Pauli matrix.
In addition, $R'$ satisfies the above four properties if and only if~$R$ does,
and so it remains a feasible solution.
(The chief difficulty is showing that~$R'$ is still PSD,
but this follows from the fact that $v^\dagger R v \geq 0$ for all vectors in which $v_I = 0$.)
 Now consider the $3n \times 3n$ submatrix of $R'$ indexed by the degree-1 Pauli matrices.
 Then this matrix has the exact same form and objective as the matrix $M(\cdot, \cdot)$ from \Cref{sec:sdp_proofs};
 following the steps in that proof, one can then convert $R$ into a solution to the basic SDP.
 This completes the proof.
}
\section{Other Lemmas}

\subsection{Cardinality Reduction}\label{sec:card-reduc}
\ignore{
\begin{lemma}
Let $0 \leq \eps \leq 1$. Given a vertex $u \in S^{n-1}$,
consider the hyperspherical cap
\begin{equation*}
B(u, \eps) = \{v \in S^{n-1} : \Vert u - v \Vert_2^2 \leq \eps\}.
\end{equation*}
Then the measure of $B(u, \eps)$ can be lower-bounded by~XXX
and upper-bounded by~XXX.
\end{lemma}
\begin{proof}
Suppose without loss of generality that $u = (1, 0, \ldots, 0)$.
Then the point
\begin{equation*}
v = (1-\eps^2/2, \sqrt{\eps^2-\eps^4/4},0, \ldots, 0)
\end{equation*}
has distance exactly~$\eps$ from~$u$.
Let $\phi$ be the angle between $u$ and $v$,
i.e.\ $\cos(\phi) = \langle u, v\rangle = 1 - \eps^2/2$.
By~XXX, the surface areas of the $B(u, \eps)$ is exactly
\begin{equation*}
\frac{2\pi^{(n-1)/2}}{\Gamma((n-1)/2)} \int_0^{\phi} \sin^{n-2}\theta \mathrm{d}\theta.
\end{equation*}
respectively. Dividing the first by the second, this means the measure of $B(u, \eps)$ on the sphere is
\begin{equation*}
\frac{\Gamma(n/2)}{\sqrt{\pi}\Gamma((n-1)/2)} \int_0^{\phi} \sin^{n-2}\theta \mathrm{d}\theta.
\end{equation*}
\end{proof}
}

In this section, we prove \Cref{lem:discretization}, using an argument which closely follows \cite[Appendix B]{OW08}.
All of these transformations are standard.

\begin{proof}
First, we give a series of transformations to yield a well-behaved finite graph. For each transformation, we argue that SDP value and product state value are within $\epsilon$ of the original graph. Finally, we show that in our final graph $G'$, using \Cref{cor:BH-nonuniform-easy-to-use}, $\prodval(G') \geq \heis(G') - \epsilon$, concluding the proof.

Let $G_0 = G$. We will first construct $G_1$, which restricts the Gaussian graph $\mathcal G^n_\rho$ to the sphere graph $S^{n-1}$. Next, $G_2$ will be a graph on a finite vertex set. Following that, our final graph $G' = G_3$ will remove self-loops, leaving us with a weighted, simple graph with finite vertex set. In this proof, we will identify a graph by the distribution on its edges. For $u,v \subseteq S^{n-1}$, we will write $G(u,v)$ for the probability weight $G$ puts on edges $(u,v)$.

We start with the construction for $G_1$.
For $G_0$, let $f : \mathbb R^n \rightarrow S^{n-1}$ be an SDP assignment obtaining $\hsdp(G_0)$.
(Note that this is also an optimal SDP assignment for $\prodsdp(G_0)$.)
Let $G_1$ be the graph in which $G_1(u,v) = G_0(f^{-1}(u),f^{-1}(v))$. Then if we take the identity map as the SDP embedding, we see that
$$\hsdp(G_1)
\geq \E_{(\bu,\bv) \sim G_1}[\tfrac{1}{4}-\tfrac{3}{4}\langle \bu,\bv\rangle]
= \E_{\bu \sim_\rho \bv}[\tfrac{1}{4} - \tfrac{3}{4}\langle f(\bu), f(\bv)\rangle]
= \hsdp(G_0) =: c_{\mathrm{H}}.$$
A similar argument shows that $\prodsdp(G_1) \geq \prodsdp(G_0) =: c_{\text{\sc Prod}}$.
Furthermore, for any assignment $h : S^{n-1} \rightarrow S^{2}$ on $G_1$, the assignment $h \circ f$ yields an assignment for $G_0$ and thus $\prodval(G_1) \leq \prodval(G_0)$.

To construct $G_2$, we use an argument originally from \cite{FS02}.
Pick some $\epsilon$-net $\mathcal N$ over $S^{n-1}$,
so that every point in $S^{n-1}$ is within distance $\eps$ to some point in~$\mathcal N$;
it is known that constructions exist with $|\mathcal{N}| \leq 1/\epsilon^{O(d)}$.
Then partition $S^{n-1}$ using Voronoi cells $\{C_v\}_{v \in \mathcal N}$ based on $\mathcal N$. For each $v \in \mathcal N$, the corresponding cell $C_v \subseteq S^{n-1}$ consists of all points in $S^{n-1}$ which are closer to $v$ than any other $u \in \mathcal N$.
Then $G_2$ is the finite graph on vertex set $\mathcal N$ in which $G_2(u,v) = G_1(C_u,C_v)$. We first observe that
$$\prodval(G_2) \leq \prodval(G_1) = s$$
since any assignment $f$ on $G_2$ can be extended to an assignment of equal value on $G_1$. Furthermore, we claim
$$\hsdp(G_2) \geq  c_{\mathrm{H}} - 3\epsilon.$$
To see this, consider the SDP assignment $f:\mathcal N \rightarrow S^{n-1}$
which maps each $v \in \mathcal N$ to itself.
We can extend this to a function with domain all of $S^{n-1}$ by setting $f(u) = v$ for each $u \in C_v$.
Then
\begin{equation}\label{eq:using-weird-assignment}
\hsdp(G_2)
\geq \E_{(\bu, \bv) \sim G_2}[\tfrac{1}{4} - \tfrac{3}{4} \langle f(\bu), f(\bv) \rangle]
= \E_{(\bx, \by) \sim G_1}[\tfrac{1}{4} - \tfrac{3}{4} \langle f(\bx), f(\by) \rangle].
\end{equation}
Let $C_{\bu}$ be the Voronoi cell $\bx$ falls inside and $C_{\bv}$ be the Voronoi cell $\by$ falls inside.
Then because $\mathcal N$ is an $\eps$-net, we can write $\bx = (\bu + \boldsymbol{\eta}_1)$ and $\by = (\bv + \boldsymbol{\eta}_2)$,
where $\boldsymbol{\eta}_1$ and $\boldsymbol{\eta}_2$ have length at most $\epsilon$.
Thus,
\begin{equation*}
\langle \bx,\by\rangle = \langle \bu + \boldsymbol{\eta}_1, \bv+\boldsymbol{\eta}_2\rangle \geq \langle \bu,\bv\rangle - 3\epsilon
= \langle f(\bx), f(\by) \rangle - 3\epsilon.
\end{equation*}
As a result,
\begin{equation*}
\eqref{eq:using-weird-assignment}
\geq \E_{(\bx, \by) \sim G_1}[\tfrac{1}{4} - \tfrac{3}{4} \langle \bx, \by \rangle] - 3\eps
= c_{\mathrm{H}} - 3\eps.
\end{equation*}
A similar argument shows that $\prodsdp(G_2) \geq c_{\text{\sc Prod}} - 3\eps$.

Finally, we use a simple construction appearing in \cite{KO06} (and originally due to \cite{ABH+05})
in order to remove self-loops.
Conveniently,
this construction will also make it easy to show that the product state value and maximum energy are close.
Our graph $G_3$ will be parameterized by an integer $M$ which we will select later but which is at least $1/\eps$.
For each vertex $v \in G_2$,
we will create $M$ many vertices $\{(v, j)\}_{j \in [M]}$,
each with weight $\frac{1}{M}$ of the original.
To sample a random edge in $G_3$, we simply sample $(\bu, \bv)$ from $G_2$,
let $\bi, \bj \in [M]$ be independent, uniformly random,
and output the edge between $(\bu, \bi)$ and $(\bv, \bj)$.

It is clear that $\prodval(G_3) \geq \prodval(G_2)$
because any assignment $f: \mathcal V \rightarrow S^{2}$
can be converted into an assignment $f'$ for $G_3$ of equal value
by setting $f'(u, i) = f$.
On the other hand, $\prodval(G_2)  \geq \prodval(G_3)$ as well.
To see this, consider a product state assignment $f: \mathcal V \times [M] \rightarrow S^2$ for $G_3$.
It has value
\begin{align*}
\E_{(\bu, \bv) \sim G_2}\E_{\bi, \bj \sim [M]}[\tfrac{1}{4} - \tfrac{1}{4} \langle f(\bu, \bi), f(\bv, \bj) \rangle]
= \E_{\bu, \bv \sim G_2}[\tfrac{1}{4} - \tfrac{1}{4} \langle \E_{\bi} f(\bu, \bi), \E_{\bj} f(\bv, \bj)\rangle].
\end{align*}
This is the value that the assignment $f': \mathcal V \rightarrow B^3$ defined as $f'(u) = \E_{\bi} f(u, \bi)$ achieves on the graph $G_2$.
As $f'$ has range $B^3$, there exists a function with range $S^2$ whose value is at least as high.
As a result, $\prodval(G_3) = \prodval(G_2)$.
A similar argument shows that the two SDP values remain the same as well.
Now, the total weight of self-loops in this graph is at most $\frac{1}{M}$,
and so by removing these edges and scaling the remaining weights to sum to one
we produce a graph $G'$ with no self-loops in which the SDP and product state values have increased by at most $\frac{1}{M} \leq \epsilon$.

Next, we use \Cref{cor:BH-nonuniform-easy-to-use} to relate the product state value to the optimal state value. 
Recall that we need to bound the quantity
\begin{equation}\label{eq:bh-bound}
20\cdot (n \cdot \max_{(u,i), (v,j)} \{A_{(u,i), (v,j)}\} \cdot \max_{(u,i)} \{p_{(u,i)}\})^{1/8} + \max_{(u,i)} \{p_{(u,i)}\}
\end{equation}
for $G'$. Here, $n$ is the number of vertices in $G_3$, $A_{(u,i), (v,j)}$ is the probability of an edge ending in $(u,i)$ conditioned on starting from $(v,j)$, and $p_{(u,i)}$ is one half the total weight of edges on $(u,i)$.
We'll give a bound for the above quantity \textit{before} removing self loops (i.e.\ for the graph $G_3$). However, observe that self loops consist of at most $1/M$ of the total edge weight and thus,
$$p'_{u,i} \leq p_{u,i}/(1 - \tfrac{1}{M}),$$
$$A'_{(u,i),(v,j)} \leq A'_{(u,i),(v,j)}/(1 - \tfrac{1}{M}),$$
where $p'_{u,i}$ and $A'_{(u,i),(v,j)}$ are the quantities after removing self-loops (i.e.\ for the graph $G'$).
Choosing $M$ sufficiently large makes this difference negligible.

Now, observe that if $p^{(2)}_{u}$ is the weight function associated with $G_2$, then $p_{u,i} = p^{(2)}_u/M$, since each vertex $(u,i)$ in $G_3$ inherits $1/M$ of the total edge weight of vertex $u$. Thus, by choosing $M$ sufficiently large, we can bound the last additive term by $\epsilon/2$.
Next, since whenever $(v,j)$ is connected to a vertex $(u, i)$, it is in fact connected to all vertices $\{(u, k)\}_{k \in [M]}$ with an equal weight,
we have an easy upper bound of $1/M$ on $A_{(u,i),(v,j)}$.

Finally, using that $n = M|\mathcal V|$, where $\mathcal V$ is $G_2$'s vertex set, we can rewrite \Cref{eq:bh-bound} as,
\begin{align*}
20\cdot \Paren{M|V| \cdot \max_{(u,i), (v,j)} \{A_{(u,i), (v,j)}\} &\cdot \frac{1}{M}\max_{u} \{p_u\}}^{1/8} + \epsilon/2\\
&= 20\cdot \PAREN{|V| \cdot \max_{u} \{p_u\} \cdot \frac{1}{M}}^{1/8} + \epsilon/2
\end{align*}
Thus, noting that $|V| \cdot \max_u \{p_u\}$ is just a constant $C$ which is independent of~$M$,
we can bound this equation by~$\epsilon$ by taking~$M$ sufficiently large.
\end{proof}

\subsection{Proofs of Various Fourier Properties}\label{sec:fourier-proofs}

\begin{proposition}\label{eq:rho-stab-diff}
Let $f : \{-1,1\}^n \rightarrow \mathbb R^k$ be a function and let $\rho,\gamma \in [0,1)$. Then
$$\ABS{\Stab_\rho [f] - \Stab_{\rho(1-\gamma)^2}[f]} \leq \frac{2\gamma}{1-\rho}\Var[f].$$
\end{proposition}
\begin{proof}
Let $\eps = \rho - (1-\gamma)^2\rho \in [0,\rho)$,
so that $\rho(1-\gamma)^2 = \rho - \epsilon$. Then 
$$\ABS{\Stab_\rho [f] - \Stab_{\rho - \epsilon}[f]} = \sum_{S \subseteq [n]} (\rho^{|S|} - (\rho - \epsilon)^{|S|})\norm{\widehat f(S)}_2^2.$$
When $\epsilon = 0$, the proposition clearly holds. Assume $\epsilon \in (0,\rho)$. For $S \not= \emptyset$, we'll bound the term $(\rho^{|S|} - (\rho - \epsilon)^{|S|})$. When $S = \emptyset$ this quantity is just $0$. Let $k = |S|$, and define the function $g(\delta) = (\rho - \epsilon + \delta)^k$. This function is continuous in $\delta$, and applying the Mean Value Theorem on $\delta \in [0, \epsilon]$ yields a $\delta'$ such that
$$\frac{\mathrm{d} g}{\mathrm d \delta}(\delta')= \frac{g(\epsilon) - g(0)}{\epsilon - 0} = \frac{\rho^k - (\rho - \epsilon)^k}{\epsilon}$$
Furthermore, $g'(\delta') = k(\rho - \epsilon + \delta')^{k-1}$. For any $x \in (0,1]$ and $k \in \mathbb N^+$, we have that $(1-x)^{k-1}k \leq 1/x$. Letting $1-x = \rho - \epsilon + \delta'$, we get
$$\frac{\rho^k - (\rho - \epsilon)^k}{\epsilon} = k(\rho - \epsilon + \delta')^{k-1} \leq \frac{1}{1-\rho + \epsilon - \delta'} \leq \frac{1}{1-\rho} \implies \rho^k - (\rho - \epsilon)^k \leq \frac{\epsilon}{1- \rho}.$$
Thus, 
$$\ABS{\Stab_\rho [f] - \Stab_{\rho - \epsilon}[f]} \leq \frac{\epsilon}{1-\rho} \sum_{S \subseteq [n] : S \not= \emptyset} \norm{\hat f(S)}_2^2 = \frac{\epsilon}{1- \rho} \Var[f].$$
Substituting in $\epsilon = \rho - \rho(1-\gamma)^2 \leq 2\gamma$ concludes the proof.
\end{proof}

\begin{proposition}[Stability Bound]\label{prop:stab-bound}
Let $f : \{-1,1\}^n \rightarrow \mathbb R$ and $\rho \in [-1,0]$. Then,
$$\Stab_\rho[f] \geq \rho \cdot \E[f(\bx)^2].$$
\end{proposition}
\begin{proof}
First, note that $\Stab_\rho[f] = \sum_S \rho^{|S|} \widehat f(S)^2$. Since $\rho \in [-1,0]$, we can lower bound this by $\rho\cdot \sum_S \widehat f(S)^2$. Finally, using Parseval's theorem, this is exactly equal to $\rho\cdot \E[f(\bx)^2]$.
\end{proof}

\subsection{Proofs of Lipschitz Properties}

\begin{lemma}[Lipschitz Property of $\Psi$]\label{lem:psi-lipschitz}
Let $\Psi : \mathbb R^n \rightarrow \mathbb R$ be defined as,
$$\Psi(v) = 
\left\{\begin{array}{cl}
\norm{v}_2^2 & \text{if $\norm{v}_2 \leq 1$,}\\
1 & \text{otherwise.}
\end{array}\right.$$
Then, $\Psi$ is Lipschitz continuous with constant $2$. In particular, for any $u,v \in \mathbb R^n$,
$$\abs{\Psi(u) - \Psi(v)} \leq 2 \norm{u - v}_2$$
\end{lemma}
\begin{proof}
The proof is by reduction to the function $\mathrm{Sq} : \mathbb R \rightarrow \mathbb R$, defined as,
$$\mathrm{Sq}(x) = 
\left\{\begin{array}{cl}
0 & \text{if $x < 0$,}\\
x^2 & \text{if $x \in [0,1]$,}\\
1 & \text{if $x > 1$.}
\end{array}\right.$$
As in the proof of the Majority is Stablest theorem in \cite{OD14}, we see that $\mathrm{Sq}$ is $2$-Lipschitz. Fix any $u,v \in \mathbb R^n$. By applying the Lipschitz property of $\mathrm{Sq}$, we can write
$$\abs{\Psi(u) - \Psi(v)} = \abs{\mathrm{Sq}(\norm{u}_2) - \mathrm{Sq}(\norm{v}_2)} \leq 2\abs{\norm{u}_2 - \norm{v}_2}.$$
Finally, applying the reverse triangle inequality, we obtain an upper bound of $2\norm{u-v}_2$, which concludes the proof.
\end{proof}

\begin{lemma}[Lipschitz Property of $\Phi$]\label{lem:phi-lipschitz}
Let $\Phi : \mathbb R^n \rightarrow \mathbb R$ be defined as $\Phi(v) = v - \mathcal R(v)$, where $\mathcal R$ rounds vectors to the unit ball $B^n$ and is defined as
$$\mathcal{R}(v) = 
\left\{\begin{array}{cl}
v & \text{if $\norm{v}_2 < 1$,}\\
\tfrac{v}{\norm{v}_2} & \text{otherwise.}
\end{array}\right.$$
Then $\Phi$ is Lipschitz continuous with constant $2$. In particular, for any $u,v \in \mathbb R^n$,
$$\norm{\Phi(u) - \Phi(v)}_2 \leq 2 \norm{u - v}_2.$$
\end{lemma}
\begin{proof}
First, we show that $\mathcal R(v)$ is in fact $1$-Lipschitz.
Without loss of generality, we can assume vectors $u,v \in \mathbb R^2$ and take $u = r\cdot (1,0)$ and $v = s\cdot(v_1,v_2)$ where $\norm{(v_1,v_2)}_2 = 1$. We want to show
$$\norm{\mathcal R(u) - \mathcal R(v)}_2 \leq \norm{u - v}_2.$$
Certainly, this holds when $r,s \leq 1$. Consider the case when $r,s \geq 1$. Then
\begin{align*}
\norm{\mathcal R(u) - \mathcal R(v)}_2 &= \norm{(1,0) - (v_1,v_2)}_2\\
&= \sqrt{(1-v_1)^2 + v_2^2}\\
&= \sqrt{(1-v_1)^2 + 1 - v_1^2}\\
&= \sqrt{2 - 2 v_1}.
\end{align*}
On the other hand,
\begin{align*}
\norm{u - v}_2 &= \sqrt{(r-s v_1)^2 + s^2 v_2^2}\\
&= \sqrt{r^2 - 2rs v_1 + s^2v_1^2 + s^2(1-v_1^2)}\\
&= \sqrt{r^2 - 2rs v_1 + s^2}.
\end{align*}
Thus, it suffices to show $2-2v_1 \leq r^2 - 2rs v_1 + s^2$. Rearranging so that all terms including $v_1$ are on the LHS, we get
$$v_1 \cdot 2(rs-1) \leq^? r^2 + s^2 - 2.$$
Since $s,r \geq 1$, the LHS is maximized for $v_1 = 1$ and thus this is true if and only if
$$0 \leq^? r^2 + s^2 -2rs.$$
Factoring the RHS yields $(r-s)^2$, which is indeed at least $0$.

Finally, we consider the case when $r \leq 1$ and $s \geq 1$ (the case of $r \geq 1$ and $s \leq 1$ is symmetric and we omit it). Then we again have $\norm{u-v}_2 = \sqrt{r^2 - 2rs v_1 + s^2}$. However, we now have
$$\norm{\mathcal R(u) - \mathcal R (v)}_2 = \sqrt{(r-v_1)^2 + v_2} = \sqrt{(r-v_1)^2 + 1 - v_2^2} = \sqrt{r^2 - 2rv_1 + 1}.$$
As a result, we want to show
\begin{align*}
r^2 - 2rv_1 + 1 &\leq^? r^2 - 2rsv_1 + s^2\\
v_1 \cdot 2(rs-r) + 1 &\leq^? s^2\\
2(rs-r) + 1 &\leq^? s^2\tag{LHS maximized when $v_1 = 1$}\\
2s - 2 + 1 &\leq^? s^2\tag{LHS maximized when $r = 1$}\\
0 &\leq^? s^2 - 2s + 1 = (s-1)^2.
\end{align*}
We conclude that $\mathcal R(\cdot)$ is 1-Lipschitz. Now we show that $\Phi(u) = u - \mathcal R(u)$ is $2$-Lipschitz.
\begin{align*}
\norm{\Phi(u) - \Phi(v)}_2 = \norm{u - \mathcal R(u) - v + \mathcal R(v)}_2
&\leq \norm{u - v}_2 + \norm{\mathcal R(u) - \mathcal R(v)}_2\tag{by the triangle inequality}\\
&\leq 2 \norm{u-v}_2.\tag{by the Lipschitz property of $\mathcal R(\cdot)$}
\end{align*}
This concludes the proof.
\end{proof}

\begin{corollary}\label{cor:phi-lipschitz}
The function $\Phi_i(v) : \mathbb R^n \rightarrow \mathbb R$, defined as $\Phi_i(v) = \Phi(v)_i$, is $2$-Lipschitz.
\end{corollary}

\begin{proof}
Take any $u,v \in \mathbb R^n$. Then
\begin{equation*}
|\Phi_i(u) - \Phi_i(v)| \leq \sqrt{\sum_{i}^n (\Phi_i(u) - \Phi_i(v))^2}
= \norm{\Phi(u) - \Phi(v)}_2
\leq 2\norm{u -v}_2.\qedhere
\end{equation*}
\end{proof}

\begin{lemma}\label{lem:inner_prod_lipschitz}
Let $\bx \sim_\rho \by$ be $\rho$-correlated random variables in $k$ dimensions.  The function:
$$
\E_{\bx \sim_\rho \by} \left \langle \frac{\bx}{||\bx||}, \frac{\by}{||\by||} \right \rangle =\frac{2}{k} \left( \frac{\Gamma((k+1)/2)}{\Gamma(k/2)}\right)^2 \rho\,\, \,_2 F_1[1/2, 1/2, k/2+1, \rho^2]
$$
in $\rho$ is:
\begin{enumerate}
\item non-negative for $\rho \in [0,1]$,
\item an odd function,
\item $C$-Lipschitz for $\rho\in [-1, 1]$ for some constant $C$ which is a function of $k$ if $k\geq 3$,
\item $C$-Lipschitz for $\rho\in [-1+\epsilon, 1-\epsilon]$ for some constant $C(k, \epsilon)$ for any $\epsilon>0$ if $k=1, 2$, and
\item convex for $\rho \in [0,1]$.
\end{enumerate}
\end{lemma}
\begin{proof}
From the definition of $\,_2 F_1$:
\begin{equation}\label{eq:hypergeometric-sum-def}
\,_2F_1 [a, b, c, z]=\sum_{n=0}^\infty \frac{(a)_n (b)_n}{(c)_n} \frac{z^n}{z!},
\end{equation}
we see that $\,_2 F_1[1/2, 1/2, k/2+1, \rho^2] \geq 0$, establishing the first two properties.

For the third and fourth property, since the function $f(\rho) = \rho\,\, \,_2 F_1[1/2, 1/2, k/2+1, \rho^2]$ is differentiable in $\rho$ so it is Lipschitz if we can upper bound the absolute value of the derivative in the interval containing $\rho$.  By 15.2.1 in \cite{AS72},
\begin{equation*}
    \frac{d}{dz}\,_2F_1 [a, b, c, z] = \frac{ab}{c}\,_2F_1 [a+1, b+1, c+1, z]
\end{equation*}
\noindent Applying the product rule, taking the derivative of $f$ yields:
\begin{align}\label{eq:d_of_hypergeo}
\frac{d}{d\rho} f(\rho) =\frac{\rho^2 \, _2F_1\left(\frac{3}{2},\frac{3}{2};\frac{k}{2}+2;\rho^2\right)}{2 \left(\frac{k}{2}+1\right)}+\, _2F_1\left(\frac{1}{2},\frac{1}{2};\frac{k}{2}+1;\rho^2\right).
\end{align}
It is clear the derivative is non-negative where defined, from the definition of $\,_2F_1$, \Cref{eq:hypergeometric-sum-def}, since it is a convergent sum of non-negative numbers as in the first property.  For $k\geq 3$ the derivative is defined at all $\rho\in[-1, 1]$, whereas for $k=1, 2$ the derivative is defined for all $\rho\in (-1, 1)$.  Hence, we will consider an interval $\rho\in [-1+\delta, 1-\delta]$ where $\delta=\epsilon$ for $k=1, 2$ and $\delta=0$ for $k\geq 3$.  The derivative is increasing in $z$ for fixed $a$, $b$ and $c$ by \Cref{eq:hypergeometric-sum-def} so we may upper bound the derivative in the interval $\rho\in [-1+\delta, 1-\delta]$ as:
$$
\left| \frac{d}{d\rho} f(\rho)\right| \leq \left| \frac{d}{d\rho} f(\rho)\right|_{\rho=1-\delta}= \frac{(1-\delta)^2 \, _2F_1\left(\frac{3}{2},\frac{3}{2};\frac{k}{2}+2;(1-\delta)^2\right)}{2 \left(\frac{k}{2}+1\right)}+\, _2F_1\left(\frac{1}{2},\frac{1}{2};\frac{k}{2}+1;(1-\delta)^2\right),
$$
establishing the third and fourth property.

For the last property, we compute the second derivative of $f$ using \Cref{eq:d_of_hypergeo}:
\begin{equation*}
\frac{d^2}{d\rho^2}f(\rho) =  \frac{9\rho^3 \, _2F_1\left(\frac{5}{2},\frac{5}{2};\frac{k}{2}+3;\rho^2\right)}{4 \left(\frac{k}{2}+1\right) \left(\frac{k}{2}+2\right)} + \frac{(\rho^2 + 2\rho) \, _2F_1\left(\frac{3}{2},\frac{3}{2};\frac{k}{2}+2;\rho^2\right)}{2 \left(\frac{k}{2}+1\right)} + \, _2F_1\left(\frac{1}{2},\frac{1}{2};\frac{k}{2}+1;\rho^2\right),  
\end{equation*}
which is non-negative for $\rho \in [0,1)$ by the definition \Cref{eq:hypergeometric-sum-def}. Note that $\frac{d^2}{d\rho^2}f(\rho)$ fails to be defined at $\rho=1$ in general since $\,_2F_1[a, b;c;1]$ fails to be absolutely convergent when $a+b>c$ and $5/2+5/2$ may be larger than $k/2+3$ depending on $k$.  However, convexity in $[0, 1)$ and continuity in $[0, 1]$ of the function itself imply convexity in $[0, 1]$.
\end{proof}

\newcommand{\akmc}[1]{\alpha_{#1\mathrm{MC}}}
\newcommand{\rkmc}[1]{\rho_{#1\mathrm{MC}}}

\section{Rank-constrained \maxcut}\label{sec:rank-constrained}

We now show how our results extend to the rank-constrained \maxcut problem.
Recall from \Cref{def:rank-k-maxcut} that rank-$k$ \maxcut is the problem of computing the value
\begin{equation*}
\maxcut_k(G) = \max_{f:V \rightarrow S^{k-1}} \E_{(\bu, \bv) \sim E}[ \tfrac{1}{2} - \tfrac{1}{2} \langle f(\bu), f(\bv)\rangle].
\end{equation*}
This was introduced in the work of Briët, Oliveira, and Vallentin~\cite[Section 6]{BOV10}
as the Laplacian special case of a more general problem known as the \emph{rank-constrained Grothendieck problem}.
Their work lists numerous applications of the general rank-constrained Grothendieck problem,
though to our knowledge there are no applications of the Laplacian special case (i.e.\ the rank-constrained \maxcut problem)
aside from the $k = 3$ case, which corresponds to the product state value of \qmaxcut, as we have seen.

The BOV algorithm that we have already seen for rank-3 \maxcut
(equivalently, for the product state value of \qmaxcut)
is actually the $k=3$ special case of an algorithm for rank-$k$ \maxcut for general~$k$,
which we will also refer to as the ``BOV algorithm'' in this section.
For general $k$, the BOV algorithm first solves the standard \maxcut SDP
to produce a vector solution $f_{\mathrm{SDP}}:V \rightarrow S^{n-1}$,
which it then rounds into a random function $\boldf:V \rightarrow S^{k-1}$ using projection rounding.
To compute the algorithm's approximation ratio, they go edge-by-edge:
for each edge $(u, v) \in E$,, if we set $\rho_{u, v} = \langle f_{\mathrm{SDP}}(u), f_{\mathrm{SDP}}(v)\rangle$,
then the $f_{\mathrm{SDP}}$'s value for that edge is $\tfrac{1}{2} - \tfrac{1}{2} \rho_{u, v}$,
whereas the expectation of $\boldf$'s value for that edge is $\tfrac{1}{2} - \tfrac{1}{2} F^*(k, \rho)$
(see \Cref{thm:exact-formula-for-average-inner-product} for a definition of $F^*(k, \rho)$). 
This motivates studying the following quantity.

\begin{definition}[Approximation ratio for rank-$k$ \maxcut]
Let $k \geq 1$.
The constant $\akmc{k}$ is defined as the solution to the minimization problem
\begin{equation*}
\akmc{k} = \min_{-1 \leq \rho \leq 1} \frac{\tfrac{1}{2} - \tfrac{1}{2} F^*(k,\rho)}{\tfrac{1}{2} - \tfrac{1}{2}\rho},
\end{equation*}
and the constant $\rkmc{k}$ is defined as the minimizing value of~$\rho$.
\end{definition}

For $k = 1$, $\akmc{1} = 0.8785\ldots$ and $\rkmc{1} = -0.689\ldots$,
corresponding to the Goemans-Williamson algorithm.
For $k = 3$, $\akmc{3} = 0.9563\ldots$ and $\rkmc{3} = -0.584\ldots$,
corresponding to the rank-3 BOV algorithm.
\cite{BOV10} also compute the $k = 2$ values numerically and find $\akmc{2} = 0.9349\ldots$ and $\rkmc{2} = -0.617\ldots$.
Having defined these quantities, the expected value of $\boldf$ is
\begin{equation*}
\E_{\boldf} \E_{(\bu, \bv) \sim E}[\tfrac{1}{2} - \tfrac{1}{2}\langle \boldf(\bu), \boldf(\bv)\rangle]
= \E_{(\bu, \bv)\sim E}[\tfrac{1}{2} - \tfrac{1}{2} F^*(k, \rho_{\bu, \bv})]
\geq \akmc{k} \cdot \E_{(\bu, \bv)\sim E}[\tfrac{1}{2} - \tfrac{1}{2} \rho_{\bu, \bv}]
= \mcsdp(G),
\end{equation*}
and so the BOV algorithm has approximation ratio at least $\akmc{k}$.
This gives the following theorem.

\begin{theorem}[Performance of the BOV algorithm for rank-$k$ \maxcut~\cite{BOV10}]
The BOV algorithm for rank-$k$ \maxcut achieves approximation ratio $\akmc{k}$.
\end{theorem}

Our results on the product state value of \qmaxcut
imply that the BOV algorithm is optimal for rank-3 \maxcut.
In fact, our proofs extend in a straightforward manner
to show that the BOV algorithm is optimal for all constant values of~$k$.
One slight technicality is that our proofs require the worst-case~$\rho$ to be negative,
as this is the only regime for which our vector-valued Borell's inequality applies.
The following proposition establishes this for rank-$k$ \maxcut.

\begin{proposition}[Negative $\rho$ is the worst case]
For all $k \geq 1$, $-1 \leq \rkmc{k} \leq 0$.
\end{proposition}
\begin{proof}
The proposition follows from two claims about $F^*(k, \rho)$:
(i) that $F^*(k, \rho)$ always has the same sign as~$\rho$,
and (ii) that $|F^*(k, \rho)| \leq |\rho|$.
Together, these imply that $\tfrac{1}{2} - \tfrac{1}{2} F^*(k, \rho) \geq \tfrac{1}{2} - \tfrac{1}{2} \rho$
whenever $0 \leq \rho \leq 1$,
and so their ratio is always at least~$1$ for~$\rho$ in this range.
But $\akmc{k}$ is an approximation ratio and so is always between $0$ and~$1$,
and thus the minimizing value of~$\rho$ must be in the interval $[-1, 0]$.

Now we prove the claims. We recall from \Cref{thm:exact-formula-for-average-inner-product} that
\begin{equation*}
F^*(k, \rho) =\frac{2}{k} \left( \frac{\Gamma((k+1)/2)}{\Gamma(k/2)}\right)^2 \rho\,\, \,_2 F_1[1/2, 1/2, k/2+1, \rho^2].
\end{equation*}
Property (i) follows because, by \Cref{lem:inner_prod_lipschitz}, $F^*(k, \rho) \geq 0$ for $\rho \in [0,1]$, and $F^*$ is an odd function in $\rho$.  \Cref{lem:inner_prod_lipschitz} also establishes that  $F^*$ is convex for $\rho \in [0,1]$.  One may directly evaluate the expression for $F^*(k, \rho)$ above to see that $F^*(k, 0) = 0$ and $F^*(k, 1) = 1$.  Thus $F^*(k, \rho) \leq \rho$ for $\rho \in [0,1]$; because $F^*$ is odd in $\rho$, $F^*(k, \rho) \geq \rho$ for $\rho \in [-1,0]$, establishing (ii).
\end{proof}

Having established this, it is straightforward to extend our proofs to the case of rank-$k$ \maxcut,
and we omit the details.
(One very minor difference in the $k = 1,2$ case for the algorithmic gap is that the Lipschitz guarantee from \Cref{lem:inner_prod_lipschitz} only holds when bounded away from $-1$ and $1$. However, inspecting the proof of the algorithmic gap shows that this suffices.)
Our results for rank-$k$ \maxcut are stated as follows.

\begin{theorem}[Hardness for rank-$k$ \maxcut]
Let $k \geq 1$ be fixed.
Then the following three statements hold.
\begin{enumerate}
\item The \maxcut semidefinite program $\mcsdp(G)$, when viewed as a relaxation of $\maxcut_k(G)$, has integrality gap $\akmc{k}$.
\item The BOV algorithm for rank-$k$ \maxcut has algorithmic gap $\akmc{k}$.
\item Assuming the UGC, it is $\NP$-hard to approximate $\maxcut_k(G)$ to within a factor of $\akmc{k}+\eps$, for all $\eps > 0$.
\end{enumerate}
\end{theorem}

\bibliographystyle{alpha}
\bibliography{wright}

\newcommand{\etalchar}[1]{$^{#1}$}
\begin{thebibliography}{BdOFV10}

\bibitem[AA88]{A88}
Daniel Arovas and Assa Auerbach.
\newblock Functional integral theories of low-dimensional quantum {H}eisenberg
  models.
\newblock {\em Physical Review B}, 38(1):316, 1988.

\bibitem[AAV13]{AAV13}
Dorit Aharonov, Itai Arad, and Thomas Vidick.
\newblock Guest column: the quantum {PCP} conjecture.
\newblock {\em ACM SIGACT News}, 44(2):47--79, 2013.

\bibitem[ABH{\etalchar{+}}05]{ABH+05}
Sanjeev Arora, Eli Berger, Elad Hazan, Guy Kindler, and Muli Safra.
\newblock On non-approximability for quadratic programs.
\newblock In {\em Proceedings of the 46th Annual IEEE Symposium on Foundations
  of Computer Science}, pages 206--215, 2005.

\bibitem[AGKS21]{AGKS21}
Anurag Anshu, David Gosset, Karen~Morenz Korol, and Mehdi Soleimanifar.
\newblock Improved approximation algorithms for bounded-degree local
  {H}amiltonians.
\newblock {\em arXiv preprint arXiv:2105.01193}, 2021.

\bibitem[AGM20]{AGM20}
Anurag Anshu, David Gosset, and Karen Morenz.
\newblock Beyond product state approximations for a quantum analogue of {M}ax
  {C}ut.
\newblock In {\em Proceedings of the 15th Conference on the Theory of Quantum
  Computation, Communication and Cryptography}, pages 7:1--7:15, 2020.

\bibitem[ALM{\etalchar{+}}98]{ALM+98}
Sanjeev Arora, Carsten Lund, Rajeev Motwani, Madhu Sudan, and Mario Szegedy.
\newblock Proof verification and the hardness of approximation problems.
\newblock {\em Journal of the ACM}, 45(3):501--555, 1998.

\bibitem[AS72]{AS72}
Milton Abramowitz and Irene Stegun, editors.
\newblock {\em Handbook of mathematical functions: with formulas, graphs, and
  mathematical tables}, volume~55 of {\em Applied mathematics series}.
\newblock National {B}ureau of {S}tandards, 1972.

\bibitem[AS98]{AS98}
Sanjeev Arora and Shmuel Safra.
\newblock Probabilistic checking of proofs: a new characterization of {NP}.
\newblock {\em Journal of the ACM}, 45(1):70--122, 1998.

\bibitem[BBT09]{NST09}
Nikhil Bansal, Sergey Bravyi, and Barbara Terhal.
\newblock Classical approximation schemes for the ground-state energy of
  quantum and classical {I}sing spin {H}amiltonians on planar graphs.
\newblock {\em Quantum Information \& Computation}, 9(7):701--720, 2009.

\bibitem[BdOFV10]{BOV10}
Jop Bri{\"e}t, Fernando~M{\'a}rio de~Oliveira~Filho, and Frank Vallentin.
\newblock The positive semidefinite {G}rothendieck problem with rank
  constraint.
\newblock In {\em Proceedings of the 37th International Colloquium on Automata,
  Languages and Programming}, pages 31--42, 2010.

\bibitem[Bet31]{Bet31}
Hans Bethe.
\newblock Zur theorie der metalle.
\newblock {\em Zeitschrift f{\"u}r Physik}, 71(3-4):205--226, 1931.

\bibitem[BGKT19]{BGKT19}
Sergey Bravyi, David Gosset, Robert K{\"o}nig, and Kristan Temme.
\newblock Approximation algorithms for quantum many-body problems.
\newblock {\em Journal of Mathematical Physics}, 60(3):032203, 2019.

\bibitem[BH16]{BH16}
Fernando Brand{\~a}o and Aram Harrow.
\newblock Product-state approximations to quantum ground states.
\newblock {\em Communications in Mathematical Physics}, 342(1):47--80, 2016.

\bibitem[BK10]{BK10}
Nikhil Bansal and Subhash Khot.
\newblock Inapproximability of hypergraph vertex cover and applications to
  scheduling problems.
\newblock In {\em Proceedings of the 37th International Colloquium on Automata,
  Languages and Programming}, pages 250--261, 2010.

\bibitem[Blo30]{B30}
Felix Bloch.
\newblock Zur theorie des ferromagnetismus.
\newblock {\em Zeitschrift f{\"u}r Physik}, 61(3-4):206--219, 1930.

\bibitem[Bor85]{Bor:85}
Christer Borell.
\newblock Geometric bounds on the {O}rnstein-{U}hlenbeck velocity process.
\newblock {\em Probability Theory and Related Fields}, 70(1):1--13, 1985.

\bibitem[BS16]{BS16}
Boaz Barak and David Steurer.
\newblock Lecture notes from {CS229r/6.S898}:\ {P}roofs, beliefs, and
  algorithms through the lens of {S}um-of-{S}quares.
\newblock Found at \url{https://www.sumofsquares.org/public/index.html}, 2016.

\bibitem[CM16]{CM16}
Toby Cubitt and Ashley Montanaro.
\newblock Complexity classification of local {H}amiltonian problems.
\newblock {\em SIAM Journal on Computing}, 45(2):268--316, 2016.

\bibitem[CW04]{CW04}
Moses Charikar and Anthony Wirth.
\newblock Maximizing quadratic programs: extending {G}rothendieck's inequality.
\newblock In {\em Proceedings of the 45th Annual IEEE Symposium on Foundations
  of Computer Science}, pages 54--60, 2004.

\bibitem[DX13]{D13}
Feng Dai and Yuan Xu.
\newblock Spherical harmonics.
\newblock In {\em Approximation theory and harmonic analysis on spheres and
  balls}, pages 1--27. Springer, 2013.

\bibitem[Eld15]{Eld:15}
Ronen Eldan.
\newblock A two-sided estimate for the {G}aussian noise stability deficit.
\newblock {\em Inventiones mathematicae}, 201(2):561--624, 2015.

\bibitem[FBB79]{F79}
J.~N. Fields, Hendrik Bl{\"o}te, and Jill Bonner.
\newblock Renormalization group and other calculations for the one-dimensional
  spin-1/2 dimerized {H}eisenberg antiferromagnet.
\newblock {\em Journal of Applied Physics}, 50(B3):1807--1809, 1979.

\bibitem[FE12]{FE12}
Christopher Frye and Costas Efthimiou.
\newblock {\em Spherical Harmonics in $p$ Dimensions}.
\newblock 2012.

\bibitem[Fra17]{F17}
Fabio Franchini.
\newblock {\em An introduction to integrable techniques for one-dimensional
  quantum systems}.
\newblock Springer, 2017.

\bibitem[FS02]{FS02}
Uriel Feige and Gideon Schechtman.
\newblock On the optimality of the random hyperplane rounding technique for
  {MAX CUT}.
\newblock {\em Random Structures \& Algorithms}, 20(3):403--440, 2002.

\bibitem[GHLS15]{GHLS15}
Sevag Gharibian, Yichen Huang, Zeph Landau, and Seung~Woo Shin.
\newblock Quantum {H}amiltonian complexity.
\newblock {\em Foundations and Trends in Theoretical Computer Science},
  10(3):159--282, 2015.

\bibitem[GK12]{GK12}
Sevag Gharibian and Julia Kempe.
\newblock Approximation algorithms for {QMA}-complete problems.
\newblock {\em SIAM Journal on Computing}, 41(4):1028--1050, 2012.

\bibitem[GP19]{GP19}
Sevag Gharibian and Ojas Parekh.
\newblock Almost optimal classical approximation algorithms for a quantum
  generalization of {M}ax-{C}ut.
\newblock In {\em Proceedings of the 22nd International Workshop on
  Approximation Algorithms for Combinatorial Optimization Problems}, pages
  31:1--31:17, 2019.

\bibitem[GW95]{GW95}
Michel Goemans and David Williamson.
\newblock Improved approximation algorithms for maximum cut and satisfiability
  problems using semidefinite programming.
\newblock {\em Journal of the ACM}, 42(6):1115--1145, 1995.

\bibitem[Hec17]{H1917}
Erich Hecke.
\newblock {\"U}ber orthogonal-invariante integralgleichungen.
\newblock {\em Mathematische Annalen}, 78(1):398--404, 1917.

\bibitem[Hei28]{Hei28}
Werner Heisenberg.
\newblock Zur {T}heorie des {F}erromagnetismus.
\newblock {\em Zeitschrift f{\"u}r Physik}, 49(9-10):619--636, 1928.

\bibitem[HLP20]{HEP20}
Sean Hallgren, Eunou Lee, and Ojas Parekh.
\newblock An approximation algorithm for the max-2-local {H}amiltonian problem.
\newblock In {\em Approximation, Randomization, and Combinatorial Optimization.
  Algorithms and Techniques (APPROX/RANDOM 2020)}. Schloss
  Dagstuhl-Leibniz-Zentrum f{\"u}r Informatik, 2020.

\bibitem[HM17]{HM17}
Aram~W Harrow and Ashley Montanaro.
\newblock Extremal eigenvalues of local {H}amiltonians.
\newblock {\em Quantum}, 1:6, 2017.

\bibitem[HO22]{HOD22}
Matthew~B Hastings and Ryan O'Donnell.
\newblock Optimizing strongly interacting fermionic hamiltonians.
\newblock In {\em Proceedings of the 54th Annual ACM SIGACT Symposium on Theory
  of Computing}, pages 776--789, 2022.

\bibitem[HT20]{HT20}
Steven Heilman and Alex Tarter.
\newblock Three candidate plurality is stablest for small correlations.
\newblock Technical report, arXiv:2011.05583, 2020.

\bibitem[IM12]{IM12}
Marcus Isaksson and Elchanan Mossel.
\newblock Maximally stable {G}aussian partitions with discrete applications.
\newblock {\em Israel Journal of Mathematics}, 189(1):347--396, 2012.

\bibitem[Kar72]{Kar72}
Richard Karp.
\newblock Reducibility among combinatorial problems.
\newblock In {\em Proceedings of the Symposium on the Complexity of Computer
  Computations}, pages 85--103, 1972.

\bibitem[Kar99]{Kar99}
Howard Karloff.
\newblock How good is the {G}oemans--{W}illiamson {MAX} {CUT} algorithm?
\newblock {\em SIAM Journal on Computing}, 29(1):336--350, 1999.

\bibitem[Kho02]{Kho02}
Subhash Khot.
\newblock On the power of unique 2-prover 1-round games.
\newblock In {\em Proceedings of the 34th Annual ACM Symposium on Theory of
  Computing}, pages 767--775, 2002.

\bibitem[KKMO07]{KKMO07}
Subhash Khot, Guy Kindler, Elchanan Mossel, and Ryan O'Donnell.
\newblock Optimal inapproximability results for {M}ax-{C}ut and other
  2-variable {CSP}s?
\newblock {\em SIAM Journal on Computing}, 37(1):319--357, 2007.

\bibitem[KO09]{KO06}
Subhash Khot and Ryan O'Donnell.
\newblock {SDP} gaps and {UGC}-hardness for {M}ax-{C}ut-{G}ain.
\newblock {\em Theory of Computing}, 5:83--117, 2009.

\bibitem[KR08]{KR08}
Subhash Khot and Oded Regev.
\newblock Vertex cover might be hard to approximate to within $2- \varepsilon$.
\newblock {\em Journal of Computer and System Sciences}, 74(3):335--349, 2008.

\bibitem[KSV02]{KSV02}
Alexei Kitaev, Alexander Shen, and Mikhail Vyalyi.
\newblock {\em Classical and quantum computation}.
\newblock American Mathematical Society, 2002.

\bibitem[LM62]{LM62}
Elliott Lieb and Daniel Mattis.
\newblock Ordering energy levels of interacting spin systems.
\newblock {\em Journal of Mathematical Physics}, 3(4):749--751, 1962.

\bibitem[LRS15]{LRS15}
James Lee, Prasad Raghavendra, and David Steurer.
\newblock Lower bounds on the size of semidefinite programming relaxations.
\newblock In {\em Proceedings of the 47th Annual ACM Symposium on Theory of
  Computing}, pages 567--576, 2015.

\bibitem[Men13]{M13}
J~Ricardo~G Mendon{\c{c}}a.
\newblock Exact eigenspectrum of the symmetric simple exclusion process on the
  complete, complete bipartite and related graphs.
\newblock {\em Journal of Physics A: Mathematical and Theoretical},
  46(29):295001, 2013.

\bibitem[MN15]{MN:15}
Elchanan Mossel and Joe Neeman.
\newblock Robust optimality of {G}aussian noise stability.
\newblock {\em Journal of the European Mathematical Society}, 17(2):433--482,
  2015.

\bibitem[MN18]{MN:18}
Emanuel Milman and Joe Neeman.
\newblock The {G}aussian double-bubble conjecture.
\newblock Technical report, arXiv:1801.09296, 2018.

\bibitem[MOO10]{MOO10}
Elchanan Mossel, Ryan O'Donnell, and Krzysztof Oleszkiewicz.
\newblock Noise stability of functions with low influences: invariance and
  optimality.
\newblock {\em Annals of Mathematics}, 171(1):295--341, 2010.

\bibitem[MR15]{MR:15}
Matthew McGonagle and John Ross.
\newblock The hyperplane is the only stable, smooth solution to the
  isoperimetric problem in {G}aussian space.
\newblock {\em Geometriae Dedicata}, 178(1):277--296, 2015.

\bibitem[O'D08]{OD08}
Ryan O'Donnell.
\newblock Lecture {16} from {18-854B}:\ advanced approximation algorithms.
\newblock Found at
  \url{https://www.cs.cmu.edu/~anupamg/adv-approx/lecture16.pdf}, 2008.

\bibitem[O'D14]{OD14}
Ryan O'Donnell.
\newblock {\em Analysis of {B}oolean functions}.
\newblock Cambridge University Press, 2014.

\bibitem[OW08]{OW08}
Ryan O'Donnell and Yi~Wu.
\newblock An optimal {SDP} algorithm for {M}ax-{C}ut, and equally optimal
  {L}ong {C}ode tests.
\newblock In {\em Proceedings of the 40th Annual ACM Symposium on Theory of
  Computing}, pages 335--344, 2008.

\bibitem[PM17]{PM17}
Stephen Piddock and Ashley Montanaro.
\newblock The complexity of antiferromagnetic interactions and {2D} lattices.
\newblock {\em Quantum Information \& Computation}, 17(7-8):636--672, 2017.

\bibitem[PT21a]{PT21}
Ojas Parekh and Kevin Thompson.
\newblock {Application of the Level-2 Quantum Lasserre Hierarchy in Quantum
  Approximation Algorithms}.
\newblock In {\em 48th International Colloquium on Automata, Languages, and
  Programming (ICALP 2021)}, volume 198 of {\em Leibniz International
  Proceedings in Informatics (LIPIcs)}, pages 102:1--102:20, 2021.

\bibitem[PT21b]{PT20}
Ojas Parekh and Kevin Thompson.
\newblock {Beating Random Assignment for Approximating Quantum 2-Local
  Hamiltonian Problems}.
\newblock In {\em 29th Annual European Symposium on Algorithms (ESA 2021)},
  volume 204 of {\em Leibniz International Proceedings in Informatics
  (LIPIcs)}, pages 74:1--74:18, 2021.

\bibitem[PT22]{PT22}
Ojas Parekh and Kevin Thompson.
\newblock An optimal product-state approximation for 2-local quantum
  hamiltonians with positive terms.
\newblock Technical report, arXiv: 2206.08342, 2022.

\bibitem[Rag08]{Rag08}
Prasad Raghavendra.
\newblock Optimal algorithms and inapproximability results for every {CSP}?
\newblock In {\em Proceedings of the 40th Annual ACM Symposium on Theory of
  Computing}, 2008.

\bibitem[Rag09]{Rag09}
Prasad Raghavendra.
\newblock {\em Approximating {NP}-hard problems: efficient algorithms and their
  limits}.
\newblock PhD thesis, University of Washington, 2009.

\bibitem[Sze39]{S39}
Gabor Szeg\H{o}.
\newblock {\em Orthogonal polynomials}, volume~23.
\newblock American Mathematical Soc., 1939.

\bibitem[TSSW00]{TSSW00}
Luca Trevisan, Gregory Sorkin, Madhu Sudan, and David Williamson.
\newblock Gadgets, approximation, and linear programming.
\newblock {\em SIAM Journal on Computing}, 29(6):2074--2097, 2000.

\bibitem[VB96]{VB96}
Lieven Vandenberghe and Stephen Boyd.
\newblock Semidefinite programming.
\newblock {\em SIAM review}, 38(1):49--95, 1996.

\end{thebibliography}

\end{document}